%% file: main.tex
\documentclass[11pt,letterpaper,twoside]{report}

\usepackage[main=english,french]{babel}

\usepackage[hyphens]{url}

\usepackage{geometry}
\usepackage{setspace}
\usepackage{titlesec}
\usepackage{tocloft}
\usepackage{appendix}
\usepackage{pdflscape}
\usepackage{afterpage}
\usepackage{indentfirst}

\usepackage[bottom,marginal]{footmisc}
\usepackage{natbib}
\usepackage{apalike}

\usepackage{bibentry}
\nobibliography*

\usepackage[labelsep=colon]{caption}
\usepackage{subcaption} 
\usepackage{epsfig}
\usepackage{booktabs}
\usepackage{multicol}
\usepackage{listings}
\usepackage{color}
\usepackage{tabularx}
\usepackage{hhline}
\usepackage{multirow}
\usepackage[table]{xcolor}
\usepackage{pbox}
\usepackage{longtable}

\usepackage{enumitem}

\usepackage{amsthm}
\usepackage{amsmath}
\usepackage{amssymb}

\usepackage[binary-units=true,detect-all]{siunitx}
\input{common/units}

\usepackage{times}
\usepackage{microtype}
\usepackage{textcomp}
\usepackage{hyphenat}
\usepackage[euler]{textgreek}	
\usepackage{amsmath}

\usepackage{algorithm}
\usepackage[noend]{algpseudocode}

\usepackage{acro}	
\acsetup{
	log=true,
	page-ref = paren,
	extra-style = comma,
	only-used = true,		
	hyperref = true,
	page-style = none,		
	list-style = longtable,
	macros = false,
	cite-cmd = {\citep},
	use-barriers = true,	
	reset-at-barriers = true,
	single = true				
}

\input{common/acro-setup}

\usepackage{xspace}

\usepackage{hyperxmp}

\usepackage[hidelinks]{hyperref} 
\hypersetup{
	pdfauthor = {First M. Last},
	pdftitle = {Title Text Here},
	pdfsubject = {Dissertation},
	pdfkeywords = {keyword1, keyword2, keyword3},
	colorlinks=true,
	linkcolor= black,
	citecolor= black,
	pageanchor= true,
	urlcolor = black,
	plainpages = false,
	linktocpage
}

\usepackage[capitalise,noabbrev]{cleveref}	

\input{common/cleverref-setup}

\usepackage{ifthen}
\usepackage{lipsum}

\graphicspath{{./figs/}}

\input{common/layout}
\input{common/acronyms}
\input{common/macros}


\DeclareMathOperator*{\argmax}{arg\,max}
\DeclareMathOperator*{\argmin}{arg\,min}
\newcommand{\vectw}[2]{(#1,\: #2)}
\newcommand{\vecth}[3]{(#1,\: #2,\: #3)}
\newcommand{\learn}{\left(\pi\right)}
\newcommand{\expert}{\left(\pi^*\right)}
\newtheorem{corollary}{Corollary}

\newcommand{\flows}{flows\xspace}

\newcommand{\tratios}{turning ratios\xspace}

\begin{document}

\input{frontmatter/pages}


\input{chapters/Introduction}

\input{chapters/ITSM}

\input{chapters/SIGA}

\input{chapters/ADAPS}
\input{chapters/Conclusions}

\begin{appendices}
\end{appendices}

\clearpage
\phantomsection

{\def\chapter*#1{} 
	\begin{singlespace}
		\addcontentsline{toc}{chapter}{BIBLIOGRAPHY}
		\begin{center}
			\textbf{BIBLIOGRAPHY}
			\vspace{11pt}
		\end{center}
		
		\bibliographystyle{apalike}
		\bibliography{itsm-ref,siga-ref,adaps-ref,pz-ref}
	\end{singlespace}
}

\end{document}

%% file: common/units.tex
\DeclareSIUnit\color{color}
\DeclareSIUnit\inch{in}
\DeclareSIUnit\pixel{pixel}
\DeclareSIUnit\tick{tick}


%% file: common/acro-setup.tex
 \ExplSyntaxOn 

\ProvideAcroEnding{possessive-singular}{'s}{'s}
\ProvideAcroEnding{possessive-plural}{s'}{s'}

\ProvideAcroCommand{\acg}
{
	\acro_possessive_singular:
	\acro_use:n {#1}
}

\ProvideAcroCommand{\Acg}
{
	\acro_first_upper:
	\acro_possessive_singular:
	\acro_use:n {#1}
}

\ProvideAcroCommand{\acgp}
{
	\acro_possessive_plural:
	\acro_use:n {#1}
}

\ProvideAcroCommand{\Acgp}
{
	\acro_first_upper:
	\acro_possessive_plural:
	\acro_use:n {#1}
}

\ProvideAcroCommand{\ace}
{
	\acro_extra:n {#1}
}

\ExplSyntaxOff


%% file: common/layout.tex
\titleformat{\chapter}[block]{\fillast\bfseries}{\MakeUppercase{\chaptertitlename} \thechapter:}{1ex}{\centering}[]

\titlespacing{\chapter}{0in}{0.62in}{22pt}

\titleformat*{\section}{\bfseries}
\titleformat*{\subsection}{\bfseries}
\titleformat*{\subsection}{\bfseries}

\setlist{leftmargin=4ex}

\titlespacing{\paragraph}{0in}{0.08in}{0.07in}

%% file: common/acronyms.tex
\DeclareAcronym{AR}{
	short = {AR},
	long = {Augmented Reality},
	short-indefinite = {an},
	long-indefinite = {an}
}

\DeclareAcronym{VR}{
	short = {VR},
	long = Virtual Reality
}

\DeclareAcronym{AV}{
	short = {AV},
	long = {Augmented Virtuality},
	short-indefinite = {an},
	long-indefinite = {an}
}

\DeclareAcronym{FPGA}{
	short = {FPGA},
	long = {Field-Programmable Gate Array},
	short-indefinite = {an},
	long-indefinite = {a}
}

\DeclareAcronym{LED}{
	short = {LED},
	long = {Light-Emitting Diode},
	short-indefinite = {an},
	long-indefinite = {a}
}

\DeclareAcronym{I2C}{
	short = {I\textsuperscript{2}C},
	long = {Inter-Integrated Circuit},
	alt = {IIC},
	short-indefinite = {an},
	long-indefinite = {an},
	alt-indefinite = {an}
}

\DeclareAcronym{DMD}{
	short = {DMD},
	long = {Digital Micromirror Display},
	cite = {{Hornbeck97digitallight}}
}

\DeclareAcronym{HDR}{
	short = {HDR},
	long = {High-Dynamic Range},
	short-indefinite = {an},
	long-indefinite = {a}
}

\DeclareAcronym{SLM}{
	short = {SLM},
	long = {Spatial Light Modulator},
	short-indefinite = {an},
	long-indefinite = {a}
}

\DeclareAcronym{LCD}{
	short = {LCD},
	long = {Liquid Crystal Display},
	first-style = {reversed},
	short-indefinite = {an},
	long-indefinite = {a}
}

\DeclareAcronym{PTM}{
	short = {PTM},
	long = {Pulse Train Modulation}
}

\DeclareAcronym{PWM}{
	short = {PWM},
	long = {Pulse Width Modulation}
}

\DeclareAcronym{PDM}{
	short = {PDM},
	long = {Pulse Density Modulation}
}

\DeclareAcronym{LSB}{
	short = {LSB},
	long = {Least Significant Bit},
	short-indefinite = {an},
	long-indefinite = {a}
}

\DeclareAcronym{MSB}{
	short = {MSB},
	long = {Most Significant Bit},
	short-indefinite = {an},
	long-indefinite = {a}
}

\DeclareAcronym{RGB}{
	short = {RGB},
	long = {Red, Green, and Blue},
	first-style = {reversed},
	short-indefinite = {an},
	long-indefinite = {a}
}

\DeclareAcronym{HMD}{
	short = {HMD},
	long = {Head-Mounted Display},
	short-indefinite = {an},
	long-indefinite = {a}
}

\DeclareAcronym{OST}{
	short = {OST},
	long = {Optical See-through},
	short-indefinite = {an},
	long-indefinite = {an}
}

\DeclareAcronym{RAM}{
	short = {RAM},
	long = {Random Access Memory}
}

\DeclareAcronym{DDR}{
	short = {DDR},
	long = {Double Data Rate},
	first-style = {reversed}
}

\DeclareAcronym{SRAM}{
	short = {SRAM},
	long = {Static Random Access Memory},
	first-style = {reversed},
	short-indefinite = {an},
	long-indefinite = {a}
}

\DeclareAcronym{DRAM}{
	short = {DRAM},
	long = {Dynamic Random Access Memory},
	first-style = {reversed}
}

\DeclareAcronym{PC}{
	short = {PC},
	long = {Personal Computer},
	first-style = {short}
}

\DeclareAcronym{FOV}{
	short = {FOV},
	long = {Field of View},
	short-indefinite = {an},
	long-indefinite = {a}
}

\DeclareAcronym{GPU}{
	short = {GPU},
	long = {Graphics Processing Unit}
}

\DeclareAcronym{CPU}{
	short = {CPU},
	long = {Central Processing Unit}
}

\DeclareAcronym{OLED}{
	short = {OLED},
	long = {Organic Light-Emitting Diode},
	short-indefinite = {an},
	long-indefinite = {an}
}

\DeclareAcronym{DAC}{
	short = {DAC},
	long = {Digital-to-Analog Converter}
}

\DeclareAcronym{ADC}{
	short = {ADC},
	long = {Analog-to-Digital Converter},
	short-indefinite = {an},
	long-indefinite = {an}
}

\DeclareAcronym{XGA}{
	short = {XGA},
	long = {Extended Graphics Array},
	extra = {\resolution{1024}{768}},
	first-style = {reversed},
	short-indefinite = {an},
	long-indefinite = {an}
}

\DeclareAcronym{AMOLED}{
	short = {AMOLED},
	long = {Active Matrix Organic Light-Emitted Diode},
	short-indefinite = {an},
	long-indefinite = {an}
}

\DeclareAcronym{DOF}{
	short = {DOF},
	long = {Degree of Freedom},
	short-plural-form = {DOF},
	long-plural-form = {Degrees of Freedom}
}

\DeclareAcronym{SAR}{
	short = {SAR},
	long = {Spatial Augmented Reality}
}

\DeclareAcronym{DC}{
	short = {DC},
	long = {Direct Current},
	first-style = {reversed}
}

\DeclareAcronym{AC}{
	short = {AC},
	long = {Alternating Current},
	first-style = {reversed},
	short-indefinite = {an},
	long-indefinite = {an}
}

\DeclareAcronym{RC}{
	short = {RC},
	long = {Resistor-Capacitor},
	first-style = {reversed},
	short-indefinite = {an},
	long-indefinite = {a}
}

\DeclareAcronym{2D}{
	short = {2D},
	long = {Two-Dimensional},
	first-style = {short}
}

\DeclareAcronym{3D}{
	short = {3D},
	long = {Three-Dimensional},
	first-style = {short}
}

\DeclareAcronym{1080p}{
	short = {1080p},
	long = {Ten-Eighty Progressive Scan},
	extra = {\resolution{1920}{1080}},
	first-style = {short}
}

\DeclareAcronym{VGA}{
	short = {VGA},
	long = {Video Graphics Array},
	first-style = {reversed}
}

\DeclareAcronym{DVI}{
	short = {DVI},
	long = {Digital Video Interface},
	first-style = {reversed}
}

\DeclareAcronym{HDMI}{
	short = {HDMI},
	long = {High-Definition Multimedia Interface},
	first-style = {reversed},
	short-indefinite = {an},
	long-indefinite = {a}
}

\DeclareAcronym{DP}{
	short = {DP},
	long = {DisplayPort\texttrademark},
	first-style = {reversed}
}

\DeclareAcronym{IC}{
	short = {IC},
	long = {Integrated Circuit},
	short-indefinite = {an},
	long-indefinite = {an}
}

\DeclareAcronym{LFSR}{
	short = {LFSR},
	long = {Linear-Feedback Shift Register},
	cite = {{koeter1996s}},
	short-indefinite = {an},
	long-indefinite = {a}
}

\DeclareAcronym{ASIC}{
	short = {ASIC},
	long = {Application Specific Integrated Circuit},
	short-indefinite = {an},
	long-indefinite = {an}
}

\DeclareAcronym{VST}{
	short = {VST},
	long = {Video See-through}
}

\DeclareAcronym{IBR}{
	short = {IBR},
	long = {Image-based Rendering},
	short-indefinite = {an},
	long-indefinite = {an}
}

\DeclareAcronym{PCM}{
	short = {PCM},
	long = {Pulse Code Modulation}
}

\DeclareAcronym{PFM}{
	short = {PFM},
	long = {Pulse Frequency Modulation}
}

\DeclareAcronym{DLP}{
	short = {DLP\textregistered{}},
	long = {Digital Light Projection}
}

\DeclareAcronym{MCP}{
	short = {MCP},
	long = {Mirror Clocking Pulse},
	short-indefinite = {an},
	long-indefinite = {a}
}

\DeclareAcronym{TI}{
	short = {TI},
	long = {Texas Instruments},
	single = {Texas Instruments (TI)}
}

\DeclareAcronym{SOC}{
	short = {SoC},
	long = {System on a Chip}
}

\DeclareAcronym{HUD}{
	short = {HUD},
	long = {Heads-up Display}
}

\DeclareAcronym{MR}{
	short = {MR},
	long = {Mixed Reality},
	short-indefinite = {an},
	long-indefinite = {a}
}

\DeclareAcronym{OS}{
	short = {OS},
	long = {Operating System},
	short-indefinite = {an},
	long-indefinite = {an}
}

\DeclareAcronym{PRW}{
	short = {PRW},
	long = {Post-Rendering Warp}
}

\DeclareAcronym{CRT}{
	short = {CRT},
	long = {Cathode Ray Tube}
}

\DeclareAcronym{EEM}{
	short = {EEM},
	long = {Estimated Error Modulation},
	short-indefinite = {an},
	long-indefinite = {an}
}

\DeclareAcronym{VSYNC}{
	short = {VSYNC},
	long = {Vertical Sync},
	first-style = {reversed}
}

\DeclareAcronym{ALS}{
	short = {ALS},
	long = {Ambient Light Sensing},
	short-indefinite = {an},
	long-indefinite = {an}
}

\DeclareAcronym{CIE}{
	short = {CIE},
	long = {International Commission on Illumination},
	foreign = {Commission internationale de l'\'{e}clairage},
	foreign-lang = {french},
	first-style = {short}
}

\DeclareAcronym{HLG}{
	short = {HLG},
	long = {Hybrid Log-Gamma},
	short-indefinite = {an},
	long-indefinite = {a}
}

\DeclareAcronym{DeltaSigma}{
	short = {$\Delta\Sigma$},
	long = {Delta-Sigma Modulation},
	alt = {Delta-Sigma},
	list = {Delta-Sigma},
	sort = {9999DS}
}

\DeclareAcronym{SI}{
	short = {SI},
	long = {International System of Units},
	foreign = {Syst\`{e}me international d'unit\'{e}s},
	foreign-lang = {french},
	first-style = {short},
	short-indefinite = {an},
	long-indefinite = {an}
}



%% file: common/macros.tex
\renewcommand{\cite}{\citep}

\newboolean{isdraft}
\setboolean{isdraft}{true}

\newboolean{isapproved}
\setboolean{isapproved}{false}

\ifthenelse{\boolean{isdraft}}{\newcommand{\todo}[1]{\noindent \textcolor{red}{\textbf{[*** TODO: #1]}}}}{\newcommand{\todo}[1]{}}
\ifthenelse{\boolean{isdraft}}{\newcommand{\note}[1]{\noindent \textcolor{red}{\textbf{[*** #1]}}}}{\newcommand{\note}[1]{}}
\ifthenelse{\boolean{isdraft}}{\newcommand{\revcmt}[2]{\noindent \textcolor{orange}{\textbf{[*** #1: #2]}}}}{\newcommand{\revcmt}[2]{}}
\ifthenelse{\boolean{isdraft}}{\newcommand{\todocite}[1]{\textcolor{red}{\textbf{[*** CITE: #1]}}}}{\newcommand{\todocite}[1]{}}
\ifthenelse{\boolean{isdraft}}{\newcommand{\todolater}[1]{\noindent \textcolor{cyan}{\textbf{[*** MUCH LATER: #1]}}}}{\newcommand{\todolater}[1]{}}
\ifthenelse{\boolean{isdraft}}{}{}
\ifthenelse{\boolean{isdraft}}{}{}
\ifthenelse{\boolean{isdraft}}{\newcommand{\placeholdertext}[1][]{\textcolor{cyan}{\lipsum[#1]}}}{\newcommand{\placeholdertext}[1][]{}}

\let\oldcite=\cite
\renewcommand\cite[1]{\ifthenelse{\equal{#1}{NEEDED}}{\ifthenelse{\boolean{isdraft}}{(\textcolor{red}{citation needed})}{}}{\oldcite{#1}}}

\ifthenelse{\boolean{isdraft}}{\newcommand{\citeneeded}[1][]{(\textcolor{red}{citation needed\ifthenelse{\equal{#1}{}}{}{: #1}})}}{\newcommand{\citeneeded}[1][]{}}

\newcommand{\Virtex}[1]{Virtex\textregistered\ifthenelse{\equal{}{#1}}{}{\nobreakdash-#1}\xspace}

\newcommand{\resolution}[2]{\num{#1}~$\times$~\num{#2}}

\long\def\symbolfootnote[#1]#2{\begingroup%
\def\thefootnote{\fnsymbol{footnote}}\footnote[#1]{#2}\endgroup}

\newtheoremstyle{mylemthm}
        {6pt}
        {3pt}
        {\slshape}
        {}
        {\bfseries}
        {.}
        {.5em}
        {}

\theoremstyle{mylemthm}

\newtheorem{theorem}{Theorem}[chapter]
\newtheorem{lemma}{Lemma}[chapter]

\newtheoremstyle{mydef}
        {3pt}
        {3pt}
        {\normalfont}
        {}
        {\bfseries}
        {.}
        {.5em}
        {\thmname{#1} \thmnumber{#2}\thmnote{#3}}

\theoremstyle{mydef}

\newcommand\rightparend[1]{{%
      \unskip\nobreak\hfil\penalty50
      \hskip2em\hbox{}\nobreak\hfil\textbf{#1}%
      \parfillskip=0pt \finalhyphendemerits=0 \par}}

\newtheorem{xxexample}{Example}[chapter]



%% file: frontmatter/pages.tex
\newgeometry{left=1in,top=2in,right=1in,bottom=1in,nohead}
\pagenumbering{roman}


\input{frontmatter/titlepage}

\newgeometry{left=1in,top=8.33in,right=1in,bottom=1in,nohead}

\input{frontmatter/copyright}
\newgeometry{left=1in,top=2in,right=1in,bottom=1in,nohead}


\restoregeometry

\input{frontmatter/abstract}



\input{frontmatter/dedication}
\input{frontmatter/acknowledgements}

\acbarrier

\input{frontmatter/contents}


\input{frontmatter/tables}


\input{frontmatter/figures}



\pagenumbering{arabic}

%% file: frontmatter/titlepage.tex
\begin{titlepage}

\newlength{\titlepartseparation}
\setlength{\titlepartseparation}{61pt}		

\newlength{\titlepartseparationdbl}
\setlength{\titlepartseparationdbl}{39pt}	

\begin{center}



\vspace{2in}
\begin{singlespace}
\noindent
SIMULATION AND LEARNING FOR URBAN MOBILITY: CITY-SCALE TRAFFIC RECONSTRUCTION AND AUTONOMOUS DRIVING
\end{singlespace}

\vspace{\titlepartseparation}
\noindent
Weizi Li
\end{center}


\vspace{\titlepartseparationdbl}
\begin{center}
\begin{singlespace}
\noindent
A dissertation submitted to the faculty of the University of North Carolina at Chapel Hill
in partial fulfillment of the requirements for the degree of Doctor of Philosophy in
the Department of Computer Science.
\end{singlespace}
\end{center}

\vspace{\titlepartseparationdbl}
\begin{center}
\begin{singlespace}
\noindent
Chapel Hill\\
2019
\end{singlespace}
\end{center}


\vspace{\titlepartseparationdbl}

\newlength{\approvalwidth}
\setlength{\approvalwidth}{2.166667in}

\begin{flushright}
\begin{minipage}[t]{\approvalwidth}
\ifthenelse{\boolean{isapproved}}{Approved by:}{Approved by:} \\
Ming C. Lin \\
Dinesh Manocha \\
Ron Alterovitz \\
Julien Pettr\'{e} \\
David Wilkie
\end{minipage}
\end{flushright}

\vfill

\end{titlepage}

%% file: frontmatter/copyright.tex

\begin{center}
\begin{singlespace}
\copyright{} 2019\\
Weizi Li \\
ALL RIGHTS RESERVED
\end{singlespace}
\end{center}

\clearpage

%% file: frontmatter/abstract.tex
\begin{center}
\vspace*{52pt}
{\textbf{ABSTRACT}\pdfbookmark{ABSTRACT}{ch:Abstract}}
\vspace{11pt}

\begin{singlespace}
Weizi Li: Simulation and Learning for Urban Mobility: City-scale Traffic Reconstruction and Autonomous Driving \\
(Under the direction of Ming C. Lin)
\end{singlespace}
\end{center}

Traffic congestion has become one of the most critical issues worldwide. The costs due to traffic gridlock and jams are approximately \$160 billion in the United States, more than \pounds13 billion in the United Kingdom, and over one trillion dollars across the globe annually. As more metropolitan areas will experience increasingly severe traffic conditions, the ability to analyze, understand, and improve traffic dynamics becomes critical. This dissertation is an effort towards achieving such an ability. I propose various techniques combining simulation and machine learning to tackle the problem of traffic from two perspectives: city-scale traffic reconstruction and autonomous driving.  

Traffic, by its definition, appears in an aggregate form. In order to study it, we have to take a holistic approach. I address the problem of efficient and accurate estimation and reconstruction of city-scale traffic. The reconstructed traffic can be used to analyze congestion causes, identify network bottlenecks, and experiment with novel transport policies. City-scale traffic estimation and reconstruction have proven to be challenging for two particular reasons: first, traffic conditions that depend on individual drivers are intrinsically stochastic; second, the availability and quality of traffic data are limited. Traditional traffic monitoring systems that exist on highways and major roads can not produce sufficient data to recover traffic at scale. GPS data, in contrast, provide much broader coverage of a city thus are more promising sources for traffic estimation and reconstruction. However, GPS data are limited by their spatial-temporal sparsity in practice. I develop a framework to statically estimate and dynamically reconstruct traffic over a city-scale road network by addressing the limitations of GPS data.  

Traffic is also formed of individual vehicles propagating through space and time. If we can improve the efficiency of them, collectively, we can improve traffic dynamics as a whole. Recent advancements in automation and its implication for improving the safety and efficiency of the traffic system have prompted widespread research of autonomous driving. While exciting, autonomous driving is a complex task, consider the dynamics of an environment and the lack of accurate descriptions of a desired driving behavior. Learning a robust control policy for driving remains challenging as it requires an effective policy architecture, an efficient learning mechanism, and substantial training data covering a variety of scenarios, including rare cases such as accidents. I develop a framework, named ADAPS (Autonomous Driving via Principled Simulations), for producing robust control policies for autonomous driving. ADAPS consists of two simulation platforms which are used to generate and analyze simulated accidents while automatically generating labeled training data, and a hierarchical control policy which takes into account the features of driving behaviors and road conditions. ADAPS also represents a more efficient online learning mechanism compared to previous techniques, in which the number of iterations required to learn a robust control policy is reduced.

\clearpage

%% file: frontmatter/dedication.tex

\begin{center}
\vspace*{52pt}
For my mother, Jie Li. 
\end{center}

\pagebreak

%% file: frontmatter/acknowledgements.tex

\begin{center}
\vspace*{52pt}
{\textbf{ACKNOWLEDGMENTS}\pdfbookmark{ACKNOWLEDGMENTS}{ch:Acknowledgments}}
\end{center}

First of all, I would like to thank my advisor, Ming C. Lin, who has and continues to support me during my time at UNC and has taught me numerous lessons regarding academic life and beyond. I also want to express my sincerest appreciation to the rest of my committee members: Dinesh Manocha, Julien Pettr\'{e}, Ron Alterovitz, and David Wilkie.  

My research will not advance without wonderful collaborators I have had in years. Among them, I want to express my deep gratitude to David Wolinski, who has been serving as my mentor and behind all my important projects, and Jan Allbeck, who had led me to the wonderful research world. I am also thankful to all my supervisors at various institutes: Ellen Yi-Luen Do, Xiang Cao, and Carol O'Sullivan for their guidance and wisdom.   

It has been a great pleasure to work in the GAMMA research lab and interact with its amazing members. I would like to especially thank Shan Yang, Auston Sterling, Tetsuya Takahashi, Andrew Phillip Best, Aniket Bera, Sujeong Kim, Chonhyon Park, Junbang Liang, and Zhenyu Tang for their help and companionship. 

My life at Chapel Hill has been an enjoyable adventure because of the friends I made across the campus. This list includes but not limited to: Yingchi Liu, Xin Zhao, Fan Jiang, Xuxiang Mao, Jun Jiang, Dong Nie, Licheng Yu, Qishun Tang, Feng Shi, Yan Song, Meilei Jiang, Yaoyu Chen, Jiangyue Sun, Chuyi Du, Yunyan He, Zhaopeng Xing, Zhechang Yuan, Yang Zhan, and Tao Zhang. Thank you for sharing this adventure with me at UNC!

Lastly, I am gratefully indebted to my family members for their sacrifice and unconditional love. This dissertation would not have been possible without you.

\clearpage

%% file: frontmatter/contents.tex
\renewcommand{\contentsname}{TABLE OF CONTENTS}
\renewcommand{\cfttoctitlefont}{\hfill\bfseries}
\renewcommand{\cftaftertoctitle}{\hfill\hfill}
\renewcommand{\cftdotsep}{1.5}
\cftsetrmarg{1.0in}

\setlength{\cftbeforetoctitleskip}{61pt}
\setlength{\cftaftertoctitleskip}{28pt}

\renewcommand{\cftchapfont}{\normalfont}
\renewcommand{\cftchappagefont}{\normalfont}
\renewcommand{\cftchapleader}{\cftdotfill{\cftdotsep}}

\setlength{\cftbeforechapskip}{15pt}
\setlength{\cftbeforesecskip}{10pt}
\setlength{\cftbeforesubsecskip}{10pt}
\setlength{\cftbeforesubsubsecskip}{10pt}

\pdfbookmark{TABLE OF CONTENTS}{ch:TOC}

\begin{singlespace}
\tableofcontents
\end{singlespace}

\clearpage

%% file: frontmatter/tables.tex
\renewcommand{\listtablename}{LIST OF TABLES}
\phantomsection
\addcontentsline{toc}{chapter}{LIST OF TABLES}

\setlength{\cftbeforelottitleskip}{-11pt}
\setlength{\cftafterlottitleskip}{22pt}
\renewcommand{\cftlottitlefont}{\hfill\bfseries}
\renewcommand{\cftafterlottitle}{\hfill}

\setlength{\cftbeforetabskip}{10pt}

\begin{singlespace}
\listoftables
\end{singlespace}

\clearpage

%% file: frontmatter/figures.tex
\renewcommand{\listfigurename}{LIST OF FIGURES}
\phantomsection
\addcontentsline{toc}{chapter}{LIST OF FIGURES}

\setlength{\cftbeforeloftitleskip}{-11pt}
\setlength{\cftafterloftitleskip}{22pt}
\renewcommand{\cftloftitlefont}{\hfill\bfseries}
\renewcommand{\cftafterloftitle}{\hfill}

\setlength{\cftbeforefigskip}{10pt}
\cftsetrmarg{1.0in}

\begin{singlespace}
\listoffigures
\end{singlespace}
\clearpage

%% file: chapters/Introduction.tex
\chapter{INTRODUCTION}
\label{ch:Introduction}

Throughout history, human mobility not only characterizes our way of life but also plays an essential role in socio-economic development. In the past century, as a result of rapid motorization and urbanization, increasing human-mobility patterns are stemmed from automobiles and are appearing in urban environments. While this phenomenon reflects better accessibility to societal resources and an increase in quality of life, the resulting traffic congestion has become one of the most infamous problems across the globe. The annual costs due to traffic gridlock and jams are nearly \$160 billion in the U.S., \pounds13 billion in the U.K., and exceeding one trillion U.S. dollars worldwide. 

As urbanization, motorization, and vehicle production rates keep climbing---especially in Asia and Africa---we will witness many more metropolitan areas forming, growing, and experiencing severe traffic conditions. The ability to understand and improve traffic dynamics is thus critical more today than ever. 

Traffic can be studied at various scales. Traffic is an aggregate phenomenon, which can appear at a large spatial-temporal scale. Study of its dynamics requires a holistic view. Among various traffic engineering tasks, one crucial task is the efficient and accurate estimation and reconstruction of city-scale traffic. This task can enable many Intelligent Transportation System (ITS) applications such as identifying congestion causes, detecting network bottlenecks, planning traffic flows, and analyzing transport policies. 

Additionally, traffic is formed by individual vehicles propagating through space and time. If we can improve the efficiency and coordination of them, jointly, we have the opportunity to alleviate severe traffic conditions in metropolitan areas. While the behaviors of human-driven vehicles can be only altered through regulations and policies, the behaviors of autonomous vehicles can be precisely directed and optimized via control algorithms. Thus, the switch from human-driven vehicles to autonomous vehicles has the potential to revolutionize our transportation systems by reducing the number of crashes and alleviating traffic jams which both, by a large degree, attribute to human factors.   

This dissertation summarizes my early efforts towards understanding and improving traffic conditions via studying traffic at both the macroscopic level and the microscopic level. Macroscopically, I study \emph{city-scale traffic reconstruction}: I have combined physics-based simulation for modeling aggregate traffic behaviors and machine learning techniques for distilling travel patterns from a large volume of traffic data in order to estimate and reconstruct traffic at a metropolitan scale. Microscopically, I study \emph{autonomous driving}: I have developed simulation platforms and an online learning mechanism for effectively and efficiently learning and testing control policies for autonomous driving.

\section{City-Scale Traffic Reconstruction}
\label{sec:intro-traffic}
In order to analyze congestion causes, identify network bottlenecks, and experiment with novel transport policies, we need to be able to reconstruct and simulate city-scale traffic at the macroscopic level. I have developed methods to efficiently and accurately estimate and reconstruct large-scale traffic using mobile-sensor data while generating visual analytics in various forms.

Estimation and reconstruction of large-scale traffic is difficult due to fundamental challenges: 1) traffic dynamics is intrinsically stochastic as a result of individual drivers' behaviors, and 2) the availability and quality of traffic data are usually limited. Conventionally, traffic data are collected via in-road sensors such as loop detectors and video cameras. While these sensors produce accurate measurements, they are mostly installed on highways and major roads, which only constitute a small portion of a city. Mobile-sensor data, such as GPS reports, are more promising sources for the estimation and reconstruction task due to their broader coverage. However, GPS points usually embed a low-sampling rate, meaning that the time lapse between two consecutive reports is large (e.g., greater than 60 seconds), and exhibit spatial-temporal sparsity, meaning that the data are scarce in certain areas and time periods.

In order to adopt GPS data for reconstructing full traffic dynamics, several procedures are required to address the abovementioned features: 1) \emph{map-matching}, which maps off-the-road GPS points (due to inevitable measurement noise) onto a road network and infers the traversed path of a vehicle (illustrated in Figure~\ref{fig:issue-map-matching} LEFT and MIDDLE); 2) \emph{travel-time estimation}, which estimates the travel time of a road network through GPS timestamps (illustrated in Figure~\ref{fig:issue-map-matching} RIGHT); 3) \emph{missing-data completion}, which interpolates spatial-temporal missing measurements (illustrated in Figure~\ref{fig:issue-sparsity}). 

\begin{figure}[htb]
	\centering
	\includegraphics[width=\textwidth]{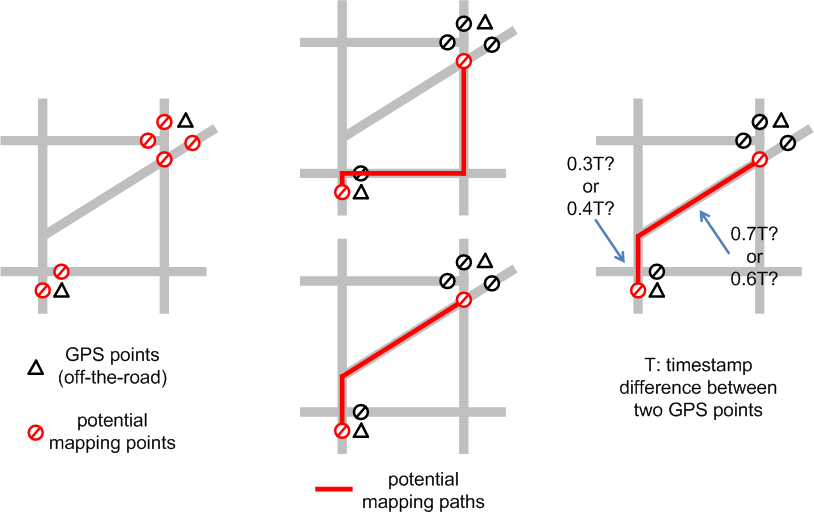}
	\caption{Illustrations of procedures required to process GPS data for traffic estimation and reconstruction. LEFT (map-matching): off-the-road GPS points need to be mapped onto a road network. MIDDLE (map-matching): after determining the matching points, the traversed path of a vehicle needs to be inferred. RIGHT (travel-time estimation): after determining the traversed path, the timestamp difference needs to be distributed to individual road segments.}
	\label{fig:issue-map-matching}
\end{figure}

\begin{figure}[htb]
	\centering
	\includegraphics[width=0.6\textwidth]{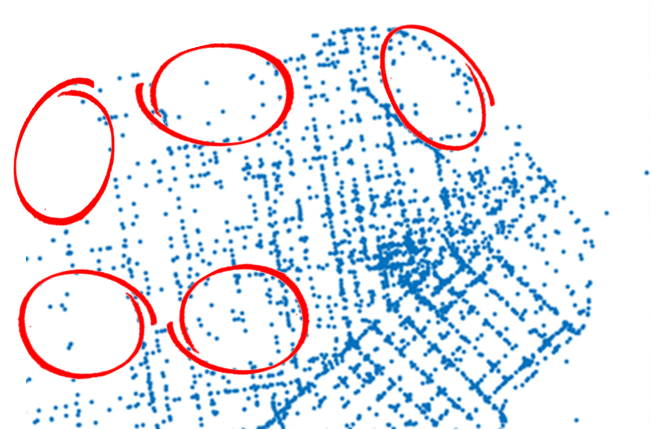}
	\caption{Illustration of the sparsity issue embedded in GPS data. One-day GPS data of downtown San Francisco from the Cabspotting project~\protect\cite{cabspotting} are plotted, marked regions showing the lack of data.}
	\label{fig:issue-sparsity}
\end{figure}

Many state-of-the-art map-matching approaches use the shortest-distance criterion~\cite{lou2009map,quddus2015shortest}, which treats the shortest-distance path between two mapped GPS points as the traversed path. This assumption, however, can lead to a considerable bias in a congested network where the shortest-distance path potentially differs from the shortest travel-time path between two GPS points, while the latter is preferred by GPS devices and experienced drivers. If we adopt a sequential pipeline to process GPS data, \emph{travel-time estimation} (the subsequent step of \emph{map-matching}) will produce inaccurate results, since the timestamp difference between GPS reports will be distributed to a wrong set of road segments. 

GPS data are usually scarce in certain time periods, such as early-morning hours, and certain areas, such as suburbs (illustrated in Figure~\ref{fig:issue-sparsity}). We need to interpolate these spatial-temporal missing traffic data in order to gain a holistic view of a city's traffic dynamics. In addition, in order to account for the dynamic nature of traffic, we need to develop methods that can take both local interactions among cars and global propagation of traffic into account. In theory, the spatial missing data can be approximated using traffic simulation. However, we have to ensure the flow consistency on the boundaries that separate areas with GPS data coverage and areas without GPS data coverage in order to accept a particular simulation result. This requires us to be able to dynamically adjust traffic simulation for matching traffic flows.


Chapter~\ref{ch:itsm} and Chapter~\ref{ch:siga} are dedicated to address the limitations of GPS data and explain how to accurately and efficiently estimate traffic conditions, interpret spatial-temporal missing traffic data, and dynamically reconstruct traffic flows at a city scale.

\section{Autonomous Driving}
\label{sec:intro-av}
While traffic dynamics can only be studied collectively, it is also formed by individual vehicles. Given over 90\% crashes are due to human errors, the conversion from human-driven vehicles to autonomous vehicles (AVs)---with improved safety and efficiency features---has the potential to alleviate the severe traffic conditions~\cite{Wu2018Stablizing}. 


Widespread research on autonomous driving has been conducted as a result of recent advancements in automation and machine learning algorithms. Before the deployment of AVs, we need to scrutinizingly test and improve their safety, control, and coordination. I have developed a framework, named ADAPS (\textbf{A}utonomous \textbf{D}riving Vi\textbf{a} \textbf{P}rincipled \textbf{S}imulations), which not only can be used to simulate, analyze, and produce driving data in various scenarios including accidents, but also represent a more efficient online learning mechanism for learning robust control policies for autonomous driving~\cite{Li2019ADAPS}.

Learning to drive is a sophisticated task considering the complexity and dynamics of an environment and the lack of precise descriptions of a desired driving behavior. These features pose challenges to conventional planning and control methods, since they usually require us to specify the dynamics of the world---which can be used to model the cost of an action at each time step and consequences of an action in the future~\cite{ratliff2009learning,silver2010learning}. Describing and parameterizing an uncontrolled and unpredictable driving environment is a daunting and often impractical task, thus preventing the conventional techniques from succeeding.

Nevertheless, driving, as a complex task, can be easily demonstrated by humans. The intricate traffic flow manifests such success. This phenomenon has inspired \emph{imitation learning}---a data-driven machine learning technique that leverages expert demonstrations to synthesize a controller for achieving a desired behavior. A number of studies have shown that imitation learning is an effective approach capable of learning policies for a juggling robot~\cite{atkeson1994using}, a quadruped robot maneuvering on rough terrains~\cite{ratliff2007imitation}, and autonomous driving on highway~\cite{pomerleau1989alvinn}, among many others.


A straightforward way to achieve imitation learning is via standard supervised learning: learning a training dataset that is a collection of observations and their corresponding behaviors. In the context of driving, this training dataset could be images from a front-facing camera and the steering angle applied while each image was taken. The effectiveness of supervised learning is commonly shown by running a learned model on a test dataset and analyzing the results. 

Supervised learning operates based on one assumption---the examples in both training and test datasets are independently and identically distributed (i.i.d.)~\cite{friedman2001elements}. While this assumption is reasonable for many problems, it fails to apply on driving, in which task the training and test examples are not i.i.d.~\cite{ross2011reduction}. 

To be specific, autonomous driving is a sequential prediction and controlled (SPC) task, meaning the system has to predict a sequence of control commands over time to achieve a goal. Given the sequential nature, in such a system, any predicted control commands will affect the following observations being taken by the system and, consequently affect the future predicted control commands (since the predictions are drawn based on the observations). Because the predictor and the observer are entangled, their resulting training and test examples violate the i.i.d. assumption. 

Although we can still use standard supervised learning to achieve small training errors, the problem occurs during the test phase. In any learning process, approximation is inevitable, which, when combined with the entanglement of the predictor and the observer, can lead a system to large test errors: a small disturbance from either the predictor or the observer is likely to lead to compounding errors. To give an example, in driving, a small disturbance to the sensor or the control module can result in a vehicle encountering an ``unseen'' observation (i.e., an observation not appearing in the training dataset). Consequently, this observation will confuse the control module, causing an unpredictable maneuver of the vehicle and lead the vehicle to more ``unseen'' observations. This mismatch between the observations used for training (i.e., training distribution) and the observations encountered during testing (i.e., test distribution) is termed \emph{covariate shift}~\cite{sugiyama2012machine}. 

In essence, a driving policy from supervised learning by treating all training examples as i.i.d. will not be robust to its own mistakes and likely lead a vehicle to dangerous situations including accidents. To alleviate this issue, first, we need driving data from both safe and, especially, dangerous situations. This is critical because most expert demonstrations (i.e., human driving data) are from safe driving. However, if a learning algorithm is not trained on recovery demonstrations in dangerous situations, when encountering, the algorithm's behavior is undefined. Second, we need an efficient online learning mechanism. For solving an SPC task, interactions with experts in a number of iterations are essential, without which, no learning algorithm can ensure robust performance~\cite{ross2011reduction}. Since learning an effective policy for an SPC task requires iterative testing and update, an efficient online learning technique that can reduce the number of learning iterations is desired. For autonomous driving, this becomes critical given that one iteration implies one incident (otherwise there is no need to update a policy by proceeding to the next iteration). Third, we need an effective architecture for the control policy. In the simplest case, a policy needs to classify road conditions and make corresponding maneuvers. The two basic road conditions are ``safe'', e.g., no obstacle on the road---the vehicle can just follow the road, and ``dangerous'', e.g., an obstacle appearing on the road---the vehicle needs to avoid it. Next, I will briefly explain and introduce my solutions to these challenges.

Obtaining recovery data of vehicles from dangerous situations in the physical world is impractical, due to the high expenditure of a vehicle and potential injuries to people both inside and outside a vehicle. In addition, even if such recovery data are collected, human experts are usually required to label them. This process could be inefficient and subject to judgmental errors~\cite{ross2013learning}. These difficulties suggest the potential use of the virtual world, in which we can simulate various dangerous scenarios including accidents and then analyze the simulated scenarios to generate recovery data. I have developed a simulation platform for this purpose. 



In order to develop an efficient online learning technique, my solution is to treat the principled simulation as the ``expert''. This ``expert'' will plan alternative safe trajectories during the analysis of a simulated accident while taking the kinematic and dynamic constraints of a vehicle into account. As a result, not only can the number of ``expert'' trajectories be generated indefinitely, but the access to such an ``expert'' is instantaneous. These features will assist in reducing the number of learning iterations.  

Lastly, I propose a hierarchical control policy using deep neural networks~\cite{lecun2015deep} and long short-term memory networks~\cite{hochreiter1997long}. My policy consists of three modules, the \emph{detection} module, which monitors a road condition and categorizes it as ``safe'' or ``dangerous'', the \emph{following} module, which directs a vehicle to follow the road if the condition is considered ``safe'', and the \emph{avoidance} module, which steers a vehicle away from an obstacle when the road condition is considered ``dangerous''. Chapter~\ref{ch:adaps} is dedicated to explain the abovementioned solutions in details. 

Although this dissertation addresses the macroscopic level and the microscopic level of traffic separately, the two levels are tightly coupled. Many applications can stem from their rich connection. For example, from macroscopic to microscopic, the reconstructed city-scale traffic can assist route planning of autonomous driving and enrich virtual environments for training purposes~\cite{Chao2019Survey}; from microscopic to macroscopic, individual autonomous vehicles can be adopted to stabilize and regulate traffic flows~\cite{Wu2018Stablizing}. These applications will be further discussed in Chapter~\ref{ch:Conclusion}. In the following, I will introduce the main results regarding \emph{city-scale traffic reconstruction} and \emph{autonomous driving}.

\section{Main Results}
\subsection{City-Scale Traffic Reconstruction}
\subsubsection{Deterministic Estimation of Traffic Conditions}
\label{sol:determine}
My first solution to the estimation of traffic conditions is an efficient deterministic approach~\cite{Li2017CityEstSparse}, for which details can be found in Chapter~\ref{ch:itsm}. My solution is based on two observations: \emph{traffic patterns exhibit weekly periodicity} and \emph{traffic conditions are quasi-static}. Using these observations, I treat one week as a traffic period and assume that traffic conditions are static within each hourly time interval of a weekly period. Then, in each time interval, I conduct \emph{map-matching} by replacing the shortest distance criterion with the shortest travel-time criterion, since the former criterion can lead to mapping errors in a congested network. 

In order to compute the shortest travel-time path between any two nodes in a network, we need the travel time information of the network. However, the initial condition of a road network usually lacks such information. My approach to alleviating this issue is to adapt a travel-time allocation method from Hellinga et al.~\cite{hellinga2008decomposing}, which was developed based on empirical observations of the real-world traffic dynamics. Using this technique, the subsequent \emph{map-matching} process can take the intermediate estimated traffic conditions into account, which helps improving the overall map-matching accuracy. 

I have evaluated the effectiveness of my solution through comparison with a state-of-the-art technique from Lou et al.~\cite{lou2009map} on a synthetic road network. I have established 30 traffic conditions corresponding to 30 congestion levels as the ground truth. For each congestion level, I test my algorithm by sampling different portions of the simulated traffic population, namely at levels 20\%, 40\%, 60\%, 80\%, and 100\%. In principle, the more GPS traces are used in traffic reconstruction, the more accurate are the reconstruction results. 


The first analysis, shown in Figure~\ref{fig:intro-nn-time}, is conducted using the aggregate travel time of the road network under each congestion level as the \emph{true travel time}. Starting with 20\% GPS traces, my technique shows close approximations to the \emph{true travel time} at all congestion levels, while the shortest-distance based technique fails to achieve the same performance. Table~\ref{tab:intro-nn-error} shows the ``absolute error to ground truth'' computed from all congestion levels. 


\begin{figure}
	\centering
	\includegraphics[width=\textwidth]{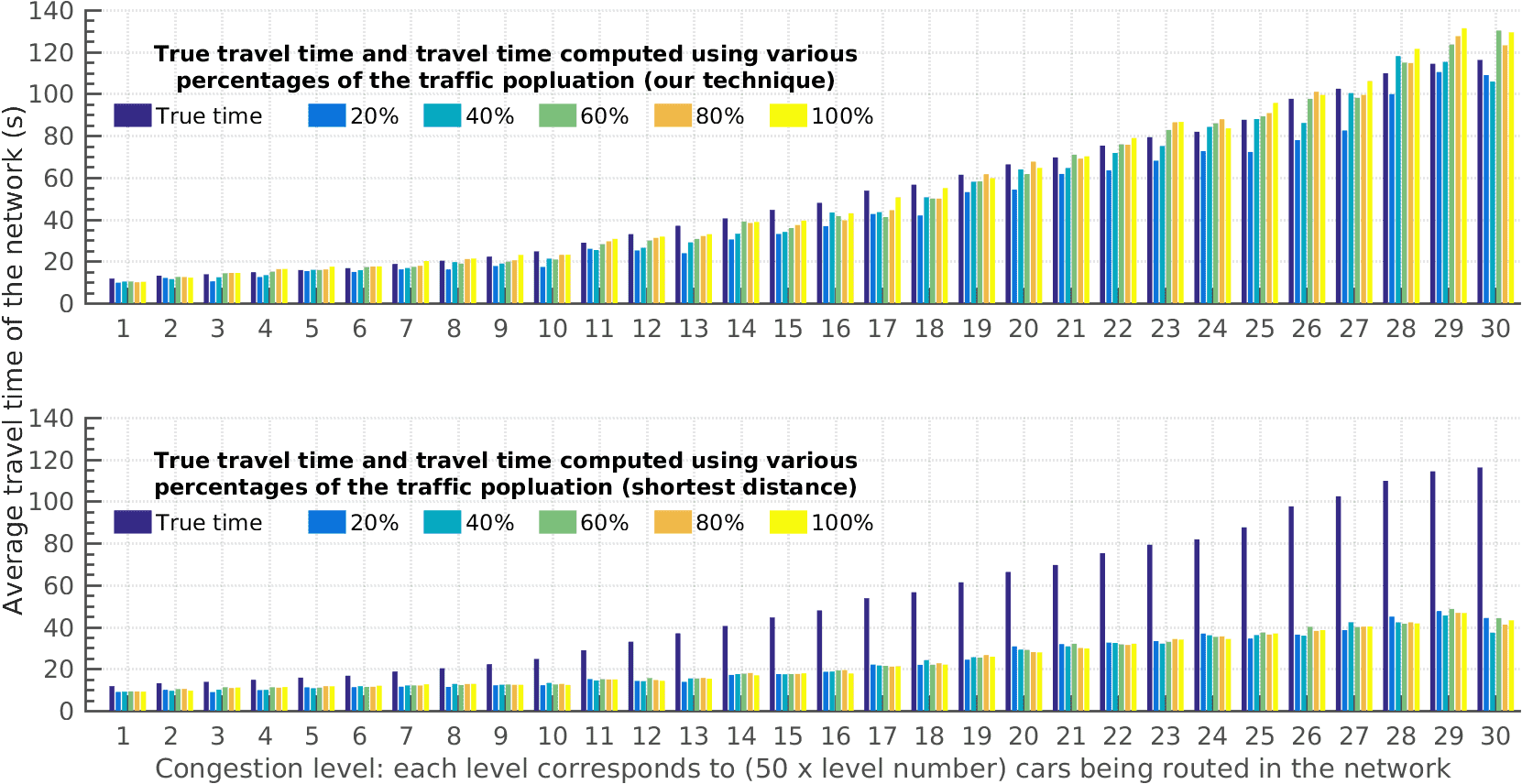}
	\caption{Recovery of average travel time on different percentages of the traffic population using my approach (TOP) vs. Lou et al.~\protect\cite{lou2009map} (BOTTOM). My technique consistently outperforms Lou et al.~\protect\cite{lou2009map} in estimating the average travel time over 30 different congestion levels.}
	\label{fig:intro-nn-time}
\end{figure}

\begin{table}[ht!]
	\centering
	\small
	\begin{tabular}{ccccc}
		\toprule
		& \multicolumn{4}{c}{Absolute errors to the ground truth}  \\
		\cmidrule(l){2-5}
		& \multicolumn{2}{c}{My technique}  & \multicolumn{2}{c}{Lou et al.~\cite{lou2009map}}  \\		
		\cmidrule(l){2-3} 	\cmidrule(l){4-5} 	
		Traffic percentage & mean ($s$) & std. ($s$) & mean ($s$) & std. ($s$)		\\
		\midrule
		20\% & 8.29 & 5.31 & 29.74 & 21.90 
		\\
		\midrule
		40\% & 4.25 & 3.38 & 29.85 & 22.66
		\\
		\midrule
		60\% & 3.67 & 3.67 & 29.33 & 22.06
		\\
		\midrule
		80\% & 3.40 & 3.32 & 29.55 & 22.46  
		\\
		\midrule
		100\% & 3.58 & 4.04 & 29.66 & 22.35 
		\\
		\bottomrule
	\end{tabular}
	\caption{The absolute errors in the recovered travel time computed using my technique vs. Lou et al.~\protect\cite{lou2009map} by using GPS traces from various percentages of the traffic population. My technique results in much smaller errors as of Lou et al.~\protect\cite{lou2009map}. }
	\label{tab:intro-nn-error}
\end{table}

The second analysis, shown in Figure~\ref{fig:intro-tt-results}, summarizes the relative improvements measured in mean squared error (MSE) of my method over Lou et al.~\cite{lou2009map}. As the congestion level increases or more GPS traces are used in recovery, my technique outperforms the other technique. In comparison, the improvements are less salient when the congestion level is low, (i.e., the first 10 congestion levels out of 30), but more salient when the congestion level is high (i.e., the rest 20 congestion levels out of 30).  

\begin{figure}
	\centering
	\includegraphics[width=\textwidth]{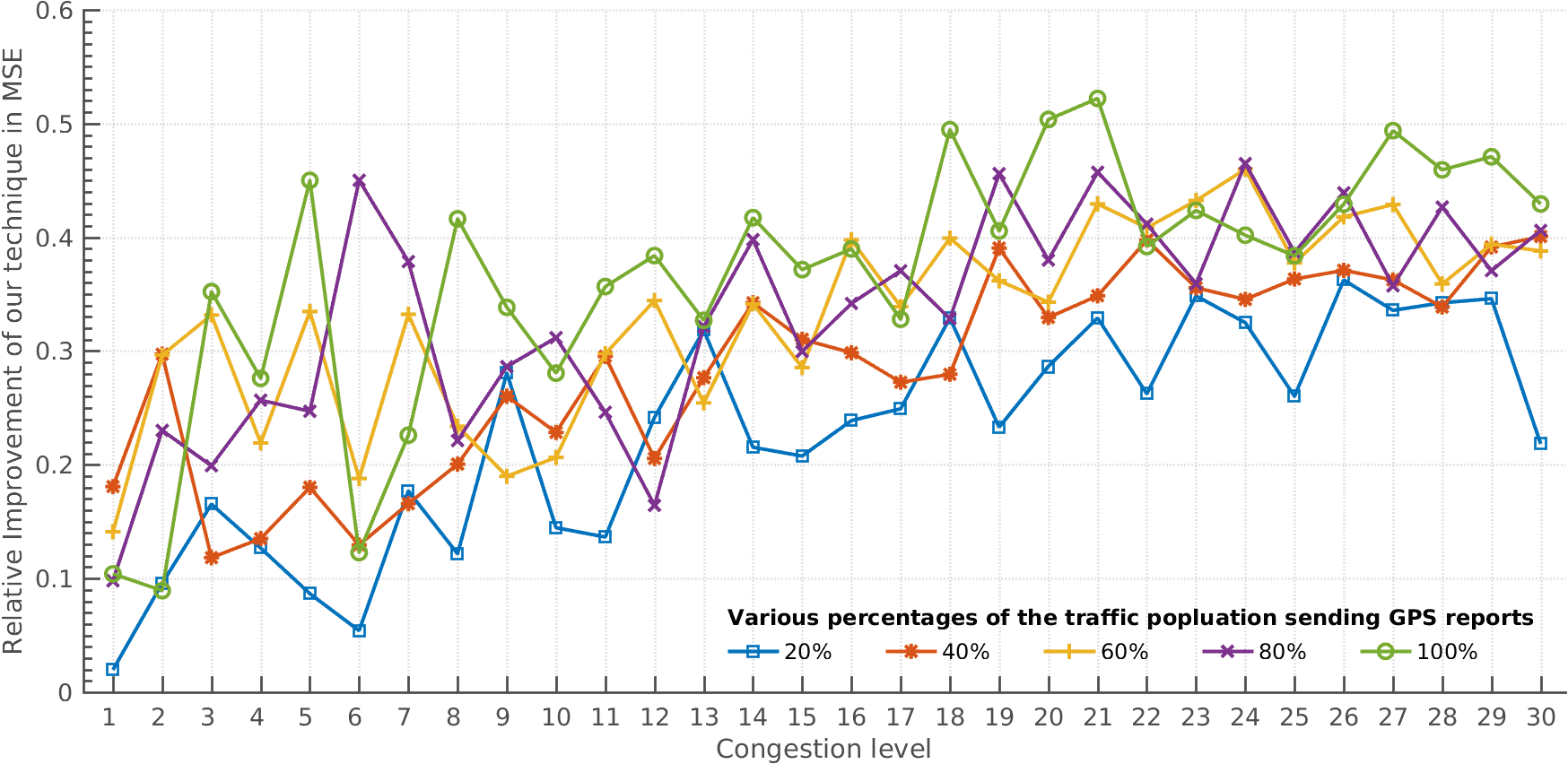}
	\caption{Relative improvements measured in MSE of my technique over Lou et al.~\protect\cite{lou2009map} on travel time. My technique outperforms Lou et al.~\protect\cite{lou2009map} as the congestion level increases or as more GPS traces become available for the recovery task.} 
	\label{fig:intro-tt-results}
\end{figure}

\subsubsection{Temporal Missing Data Completion}
Traffic data can be scarce for certain time periods such as early-morning or late-night hours. In order to interpolate the data for these time periods, I exploit the fact that the traffic pattern is intrinsically periodic (as a result of periodic human behaviors). It thus has a sparse representation in the frequency domain. Based on this observation, I have developed a technique based on Compressed-Sensing~\cite{1614066,1580791} to fill in missing travel time information and robustly recover the traffic pattern of a road segment over an entire traffic period~\cite{Li2017CityEstSparse}. The details of this approach can be found in Chapter~\ref{ch:itsm}.

\begin{figure}
	\centering
	\includegraphics[width=\textwidth]{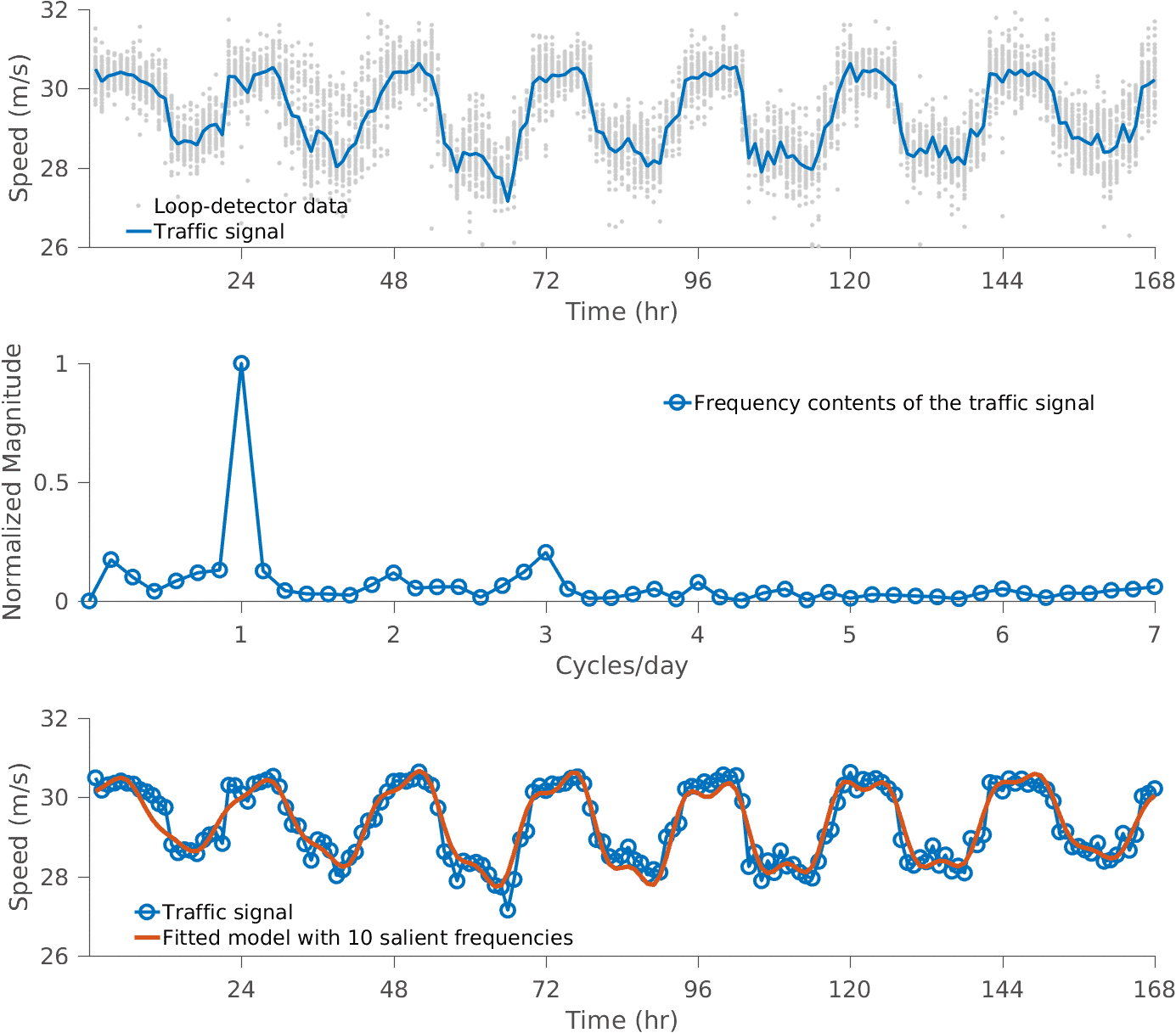}
	\caption{The average speed measurements from loop-detector data are interpreted as a \emph{traffic signal}, which exhibits a clear periodic pattern (TOP); the spectral analysis reveals that the most prominent frequency is one cycle per day (i.e., 24 hours) (MIDDLE); the traffic signal is approximated by a frequency-domain linear regression model in which 95\% energy is retained by keeping the 10 largest frequencies (BOTTOM).}
	\label{fig:intro-loop-analysis}
\end{figure}

In order to test the effectiveness of my approach, I have adopted speed measurements from 38 loop detectors in the city of San Francisco. These data represent relatively complete and accurate measurements of traffic. The hourly average speed measurements of a single loop detector for a weekly period is termed the \emph{traffic signal} (see an example in the Figure~\ref{fig:intro-loop-analysis} top panel). To analyze each traffic signal, I have performed a spectral analysis, which results are shown in Figure~\ref{fig:intro-loop-analysis}. We can see from the middle panel that the most salient oscillation appears at 24 hours and from the bottom panel that most signal energy can be captured by keeping the 10 largest frequencies of the signal. These indicate that a traffic signal indeed has a sparse representation in the frequency domain, which is emphasized by the analysis of the decaying rates of frequency magnitude presented in the top panel of Figure~\ref{fig:intro-sparse}.


\begin{figure}
	\centering
	\includegraphics[width=\textwidth]{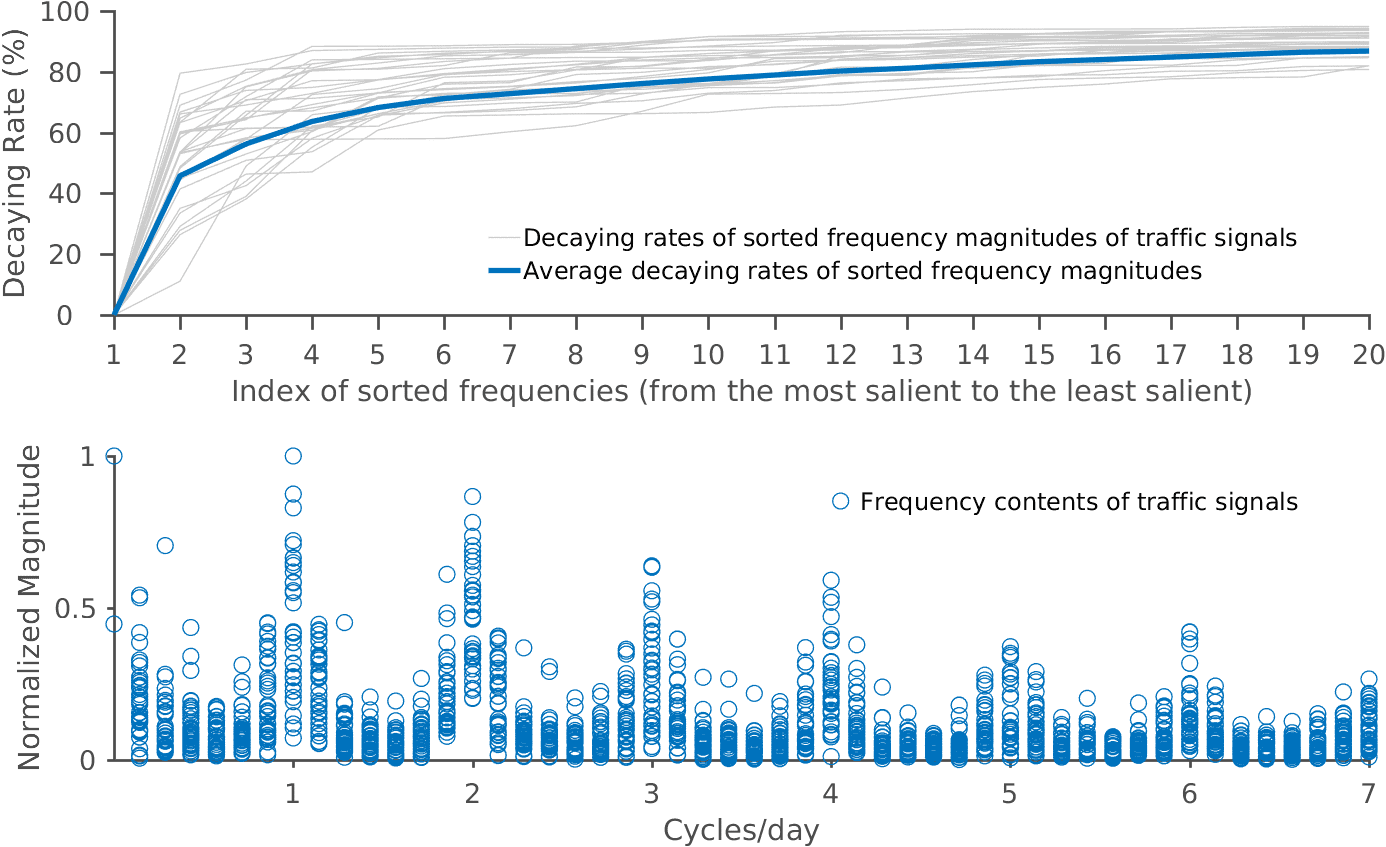}	
	\caption{The top panel shows the decaying rates of frequency magnitudes of all traffic signals; the bottom panel shows the locations and normalized magnitudes of the frequency components of all traffic signals. The rapid growth of decaying rates and randomly distributed frequency structures indicate that Compressed Sensing is applicable for recovering a traffic signal.}
	\label{fig:intro-sparse}
\end{figure}

According to the Nyquist-Shannon Theorem, we need at least 168 measurements to fully recover a traffic signal, which are lacking as a result of the temporal sparsity of GPS data. However, Compressed Sensing promises that if a signal has a sparse representation and randomly distributed frequency components, we can have a accurate recovery of the signal. Those features do appear in a traffic signal as shown in Figure~\ref{fig:intro-sparse}. This indicates that we can recover a traffic signal using Compressed Sensing. An example is shown in Figure~\ref{fig:intro-loop-recover}. The robustness analysis of my approach can be found in Figure~\ref{fig:intro-comsense-error}.  


\begin{figure}[ht]
	\centering
	\includegraphics[width=\textwidth]{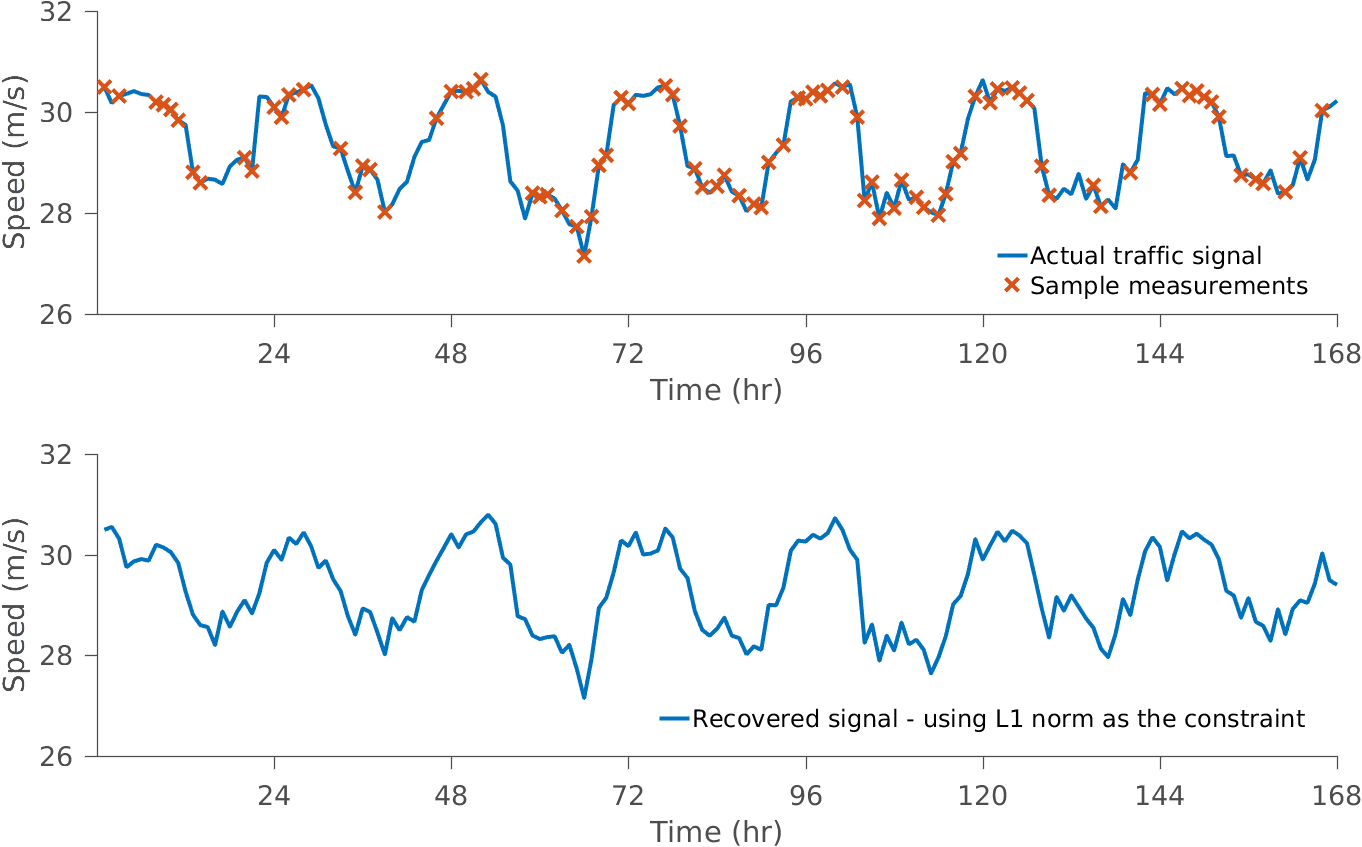}
	\caption{A recovered traffic signal using my technique highly resembles its original form.} 
	\label{fig:intro-loop-recover}
\end{figure}

\begin{figure}[ht]
	\centering
	\includegraphics[width=\textwidth]{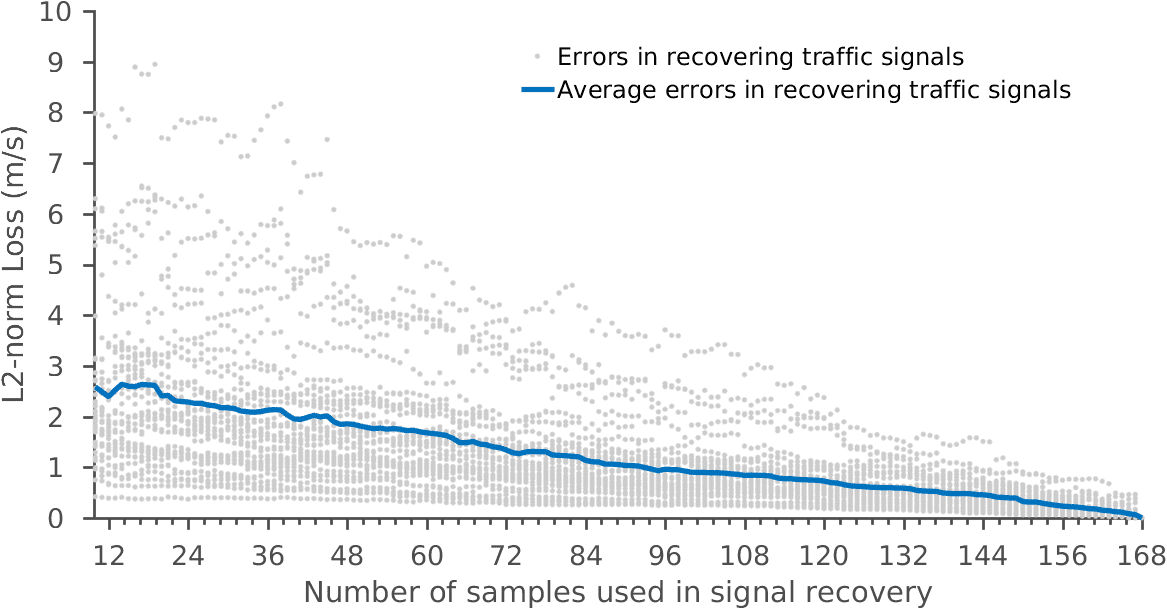}
	\caption{My technique shows robustness when the number of samples used in recovering a traffic signal decreases.} 
	\label{fig:intro-comsense-error}
\end{figure}

By confirming the applicability of Compressed Sensing algorithm to recovering traffic signals, we can proceed to apply it to real-world GPS datasets. To demonstrate that my approach can recover features of a traffic signal, I have adopted the metric $\text{\emph{fluidity}}\in [0,1]$~\cite{hofleitner2012large}, computed as the ratio of the estimated travel speed to the free-flow speed of a road segment. In Figure~\ref{fig:intro-dynamics}, I show the recovered traffic pattern using actual GPS data in San Francisco, which shows clear periodicity at one cycle per 24 hours. 

The similarity among different days in a week can also be used to show the effectiveness of a signal recovering method. I have computed the cosine distance between every pair of days for both the estimated traffic conditions and the actual traffic conditions derived from the loop-detector data from San Francisco. The results are shown in Figure~\ref{fig:intro-similarity}. The left panel shows the distance scores while the right panel shows the qualitative results: the upper triangle indicates the estimated quantities and the lower triangle indicates the actual quantities. The symmetrical pattern illustrates that my technique can produce accurate results compared to the ground-truth values.

\begin{figure}
	\centering
	\includegraphics[width=\textwidth]{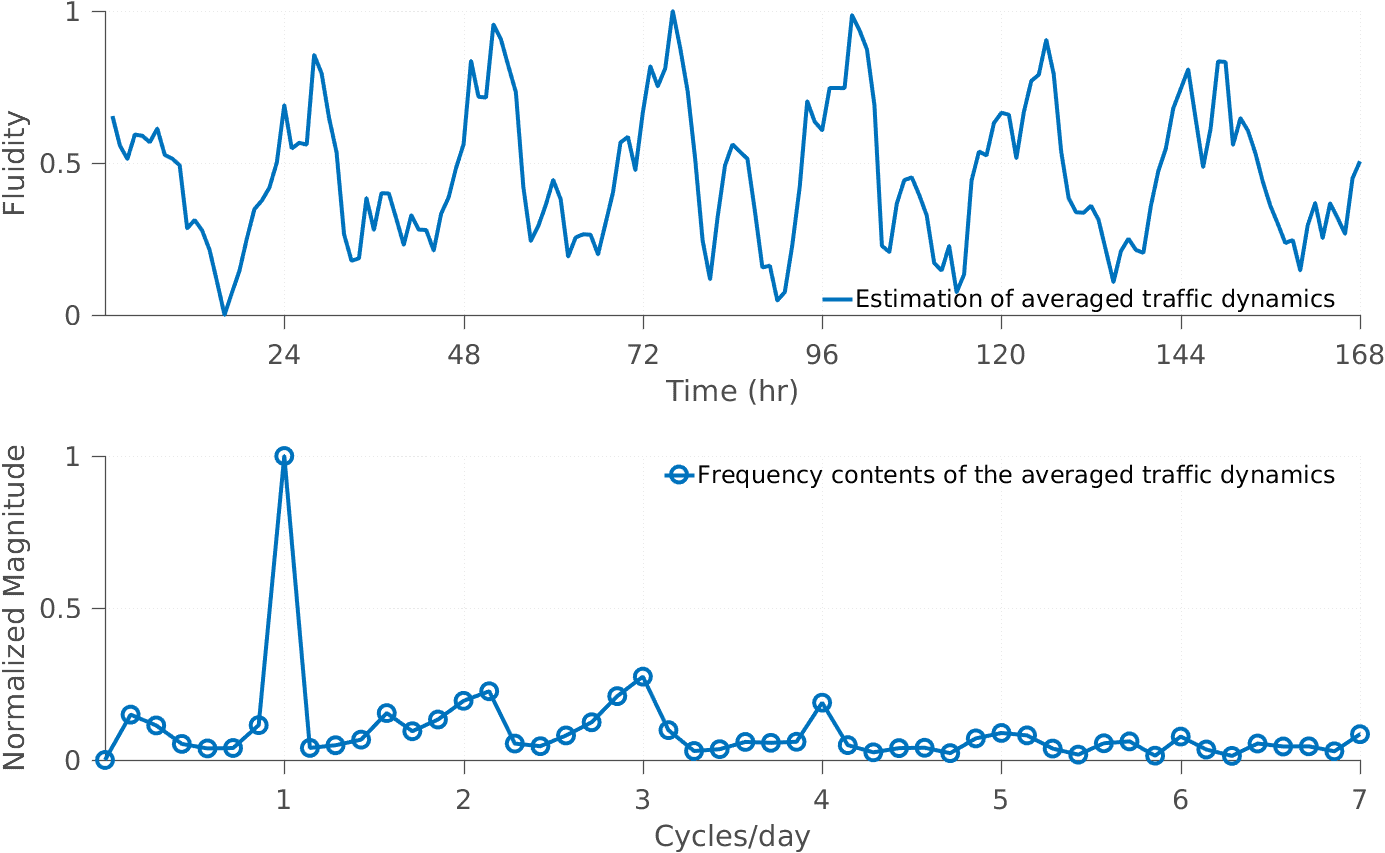}
	\caption{Estimated traffic pattern of downtown San Francisco (TOP) and its spectral analysis (BOTTOM). The result from my technique demonstrates a clear daily trend, which is consistent with the periodic feature observed in loop-detector data from the same area.}
	\label{fig:intro-dynamics}
\end{figure}

\begin{figure}
	\centering
	\includegraphics[width=\textwidth]{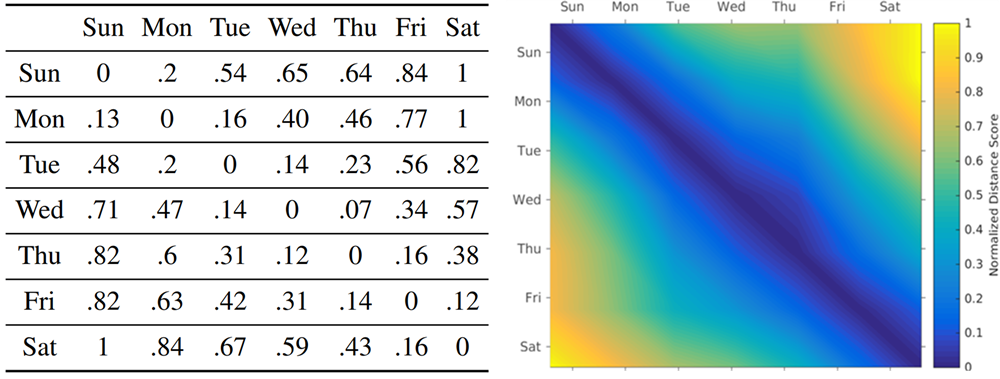}
	\caption{Correlation between every pair of days in a week. The left panel lists normalized similarity scores calculated using the cosine distance, and the right panel provides the qualitative results. For both, the upper triangular matrix is derived using the estimated traffic conditions from GPS data, and the lower triangular matrix is computed using loop-detector data from the same area. When data and patterns are compared across the diagonal line, my estimated results exhibit high similarity to the loop-detector data. }
	\label{fig:intro-similarity}
\end{figure}

\subsubsection{Iterative Estimation and Spatial Missing Data Completion}
While the deterministic approach introduced in the above section for estimating traffic conditions is efficient and effective, in order to further improve the accuracy for interpolating spatial missing data, I have developed an iterative algorithm that embeds \emph{map-matching} and \emph{travel-time estimation} as its sub-routines~\cite{Li2018CityEstIter}. 

To be specific, in areas with GPS data coverage, I first conduct a coarse inference of travel times at each hourly time interval by solving a convex optimization program inspired by Wardrop's Principles~\cite{wardrop1952road,sheffi1985urban}. Then, I refine the inferred results via executing \emph{map-matching} and \emph{travel-time estimation} iteratively. Next, I proceed to establish baseline estimation of traffic conditions in areas without GPS data coverage using a nested optimization procedure~\cite{yang1992estimation} to ensure that certain traffic flow characteristics are met. The details can be found in Chapter~\ref{ch:siga}. 

I have evaluated my approach using a real-world road network, which contains 5407 nodes and 1612 road segments. Additionally, 34 ground-truth travel times are established and over 10 million synthetic GPS traces are sampled based on the established heuristic travel times. I have compared my approach with two state-of-the-art methods, Hunter et al.~\cite{hunter2014large} and Rahmani et al.~\cite{rahmani2015non}. The results show that my algorithm offers the lowest error rate and up to 97\% relative improvement in estimation accuracy (see Figure~\ref{fig:estimation}).

\begin{figure}
	\centering
	\includegraphics[width=\textwidth]{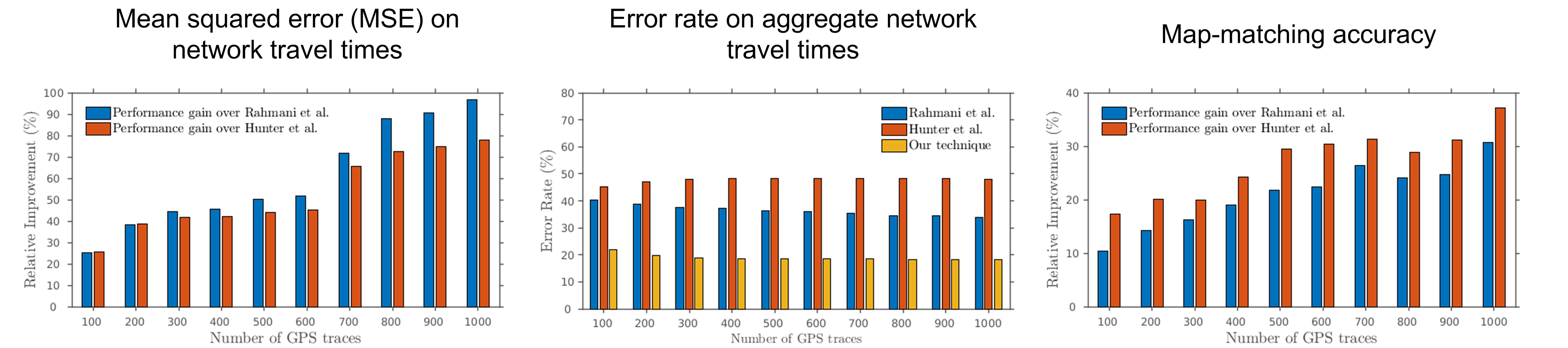}
	\caption{My algorithm on map-matching and travel-time estimation achieves consistent improvements over the previous techniques~\protect\cite{rahmani2015non,hunter2014large} in various measures.}
	\label{fig:estimation}
\end{figure}

Currently, most existing methods interpolate spatial missing measurements statically~\cite{hellinga2008decomposing,rahmani2015non,herring2010estimating,hunter2014large,tang2016estimating,hunter2014large,Li2017CityEstSparse}. In order to account for the dynamic nature of traffic, I have leveraged traffic simulation for the interpolation task and developed an algorithm to guarantee the consistency of traffic flows on the boundaries of areas with and without GPS data~\cite{Li2017CityFlowRecon}. 

An effective way to dynamically interpolate spatial missing data is to run traffic simulation in the ``empty'' areas. However, arbitrary simulations will not respect the boundary conditions from the previously estimated traffic conditions in data-rich areas. This indicates we need to fine tune the simulation with the objective of enforcing minimal discrepancy between simulated traffic flows and previously estimated traffic flows. For this goal, I allow the simulation algorithm to alter ``turning ratios'' at intersections---parameters that determine how traffic will distribute itself to downstream road segments at an intersection. Using this design perspective, the objective then becomes deriving the optimal ``turning ratios'' such that the simulated traffic flows will respect the estimated traffic flows (from GPS data) at the boundaries that join data-deficient and data-rich regions. 

The altered objective leads me to use simulation-based optimization for finding the optimal solution. However, as our goal is to reconstruct city-scale traffic, an optimization program at that scale could be computationally cost prohibitive. To remedy this issue, I have adopted a metamodel-based simulation optimization~\cite{osorio2015metamodel}. Compared to a stochastic microscopic traffic simulator, the metamodel is a deterministic function, thus is much more tractable and computationally efficient when running within an optimization routine. In order to make sure that the metamodel behaves similarly to the simulator, the metamodel is trained using simulations. This way we only need to run traffic simulation dozens of times (for training the metamodel) instead of hundreds or even thousands of times (for running simulation-based optimization).

I have compared my technique to a baseline approach that only uses a simulator (i.e., simulation-only approach). I set up the experiments using various origin-destination demands on the road network of San Francisco to establish ground-truth traffic conditions. Tested on these conditions, my approach maintains on average a 7\% error rate while achieving up to 90 times speedup compared to the simulation-only approach. These statistics can be seen in Figure~\ref{fig:intro-meta-error} and Figure~\ref{fig:intro-meta-speedup}, respectively.

\begin{figure}
	\centering
	\includegraphics[width=\textwidth]{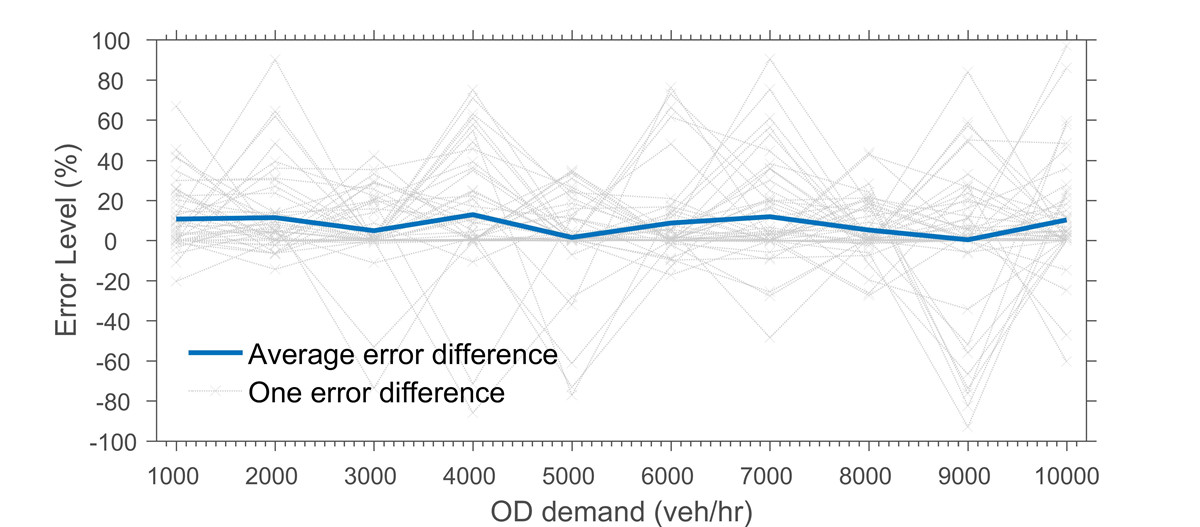}
	\caption{The error level of my technique vs. simulation-only approach: For a given road network and a specific origin-destination demand, I first compute the differences between the two methods with respect to the ground truth. Then, I subtract these two differences to obtain one error difference measure (indicated by a gray cross). The mean, minimum, and maximum values of the average error level (indicated by the solid line) are respectively 7.8\%, 0\%, and 13\%. In many cases, my technique even outperforms the simulation-only approach with much smaller differences to the ground truth, shown by the negative values.}
	\label{fig:intro-meta-error}
\end{figure}

\begin{figure}
	\centering
	\includegraphics[width=\textwidth]{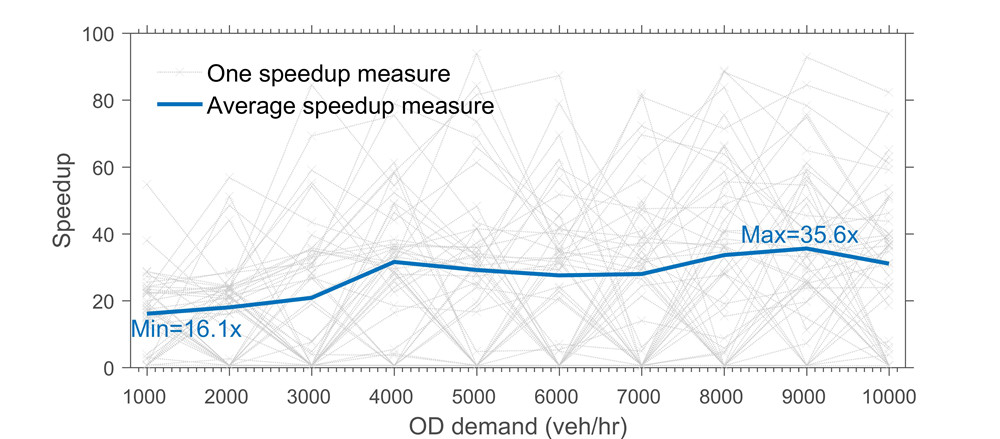}
	\caption{The performance speedup of my technique over the simulation-only approach: my technique is \emph{on average} about $27.2x$ faster, with maximum and minimum performance gains of $35.6x$ and $16.1x$. The maximum observed performance gain of a single speedup measure is \emph{over $90x$} (at origin-destination demand = 9000).}
	\label{fig:intro-meta-speedup}
\end{figure}

After interpolating spatial missing data, I have fully reconstructed spatial-temporal traffic at the scale of a city. The reconstructed traffic can be visualized in many ways such as 2D flow map, 2D animation, and 3D animations. Some examples are shown in Figure~\ref{fig:intro-visual-2d}, Figure~\ref{fig:intro-visual-3d}, Figure~\ref{fig:intro-city}, and Figure~\ref{fig:intro-pattern}. These visual representations can be used to enable 1) analysis of traffic patterns at street level, region level, and the city level, and 2) virtual environment applications such as virtual tourism and the training of autonomous vehicles.  


\begin{figure}
	\centering
	\includegraphics[width=\textwidth]{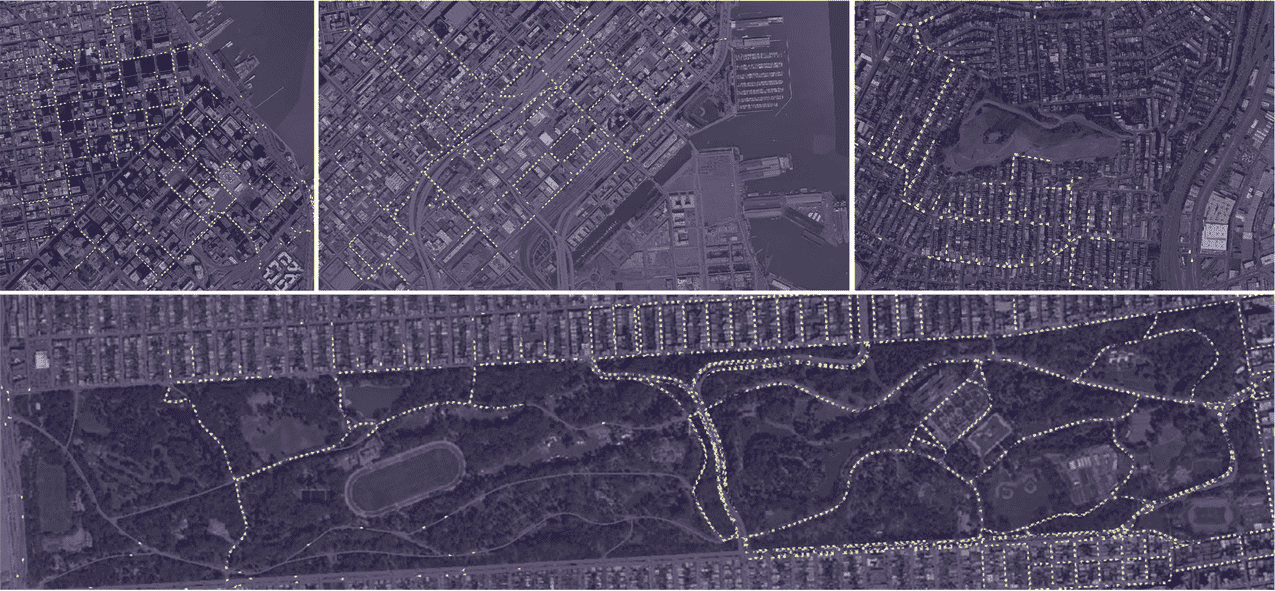}
	\caption{2D traffic animation of regions in San Francisco: Northeast (top left), Central-East (top center), Central (top Right), Northwest (bottom). I have exaggerated the headlights and adopted an evening time period (i.e., Friday 7PM) to make vehicles more visible.}
	\label{fig:intro-visual-2d}
\end{figure}

\begin{figure}
	\centering
	\includegraphics[width=\textwidth]{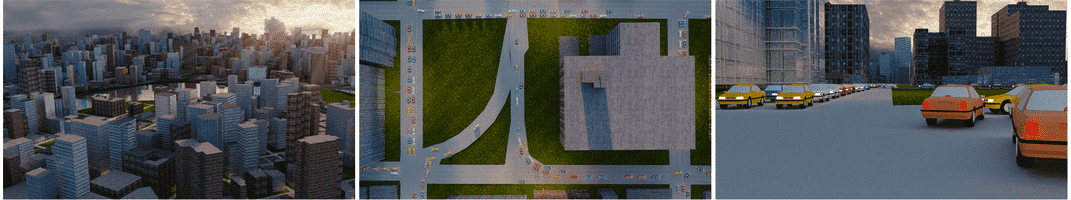}
	\caption{3D traffic animation: a perspective overview (left), a topdown view (center), and a driver's view (right).}
	\label{fig:intro-visual-3d}
\end{figure}

\begin{figure}
	\centering
	\includegraphics[width=\textwidth]{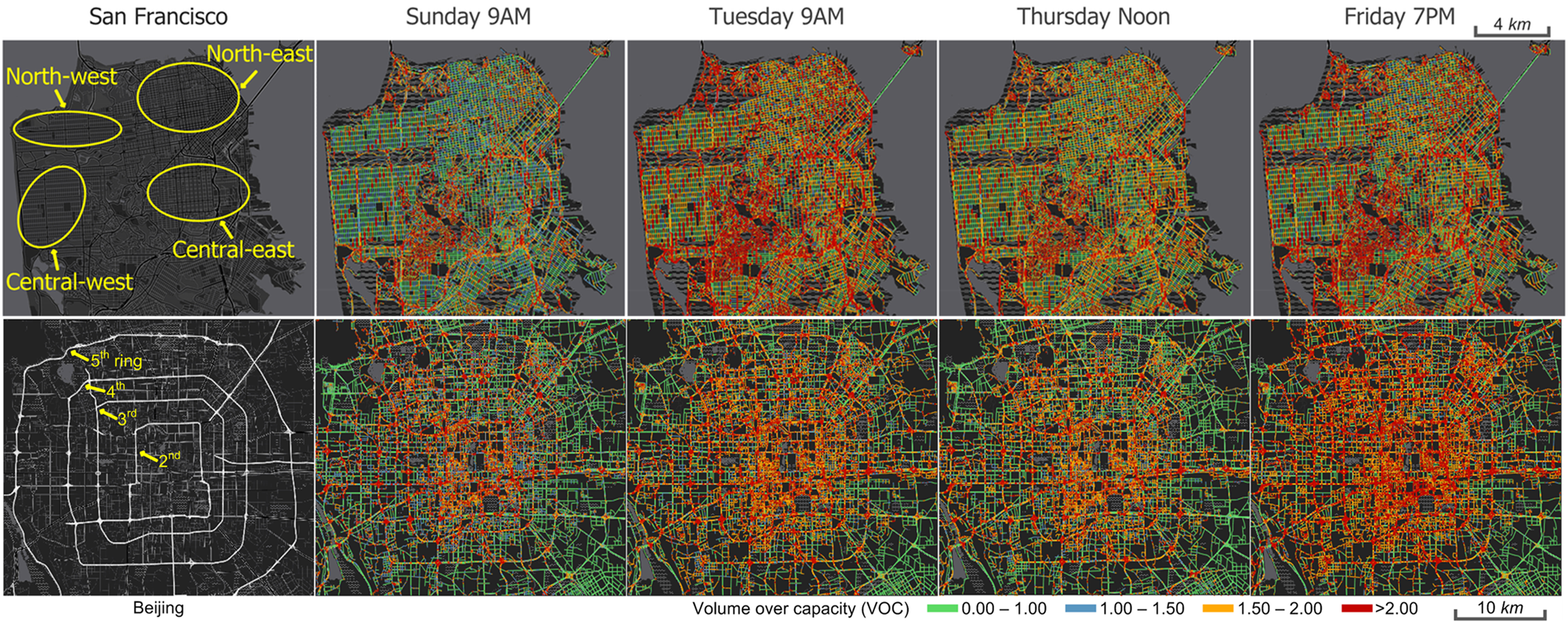}
	\caption{Visualization of traffic patterns in San Francisco and Beijing. Four time periods of a week, namely Sunday 9AM, Tuesday 9AM, Thursday Noon, and Friday 7PM, are selected to illustrate weekend vs. weekday and morning vs. evening traffic. The traffic is measured by Volume Of Capacity (VOC). All computations are conducted in epoch time.}
	\label{fig:intro-city}
\end{figure}

\begin{figure}
	\centering
	\includegraphics[width=\textwidth]{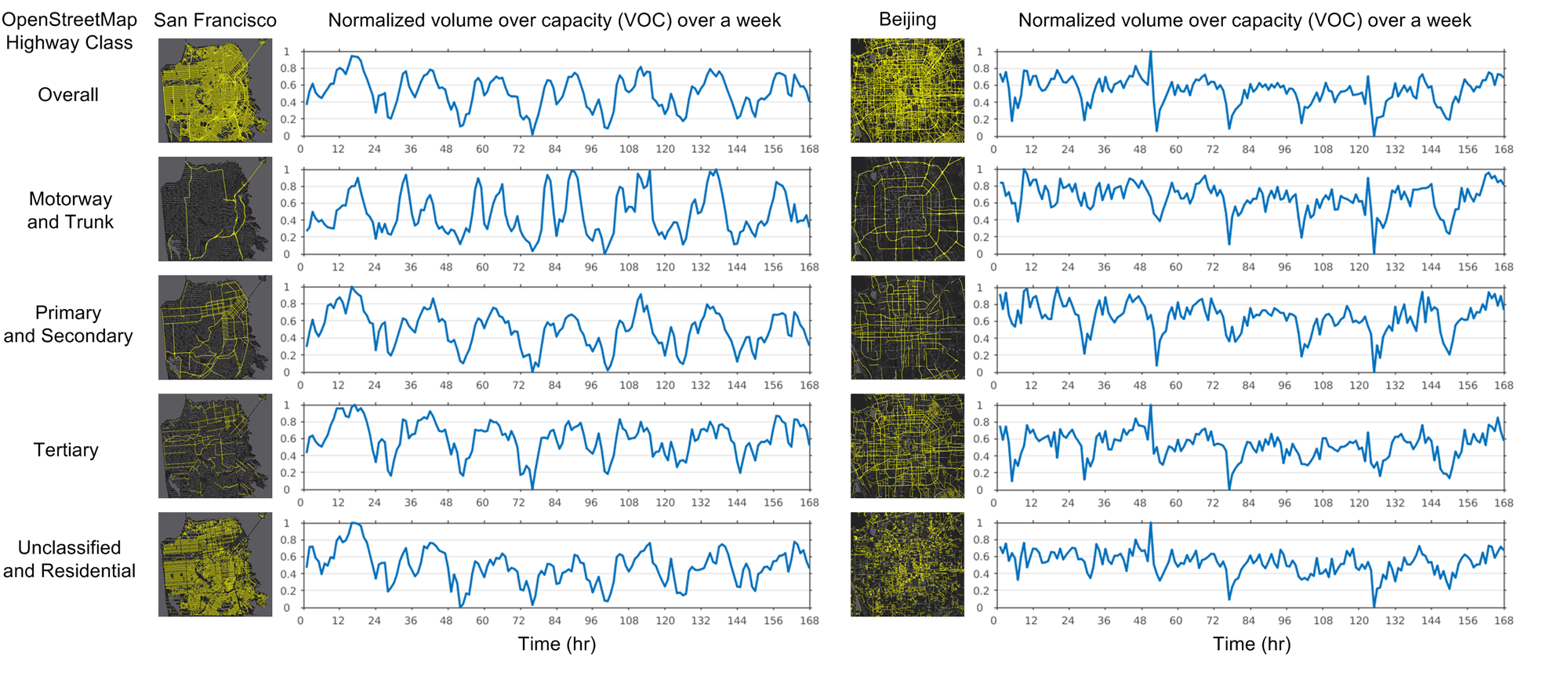}
	\caption{The estimated traffic conditions measured in average volume over capacity (VOC) of San Francisco and Beijing for various types of roads. All computations are conducted in epoch time. My technique successfully recovers the periodic phenomena in all cases. In San Francisco, saddle shapes appear on \textit{motorway and truck} and overall roads for several days indicating mid-day traffic relief. Such phenomena are not observed in Beijing, which suggests the similar usage of different types of roads and the congestion forming throughout daytime.}
	\label{fig:intro-pattern}
\end{figure}

\subsection{Autonomous Driving}
I have developed the framework ADAPS for learning and testing the control policy of autonomous driving. ADAPS consists of two simulation platforms and a hierarchical control policy. The first simulation platform resides in a 3D virtual environment and is utilized to test a learned policy, simulate accidents, and generate labelled data. The second simulation platform operates in a 2D environment and serves as an ``expert'' to analyze and resolve an accident via planning alternative safe trajectories for a vehicle by considering its kinematic and dynamic constraints.

In addition, ADAPS represents a more efficient online learning mechanism compared to existing techniques such as DAGGER~\cite{ross2011reduction}. The reason for the improvement is due to a switch from a ``reset modeling'' approach, in which we can only sample observations by putting an agent in its initial state distribution, to ``generative modeling'' approach, in which we can put an agent in an arbitrary state and sample observations by executing any action in that state. The reason that I can make this switch is because of two reasons: 1) the generation of training examples is from simulations that take vehicle kinematic and dynamic constraints into account, rather than merely executing a policy; 2) the assumption that we have access to all agent states in a simulation is because the simulation is conducted in retrospect of an accident, thus viable. 

In order to understand the theoretical results, I will briefly introduce the notation and definitions used in the analysis. The problem we consider is a $ T $-step control task. Given the observation $ \phi=\phi(s) $ of a state $ s $ at each step $ t \in [\![1, T]\!] $, the goal of a learner is to find a policy $ \pi \in \Pi $ such that its produced action $ a=\pi(\phi)$ will lead to the minimal cost:

\begin{equation}
\hat{\pi} = \argmin_{\pi \in \Pi} \sum_{t=1}^{T} C\left(s_{t}, a_{t}\right),
\end{equation}

\noindent where $ C\left(s,a\right) $ is the expected immediate cost of performing $ a $ in  $ s $. For many tasks such as driving, we may not know the true value of $ C$. Instead, the observed surrogate loss $ l(\phi,\pi, a^*) $ is commonly minimized. This loss is assumed to upper bound $ C $, based on the approximation of the learner's action $ a=\pi(\phi) $ to the expert's action $ a^*=\pi^*(\phi) $. I denote the distribution of observations at $ t $ as $ d_{\pi}^t $, which is the result of executing $ \pi $ from timestep $ 1 $ to $ t-1 $. Consequently, $ d_{\pi}=\frac{1}{T}\sum_{t=1}^{T}d_{\pi}^t $ is the average distribution of observations by executing $ \pi $ for $ T $ steps. The goal of solving an SPC task is to obtain $ \hat{\pi} $---a policy that can minimize the observed surrogate loss under its own induced observations with respect to an expert's actions for those observations:

\begin{equation}
\hat{\pi} = \argmin_{\pi \in \Pi} \mathbb{E}_{\phi \sim d_{\pi}, a^* \sim \pi^*(\phi)} \left[l\left(\phi,\pi, a^*\right)\right].
\end{equation}

\noindent I further denote $ \epsilon =  \mathbb{E}_{\phi \sim d_{\pi^*}, a^* \sim \pi^*(\phi)} \left[l\left(\phi,\pi, a^*\right)\right] $ as the expected loss under the training distribution induced by the expert's policy $ \pi^* $, and the cost-to-go over $ T $ steps of $ \hat{\pi} $ as $ J\left(\hat{\pi}\right) $ and of $ \pi^* $ as $ J\left(\pi^*\right) $. 

By simply treating expert demonstrations as i.i.d. samples, the discrepancy between $ J\left(\hat{\pi}\right) $ and $ J\left(\pi^*\right) $ is $ \mathcal{O}(T^2\epsilon) $~\cite{syed2010reduction,ross2011reduction}. Given the error of a typical supervised learning is $ \mathcal{O}\left(T\epsilon\right) $, this demonstrates the additional cost due to covariate shift when solving an SPC task via standard supervised learning. 

DAGGER~\cite{ross2011reduction} has been used to solve an SPC task by keeping the $ \mathcal{O}\left(T\epsilon\right) $ error. To illustrate its result, I introduce more definitions: the best policy at the $ i $th iteration (trained using all observations from the previous $ i-1 $ iterations) is denoted as $ \pi_{i} $; for any policy $ \pi \in \Pi $, we have its expected loss under the observation distribution induced by $ \pi_{i} $ as $ l_{i}\left(\pi\right) =  \mathbb{E}_{\phi \sim d_{\pi_{i}}, a^* \sim \pi^*(\phi)} \left[l_{i}\left(\phi,\pi, a^*\right)\right], l_{i}\in \left[0,l_{max}\right]$; the minimal loss in hindsight after $ N \ge i $ iterations is denoted as $ \epsilon_{min} = \min_{\pi \in \Pi} \frac{1}{N}\sum_{i=1}^{N}l_{i}(\pi) $ (i.e., the training loss after $ N $ iterations); the average regret is $ \epsilon_{regret} = \frac{1}{N}\sum_{i=1}^{N}l_{i}(\pi_{i}) - \epsilon_{min}$. Then, the accumulated error after $ T $-step using DAGGER~\cite{ross2011reduction} is bounded by the summation of three terms: 

\begin{equation}
J\left(\hat{\pi}\right) \le T\epsilon_{min} + T\epsilon_{regret} +  \mathcal{O}(\frac{f\left(T, l_{max}\right)}{N}),
\end{equation}

\noindent where $ f\left(\cdot\right) $ is the function of fixed $ T $ and $ l_{max} $. The second term tends to $ 0 $ if a no-regret algorithm is used~\cite{hazan2007logarithmic}. The third term tends to $ 0 $ if $ N \rightarrow \infty $. 

\subsubsection{Theoretical Results}
With the above-introduced background, I can introduce the following theoretical results. 

\begin{theorem}
	If the surrogate loss $ l $ upper bounds the true cost $ C $, by collecting $ K $ trajectories using ADAPS at each iteration, with probability at least $ 1 - \mu $, $\mu \in(0,1)$, ADAPS offers the following guarantee:
	\begin{equation*}
	J\left(\hat{\pi}\right) \le J\left(\bar{\pi}\right) \le T\hat{\epsilon}_{min} + T\hat{\epsilon}_{regret} + \mathcal{O}\left(Tl_{max}\sqrt{\frac{\log{\frac{1}{\mu}}}{KN}}\right)
	\end{equation*}
\end{theorem}

This theorem bounds the expected cost-to-go of the best policy $ \hat{\pi} $ based on the empirical loss of the best policy in $ \Pi $ (i.e., $ \hat{\epsilon}_{min} $) and the empirical average regret of the learner (i.e., $ \hat{\epsilon}_{regret} $). Similar to DAGGER~\cite{ross2011reduction}, the second term can be eliminated if a no-regret algorithm such as Follow-the-Leader~\cite{hazan2007logarithmic} is used. The third term implies the number of training examples $ KN $ needs to be $\mathcal{O}\left(T^2l^2_{max}\log{\frac{1}{\mu}}\right) $ in order to become negligible. We can achieve $\mathcal{O}\left(T^2l^2_{max}\log{\frac{1}{\mu}}\right) $ samples easily as ADAPS uses principled simulations for generating them. With these changes, this theorem can lead to the following Corollary.

\begin{corollary}
	If $ l $ is convex in $ \pi $ for any $ s $ and it upper bounds $ C $, and Follow-the-Leader is used to select the learned policy, then for any $ \epsilon > 0 $, after collecting $ \mathcal{O}\left(\frac{T^2l^2_{max}\log{\frac{1}{\mu}}}{\epsilon^2}\right) $ training examples, with probability at least $ 1 - \mu $, $\mu \in(0,1)$, ADAPS offers the following guarantee:
	\begin{equation*}
	J\left(\hat{\pi}\right) \le J\left(\bar{\pi}\right) \le T\hat{\epsilon}_{min} + \mathcal{O}\left(\epsilon\right)
	\end{equation*}
\end{corollary}

\noindent Now, we only need the training error $ \hat{\epsilon}_{min} $ to be minimal, which is achievable via standard supervised learning. 

\subsubsection{Experimental Results}
I have tested my method empirically in three simulated scenarios, namely a straight road (representing a linear geometry), a curved road (representing a non-linear geometry), and an open ground. The straight and curved roads represent an \emph{on-road} scenarios with a static obstacle in the form of a traffic cone. The open ground represents an \emph{off-road} scenario with a dynamic obstacle in the form a vehicle.

I have compared my policy to the technique from Bojarski et al.~\cite{Bojarski2016}, as it is one of the representative approaches for end-to-end autonomous driving. Usually, this type of approach is limited to single-lane following~\cite{chen2015deepdriving,Xu2017end,zhang2017query,codevilla2017end} or \emph{off-road} collision avoidance~\cite{LeCun2006off} behaviors.

For the two on-road scenarios, I collect training datasets from \emph{straight road with or without an obstacle} and \emph{curved road with or without an obstacle}. This separation enables six policies for testing the effectiveness of \emph{learning from accidents}:

\begin{itemize}
	\item My policy: trained with just the lane-following data $O_{follow}$; $O_{follow}$ additionally trained after analyzing an accident on the straight road $O_{straight}$; and $O_{straight}$ additionally trained after analyzing an accident on the curved road $O_{full}$.
	\item Similarly, for the policy from Bojarski et al.~\cite{Bojarski2016}: $B_{follow}$, $B_{straight}$, and $B_{full}$.
\end{itemize}

Being evaluated first are $O_{follow}$ and $B_{follow}$. The evaluation strategy is to run these two policies on both the straight and curved roads, and count how many laps (out of 50) can be safely finished. The results showing both policies can finish all 50 laps safely. Then, I add an obstacle to the straight road and test both policies again. Since both policies have not trained on the accident data yet, they both run into the obstacle and cause an accident. 

After analyzing the occurred accident and incorporating the accident data into training, I obtain two new policies $O_{straight}$ and $B_{straight}$. As expected, $O_{straight}$ avoids the obstacle, while $B_{straight}$ continues to cause collision.

I proceed to add an obstacle to the curved road, after a similar process, I obtain $ O_{full}$ and $ B_{full} $. Again, as expected, $ O_{full}$ manages to perform both lane-following and collision avoidance behaviors. $B_{full}$, in contrast, leads the vehicle to an accident. These two cases are illustrated in Figure~\ref{fig:examples-intro} LEFT and CENTER.

\begin{figure}[th]
	\centering
	\includegraphics[width=\textwidth]{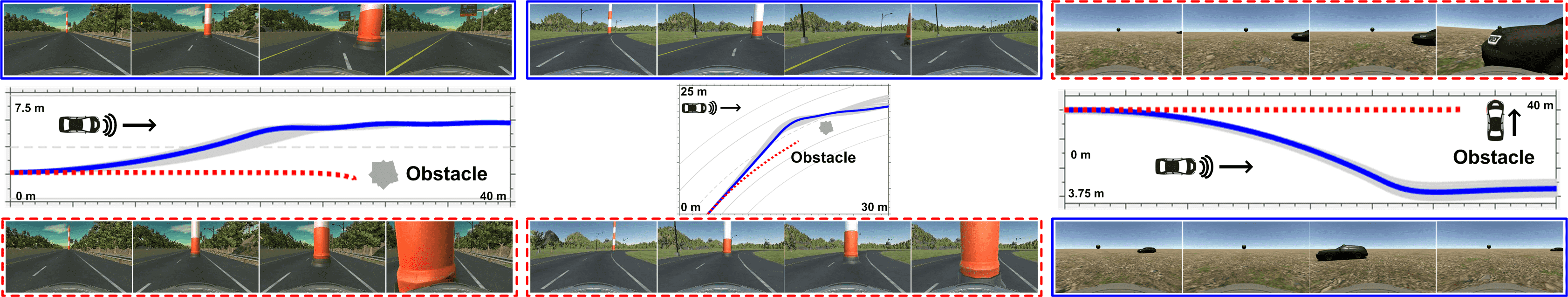}
	\caption{LEFT and CENTER: the comparisons between my policy $ O_{full} $ (TOP) and Bojarski et al.~\protect\cite{Bojarski2016}, $ B_{full} $ (BOTTOM). $B_{full} $ causes collision while $ O_{full} $ steers the AV away from the obstacle. RIGHT: the accident analysis results on the open ground. I show the accident caused by an adversary vehicle (TOP); then I show, after additional training, the AV can avoid the adversary vehicle (BOTTOM).}
	\label{fig:examples-intro}
\end{figure}

In order to test the generalization of my policy, I uniformly sampled 50 positions to place the obstacle on a $3$ meters line segment that is perpendicular to the direction of a road and in the same lane as the vehicle. The success rate (i.e., avoid an obstacle and resume normal driving) are documented in Table~\ref{tb:acc-intro}.

\begin{table}[ht!]
\centering
\small
\tabcolsep=0.1cm
\scalebox{1}{
	\begin{tabular}{ccccccc}
		\toprule
		& \multicolumn{6}{c}{Test Policy and Success Rate (out of 50 runs)}  \\        
		\cmidrule(l){2-7}       
		Scenario & $ B_{follow} $ & $ O_{follow} $  & $ B_{straight} $ & $ O_{straight} $ & $ B_{full} $ & $ O_{full} $   \\
		\midrule
		Straight road / Curved road & 100\% & 100\%  & 100\%  & 100\%  & 100\%  & 100\% 
		\\
		\midrule
		Straight road + Static obstacle & 0\% & 0\%  & \textbf{0\%}  & \textbf{100\%}  & \textbf{0\%}  & \textbf{100\%}
		\\
		\midrule
		Curved road + Static obstacle & 0\% & 0\%  & 0\%  & 0\%  &\textbf{ 0\% } & \textbf{100\%}
		\\
		\bottomrule
	\end{tabular}}
	\caption{Test Results of On-Road Scenarios: my policies $ O_{straight} $ \& $ O_{full} $ can lead to robust collision avoidance and lane-following behaviors. }
	\label{tb:acc-intro}
\end{table}

For the off-road scenario, I first train the AV to head towards a green sphere as its target. Then, I scripted an adversary vehicles to collide with the AV on its default course. By having and addressing the accident, my policy can steer the vehicle away from the adversary vehicle and resume its direction to the target. This is illustrated in Figure~\ref{fig:examples-intro} RIGHT.



The key to rapid improvements of a policy is the generation of sufficient and heterogeneous training data. In Figure~\ref{fig:intro-tsne}, I show the visualization results of images collected via my method and DAGGER~\cite{ross2011reduction} in one learning iteration. By progressively increasing the number of sampled trajectories, my method results in much more heterogeneous training data, which, when produced in a large quantity, can greatly facilitate the update of a policy.

\begin{figure}[th]
	\centering
	\includegraphics[width=0.7\columnwidth]{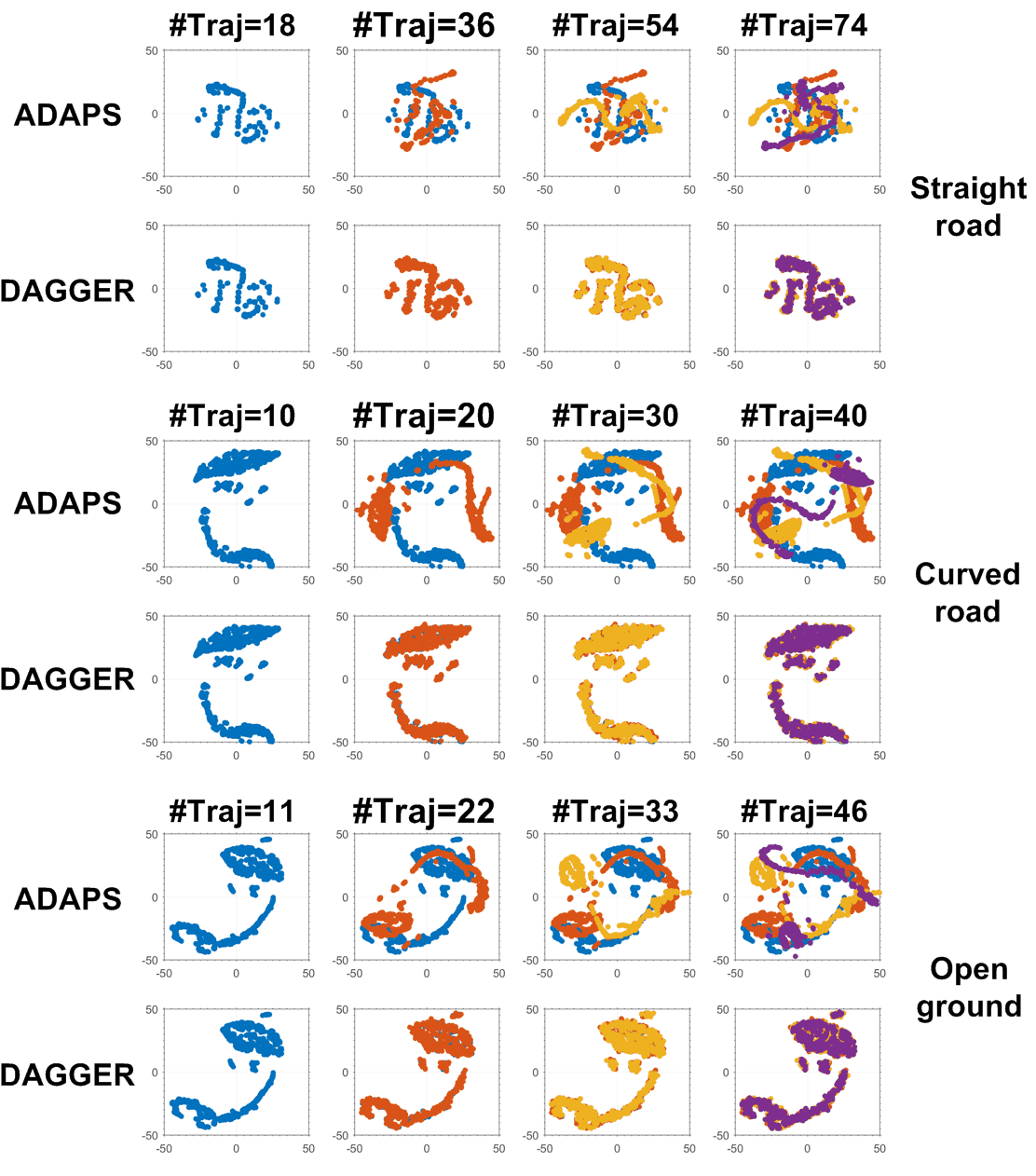}
	\caption{The visualization results of collected images using t-SNE~\protect\cite{maaten2008visualizing}. My method generates more heterogeneous training data compared to DAGGER~\protect\cite{ross2011reduction} in one learning iteration. }
	\label{fig:intro-tsne}
\end{figure}

\section{Thesis Statement}
\label{sec:ThesisStatement}
My thesis statement is as follows: 

\emph{Urban mobility can benefit from better designed intelligent transportation systems which can be further improved with 1) accurate and efficient reconstruction of city-scale traffic at the macroscopic level and 2) enhanced control and learning mechanisms for individual vehicles at the microscopic level.}

To support this thesis, at the macroscopic level, I have developed methods to efficiently estimate traffic conditions, accurately interpolate spatial-temporal missing traffic data, dynamically reconstruct traffic flows, and produce visual analytics in various forms. At the microscopic level, I have developed a framework to simulate and analyze various driving scenarios while automatically producing labeled training data, which, when combined with an efficient online learning mechanism and an effective policy architecture, can lead to robust control policies for autonomous driving.

\section{Organization}

The remainder of this dissertation is organized as follows. In Chapter~\ref{ch:itsm}, I describe my approach to deterministically estimate traffic conditions and interpolate temporal missing traffic data. In Chapter~\ref{ch:siga}, I describe my approach to iteratively estimate traffic conditions and interpolate spatial missing traffic data. In Chapter~\ref{ch:adaps}, I present ADAPS, a framework for obtaining robust control polices for autonomous driving. Finally in Chapter~\ref{ch:Conclusion}, I conclude this dissertation with discussion of potential future research directions.

%% file: chapters/ITSM.tex
\chapter{CITYWIDE ESTIMATION OF TRAFFIC DYNAMICS VIA SPARSE GPS TRACES}
\label{ch:itsm}

\input{itsm/sections/1-intro.tex}

\input{itsm/sections/2-related.tex}

\input{itsm/sections/3-reconstruction.tex}
\input{itsm/sections/4-missingvalue.tex}

\input{itsm/sections/5-dynamics.tex}
\input{itsm/sections/6-conclusion.tex}

%% file: itsm/sections/1-intro.tex
\section{Introduction}
\label{sec:itsm-intro}

Traffic is ubiquitous in modern cities, impacting their social, economic, and environmental developments. However, ever-present gridlock and congestion keep challenging transportation researchers and urban planners. According to the 2015 Urban Mobility Scorecard~\cite{schrank2015}, traffic congestion causes an extra 6.9 billion travel hours and 3.1 billion gallons of fuel consumption annually in the United States, which costs are approximately \$160 billion. As an increasing number of metropolitan areas experience severe traffic conditions and the overall cost is estimated more than one trillion U.S. dollars worldwide, the ability to analyze and understand traffic dynamics is becoming crucial.

In order to understand traffic congestion, first, we need to obtain its measurements. Traditionally, traffic measurements are collected via in-road sensors such as loop detectors and video cameras~\cite{leduc2008road}. While these sensors produce relatively accurate records, the high expenditures for installation and maintenance prevent them from being deployed over an entire city, rather than major roads and highways. Consequently, the lack of sensing infrastructure for arterial streets---which comprise the majority of a city---has made the large-scale traffic measuring task difficult.

Mobile data such as GPS traces, in contrast, are more promising sources for understanding citywide traffic dynamics due to their much boarder coverage. However, such data are limited in two aspects: 1) inevitable errors in measurement and transmission often yield reported locations off the road, and 2) due to energy and privacy concerns, GPS data commonly have a \emph{low sampling rate}, meaning that the time difference between consecutive points can be large (e.g., greater than 60 seconds), and a \emph{low penetration rate}, meaning that only a small percentage of traffic population is willing to send location reports.

Together, these limitations of GPS data bring a large degree of uncertainty into using them for studying and estimating traffic dynamics, for which purpose, several processing steps are required. First, off-the-road GPS points need to be mapped onto the road network, and the true traversed paths of vehicles need to be inferred. This process is called \emph{map-matching}. Second, the time taken to travel each road segment must be accurately estimated. Because of the \emph{low sampling rate}, the inferred path between two consecutive GPS points is likely to consist of multiple road segments, but only the aggregate travel time (i.e., the difference between the GPS timestamps) is known. The aggregate travel time needs to be distributed to individual road segments, which process is named \emph{travel time allocation}. Third, in order to understand the full traffic dynamics of a city, traffic data are needed for an entire traffic period for each road segment of a city's road network. However, GPS data often do not provide complete temporal coverage as they are commonly scarce in late night and early morning hours. The process for interpolating the missing temporal information is named \emph{missing value completion}.

Many efforts have been made towards improving the effectiveness of the abovementioned three processing steps. To be specific, the low sampling rate introduces issues: two consecutive GPS points are likely far apart from each other and multiple paths can exist for connecting them (especially in a complex urban environment). Thus, inferring the true traversed path between them is challenging. Many state-of-the-art \emph{map-matching} approaches use the shortest-distance criterion~\cite{lou2009map,yuan2010interactive,miwa2012development,hunter2014path,chen2014map,quddus2015shortest} to infer the traversed path. While this criterion is viable when a road network is under or close to a free-flow condition, it can introduce errors in a congested environment where other paths (not the shortest-distance path) can be traveled with less time and will be preferred by GPS devices and most drivers. If we have a wrongly inferred path, the timestamp difference (i.e., aggregate travel time) will be distributed to a wrong set of road segments. In other words, the introduced errors will be carried over to the subsequent step \emph{travel time allocation}, causing the overall estimation of traffic conditions deteriorated. 

I have developed a novel framework to estimate citywide traffic dynamics using GPS data. My framework is developed based on two observations: 1) traffic patterns exhibit weekly periodicity~\cite{hofleitner2012large}, and 2) traffic conditions are quasi-static~\cite{hunter2014path}. Based on these observations, I treat one week as a traffic period and assume the traffic conditions within each hourly time interval of a weekly period as static. Computationally, my framework consists of two phases. The first phase is conducted on individual time intervals and the second phase is performed over all time intervals. In the first phase, based on Wardrop's Principles~\cite{wardrop1952road}, I use the shortest travel time criterion instead of the shortest distance criterion to perform map-matching. Along with a travel time allocation technique adapted from Hellinga et al.~\cite{hellinga2008decomposing}, a novel computation scheme is designed to reconstruct traffic dynamics of a road network. In the second phase, exploiting the sparsity embedded in traffic patterns, I have developed a novel method based on Compressed Sensing~\cite{1614066,1580791} to interpolate missing travel information over an entire traffic period. The overview of my framework pipeline is shown in Figure~\ref{fig:itsm-system}.

\begin{figure}
	\centering
	\includegraphics[width=\textwidth]{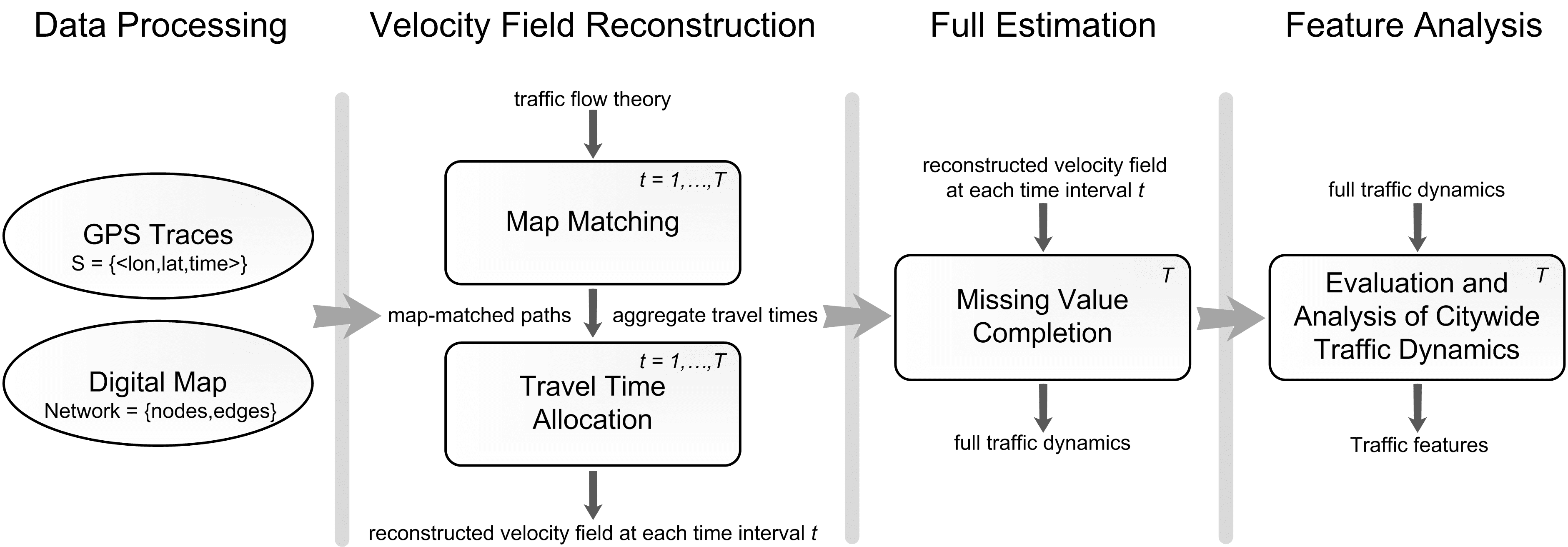}	\caption{Pipeline of my framework. Map Matching and Travel Time Allocation are applied on individual time intervals, while Missing Value Completion is performed over all time intervals.}
	\label{fig:itsm-system}
\end{figure}

I have extensively evaluated and tested the effectiveness of my approach and compared it to the method developed by Lou et al.~\cite{lou2009map}, using GIS data from a synthetic road network and the city of San Francisco~\cite{cabspotting}. The results demonstrate major improvements over the existing method~\cite{lou2009map} in various steps during the estimation process. In summary, the contributions of this work are the following: 

\begin{itemize}
	\item A novel perspective (switching from the shortest-distance criterion to the shortest travel-time criterion) in addressing sparse GPS traces during map-matching;
	\item An improved travel time allocation technique which incorporates estimated intermediate traffic conditions into computation;
	\item An efficient and robust method for interpolating missing traffic measurements in traffic patterns.
\end{itemize}

The rest of this chapter is organized as follows. In Section~\ref{sec:itsm-related}, I discuss the related work and highlight the differences between my work and existing studies. In Section~\ref{sec:itsm-reconstruction}, I take a holistic view of map matching and travel time allocation, detailing the process of reconstructing the velocity field of a road network. In Section~\ref{sec:itsm-missingvalue}, I explain how to interpolate missing travel information. In Section~\ref{sec:itsm-dynamics}, I show the estimated traffic dynamics of San Francisco and the analysis of the key features of the recovered traffic patterns. Finally, I conclude with a discussion of future work in Section~\ref{sec:itsm-conclusion}.

%% file: itsm/sections/2-related.tex
\section{Related Work}
\label{sec:itsm-related}

Over the last few decades, the estimation of traffic conditions has gained increasing scholarly attention~\cite{celikoglu2007dynamic,celikoglu2009node,gao2012modeling,AbadiRajabiounIoannou2015,kachroo2016travel,agarwal2016dynamic}.  Early studies on traffic estimation have focused on traffic states on highways using accurate measurements from stationary sensors such as loop detectors and video cameras~\cite{leduc2008road}. Recent advancements have shifted to combining multiple data sources and traffic simulation models for achieving better estimation results~\cite{work2010traffic,sun2014distributed,gning2011interval,li2014multimodel,celikoglu2016extension,hajiahmadi2016integrated}. However, the scenarios of interest in these  studies were limited to road segments with lengths of a few kilometers. 

The increasing availability of GPS data provides new means for conducting large-scale estimation of traffic conditions. However, as GPS data are inherently noisy, the estimated traffic conditions usually do not satisfy the flow conservation requirement assumed by many simulation models~\cite{phan2011interpolating,kong2013efficient,zhang2013aggregating}. As a result, new studies that consist of several steps for the estimation task, are emerging.

The first step, \emph{map-matching} addresses the problem of mapping off-the-road GPS points onto a road network and identifies the true traversed path between consecutive GPS points. However, GPS data could contain a \emph{low sampling rate}, which causes points to be far away from each other and making the selection among multiple paths connecting the points difficult. In order to determine the ``actual'' traversed path of a vehicle, a common approach is to use the shortest-distance criterion to connect two GPS points on a road network~\cite{lou2009map,yuan2010interactive,miwa2012development,hunter2014path,chen2014map,quddus2015shortest}. Nonetheless, the shortest-distance assumption can lead to errors in a congested environment, in which alternative paths can be traveled faster than the shortest-distance path and preferred by GPS devices and drivers. Essentially, the shortest-distance criterion only uses spatial information (i.e., longitude and latitude of GPS points, and the geometry of a road network), while ignoring the temporal information (i.e., timestamps) recorded in GPS reports. This happens primarily due to the travel times of a road network are largely unknown, causing the temporal information has nothing to be compared with~\cite{tang2015estimating,rahmani2013path}.

After \emph{map-matching}, travel time of individual road segments need to be estimated. To provide a few examples, Hellinga et al.~\cite{hellinga2008decomposing} have developed an analytical solution to estimate travel times of road segments using intuitive and empirical observations of traffic patterns in real life. Rahmani et al.~\cite{rahmani2015non} take a non-parametric approach, performing an estimation using a kernel-based method. Probabilistic frameworks have also been developed to conduct an estimation of traffic conditions~\cite{khosravi2011prediction,westgate2013travel,herring2010estimating,hofleitner2012learning,kuhi2015using}. While significant improvements have been achieved, these methods all perform the estimation steps sequentially, causing the limitations of \emph{map-matching} will be carried over to its subsequent steps and eventually deteriorate the overall estimation accuracy.

Researchers have also proposed solutions to \emph{missing value completion}. For example, tensor-based approaches~\cite{wang2014travel,asifmatrix}, which explore correlations among nearby road segments, have been developed. From Zhu et al.~\cite{zhu2013compressive} and Mitrovic et al.~\cite{mitrovic2015low}, algorithms based on Compressed Sensing have been proposed by taking an entire road network as the study subject. Interpolating missing values has also been addressed in an online setting~\cite{anava2015online}. Nevertheless, the abovementioned methods were not designed to tackle the problem of estimating \emph{full} traffic dynamics of individual road segments over an entire city---a subject for which little progress has been made~\cite{hofleitner2012large}.

%% file: itsm/sections/3-reconstruction.tex
\section{Traffic Velocity Field Reconstruction}
\label{sec:itsm-reconstruction}
I take a holistic view of \emph{map-matching} and \emph{travel time allocation}, and propose a noval method to reconstruct the velocity field of a road network. Starting with some definitions in this section, I then discuss methodologies and implementation details of my approach. My algorithm is evaluated and validated using a synthetic road network with microscopic traffic simulation.

\subsection{Notation and Definitions}
A \emph{road network} is defined as a directed graph $ G = (V,E) $ in which edges $ E $ denote road segments and nodes $ V $ represent intersections and terminal points. Each road segment $ e \in E $ contains several attributes: the length $ e.len $, the maximum/free-flow travel speed  $ e.v_{max} $, the minimum/free-flow travel time $ e.t_{min} = \frac{e.len}{e.v_{max}} $, and the maximum/jam density $ e.k_{max} $.

A \emph{path} from node $ g $ to node $ h $ on a network $ g \overset{\text{$ p $}}\leadsto  h  $ is a collection of road segments $ p = \{e_{1}, e_{2}, \dots, e_{n}\} $, where $ g $ is the starting node of $ e_{1} $ and $ h $ is the ending node of $ e_{n} $. A \emph{trace} is a sequence of GPS points $ S = \{ s_{1}, s_{2}, \dots, s_{n} \}$ in which each point is a tuple $ s_{i} = <s_{i}.x, s_{i}.y, s_{i}.t> $ containing longitude, latitude, and a timestamp. 

\subsection{Velocity Field Estimation}
Given the periodicity of traffic patterns over a week~\cite{hofleitner2012large}, I study traffic dynamics over the region of interest in a weekly period. I discretize one week into hourly time intervals and assume that traffic conditions remain the same within one hour on a road segment. For simplicity, I restrict my discussion of estimating the velocity field to one time interval (i.e., one hour). The process can be trivially extended to cover other time intervals of an entire traffic period. 

Ideally, if the actual traversed path of a vehicle is known and the generated GPS points are exactly on the road, I can derive the average travel speed of a path $ p $ that connects GPS points $ s_{i} $ and $ s_{i+1} $ as $ p.t = \frac{\sum_{e \in p} e.len}{s_{i+1}.t - s_{i}.t}$. However, GPS points are often off-the-road due to inherent measurement and sensing errors, and the underlying path of a vehicle is also unknown. To address these issues, a number of candidate nodes of the network are considered for mapping a GPS point, based on their distances to the point. Then, one of the paths connecting a pair of candidate nodes of two consecutive GPS points is selected to represent the actual path. As mentioned earlier, one common approach for choosing such a path is taking the shortest distance criterion~\cite{lou2009map,yuan2010interactive, miwa2012development,hunter2014path,chen2014map,quddus2015shortest}, which can produce errors in a congested environment (an illustrative example is shown in Figure~\ref{fig:fail}).

\begin{figure}
	\centering
	\includegraphics[width=\textwidth]{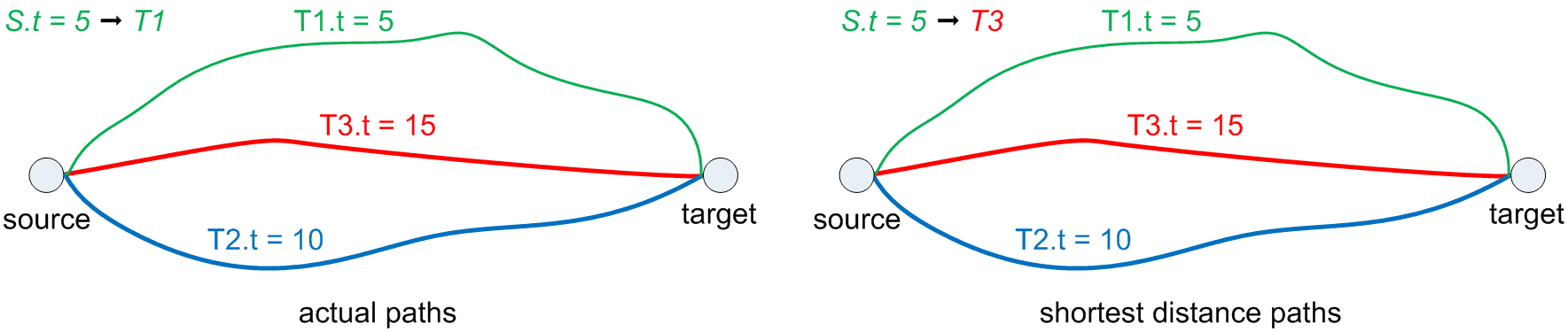}
	\caption{An example illustrating a failure using the shortest distance criterion for map-matching trace $ S $ to $ T1 $, $ T2 $, or $ T3 $. LEFT: By matching travel time of trace $S$ and road conditions, the correct path $ T1 $ is identified. RIGHT: By only considering the shortest distance path, $S$ is mismatched to $T3$.}
	\label{fig:fail}
\end{figure}

According to Wardrop's Principles~\cite{wardrop1952road}, the traffic in congested networks would move in a way such that no vehicle can reduce its travel time by switching routes. This state is called \emph{user equilibrium}, and is a result of every vehicle non-collaboratively attempting to minimize its traveling cost---which commonly appears to be the travel time. Under such an equilibrium state, the average travel time is balanced (i.e., the same) for all users of the network. 

\noindent \textbf{Assumption 1.} Based on Wardrop's Principles~\cite{wardrop1952road} and the observation that modern GPS devices largely adopt the fastest route, I assume \emph{all GPS traces are planned using the shortest travel time criterion}. 

Denoting the network with ground-truth traffic conditions as $ G_{true} $, according to Assumption 1, all GPS traces represent the fastest routes planned on $ G_{true} $. However, since $ G_{true} $ is unknown, my goal is to use available GPS traces and the initial road network $ G_{est} $, with all road segments set to their speed limits, to estimate $ G_{true} $, i.e., reconstruct the velocity field of $ G_{true} $. 

\begin{corollary}
A GPS trace with travel time $ t $ matching a path $ g \overset{\text{$ p $}}\leadsto  h  $ implies that no path in $ G_{true} $ from $ g $ to $ h $ has a travel time smaller than $ t $. 
\label{cor:itsm1}
\end{corollary}

\begin{proof}
According to Assumption 1, the traffic along every traveled route between $ g $ and $ h $ in $ G_{true} $ is in user equilibrium, meaning all routes between $ g $ and $ h $ have the equivalent lower-bound travel time. 
\end{proof}

\begin{corollary}
A pair of GPS points from a trace matched to locations $ g $ and $ h $ are sufficient to bound the travel time for all paths connecting $ g $ and $ h $. 
\label{cor:itsm2}
\end{corollary}

\begin{proof}
According to Assumption 1, the path indicated by a GPS trace is a time-optimal path (i.e., the shortest travel-time path). Therefore, it has an optimal substructure, in which a subset of two is time-optimal and provides a bound on the travel time for paths from $ g $ to $ h $. 
\end{proof}

As Collorary~\ref{cor:itsm2} implies, during the reconstruction of a velocity field, I can inspect two consecutive GPS points at a time (due to the optimal substructure). Consider two arbitrary GPS points $ s_{i} $ and $ s_{i+1} $, the true path $ s_{i} \overset{\text{$ p_{true} $}}\leadsto  s_{i+1}  $ is the fastest path between $ s_{i} $ and $ s_{i+1} $ on $ G_{true} $. Based on Collorary~\ref{cor:itsm1}, the travel time $ p_{true}.t = s_{i+1}.t - s_{i}.t $ is the lower bound of all travel times between the two points. This means if a path $ p_{est} $ has a higher travel time, the speed of the road segments on $ p_{est} $ should be decreased utill $ p_{est}.t \geq p_{true}.t$. I refer to such $ p_{est} $ as an \emph{overestimated} path. In practice, $ p_{est} $ connects candidate nodes of $ s_{i} $ and $ s_{i+1} $ rather than $ s_{i} $ and $ s_{i+1} $ themselves, and there exists a set of paths $  \{p_{est}\}_{all} $ for $ s_{i} $ and $ s_{i+1} $, in which one is the ``closest'' to $ p_{true} $. Denoting an arbitrary element in $  \{p_{est}\}_{all} $ as $ p_{est} $, if $ p_{est} $ and $ p_{true} $ happen to be the same path (i.e., containing the same set of road segments), $ p_{est}.t$ should be equal to $ p_{true}.t$, otherwise $ p_{est}.t $ should be larger than $ p_{true}.t$. Since there is no further information for me to estimate the excessive time of $ p_{est}.t $ over $p_{true}.t$, I conservatively set $ p_{est}.t = p_{true}.t$. I refer to this process of addressing overestimated paths in $ \{p_{est}\}_{all} $ as \emph{relaxation}. 

The relaxation will make all paths in $  \{p_{est}\}_{all} $ having the same travel time (i.e., $ p_{true}.t $). However, it is difficult to deterministically derive the ``closest'' path to $ p_{true}.t $ using a single GPS trace. To remedy this issue, I rely on the ``collective intelligence'' of multiple GPS traces that share a partial or full set of road segments. These segments will become gradually updated during relaxation of each GPS trace and eventually assist in differentiating the paths that include them from other paths in terms of the travel time. An illustration of this process is shown in Figure~\ref{fig:itsm-relaxation}. The relaxation is essentially a fulfillment of Wardrop's Principles, and is conducted in a greedy fashion: I repeatedly extract the fastest path in $  \{p_{est}\}_{all} $ and relax it, until no path in $  \{p_{est}\}_{all} $ has its travel time smaller than $p_{true}.t$. Given that there may exist multiple paths connecting two nodes in a network, causing the number of elements in $  \{p_{est}\}_{all} $ to be large, a sub-network with a specified radius is extracted from the entire network. This sub-network encompasses $ s_{i} $, $ s_{i+1} $, and their mapping candidates but no more. Using this approach, the number of paths in $  \{p_{est}\}_{all} $ is reduced. The rationale behind this choice is through empirical observations that a vehicle rarely takes an opposite direction or arbitrary long detour from a GPS point to the next one.

\begin{figure}
	\centering
	\includegraphics[width=\textwidth]{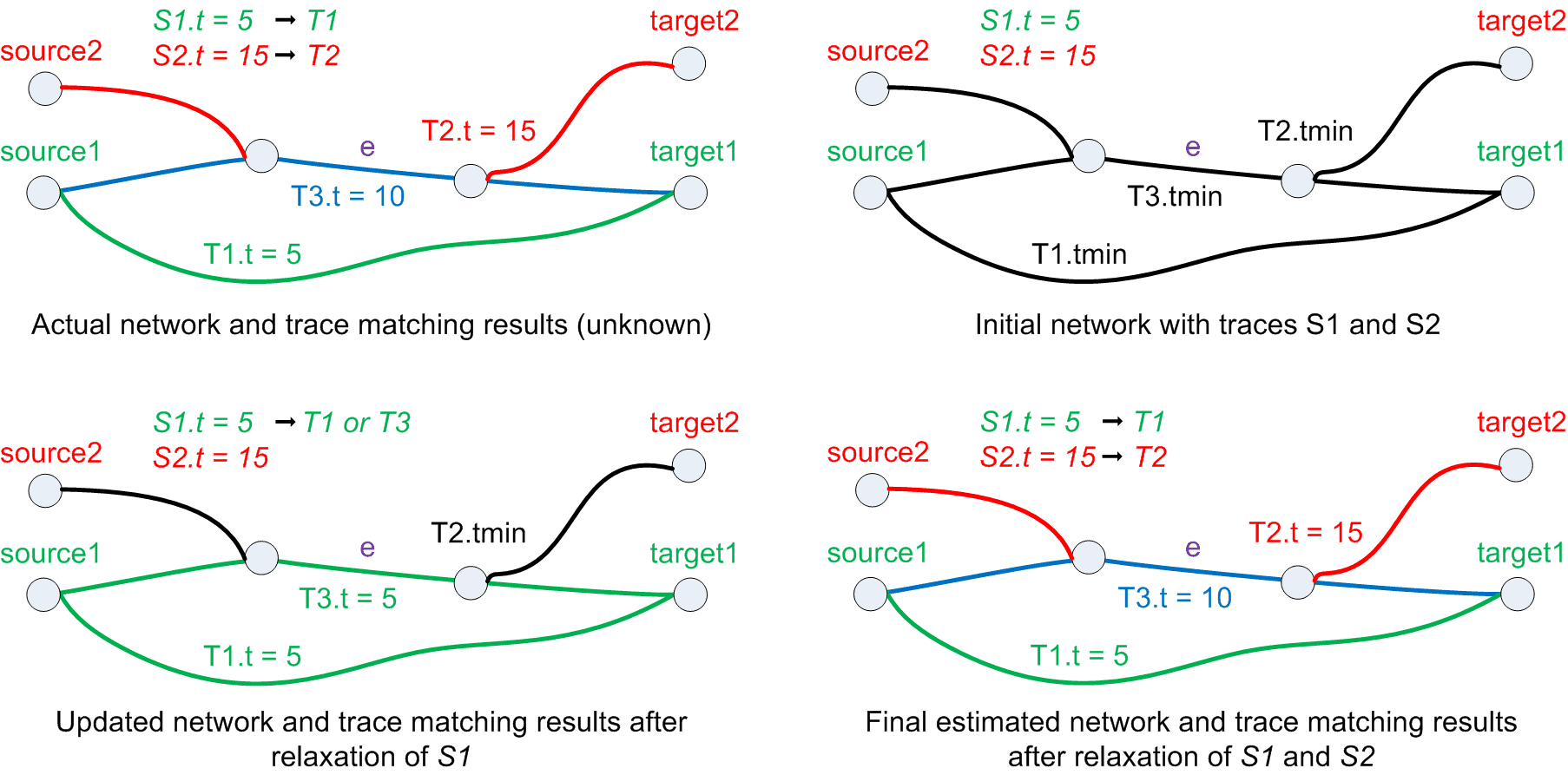}
	\caption{An illustration of the relaxation process of two traces $ S1 $ and $ S2 $. The actual traffic conditions and trace matching results are shown in the TOP-LEFT panel. The input listed in the TOP-RIGHT panel contains the initial network and trace information (i.e., sources, targets, and travel time). After relaxation of trace $ S1 $, paths $ T1 $ and $ T3 $ have their travel times increased to $ S1.t $. While $ S1 $ can be matched to either $ T1 $ or $ T3 $, the situation is resolved after relaxation of $ S2 $ due to further increase in the travel time of the road segment $ e $.}
	\label{fig:itsm-relaxation}
\end{figure}

\begin{theorem}
The speed of a road segment is monotonically decreasing during relaxation.
\label{thm:spd_road}
\end{theorem}

\begin{proof}
I prove this theorem by contradiction. Assume two overestimated paths $ p_{est}^{1} $ and $ p_{est}^{2} $ share a road segment $ e $. Without loss of generality, consider the relaxation of $ p_{est}^{1} $, the road segment $ e $ has its speed $ e.v $ decreased to $ e.v' $ and travel time $ e.t $ increased to $ e.t' $ so that $ p_{est}^{1}.t = p_{true}^{1}.t$. During relaxation of $ p_{est}^{2} $, $ e.v' $ and $ e.t' $ are subject to change. If instead of monotonically decreasing, $ e.v' $ gets increased and $ e.t' $ gets decreased, then $ p_{est}^{1}.t < p_{true}^{1}.t$, which invalidates the previous relaxation process and further contradicts Collorary~\ref{cor:itsm1}. 
\end{proof}

Taking advantage of Theorem~\ref{thm:spd_road}, further reduction in computation can be achieved by retaining reduced speed of each path in $ \{p_{est}\}_{all} $. To be specific, as $ \{p_{est}\}_{all} $ is generated for $ s_{i} $ and $ s_{i+1} $ in a sub-network, many paths in $ \{p_{est}\}_{all} $ will have shared road segments. Therefore, the speed reduction in these road segments will make multiple paths in $ \{p_{est}\}_{all} $ to have increased travel times. As a result, the greedy process of relaxation is much more efficient than the brute-force enumeration.

\begin{algorithm}
{\bf Input}: Initial estimated road network $ G_{est} = (V,E)$ with $ e.v = e.v_{max}$,  $\forall e \in E$; GPS traces $S = \{S_{1}, \dots, S_{m}\} $; Discretized time intervals $ \{1, \dots, D\} $; Maximum distance for computing candidate nodes of a GPS point $ cDis $; Maximum number of candidate nodes $ cNum $  \\
{\bf Output}: Reconstructed road network $ G_{est} $
  \caption{Velocity Field Reconstruction}
  \begin{algorithmic}[1]
  \For{$ \text{each time interval } d \in (1, \dots, D) $} :
  \State $ S^{d} = ExtractGPSTraces(S,d) $ 
  \For {$ \text{each trace } S_{j}^{d} \in S^{d}$} :
  \For{$ \text{consecutive GPS points } s_{i},s_{i+1} \in S_{j}^{d}$} :
  \State $ radius = \frac{dist(s_{i}, s_{i+1})}{2}  + cDis $
  \State $ H = ExtractSubgraph(G_{est}, radius, s_{i},s_{i+1}) $
  \State $ C_{1} = GetCandidateNodes(H, s_{i}, cDis, cNum) $
  \State $ C_{2} = GetCandidateNodes(H, s_{i+1}, cDis, cNum) $
  \State $ p_{true}.t = s_{i+1}.t - s_{i}.t $
  \State $ H_{relax} = RelaxNetwork(H, p_{true}.t, C_{1}, C_{2}) $
  \State $ G_{est} = UpdateNetwork(G_{est},H_{relax}) $
  \EndFor
  \EndFor
  \EndFor
  \State \Return $ G_{est} $
  \end{algorithmic}
  \label{algo:main}
\end{algorithm}

The overall process is described in Algorithm~\ref{algo:main}. Subroutines \emph{RelaxNetwork} and \emph{Relaxation} are specified in Algorithm~\ref{algo:relaxnetwork} and Algorithm~\ref{algo:relaxation}, respectively. In particular, two types of paths are considered as outliers, which are excluded from the computation: one has a travel time shorter than its free-flow travel time and one has a travel time longer than the travel time under a jam density. The procedure \emph{TravelTimeAllocation} in Line~\ref{line:tta} of Algorithm~\ref{algo:relaxation} is discussed in the next section.

\begin{algorithm}
	{\bf Input}: Subgraph $ H $; True travel time $ p_{true}.t $; Candidate node sets $ C_{1} $, $ C_{2} $ \\
	{\bf Output}: Relaxed subgraph $ H_{relax} $
	\caption{RelaxNetwork}
	\begin{algorithmic}[1]
		\For{$ \text{each node } n_{1} \in C_{1} $} :
		\For{$ \text{each node } n_{2} \in C_{2} $} :
		 \If{$ n_{1} == n_{2} $ or $ distance(n_{1},n_{2}) == 0 $} 
		 \State {continue}  
		 \Comment{no valid path exists}{}
		 \EndIf		 
		\State $ p_{est} = GetFastestPath(H,n_{1},n_{2}) $
		 \If{$ IsOutlier(p_{est}) == true $} \label{line:outlier}
		 \State {continue}  
		 \EndIf		
		\While {$ true $} :
		\If{$ t_{est} \geq t_{true} $} 
		\State {break}  
		\Comment{not an overestimated path}{}
		\EndIf
		\If{$ NumberOfNodes(p) < 2 $} 
		\State {break}  
		\Comment{not a valid path}{}
		\EndIf
		\State $ E_{relax} = Relaxation(p_{est}, p_{true}.t) $		
		\State $ H = UpdateNetwork(H,E_{relax}) $
		\Comment{update travel times of $ H $ using $ E_{relax} $}{}
		\State $ p_{est} = GetFastestPath(H,n_{1},n_{2}) $
		\Comment{get the shortest travel time path between $ n_{1} $ and $ n_{2} $ on $ H $}{}
		\EndWhile
		\EndFor
		\EndFor
		\State 
		\Return $ H_{relax} = H $
	\end{algorithmic}	
	\label{algo:relaxnetwork}
\end{algorithm}

\subsection{Travel Time Allocation}
\label{sec:tta}

During relaxation, each overestimated path $ p_{est} $ between $ s_{i} $ and $ s_{i+1} $ needs to be addressed in order to achieve $ p_{est}.t = p_{true}.t = s_{i+1}.t - s_{i}.t $. Due to low sampling rate, $ p_{est} $ often consists of several road segments. The aggregate travel time $ p_{true}.t $ needs to be appropriately allocated to individual road segments of $ p_{est} $. To address this issue with respect to traffic flow analysis, I adapt the solution proposed by Hellinga et al.~\cite{hellinga2008decomposing}. 

According to Hellinga et al.~\cite{hellinga2008decomposing}, the travel time of a road segment $ e $ can be decomposed into three components: free-flow travel time $ \tau_{e,f} $, congestion time $ \tau_{e,c} $, and stopping time $ \tau_{e,s} $. For an overestimated path $ p_{est} $, I denote its total free-flow travel time as $ T_{f} = \sum_{e \in p_{est}}\tau_{e,f} $, total congestion time as $ T_{c} = \sum_{e \in p_{est}}\tau_{e,c} $, total stopping time as $ T_{s} = \sum_{e \in p_{est}}\tau_{e,s} $, and the allocation time as $ T_{a} = p_{true}.t $.  To validate $ p_{est} $, the following relationship needs to be ensured: $ T_{f} + T_{c} + T_{s}  = T_{a} $. While $ T_{f} $ can be derived trivially as $ T_{f} = \sum_{e \in p_{est}}\tau_{e,f} = \sum_{e \in p_{est}}\frac{e.len}{e.v_{max}}$, the computations of $ T_{c} $ and $ T_{s} $ require additional considerations.

\begin{algorithm}
	{\bf Input}: A set of road segments $ E $ (initially contains all road segments in $ p_{est} $); Time budget $ T $ (initially set to $ p_{true}.t $) \\
	{\bf Output}: A set of relaxed road segments $ E_{relax} $
	\caption{Relaxation}
	\begin{algorithmic}[1]
		\State $ avgSpeed = \frac{\sum_{e \in E} \|e\|}{T} $
		\Comment{average travel speed}{}
		\State $ T' = T  $
		\Comment{store the time budget}{}
		\For{$ \text{each road segment } e \in E $} :
		\If{ $ e.v \leq avgSpeed $} 
		\State {$ T = T - e.t $} 
		\Comment{the leftover time budget}{}
		\State {$ E = E \setminus \{e\} $} 
		\Comment{exclude $ e $ based on Theorem~\ref{thm:spd_road}}{}
		\EndIf
		\EndFor
		\If{ $ T' \neq T $} 		
		\Comment{some $ e $ have been excluded}{}
		\State $ E_{relax} = Relaxation(E,T) $ 
		\Comment{recursive call}{}
		\Else 
		\State {$ E_{relax} = TravelTimeAllocation(E,T)$ } \label{line:tta}
		\EndIf
		\State
		\Return $ E_{relax} $
	\end{algorithmic}	
	\label{algo:relaxation}
\end{algorithm}

By assuming nearby road segments have similar traffic conditions, the \emph{path congestion level} is defined as $ w = \frac{T_{c}}{T_{c} + T_{f}} $. The minimum value $ w_{min} = 0 $ is reached when a path can be traveled using its free-flow speed, and the maximum value $ w_{max} = \frac{T_{c,max}}{T_{c,max} + T_{f}} $ where $ T_{c,max} = T_{a} - T_{f} $ is reached when a path is congested and $ T_{s} = 0 $. Using a specific path congestion level $ w_{min} \leq w \leq w_{max} $, the probability
that congestion occurred on $ p_{est} $ is computed as follows:

\begin{equation}
P_{c}(w) = min\{1, \frac{T_{c,max}^{prev} + T_{c,max}}{{T_{a}^{prev} + T_{a}}} \cdot \frac{1}{w}\},
\label{eq:congestion}
\end{equation}

\noindent where $ T_{c,max}^{prev} $ and $ T_{a}^{prev} $ represent the maximum congestion time and the allocation time of path $ s_{j} \overset{\text{$ p_{est}^{prev} $}}\leadsto  s_{j+1}, \, j+1 \leq i $. From Hellinga et al.~\cite{hellinga2008decomposing}, $ p_{est}^{prev} $ denotes the path that has an allocation time longer than the jam-density travel time with maximum possible $ j $. As I have excluded the path outliers,  $ p_{est}^{prev} $ indicates the path connecting $ s_{i-1} $ and $ s_{i} $, and is the directly preceding path of $ p_{est} $. $ P_{c}(w) $ is defined under assumptions: when all variables are fixed, the probability of a specific level of congestion occurring increases as $ T_{c,max} $ increases; given a particular $ T_{c,max} $, higher level congestion is less likely to occur than lower level congestion.

Next, the \emph{stopping likelihood function} is defined for computing the probability of stopping. Since the original formula does not take estimated traffic conditions into account, I alter it to the following: 

\begin{equation}
L_{s,e}(w) = \beta w + (1-\beta)\frac{e.k}{e.k_{max}},
\label{eq:stop}
\end{equation}

\noindent where $ \beta $ is the weighting factor (set to 0.5), and $ e.k $ is the estimated density as a result of possible previous relaxation (otherwise $ e.k = 0$). Equation~\ref{eq:stop} computes the likelihood by leveraging the \emph{path congestion level} $ w $ and the \emph{road-segment congestion level} $ \frac{e.k}{e.k_{max}} $. Intuitively, the road-segment congestion level becomes higher when the estimated density $ e.k $ approaches the jam density $ e.k_{max} $. In order to derive $ e.k $ from the estimated speed $ e.v $, I utilize the Greenshield's model~\cite{greenshields1935study}: 

\begin{equation}
e.k = e.k_{max}\left(1-\frac{e.v}{e.v_{max}}\right).
\end{equation}

\noindent By having the \emph{stopping likelihood function}, the probability that a vehicle stopped on a particular road segment is stipulated by assuming the vehicle stops at most once on $p_{est} $:

\begin{equation}
P_{s,e_{i}}(w) = L_{s,e_{i}}(w)\prod_{e_{j} \in p_{est}, j \neq i}(1- L_{s,e_{j}}(w)).
\label{eq:stopping}
\end{equation}

\noindent Using Equation~\ref{eq:congestion} and Equation~\ref{eq:stopping}, the congestion time on a single road segment is calculated by integrating all path congestion levels together: 

\begin{equation}
\tau_{e,c} = \int_{0}^{w_{max}} \tau_{e,f}\frac{w }{1-w} \frac{P_{c}(w) \sum_{e}P_{s,e}(w)}{Q_{s}} dw,
\end{equation}

\noindent where $ Q_{s} = \int_{0}^{w_{max}} P_{c}(w) \sum_{e}P_{s,e}(w) dw $ is the normalizing factor. After summing the congestion time of all road segments $ T_{c} = \sum_{e \in p_{est}} \tau_{e,c} $, the total stopping time can be derived easily: $ T_{s} = T_{a}-T_{f}-T_{c} $. Finally, I can calculate the stopping time on each road segment using the following formula: 

\begin{equation}
\tau_{e,s} = \int_{0}^{w_{max}} T_{s} \frac{P_{c}(w)P_{s,e}(w)}{Q_{s}} dw.
\end{equation}

\noindent To this point, I have solved the travel time allocation problem by having $ \tau_{e,f} $, $ \tau_{e,c} $, and $ \tau_{e,s} $ for all road segments of $ p_{est} $ and fulfilled $ T_{f} + T_{c} + T_{s}  = T_{a} $ (i.e., $p_{est}.t = p_{true}.t $).

\subsection{Evaluation on A Synthetic Network Using Traffic Simulation}
In order to evaluate my technique on the effectiveness of estimating a velocity field, I use a synthetic road network and an agent-based traffic simulator SUMO~\cite{krajzewicz2012recent}. The road network is modeled as a grid with $ 5 \times 5 $ intersections. By treating one hour as a time interval for a specific congestion level, a set of vehicles is routed in the network. The average travel times of all road segments are taken as the ground truth for this congestion level. All traces are simulated by randomly selecting nodes of the network as sources and targets using the fastest route strategy. As my method operates on size two subsets of GPS traces, I emit pairs of points at the source and target for each simulated trace to resemble the low sampling rate feature of real-world GPS data. This design choice enables the testing of the travel-time allocation algorithm developed in the previous section.

All road segments share the same setting: length of $150~m$, maximum speed at $17.88~m/s$, and a maximum density of $0.15~vehicles/meter$ . In total, 30 congestion levels are created by simulating 50 to 1500 vehicles with an increment of 50 vehicles per time interval. In addition, for each time interval, five levels of the vehicle population, namely at 20\%, 40\%, 60\%, 80\%, and 100\%, are sampled to generate GPS reports.

The first analysis is conducted by treating the network as a whole: the aggregate travel time over the entire network is used for evaluation. The comparison between the estimated values computed using my technique and Lou et al.~\cite{lou2009map}, respectively to the ground-truth values are shown in Figure~\ref{fig:nn-time}. Starting with 20\% of vehicles, my technique demonstrates close approximations to the ground truth at all congestion levels, while the other technique fails to achieve the same performance. Shown in Table~\ref{tab:nn-error}, I examine the influence of different traffic population percentages in estimation by computing the absolute error to the ground truth across all congestion levels. The smallest error (mean = $3.4~s$, std = $3.3~s$) is achieved using 80\% of GPS traces from the simulated traffic. The slight increase in errors when using 100\% GPS traces is mainly due to the stochastic aspect of the travel-time allocation method that I adapted. 

\begin{figure}
	\centering
	\includegraphics[width=\textwidth]{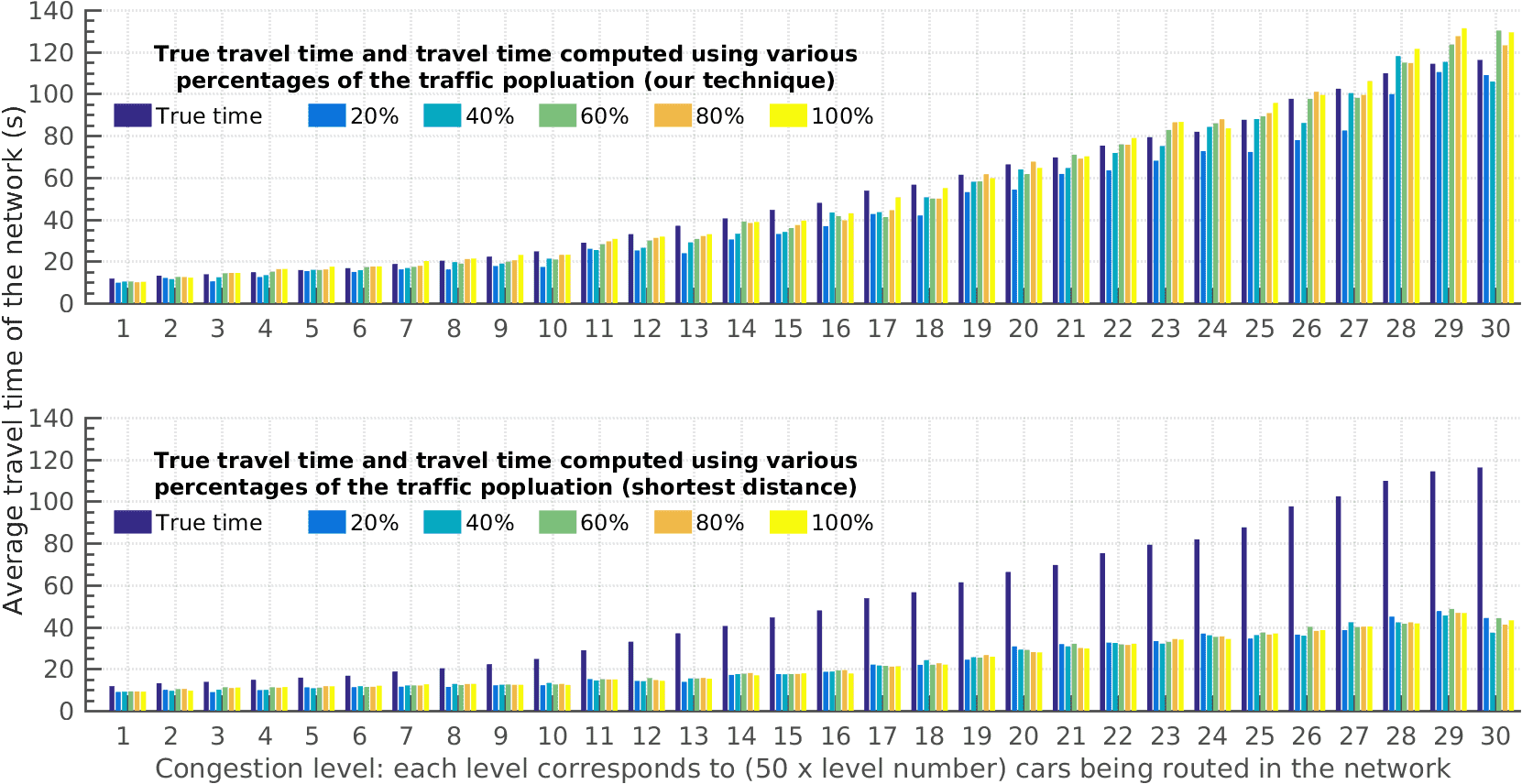}
	\caption{Recovery of average travel time on different percentages of the traffic population using my approach (TOP) vs. Lou et al.~\protect\cite{lou2009map} (which uses the shortest-distance criterion) (BOTTOM). My technique consistently outperforms Lou et al.~\protect\cite{lou2009map} in estimating the average travel time over 30 different congestion levels.}
	\label{fig:nn-time}
\end{figure}

\begin{table}[ht!]
	\centering
	\small
	\begin{tabular}{ccccc}
		\toprule
		& \multicolumn{4}{c}{Absolute errors to the ground truth}  \\
		\cmidrule(l){2-5}
		& \multicolumn{2}{c}{My technique}  & \multicolumn{2}{c}{Lou et al.~\cite{lou2009map}}  \\		
		\cmidrule(l){2-3} 	\cmidrule(l){4-5} 	
		Traffic percentage & mean ($s$) & std. ($s$) & mean ($s$) & std. ($s$)		\\
		\midrule
		20\% & 8.29 & 5.31 & 29.74 & 21.90 
		\\
		\midrule
		40\% & 4.25 & 3.38 & 29.85 & 22.66
		\\
		\midrule
		60\% & 3.67 & 3.67 & 29.33 & 22.06
		\\
		\midrule
		80\% & 3.40 & 3.32 & 29.55 & 22.46  
		\\
		\midrule
		100\% & 3.58 & 4.04 & 29.66 & 22.35 
		\\
		\bottomrule
	\end{tabular}
	\caption{The absolute errors in the recovered travel times computed using my technique vs. Lou et al.~\protect\cite{lou2009map}, using GPS traces from various percentages of the traffic population. My technique results in much smaller errors as of Lou et al.~\protect\cite{lou2009map}.}
	\label{tab:nn-error}
\end{table}

The absolute errors in the recovered travel time computed using my technique vs. Lou et al. (Lou
et al., 2009) by using GPS traces from various percentages of the traffic population. My technique results in
much smaller errors as of Lou et al. (Lou et al., 2009).

The second analysis computes the relative improvements measured in mean squared error (MSE) as $\frac{1}{n}\sum_{i = 1}^n (\hat{e.t} - e.t )^2 $ of my method over Lou et al.~\cite{lou2009map}. The results are shown in Figure~\ref{fig:tt-results}. As the congestion level increases or a higher percentage of GPS traces becomes available in estimation, my technique outperforms the other method. The improvements are less clear when a congestion level is low (i.e., $ < 10 $), but apparent when a congestion level is high (i.e., $ \geq 10 $). 

\begin{figure}
	\centering
	\includegraphics[width=\textwidth]{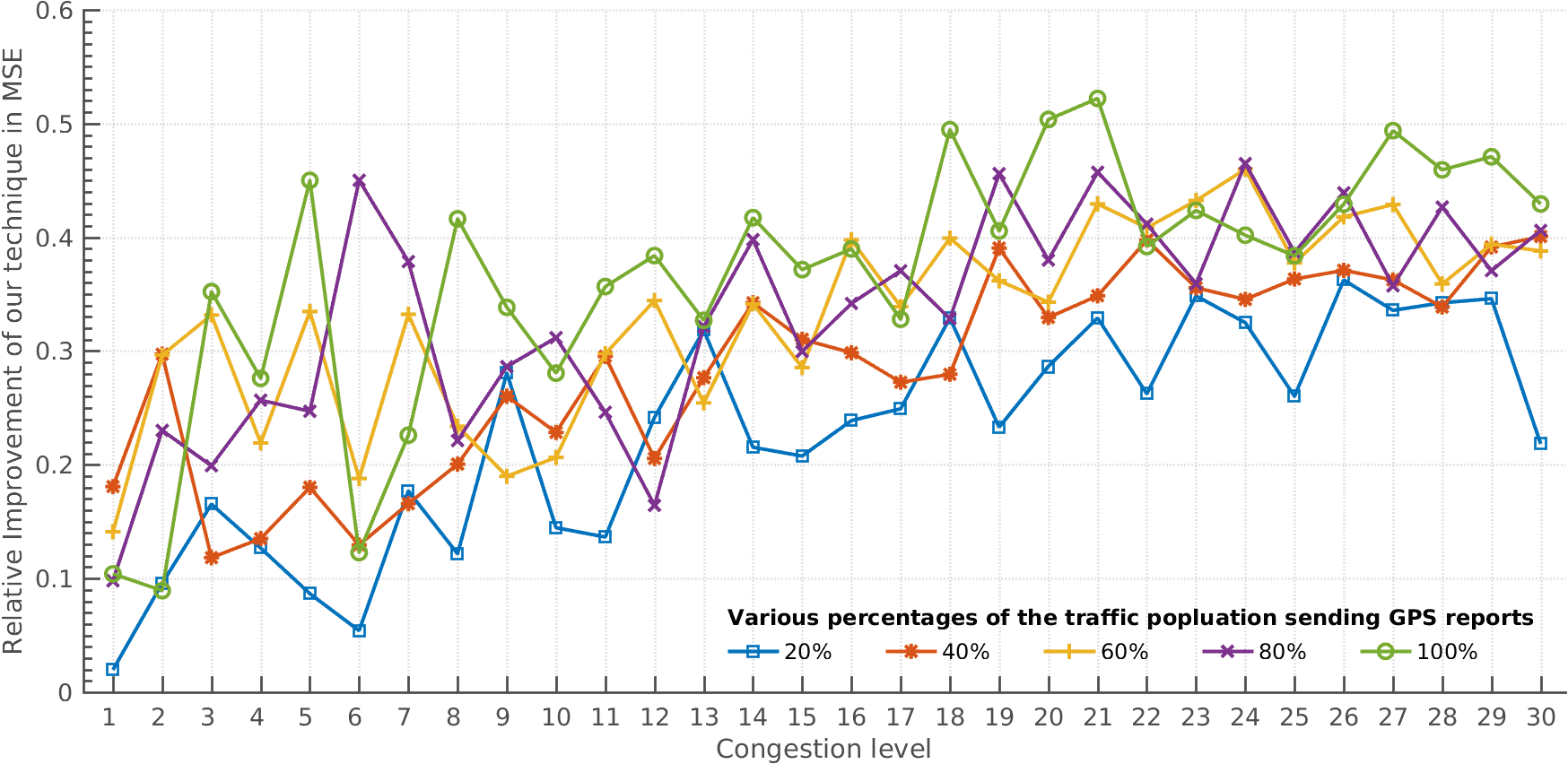}
	\caption{Relative improvements measured in MSE of my technique over Lou et al.~\protect\cite{lou2009map}. My technique outperforms Lou et al.~\protect\cite{lou2009map} as the congestion level increases or more GPS traces become available in estimation.} 
	\label{fig:tt-results}
\end{figure}

%% file: itsm/sections/4-missingvalue.tex
\section{Missing Value Completion}
\label{sec:itsm-missingvalue}

The temporal sparsity of GPS data indicates that traffic data can be missing for certain time intervals within a weekly period. The missing data inhibit us to accurately estimate full traffic dynamics. To address this issue, I explore sparsity embedded in traffic patterns and propose a novel technique based on the Compressed Sensing algorithm~\cite{1614066,1580791} for traffic reconstruction.

I have adopted loop-detector data\footnote{The loop-detector data are obtained from Caltrans Performance Measurement System (PeMS): \url{http://pems.dot.ca.gov/}.}, which represent complete and accurate measurements of traffic, to explore features of traffic patterns. I use speed measurements from 38 loop detectors installed in San Francisco for this purpose. The time range of these data is the same as that of the Cabspotting dataset~\cite{cabspotting} (GPS data) in San Francisco. I refer to the average speed measurements over a weekly period from a loop detector as a \emph{traffic signal}. An example is shown in the top panel of Figure~\ref{fig:loop-analysis}, which exhibits clear periodicity.

I performed a spectral analysis on a traffic signal by setting the frequencies as the reciprocal of the signal length (i.e., 1/168, an hourly interval within a weekly period) and subtracting the signal from its mean to make the oscillations easier to observe. The results, shown in the middle panel of Figure~\ref{fig:loop-analysis}, reveal that the period of the most salient oscillation is 24 hours (i.e., one cycle per day). In addition, the signal is sparse in the frequency domain, which is reflected as over 95\% energy is preserved by retaining its 10 largest frequencies (Figure~\ref{fig:loop-analysis}, BOTTOM).

\begin{figure}
	\centering
	\includegraphics[width=\textwidth]{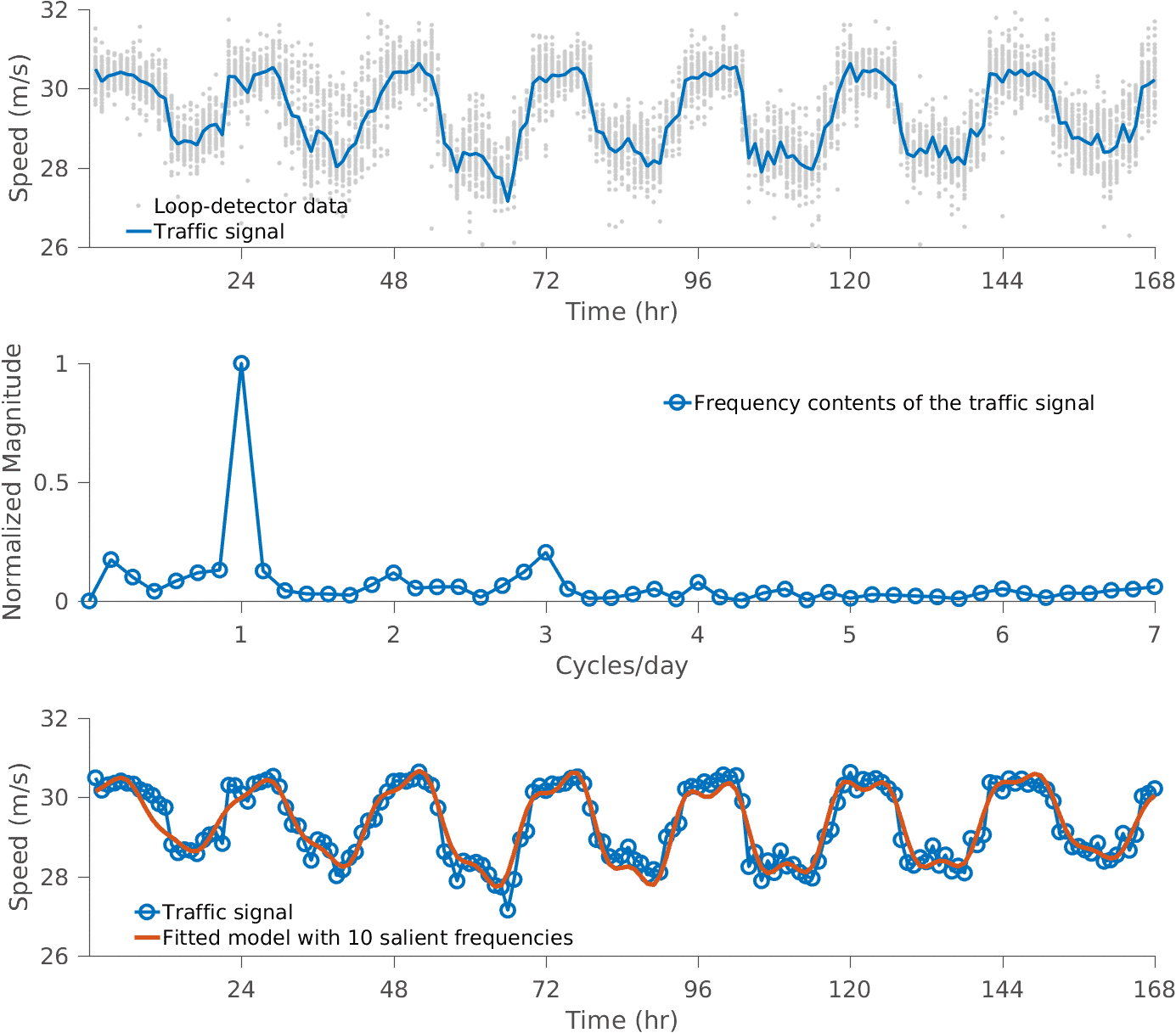}
	\caption{The average speed measurements from a loop detector are interpreted as a \emph{traffic signal}, which exhibits a clear periodic pattern (TOP); the spectral analysis reveals that the most prominent frequency is one cycle per day (MIDDLE); the traffic signal is approximated by a frequency-domain linear regression model, in which 95\% energy is retained by keeping the 10 largest frequencies (BOTTOM).}
	\label{fig:loop-analysis}
\end{figure}

According to the specification of a traffic signal, the highest frequency being supported is one cycle per two hours. To recover a traffic signal from its samples, the classical Nyquist-Shannon Theorem requires at least 168 measurements, which we do not have due to the temporal sparsity of GPS data. However, the Compressed Sensing algorithm~\cite{1580791,1614066} suggests that a signal can be recovered with a small set of samples if the signal has a sparse representation. This feature is manifested in Figure~\ref{fig:loop-analysis} and the top panel of Figure~\ref{fig:sparse}, which shows the rapid decaying of the sorted frequency magnitudes of traffic signals. On average, the decaying rate reaches 45.8\% with the most prominent frequency and 78.9\% with the 10 most prominent frequencies. After the $15$th frequency, the decaying rates become ineligible. 

Another merit of Compressed Sensing is that it does not require any prior knowledge of the sparse structure of a signal, such as the locations of large frequency components. This generality is ideal for recovering the traffic condition of a road segment, given that traffic is intrinsically stochastic and the sparse structure of a traffic signal varies from one road segment to another. In the bottom panel of Figure~\ref{fig:sparse}, I plot locations and amplitudes of the frequency components of all obtained traffic signals in my analysis: besides the appearance of the prominent oscillations at 24, 48, and 72 hours, other oscillations are spread out along the frequency axis. Together, these observations and features confirm the applicability of using Compressed Sensing on traffic signal recovery.

\begin{figure}
	\centering
	\includegraphics[width=\textwidth]{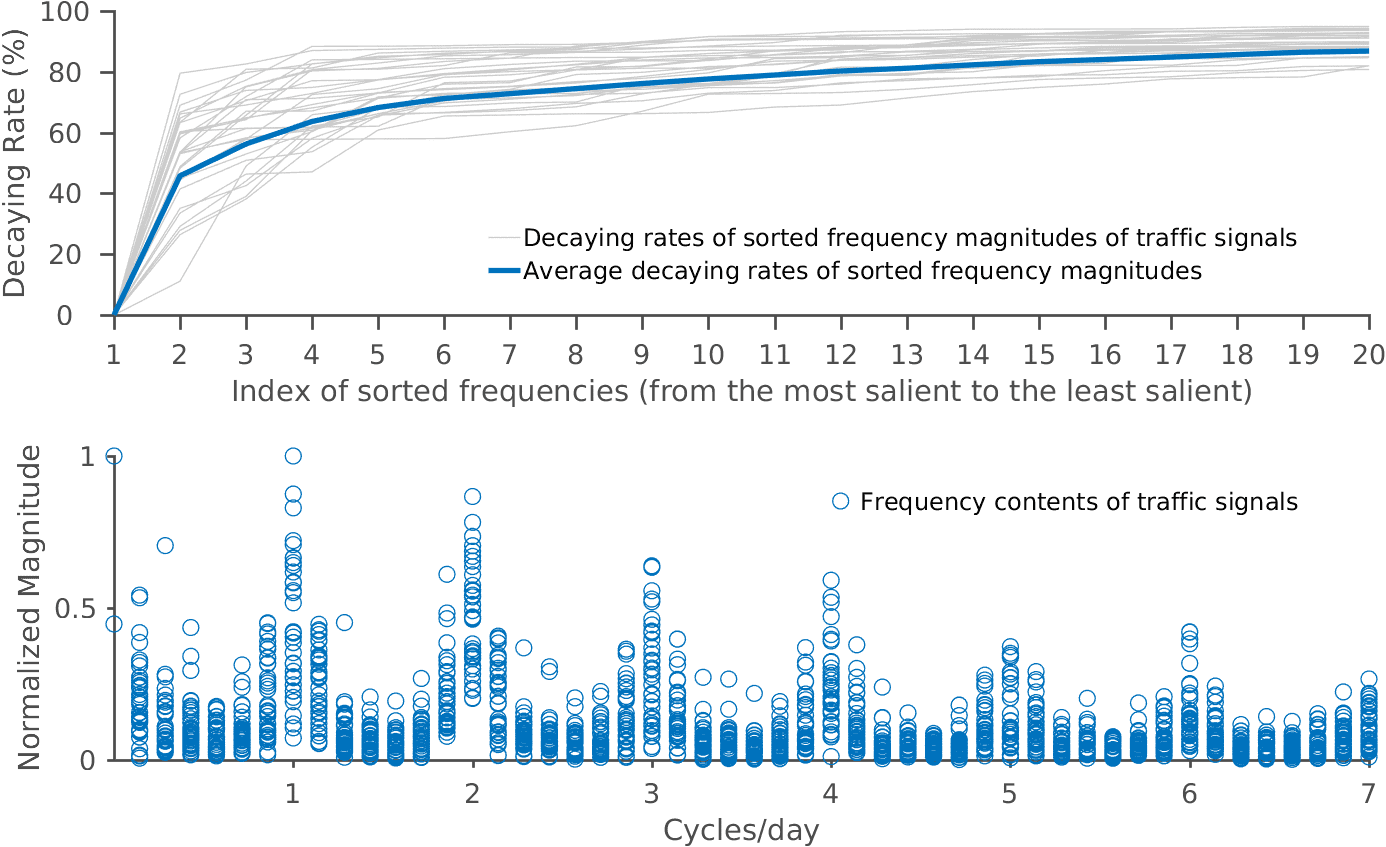}	
	\caption{The top panel shows the decaying rates of frequency magnitudes of all traffic signals; the bottom panel shows the locations and normalized magnitudes of the frequency components of all traffic signals. The rapid growth of decaying rates and randomly distributed frequency structures indicate that Compressed Sensing is applicable for recovering a traffic signal.}
	\label{fig:sparse}
\end{figure}

Given a signal $ f \in \mathbb{R}^n$ and its measurements $ b \in \mathbb{R}^m $, I consider the undersampled case in which the number of measurements $ m $ is smaller than the signal's dimension $ n $. The goal is to derive an estimated signal $ \hat{f} \in \mathbb{R}^n $ from $ b \in \mathbb{R}^m $ such that the error $ \lVert f - \hat{f} \rVert_{\mathcal{L}_{2}} $ is minimized. In general, the better the desired reconstruction quality, the more measurements are needed. In order to achieve a predefined accuracy level, signal reconstruction requires a minimum number of measurements $ m_{min} $. According to Cand\'{e}s and Wakin~\cite{candes2008introduction}, $ m_{min} $ is on the order of $ \mu^2 S \log(n) $, where $ \mu $ is the coherence between a measurement basis and a representation basis, $ S $ is the signal's sparsity level, and $ n $ is the signal's dimension (in this chapter $ n=168 $). I estimate $ S $ by averaging the number of frequencies in preserving 95\% energy of all traffic signals, which results in $ S=17.63 $. The minimum coherence value $ \mu = 1 $ is obtained by performing a discrete cosine transform (DCT) on a traffic signal $ f $: 

\begin{equation}
{\bf \Phi} f_{dct} = f,
\end{equation}

\noindent where $f_{dct}$ is the representation of  $ f $ in the DCT domain and  ${\bf \Phi}_{n \times n}$ is the DCT matrix. With these estimated values, the minimum number of samples required to recover a traffic signal can be computed: $ m_{min} = \mu^2 S \log(n)  = 1^2 \cdot 17.63 \cdot \log(168) \approx 90 $.

\begin{figure}
	\centering
	\includegraphics[width=\textwidth]{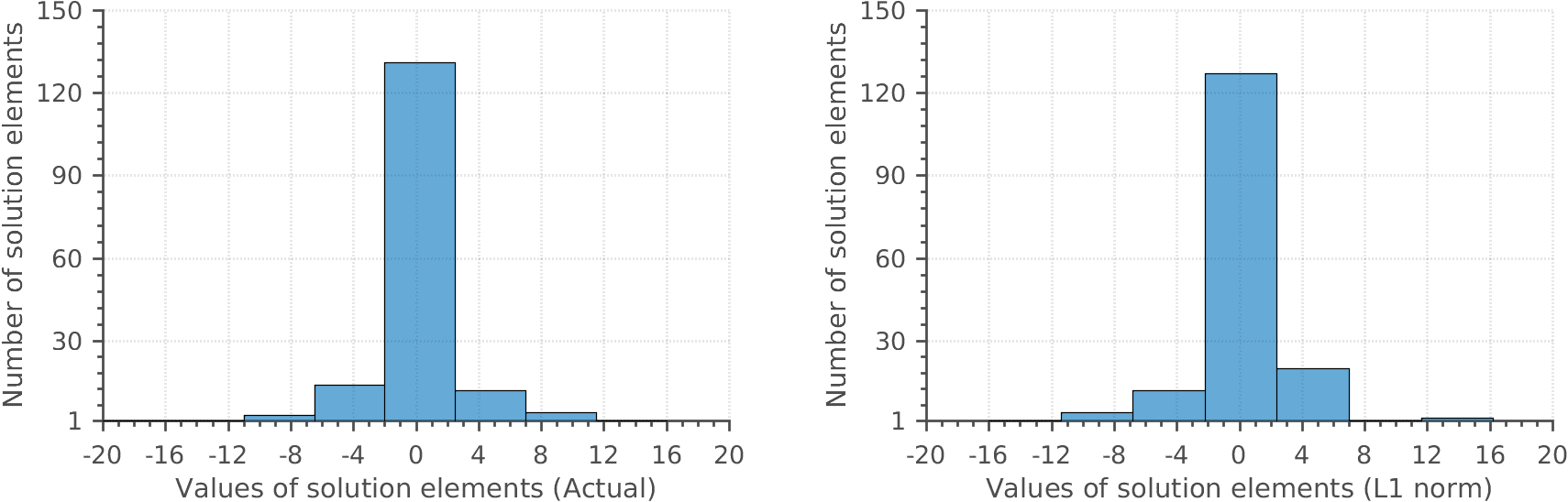}
	\caption{Solution elements (in total 168) by solving an underdetermined system via convex optimization. Both the actual and the $\mathcal{L}_1$-norm based recovery demonstrate sparsity (i.e., most solution elements are approximately zero).}
	\label{fig:sparse-solution}
\end{figure}

\begin{figure}[ht]
	\centering
	\includegraphics[width=\textwidth]{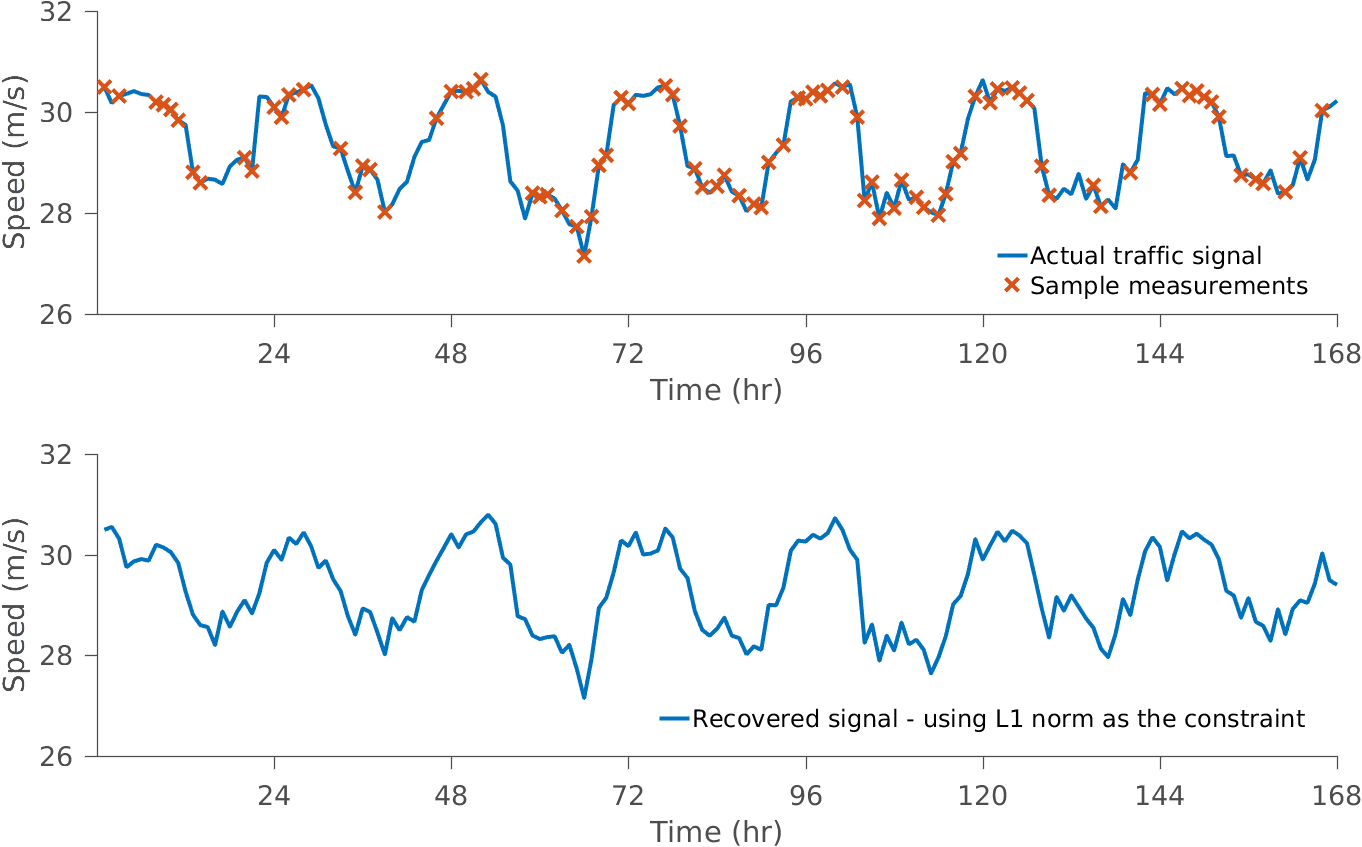}
	\caption{Recovery of a traffic signal via Compressed Sensing. The actual signal and its 90 random measurements are plotted (TOP). The $ \mathcal{L}_{1} $-norm based recovery (BOTTOM) shows high similarity to the actual signal.} 
	\label{fig:loop-recover}
\end{figure}

I test the performance of my method by first obtaining random measurements via sampling:  $b_{m \times 1} = {\bf \Psi}_{m \times n} f_{n \times 1} $, where $ {\bf \Psi}$ is the sampling matrix constructed by randomly permuting rows of the identity matrix. Then, I derive the recovered signal in the DCT domain $ \hat{f}_{dct}$ from $b$  by solving the following linear system: 

\begin{equation}
{\bf A}\hat{f}_{dct} = {\bf \Psi}{\bf \Phi} \hat{f}_{dct}= b.
\label{eq:linear}
\end{equation}

\noindent Equation~\ref{eq:linear} represents an \emph{underdetermined} system, which has infinitely many solutions $ \hat{f}_{dct}$. Among all solution candidates, the desired $ \hat{f}_{dct}$ should exhibit sparsity as observed in $ f_{dct}$, which can be computed as follows:

\begin{equation}
\begin{aligned}
& \text{$ {min} $ } \lVert \hat{f}_{dct} \rVert_{\mathcal{L}_{1}},  \\
& \text{$ s.t. \,\,\,$ }  {\bf A} \hat{f}_{dct}= b.
\end{aligned}
\label{eq:comsense}
\end{equation}

An example solution to Equation~\ref{eq:comsense} is shown in Figure~\ref{fig:sparse-solution}, where the actual solution elements of $ f_{dct} $ and the recovered solution elements of $ \hat{f}_{dct} $ both demonstrate sparsity. The final recovered signal $ \hat{f} $ is acquired by performing an inverse DCT on $ \hat{f}_{dct} $. Figure~\ref{fig:loop-recover} gives an example in which the recovered signal exhibits high similarity to the original signal. A more thorough analysis of the recovering performance can be found in Figure~\ref{fig:comsense_error}. As a result, as the number of measurements used in recovery increases, both the standard deviation and the expectation of the $ \mathcal{L}2 $ loss decrease (the average error of using $ 90 $ measurements is $ 1.4~m/s $). This analysis also demonstrates the robustness of my approach as the error only increases linearly when the number of samples used in recovery decreases.

\begin{figure}[ht]
	\centering
	\includegraphics[width=\textwidth]{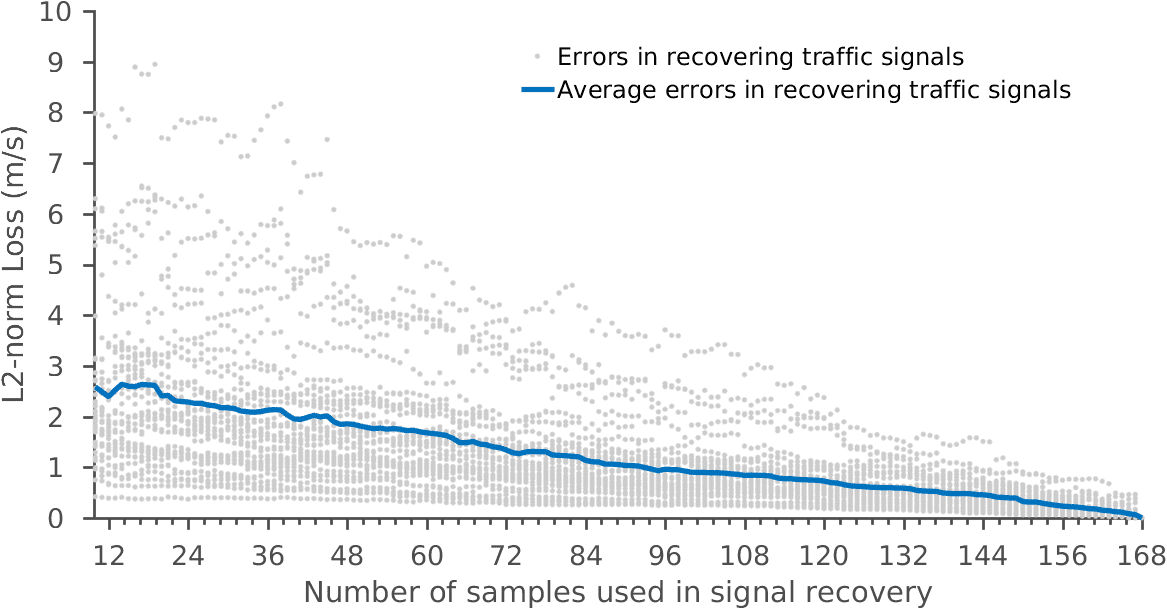}
	\caption{Errors between recovered and actual traffic signals. As more samples are used in signal's recovery, smaller errors are observed. This analysis also demonstrates the robustness of my approach as the error only increases linearly when the number of samples used in recovery decreases.} 
	\label{fig:comsense_error}
\end{figure}

To apply Compressed Sensing on GPS data, the measurements will be obtained from \emph{travel time estimation} rather than from a sampling operation performed on a traffic signal. In this case, $ {\bf \Phi} $ is set to $DCT\left(diag\left(1, \dots ,1\right)\right)$, and $ {\bf A} = {\bf \Psi}{\bf D} $ is taken to solve Equation~\ref{eq:comsense}. Since I have established that $ m_{min} = 90 $, it is worth mentioning I only address road segments that have measurements in more than 90 time intervals. Compared to techniques from Herring et al.~\cite{herring2010real} and Hofleitner et al.~\cite{hofleitner2012large}, I have reduced the minimum number of measurements required to recover a traffic signal from 168 to 90 by 46.4\%.

%% file: itsm/sections/5-dynamics.tex
\section{Estimating Traffic Dynamics Via GPS Data}
\label{sec:itsm-dynamics}
One of the hallmarks of traffic dynamics is the periodicity~\cite{hofleitner2012large}: traffic patterns show a clear trend over the course of a day and collectively over the course of a week. In this section, I first demonstrate that this phenomenon can be recovered using my technique. Then, I will analyze and discuss features revealed in the reconstructed traffic patterns of San Francisco using the GPS dataset Cabspotting~\cite{cabspotting}. 

To assist visualization and analysis, the metric $\text{\emph{fluidity}}\in [0,1]$~\cite{hofleitner2012large}, computed as the ratio of the estimated travel speed to the free-flow speed, is adopted for each road segment. In Figure~\ref{fig:dynamics}, I show the estimated traffic dynamics using the Cabspotting dataset~\cite{cabspotting}, denoted by the average fluidity, across the road network of downtown San Francisco. From the demonstration, it is clear that my technique recovers the periodicity of the traffic pattern through the dominant frequency showing at one cycle per day. This characteristic resembles the one observed in loop-detector data from the same area (see Figure~\ref{fig:loop-analysis} and~\ref{fig:loop-recover}), thus proving the effectiveness of my technique.

\begin{figure}
	\centering
	\includegraphics[width=\textwidth]{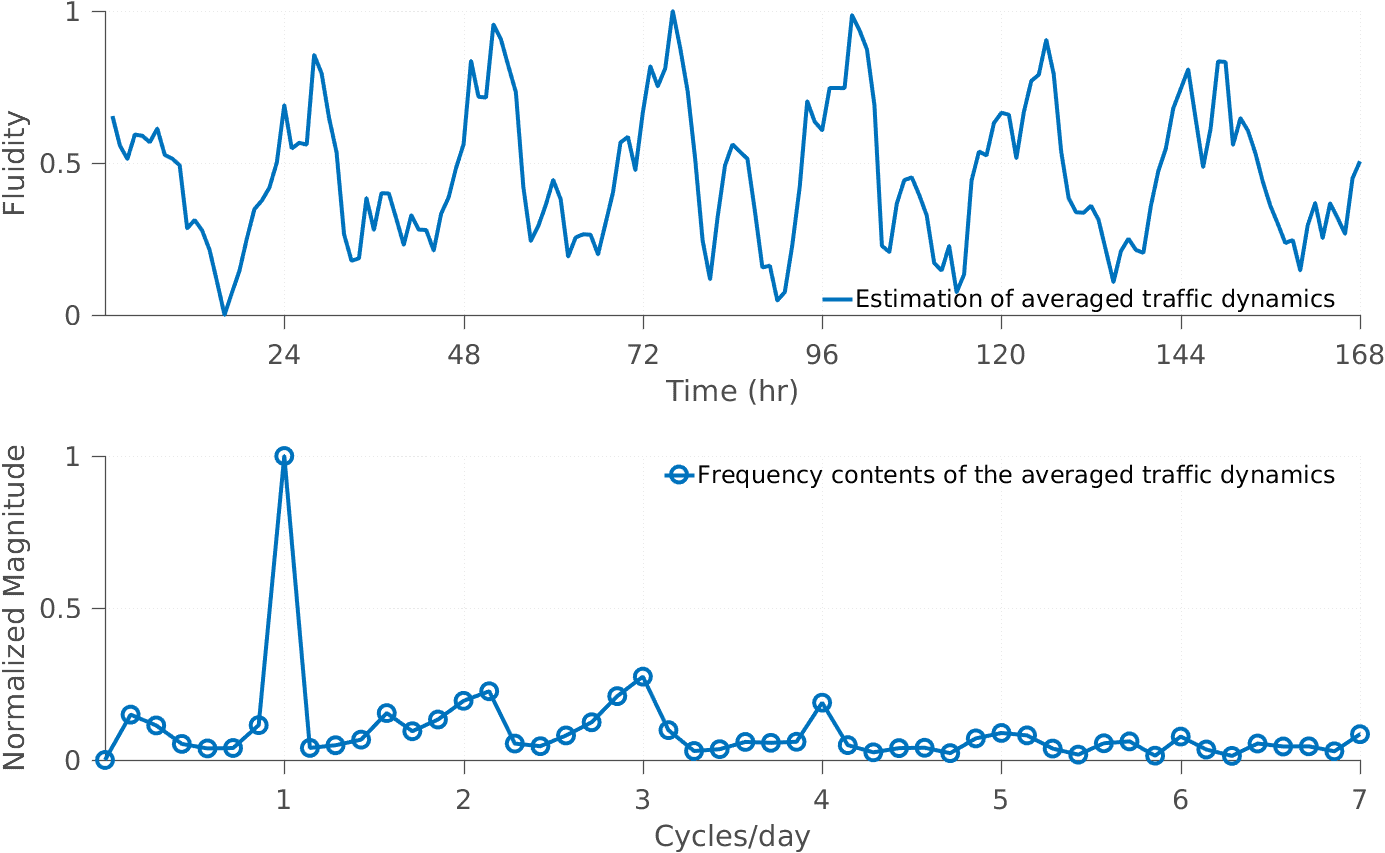}
	\caption{Estimated traffic pattern of downtown San Francisco using the Cabspotting dataset~\protect\cite{cabspotting} (TOP) and its corresponding spectral analysis (BOTTOM). The recovery from my technique shows a clear daily trend, which is consistent with the features observed in loop-detector data from the same area.}
	\label{fig:dynamics}
\end{figure}

The affinity between different days in a week can also be used to illustrate the quality of a recovery method. I have computed the correlation of every pair of days using the cosine distance for both the estimated traffic conditions using GPS data and the actual traffic conditions derived from loop-detector data. In the left panel of Figure~\ref{fig:similarity}, I provide all distance scores: the upper triangular matrix is derived using the estimated values using GPS data, and the lower triangular matrix is computed using loop-detector data. In the right panel of Figure~\ref{fig:similarity}, I provide the qualitative result for visual inspection. The symmetrical pattern across the diagonal line indicates that the estimated traffic states largely agree with the loop-detector data. 

Based on the distance scores, a hierarchical clustering is performed to reveal the similarity between various day pairs. The closest pair is Wednesday and Thursday, followed by Friday and Saturday, and Monday and Tuesday. In the second level of the hierarchy, Sunday joins Monday and Tuesday. These three day-pairs suggest that a typical week of San Francisco can be roughly divided into three stages: \emph{beginning of the week} (Sunday, Monday, and Tuesday), \emph{middle of the week} (Wednesday and Thursday), and \emph{end of the week} (Friday and Saturday).

\begin{figure}
	\centering
	\includegraphics[width=\textwidth]{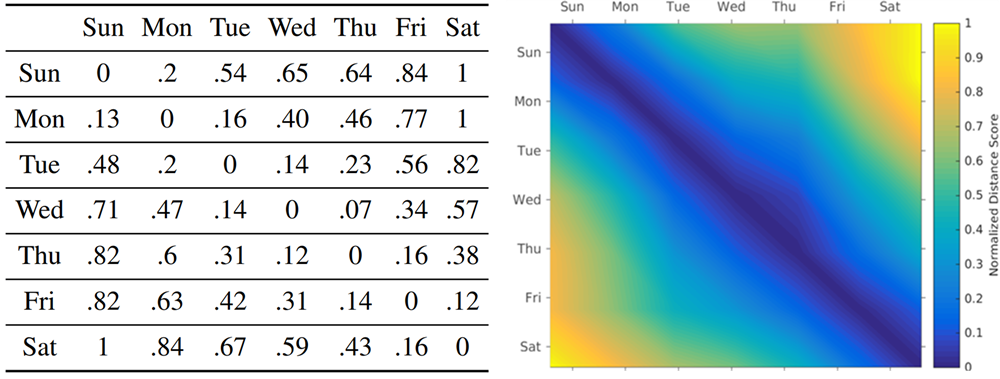}
	\caption{Correlation between every pair of days in a week. The left panel lists the normalized similarity scores calculated using the cosine distance, and the right panel provides the qualitative results. In both, the upper triangular matrix is derived using the estimated traffic conditions, and the lower triangular matrix is computed using loop-detector data from the same area. When data and patterns are compared across the diagonal line, my estimated results exhibit high similarity to the loop-detector data. }
	\label{fig:similarity}
\end{figure}

%% file: itsm/sections/6-conclusion.tex
\section{Summary and Future Work}
\label{sec:itsm-conclusion}

I have presented a novel computational scheme for estimating travel times, traversed paths, and missing values over a large-scale road network using spatially and temporally sparse GPS traces. Specifically, an approach based on the shortest travel time is performed to reconstruct the velocity field of a road network. In addition, an algorithm based on the Compressed Sensing algorithm has been developed to interpolate missing travel information over an entire traffic period. Lastly, I have extensively evaluated my approach and compared my technique to a state-of-the-art technique. My technique demonstrates consistent improvements over the previous technique in various traffic scenarios.

There are several possible future directions. To start with, as abundant GPS data are becoming available, processing them will be computationally expensive. Besides using the power of distributed computing, it is promising to explore sparsity and periodicity of traffic patterns to further reduce the amount of data needed for various Intelligent Transportation System (ITS) applications~\cite{Lin2019ComSenseTRB}.  Similarly, correlation in traffic patterns due to proximity/spatial coherence can also be examined. Lastly, it might be beneficial to integrate data mining of historical data, real-time traffic reconstruction from current data, and predictive traffic simulation to achieve more comprehensive and accurate estimations of travel conditions over a metropolitan area.

While in this chapter, I have introduced a deterministic approach to estimate traffic conditions in areas with GPS data coverage and a novel algorithm based on Compressed Sensing to interpolate temporal missing values, an effective method to address the spatial sparsity, i.e., estimating traffic conditions in areas without GPS data coverage, is still needed, given our ultimate goal is to reconstruct city-scale traffic. One natural way to interpolate spatial missing traffic data is to run traffic simulation in the ``empty'' areas and adopt the simulation results. However, the simulated traffic flows, when propagate to areas with GPS data coverage, need to respect the estimated traffic flows in those areas. In order to minimize the discrepancy between estimated and simulated traffic flows, we need to not only ensure the estimated flows are as accurate as possible but also be able to dynamically tune traffic simulation so that the resulting simulated flows will respect the estimated flows. These are the topics of the next chapter.

%% file: chapters/SIGA.tex
\chapter{CITY-SCALE TRAFFIC ANIMATION USING STATISTICAL LEARNING AND METAMODEL-BASED OPTIMIZATION}
\label{ch:siga}

\input{siga/sections/1-intro.tex}

\input{siga/sections/2-related.tex}

\input{siga/sections/3-overview.tex}
\input{siga/sections/4-notation.tex}
\input{siga/sections/5-recon.tex}
\input{siga/sections/6-sim.tex}
\input{siga/sections/7-results.tex}

\input{siga/sections/8-conclude.tex}

%% file: siga/sections/1-intro.tex
\section{Introduction}
\label{sec:siga-intro}

Many efforts have been invested in digitalizing and visualizing urban environments, for example, software tools like Google Maps and Virtual Earth. As new technologies like VR systems~\cite{wang2005steering,mit,tacvr} and autonomous vehicles emerge, there is an increasing demand to incorporate realistic traffic flows into virtualized cities. The potential applications range from virtual tourism, networked gaming, navigation services, urban design, to training of autonomous driving. The ability to reconstruct city-scale traffic from mobile sensor data can enable the visualization and animation of realistic real-world traffic conditions, thus contributes to those applications.

Traditional traffic data collection methods (e.g., in-road sensors such as loop detectors and video cameras) are costly; new and cheaper data sources such as GPS devices are becoming increasingly ubiquitous. Especially, taxicabs and shared ride services (e.g., Uber and Lyft) are systematically equipping their car fleets with these devices. As a result, GPS traces are part of the most promising data sources to estimate citywide traffic conditions, attributing to their broad coverage. However, GPS data usually contain a low-sampling rate, i.e., the time difference between two consecutive points is large (e.g., greater than 60 seconds). This can cause difficulty to the inference of the actual traversed path of a vehicle, since there are may be multiple paths for connecting the two far-apart points in a complex urban environment. In addition, GPS data exhibit spatial-temporal sparsity, i.e., data can be scarce in certain areas (e.g., suburbs) and time periods (e.g., early-morning hours), which makes city-scale traffic reconstruction challenging.

While it is already challenging to reconstruct traffic conditions in areas with GPS data coverage, in order to reconstruct city-scale traffic, we need to reconstruct traffic conditions in areas without GPS data coverage as well. In theory, the local traffic in these areas can be approximated using traffic simulation. However, it is critical to ensure the consistency of traffic flows on the boundaries of areas with and without GPS data. This requires that traffic simulation must be \emph{dynamically} tuned to ensure the matching flows with the reconstructed traffic in regions with heavy and complete GPS data coverage.

In general, reconstructing city-scale traffic dynamics using GPS data presents a number of challenges: 1) processing the available data, 2) coping with insufficient data coverage, and 3) reconstructing local traffic flows that are consistent with the global traffic dynamics at large spatial-temporal scale. In order to address these challenges, I propose a systematic approach which takes a digital map (processed using the technique from Wilkie et al.~\cite{wilkie2012transforming}) and GPS data (collected from taxicabs in San Francisco~\cite{cabspotting}) as input, and reconstructs the city-scale traffic using a two-phase process. In the first phase, I reconstruct and progressively refine estiamted traffic flows on individual road segments from the sparse GPS data using statistical learning, optimization, \emph{map-matching}~\cite{quddus2015shortest}, and \emph{travel-time estimation}~\cite{Li2017CityEstSparse} techniques.
In the second phase, I use a metamodel-based simulation optimization to efficiently propagate the reconstructed results from the previous phase, along with a microscopic simulator~\cite{SUMO2012} to dynamically interpolate missing traffic data. To ensure that the reconstructed traffic is correct, I fine tune the simulation flows with respect to city-wide boundary (traffic) constraints and the estimated flows from the first phase, which objective is enforced through an error approximation of the traffic flow computed using my novel metamodel-based formulation. In summary, I address the problem of learning-based traffic animation and visualization using GPS data with the following contributions:

\begin{itemize}
\item Accurate estimation of traffic conditions in areas with GPS data coverage using statistical learning;
\item Dynamic interpolation of traffic conditions in areas without GPS data coverage using metamodel-based simulation optimization;
\item City-scale traffic reconstruction for traffic animation and visualization.
\end{itemize}

The rest of this chapter is structured as follows. I survey related work in traffic reconstruction and simulation in Section~\ref{sec:siga-related}. A general overview of my framework pipeline is provided in Section~\ref{sec:siga-overview}. I detail the two phases of my approach in Sections~\ref{sec:siga-recon} and Section~\ref{sec:siga-sim}, respectively. I present results for reconstruction and simulation, as well as system validation and application demonstration in Section~\ref{sec:siga-results}. Finally, I conclude and discuss future work in Section~\ref{sec:siga-conclusion}.

%% file: siga/sections/2-related.tex
\section{Related Work}
\label{sec:siga-related}

The modeling of urban environments has received considerable attention in recent years. Many systems have been developed to describe a variety of aspects including layout, vehicle traffic, and pedestrian motion~\cite{thomas2000,willemsen2006,musialski2013survey,li2016manhattan,garcia2017fast}. Regarding vehicle motion, a number of driving simulators and vehicle behavioral models have been proposed~\cite{wang2005steering,mit,tacvr}, with several designed for virtual reality systems~\cite{kuhl1995,bayarri1996,wang2005steering}. Many advancements have also been made in visualizing traffic~\cite{andrienko2006exploratory,ferreira2013visual,wang2014visual}, and developing open-world video games (e.g., Grand Theft Auto, Burnout Paradise, and Watch Dogs).

Traffic simulation has received renewed interest in the past decade, for example, with the introduced concept of ``virtualized traffic''~\cite{Berg2009,sewall2011virtualized}, and the increasing efforts towards modeling realistic traffic flows~\cite{Wilkie2015Virtual}. Some recent developments include an extension of macroscopic models to generate detailed animations of traffic flows~\cite{sewall10}, and their integration with existing microscopic models to produce traffic animation on urban road networks~\cite{Sewall:2011:IHS:2070781.2024169}. To provide some other examples, Shen and Jin~\cite{shen2012detailed} and Mao et al.~\cite{mao2015efficient} have improved existing microscopic models to produce believable traffic animations. A technique from texture synthesis has been proposed to enhance the visual quality of traffic flows~\cite{chao2017realistic}. Characterization of heterogeneous vehicle types~\cite{lin2016generating} and driver personalities~\cite{lu2014personality} have also been explored.

Real-world data have been used to calibrate simulated traffic flows. Recent studies, including van den Berg et al.~\cite{Berg2009}, Sewall et al.~\cite{sewall2011virtualized}, and Wilkie et al.~\cite{Wilkie:2013:FRD:2461912.2462021}, have explored in-road sensors to reconstruct the traffic flow; Chao et al.~\cite{chao2013video} acquired individual vehicle characteristics from video cameras; Bi et al.~\cite{bi2016data} adopted a data-driven method to enrich the lane-changing behaviors of traffic simulations. Finally, Garcia-Dorado et al.~\cite{garcia2014designing} endowed users the flexibility to assign desired vehicular behaviors to a road network. 

Data-driven modeling has been studied by many researchers in the context of multi-agent simulation~\cite{Lerner2007,Lee2007,Ju:2010:MC:1882261.1866162,Charalambous2014,Li2015Insects}. There are however two main distinctions between these studies and our work. The first distinction is the \emph{scale}: while data-driven crowds are often limited to a few hundreds of individuals, city-scale traffic reconstruction concerns tens of thousands of vehicles. The second distinction is the nature of the data: in data-driven crowds, at least several positions and joint angles are recorded per second for each agent, while mobile traffic data is often very sparse and individual trajectories are not known and cannot be assumed.

Traffic reconstruction has drawn much attention in the field of transportation engineering~\cite{kachroo2016travel}. In order to achieve high reconstruction accuracy, multiple data sources and traffic simulation models have been investigated~\cite{work2010traffic,li2014multimodel,perttunen2015urban}. While significant results have been achieved, these methods along with other projects such as Mobile Century~\cite{HERRERA2010568} are largely restricted to highway segments with lengths of a few kilometers. In order to expand reconstructions to arterial roads and surface streets, recent studies~\cite{kong2013efficient,castro2012urban,zhang2013aggregating} have adopted GPS data for the reconstruction task. However, due to the uncertainty and sparsity embedded in GPS data, several processing steps are required and are usually executed sequentially.

The first processing step is \emph{map-matching}, which addresses off-the-road GPS points and infers the truly traversed path of a vehicle. One of the commonly used techniques for \emph{map-matching} is the shortest-distance criterion~\cite{hunter2014path,quddus2015shortest}. Unfortunately, this criterion can introduce errors in congested environments, where the shortest-distance path differs from the shortest \emph{travel-time} path, and the latter is preferred by GPS devices and most drivers~\cite{tang2016estimating,hunter2014large,Li2017CityEstSparse}. Although the shortest travel-time criterion is a more reasonable assumption, finding the shortest travel-time path requires accurate estimations of traffic conditions, which are difficult to obtain due the lack of traffic monitoring sensors on arterial roads.

The second processing step, \emph{travel-time estimation}, tries to distribute the time difference between consecutive GPS points to the road segments of a map-matched path. To list a few examples, Hellinga et al.~\cite{hellinga2008decomposing} developed an analytical solution based on empirical observations of real-world traffic patterns. Rahmani et al.~\cite{rahmani2015non} took a non-parametric approach and adopted a kernel-based estimator. Other approaches were developed based on probability theory~\cite{herring2010estimating,hunter2014large}. While promising, these methods suffer from the inherent limitation of a sequential pipeline, in which errors from \emph{map-matching} will be carried over to its following procedures. Especially, in a congested network, if we use the shortest-distance criterion for \emph{map-matching}, we will likely to have a wrong map-matched path causing the aggregate travel time to be distributed to a wrong set of road segments. The main difference between my approach for addressing GPS data and the previous approaches~\cite{tang2016estimating,hunter2014large,Li2017CityEstSparse} is that I take an iterative perspective combining \emph{map-matching} and its follow-on process \emph{travel-time estimation} such that the errors generated from either step will get gradually attenuated. 

My entire framework differs from previous efforts mainly due to the challenges brought by the traffic reconstruction task on arterial roads (which constitute the majority of a city) than that on highways and major roads. On highways and major roads, in-road sensors are commonly found, which can provide accurate and complete traffic measurements. The modeling of traffic flows on major roads is also easier to suffice macroscopic traffic features because of the limited branching and merging of the major roads. Due to these features, many state-of-the-art simulation techniques are developed base on accurate traffic data and macroscopic traffic assumptions~\cite{Berg2009,Wilkie:2013:FRD:2461912.2462021}. However, these features do not appear on arterial roads: only noisy GPS data are available and macroscopic traffic assumptions are difficulty to maintain~\cite{kong2013efficient}. 

My framework addresses these issues by first conducting statistical learning on GPS data to reconstruct traffic in areas with rich GPS data coverage, then dynamically completing the reconstruction in areas where GPS data are insufficient or missing. Although efficient and large-scale traffic simulation techniques have been developed~\cite{sewall10,Sewall:2011:IHS:2070781.2024169}, they are not designed to ensure simulation fidelity and flow consistency across the areas with and without GPS data coverage. My approach, in contrast, can satisfy city-wide boundary constraints and derive consistent traffic reconstruction at a city scale.

%% file: siga/sections/3-overview.tex
\section{Overview}
\label{sec:siga-overview}

Here, I provide an overview of my approach and define the notation used in this chapter. 

\begin{figure*}[th]
	\centering
	\includegraphics[width=\textwidth]{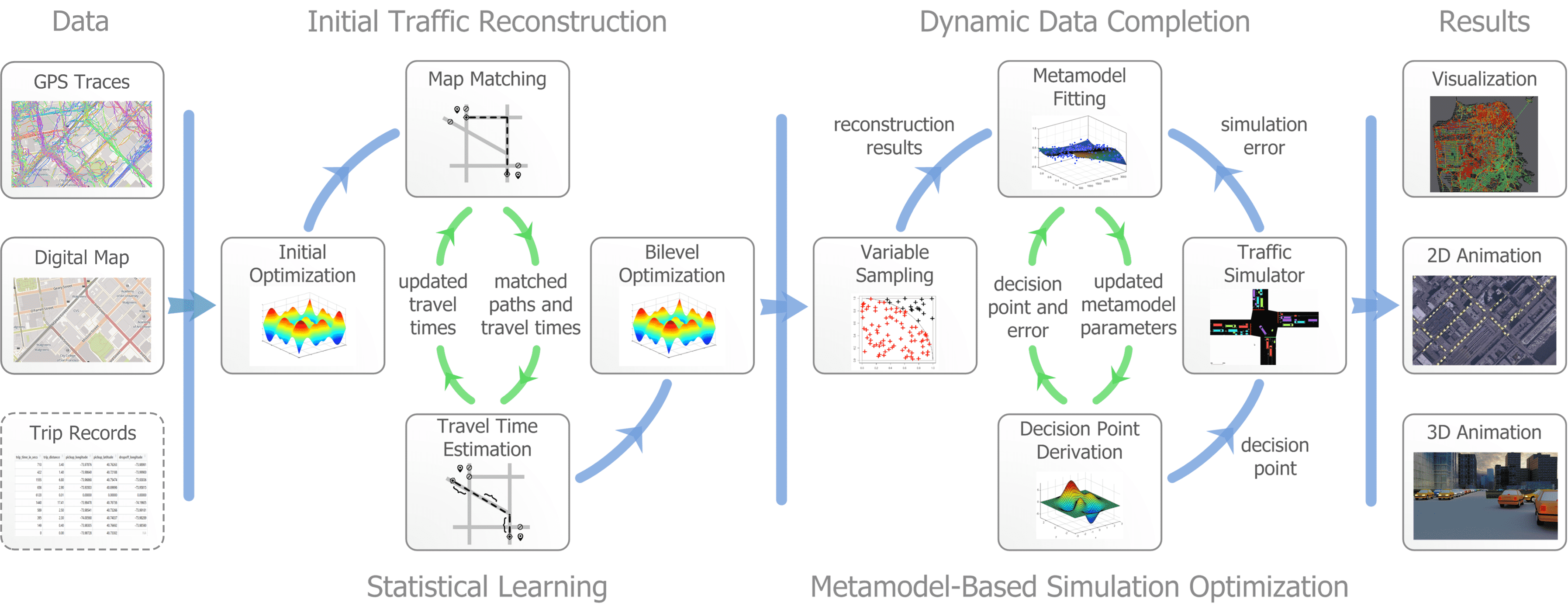}
	\caption{The systematic view of my framework. Trip records are optional as they can be inferred from GPS traces on a digital map.}
	\label{fig:overview}
\end{figure*}

\subsection{System}

My system enables city-scale traffic visualization and animation through two phases: 1) an initial traffic reconstruction phase (center-left column in Figure~\ref{fig:overview}, and Section~\ref{sec:siga-recon}), which estimates \emph{initial} traffic flows on all road segments of a road network, and 2) a dynamic data completion phase (center-right column in Figure~\ref{fig:overview}, and Section~\ref{sec:siga-sim}), which yields a \emph{dynamic} and much more accurate reconstructed traffic at the city scale.

For the \textit{Initial Traffic Reconstruction (Section~\ref{sec:siga-recon})} phase, a main objective is to reconstruct road-segment \flows from the input data. As mentioned before, this problem presents two challenges, which I will address in order.

The first one is the sparseness of GPS data, in which consecutive data-points are on average 60 seconds apart. To cope with this issue, I embed \emph{map-matching}~\cite{quddus2015shortest} and \emph{travel-time estimation}~\cite{Li2017CityEstSparse} into an iterative process (top and bottom boxes, center-left column in Figure~\ref{fig:overview}, and Section~\ref{section:reconstruction_iter}), where the output of one is treated as the input of the other and vice-versa. This process progressively refines the estimation of travel times on road segments and is initialized with a starting estimation obtained through a ``na\"ive'' optimization (left box, center-left column in Figure~\ref{fig:overview}, and Section~\ref{section:reconstruction_init}). The results of this iterative process are estimated traffic \flows for all road segments in areas with GPS data coverage.

The second challenge is the uneven coverage of GPS data. In order to reconstruct traffic \flows in areas with little to no traffic data, I perform a bilevel optimization (right box, center-left column in Figure~\ref{fig:overview}, and Section~\ref{section:reconstruction_bilevel}). The main principle of this phase is to use \emph{existing data} (either outdated traffic data or estimated results) that is very coarse, and update it using \emph{new information} (i.e., the reconstruction from the previous phase). The results of this phase are estimated traffic \flows for road segments in areas without GPS data coverage. Note however that in these areas, flows are only accurate at a large scale but not at the scale of individual road segments. Additionally, the reconstruction is static and does not account for interactions among individual vehicles.

The \textit{Dynamic Data Completion (Section~\ref{sec:siga-sim})} phase addresses the abovementioned issues by leveraging traffic simulation for data-deficient areas. However, the simulation traffic may not respect boundary conditions and previously reconstructed traffic in data-rich regions during propagation. To address this issue, I allow the simulation algorithm to change the ``\tratios'' at intersections, i.e., deciding how traffic would distribute itself to the downstream road segments at an intersection. The objective is then finding a set of optimal \tratios such that the discrepancy between simulated \flows and estimated flows is minimized. 

One way to derive the optimal \tratios is to run simulation-based optimization, which can be cost prohibitive due to the use of the microscopic simulator in optimization. I accelerate this procedure by approximating the cost function with a \emph{metamodel}~\cite{osorio2015metamodel} (top box, center-right column in Figure~\ref{fig:overview}), which can make the optimization problem more tractable. The optimization scheme is initialized randomly (left box, center-right column in Figure~\ref{fig:overview}, and Section~\ref{section:simulation_alg}). The resulting \tratios of the optimization are adopted lastly for obtaining various visualization results (right box, center-right column in Figure~\ref{fig:overview}, and Section~\ref{section:simulation_alg}).

%% file: siga/sections/4-notation.tex
\subsection{Notation}
\label{sec:siga-notation}

A road network is represented as a directed graph $ \mathcal{G} = (\mathcal{V},\mathcal{E}) $, in which the nodes $ \mathcal{V} $ and edges $ \mathcal{E} $ represent intersections (or terminal points) and road segments, respectively. A path on the road network is denoted as a set of road segments $k \in \mathcal{K}$ where $\mathcal{K}$ represents the set of all paths.

Geographically, a city can be divided into many traffic analysis zones (TAZs) based on socio-economic data. The centroids of these TAZs are considered traffic-flow \emph{origins} and \emph{destinations}.
The sets of origins and destinations are respectively denoted as $ \mathcal{O} \subseteq \mathcal{V}$ and $ \mathcal{D} \subseteq \mathcal{V}$. The traffic flow is considered to take place between origin-destination pairs (OD pairs), and the average flow from $ o \in \mathcal{O}$ to $ d \in \mathcal{D} $ during a certain time interval is noted $ u_{od} $ (in accordance with transportation engineering literature, this flow is the average number of vehicles). With $ \mathcal{K}_{od} $ as the set of paths connecting $ o $ and $ d $, the flow on path $ k \in \mathcal{K}_{od} $ is $ u(k) $. The flow inside an OD pair is a summation of flows on all paths in this pair:

\begin{equation}
	u_{od} = \sum_{k \in \mathcal{K}_{od}} u(k), \; u(k) \geq 0, \forall o \in \mathcal{O}, d \in \mathcal{D}.
	\label{eq:ue_c1}
\end{equation}

\noindent The flow $ f(e) $ on a road segment $ e \in \mathcal{E} $ is defined as the summation of all path flows traversing $ e $:

\begin{equation}
	f(e) = \sum_{(o,d) \in \mathcal{O}\times\mathcal{D}}\sum_{k \in \mathcal{K}_{od}} \delta_{e}^{k} u(k),
	\label{eq:ue_c2}
\end{equation}

\noindent where $ \delta_{e}^{k} $ is $ 1 $ if path $ k $ contains road $ e $, and $ 0 $ otherwise.
This means that $ \forall e \in \mathcal{E}, \forall k \in \mathcal{K}_{od} $, if $ e \in k $, then $ e $ contributes a certain portion $ p_{e, od} $ of $ u_{od} $.
By arranging the flows of all OD pairs in a vector $ \mathbf{u} = [u_{od}]^{\intercal}_{(o,d) \in \mathcal{O}\times\mathcal{D}}$, and the contributed flow portions of a road segment $e$ as $ \mathbf{P}_{e} = [p_{e, od}]_{(o,d) \in \mathcal{O}\times\mathcal{D}} $, the following relation can be derived:

\begin{equation}
	f(e) = \mathbf{P}_{e}\mathbf{u}.
\end{equation}

\noindent With all road-segment flows in a network $ \mathbf{f}_{\mathcal{E}} = [f(e)]_{e \in \mathcal{E}}^{\intercal}$ and the portions of all road segments $ \mathbf{P}_{\mathcal{E}} = [\mathbf{P}_{e}]_{e \in \mathcal{E}}^{\intercal}$, the general relation between flows in road segments and OD pairs can be obtained:

\begin{equation}
	\mathbf{f}_{\mathcal{E}} = \mathbf{P}_{\mathcal{E}}\mathbf{u},
\end{equation}

\noindent with $ \mathbf{P}_{\mathcal{E}} $ termed \emph{assignment matrix}. 

I list other notation used in this chapter as follows:

\begin{itemize}
	\item $\mathcal{S}$ denotes the GPS data, then $\forall s \in \mathcal{S}$ contains a longitude, a latitude, and a timestamp; a pair of successive GPS data points is noted $(s_1, s_2) \in \mathcal{S}_p$, where $\mathcal{S}_p$ represents the set of all pairs of consecutive GPS data points; $\mathcal{K}_{s_1,s_2}$ denotes the set of all paths connecting $s_1$ to $s_2$.
	\item $t$ denotes travel time; for instance, $t(e)$ is the travel time of a road segment $e \in \mathcal{E}$, $t(k) = \sum_{e \in k} t(e)$ is the travel time of a path $k \in \mathcal{K}$, and $t(s_1, s_2)$ is the travel time between a pair of two consecutive GPS data points $(s_1,s_2) \in \mathcal{S}_p$ given by their timestamps.
\end{itemize}

%% file: siga/sections/5-recon.tex
\section{City-Scale Traffic Reconstruction}
\label{sec:siga-recon}

There are three steps in the \emph{initial traffic reconstruction} phase: 1) initial estimation of travel times on road segments, 2) iterative refinement of these travel times, and 3) bilevel optimization for filling data-lacking areas.

\subsection{Initial Estimation}
\label{section:reconstruction_init} 

For the initial estimation of travel times, I use Wardrop's principle~\cite{wardrop1952road}, which states that traffic will arrange itself in congested networks such that no vehicle can reduce its travel cost by switching routes.
This state is termed \emph{user equilibrium} and is a result of every user non-collaboratively attempting to minimize their travel times. While the actual traffic may not form an exact \emph{user equilibrium}, this state can serve as an approximation to real-world traffic~\cite{hato1999incorporating} and motivate the generation of optimization constraints.

Following this principle, for a pair of consecutive GPS data points $(s_1,s_2) \in \mathcal{S}_p$, the travel time between them $t(s_1, t_2)$ is the minimum travel time of all paths connecting the two points: $\forall k \in \mathcal{K}_{s_1,s_2}, t(k) \geq t(s_1,s_2)$. Thus, the travel times of the road segments should satisfy the following ``Wardrop'' constraints $\mathcal{W}$:

\begin{equation}
	\mathcal{W} = \left \{ t(k) \geq t(s_1,s_2) \right \}_{\forall (s_1,s_2) \in \mathcal{S}_p,\: \forall k \in \mathcal{K}_{s_1,s_2}}.
\end{equation}

Additionally, the travel time of a road segment $e \in \mathcal{E}$ is bounded by $t_{min}(e)$ and $t_{max}(e)$, which represent free-flow travel time (at $120\%$ of the road segment's speed limit\footnote{The threshold is set to be $120\%$ because of two reasons: 1) it regulated the range of traffic flows, and 2) by examining map-matched GPS data, traffic flow speed that exceeds $120\%$ is rare.}) and travel time at jam density (at $0.5~m.s^{-1}$), respectively.

In order to derive a solution with respect to typical traffic patterns, I pose a regularization term, $ \mathcal{R} $, on $ \mathbf{t}_\mathcal{E} = [t(e)]_{e \in \mathcal{E}}^{\intercal} $ to model the correlation in traffic patterns of road segments in close proximity~\cite{sheffi1985urban,zheng2011urban,7501704}.
$ \mathcal{R} $ is formed as a 2-dimensional fused Lasso penalty~\cite{tibshirani2005sparsity}, in which each row represents a pair of road segments connecting at a node.
The entries corresponding to the pair of road segments are set to -1 and 1, accordingly.
As an example, the $ \mathcal{R} $ of a three-way intersection with road segments $ \{e_1,e_2,e_3\}$ takes the form: $ \mathcal{R} = \{[1,1,0]^{\intercal},[-1,0,1]^{\intercal},[0,-1,-1]^{\intercal}\}$.
$ \mathcal{R} $ is set via enumerating all pairs of road segments at each $ v \in \mathcal{V} $, which results in $ \mathcal{R} \in \mathbb{R}^{m \times n} $, where $ m = \sum_{v \in \mathcal{V}} {deg(v) \choose 2} $, $ n = | \mathcal{E} | $, and $ deg(v) $ is the total degree of $ v $.
The initial estimation $\bar{\mathbf{t}}_\mathcal{E}$ of travel times is formed as follows:

\begin{equation}
	\begin{aligned}
		& \bar{\mathbf{t}}_\mathcal{E} = \underset{\mathbf{t}_\mathcal{E}}{\text{argmin}}
		& & \lVert \mathcal{R} \mathbf{t}_\mathcal{E} \rVert_{1}, \\
		& \text{subject to}
		& & \mathcal{W}, \,\, \mathbf{t}_{\mathcal{E},min} \leq \bar{\mathbf{t}}_\mathcal{E} \leq \mathbf{t}_{\mathcal{E},max}.
	\end{aligned}
	\label{eq:init_cvx}
\end{equation}

\subsection{Iterative Estimation}
\label{section:reconstruction_iter}

I refine $\bar{\mathbf{t}}_\mathcal{E}$ with an iterative process that alternates between \emph{map-matching} and \emph{travel-time estimation}. This design is based on the observation that many approaches take a sequential perspective~\cite{kong2013efficient,zhang2013aggregating,hunter2014large,Li2017CityEstSparse}. As a consequence, the bias generated during \emph{map-matching} (especially under the shortest-distance criterion) will get cascaded into \emph{travel-time estimation}, thus affecting the overall estimation accuracy. By resorting to an iterative process, as newly estimated travel times get more accurate, so does the map-matching, and vice-versa.

\subsubsection{Map Matching}
I adopt a simple strategy to compute the ``true'' path $\bar{k}_{s_1,s_2}$ between a pair of successive GPS data points $(s_1,s_2) \in \mathcal{S}_p$.
Assuming $\mathcal{Q}_1$ and $\mathcal{Q}_2$ to denote the sets of candidate positions on the map for $s_1$ and $s_2$, respectively. This step returns $\mathcal{N} = \{ \bar{k}_{s_1,s_2} \}_{\forall (s_1,s_2) \in \mathcal{S}_p}$, where:

\begin{equation}
	\begin{aligned}
		& \bar{k}_{s_1,s_2} = \underset{\bar{k} \in \bar{\mathcal{K}}_{s_1,s_2}}{\text{argmin}} \| t(s_1,s_2) - \bar{t}(\bar{k}) \|,\\
		&\text{with}\; \bar{\mathcal{K}}_{s_1,s_2} = \left \{ \underset{k \in \mathcal{K}_{q_1,q_2}}{\text{argmin}} \bar{t}(k) \right \}_{\forall (q_1,q_2) \in \mathcal{Q}_1 \times \mathcal{Q}_2}.
	\end{aligned}
\end{equation}

\subsubsection{Travel Time Estimation}

Given the paths $\mathcal{N}$ assigned to each pair of successive GPS data points by the map-matching step, the next step is to estimate a more accurate set of travel times $\bar{\mathbf{t}}_\mathcal{E}$ on individual road segments. 

The travel time of a road segment is modeled to follow a probability distribution $\bar{t}(e) \sim \pi_e, \forall e \in \mathcal{E}$, which is parameterized by $\boldsymbol{\theta} = \{ \theta_e \} _{e \in \mathcal{E}}$ learned through \emph{maximum likelihood estimation} (MLE):

\begin{equation}
	\underset{\boldsymbol{\theta}}{\text{maximize}} \; \mathcal{L}(\boldsymbol{\theta} | \mathcal{N}) = \sum_{\forall (s_1,s_2) \in \mathcal{S}_p} \text{log} \: \pi(t(s_1,s_2)|\bar{k}_{s_1,s_2};\boldsymbol{\theta}),
	\label{eq:mle}
\end{equation}

\noindent
where $\mathcal{L}$ is the likelihood function.
Following the methodology from Hofleitner et al.~\cite{hofleitner2012probability}, I assume that the travel time of a road segment can be modeled by a univariate distribution, and the travel-time distributions of road segments are pairwise independent.
Using these assumptions, Equation~\ref{eq:mle} is solved via \emph{expectation maximization} (EM). The maximization step within the M-step is conducted for each road segment:

\begin{equation}
	\underset{\theta_e}{\text{maximize}} \; \sum_{\omega \in \Omega} w_{\omega}(e) \: \text{log} \: \pi_e(t_{\omega}(e);\theta_e),
	\label{eq:em}
\end{equation}

\noindent
where $\pi_e$, $\theta_e$, $\Omega$, $t_{\omega}$, and $w_{\omega}$ are as follows:

\paragraph{$\bullet$ $\pi_e$, $\theta_e$:}

Following Hunter et al.~\cite{hunter2014large}, $\pi$ is taken to be the Gamma distribution $\Gamma$, which has a positive domain and is robust to long-tail observations. These features make $\Gamma$ suitable for modeling the travel time of a road segment. Consequently, $\theta_e = (\alpha_e,\beta_e)$ where $\alpha_e$ is the shape and $\beta_e$ is the scale. 


\paragraph{$\bullet$ $\Omega$:}

As a road segment can be part of multiple paths, which link different pairs of successive GPS data points, samples regrading each pair must be included in Equation~\ref{eq:em}. In this work, I compute 100 samples for each pair: $\Omega = \{ (s_1,s_2,i) \}_{\forall (s_1,s_2) \in \mathcal{S}_p, \forall i \in [\![ 1,100 ]\!] }$.

\paragraph{$\bullet$ $t_{\omega}$:}

Using the previous choices for $\pi_e$, $\theta_e$ and $\Omega$, $t_{\omega}$ is computed as follows according to Hunter et al.~\cite{hunter2014large}:

\begin{equation} 
	\begin{split}
		\forall \omega \in \Omega, \; \forall e \in \bar{k}_{s_1,s_2}, \;\; t_{\omega}(e) &= t(s_1,s_2)\frac{A_{\omega}(e)}{\sum_{e \in \mathcal{E}} A_{\omega}(e)} \\
		\text{with} \; A_{\omega}(e) &\sim \Gamma(\alpha_{e}, \frac{\beta_{e}}{t(s_1,s_2)}), \\
	\end{split}
	\label{eq:gamma}
\end{equation}

\paragraph{$\bullet$ $w_{\omega}$:}

A weight $w_{\omega}(e)$ is then computed as the distance between $t_{\omega}(e)$ and $A_{\omega}(e)$.

After solving Equation~\ref{eq:em}, a new estimation of $\bar{\mathbf{t}}_\mathcal{E}$ is computed, where for each road segment $e \in \mathcal{E}$, $\bar{t}(e)$ is computed as the mean of $\pi_e$.
From here, the iterative estimation loops back to the map-matching step.
The entire process stops after 10 iterations (determined empirically), with the refined $\bar{\mathbf{t}}_\mathcal{E}$ as the output.

\subsection{Bilevel Optimization}
\label{section:reconstruction_bilevel}

The previous steps have estimated travel times $\bar{\mathbf{t}}_\mathcal{E}$ of road segments with GPS data coverage. The objective of this step is to compute travel times of all road segments (including those without data coverage).

To proceed, $\bar{\mathbf{t}}_\mathcal{E} = [\bar{t}(e)]_{e \in \mathcal{E}}^{\intercal}$ is converted to $\bar{\mathbf{f}}_\mathcal{E} = [\bar{f}(e)]_{e \in \mathcal{E}}^{\intercal}$ by inverting the road-segment {\em performance function} proposed by the U.S. Bureau of Public Roads:

\begin{equation}
	t(e) = t_{min}(e)\left(1 + 1.5\left(\frac{f(e)}{c(e)}^4\right)\right),
	\label{eq:bpr}
\end{equation} 

\noindent where $ c(e) $ is the capacity computed as follows (the formula is accessible at \url{http://www.fhwa.dot.gov}):

\begin{equation}
c(e) =
\begin{cases}
1700+10t_{min}(e) & \text{if }t_{min}(e) \leq 70~mph, \\
2400 & \text{otherwise}.
\end{cases}
\end{equation}

Next, the target OD pairs $\bar{\mathbf{u}}$ are derived based on the ratio of estimated flows to the loop-detector measurements on the same road segment~\cite{yang2016}. Having $\bar{\mathbf{f}}_\mathcal{E}$ from the previous steps, traffic flows of all road segments $\hat{\mathbf{f}}_\mathcal{E}$ and the corresponding flows between OD pairs $\hat{\mathbf{u}}$ can be computed through the following minimization~\cite{cascetta1988unified}:

\begin{equation}
	\begin{aligned}
		& \underset{\hat{\mathbf{u}}}{\text{minimize}}
		& & \mathcal{F}_{1}(\hat{\mathbf{u}},\bar{\mathbf{u}}) + \mathcal{F}_{2}(\hat{\mathbf{f}}_\mathcal{E},\bar{\mathbf{f}}_\mathcal{E}) \\
		& \text{subject to}
		& & \hat{\mathbf{f}}_\mathcal{E} = \mathcal{M}(\hat{\mathbf{u}}), \,\, \hat{\mathbf{u}} \geq 0,\\
	\end{aligned}
	\label{eq:bi}
\end{equation}

\noindent where $ \mathcal{F}_{1} $ and $ \mathcal{F}_{2} $ are generalized distance functions.
$ \mathcal{M} $ is the \emph{assignment map} which determines $ \mathbf{P}_\mathcal{E} $.
If $ \mathcal{M} $ follows Wardrop's principle~\cite{wardrop1952road}, Equation~\ref{eq:bi} becomes a bilevel optimization problem~\cite{yang1992estimation}: the upper level minimizes the distances of estimated OD pairs and traffic flows to their corresponding measurements, while the lower level satisfies the \emph{user equilibrium}.
For $ \mathcal{F}_{1} $ and $ \mathcal{F}_{2} $, I select the generalized least squares (GLS) estimator~\cite{bera2011estimation}, as it permits different weighting schemes of $ \bar{\mathbf{u}} $ and $ \bar{\mathbf{f}}_\mathcal{E} $.
I further assume that $ \bar{\mathbf{u}} $ and $ \bar{\mathbf{f}}_\mathcal{E} $ are results from the following stochastic system of equations:

\begin{equation}
	\bar{\mathbf{u}} = \hat{\mathbf{u}} + \mathbf{\epsilon_{1}}, \,\, \bar{\mathbf{f}}_\mathcal{E} = \hat{\mathbf{f}}_\mathcal{E} + \mathbf{\epsilon_{2}}.
\end{equation}

\noindent Using these choices, Equation~\ref{eq:bi} can be explicitly written as:

\begin{equation}
	\begin{aligned}
		& \underset{\hat{\mathbf{u}}}{\text{minimize}}
		& & \eta(\bar{\mathbf{u}}-\hat{\mathbf{u}})^{\intercal}U^{-1}(\bar{\mathbf{u}}-\hat{\mathbf{u}})  \\
		& \text{}
		& & + (1-\eta)(\bar{\mathbf{f}}_\mathcal{E}-\hat{\mathbf{f}}_\mathcal{E})^{\intercal}V^{-1}(\bar{\mathbf{f}}_\mathcal{E}-\hat{\mathbf{f}}_\mathcal{E}), \\
		& \text{subject to}
		& & \hat{\mathbf{f}}_\mathcal{E} = \mathcal{M}(\hat{\mathbf{u}}),\,\, \hat{\mathbf{u}} \geq 0, \,\, \hat{\mathbf{f}}_\mathcal{E} \geq 0, \,\, 0 \leq \eta \leq 1,\\
	\end{aligned}
	\label{eq:bi_gls}
\end{equation}

\noindent where $ U $ and $ V $ are variance-covariance matrices of $ \mathbf{\epsilon}_{1} $ and $ \mathbf{\epsilon}_{2} $, respectively; $ \mathbb{E}(\mathbf{\epsilon}_{1}) = 0 $ and $ \mathbb{E}(\mathbf{\epsilon}_{2}) = 0 $ are derived from experiments by Cascetta and Nguyen~\cite{cascetta1988unified}; $ \eta \in [0,1] $ is the weighting factor.
When $ \eta = 1 $, $ \bar{\mathbf{f}}_\mathcal{E} $ is ignored and the estimation is solely based on $ \bar{\mathbf{u}}$; when $ \eta = 0 $, the estimation is solely based on $ \bar{\mathbf{f}}_\mathcal{E}$.
Equation~\ref{eq:bi_gls} can be solved iteratively: in the upper level the GLS estimator is solved using quadratic programming; in the lower level \emph{user equilibrium} is approximated using the \emph{one-shot} function within SUMO~\cite{SUMO2012}. As a result of this procedure, I can use the newly computed $\hat{\mathbf{f}}_\mathcal{E}$ on road segments where GPS data are unavailable.

%% file: siga/sections/6-sim.tex
\section{Dynamic Data Completion}
\label{sec:siga-sim}

For detailed traffic simulation, I adopt a microscopic simulation algorithm.
To ensure that the simulation for any given area would respect flows at the boundaries from the previously reconstructed traffic, I allow the algorithm to alter the turning ratios at intersections (which decide how traffic would distribute itself to the downstream road segments at an intersection). I choose the turning ratios as optimization variables because of the following reasons: 1) turning ratios implicitly encode traffic light logic, 2) driving behaviors such as lane changing are limited at intersections, and 3) detailed road information is difficult to obtain. In most microscopic traffic simulators, specifying turning ratios is one of the main mechanisms to start a simulation. Furthermore, the split flows can be used to form a metamodel and compute estimation errors from the previous reconstructed result. My algorithm can also cope with multivariable optimization if the information of other simulation parameters is provided.

I denote the vector of all turning ratios by $\mathbf{x} = [ x_{v,e} ]_{\forall e \in v, \forall v \in \mathcal{V}}$, and name it a \emph{decision point}. In order to systematically derive a decision point $\mathbf{x}^*$ in a simulation region, which not only meets OD demands but also conforms the previous estimated traffic conditions, I rely on the following optimization task:

\begin{equation}
	\begin{aligned}
		& \mathbf{x}^* = \underset{\mathbf{x}}{\text{argmin}}
		& & \bar{F}(\mathbf{x};\rho) \equiv \mathbb{E}[F(\mathbf{x};\rho)],\\
	\end{aligned}
	\label{eq:traffic_opt}
\end{equation}

\noindent where $ \bar{F} $ is the objective function, and $ F $ is a stochastic network performance measure.
The distribution of $ F $ depends on the decision point $ \mathbf{x} $ and exogenous parameters $ \rho $, which record a network topology and road-segment metrics. Every simulation run with $ \mathbf{x} $ is a realization of $ F $, which involves sampling many other distributions that account for the stochastic nature of traffic (e.g., driver differences). Assuming $ r $ independent simulations with a given $ \mathbf{x} $ are executed, the objective function can be approximated:

\begin{equation}
	\hat{F}(\mathbf{x};\rho) = \frac{1}{r}\sum_{i=1}^{r}F_{i}(\mathbf{x};\rho).
	\label{eq:objective}
\end{equation}

\noindent However, using a simulation algorithm in optimization can be very costly, especially consider the details a microscopic traffic simulator models and the scale it is applied to.
This motivates me to adopt metamodel-based simulation optimization for efficiency.

\subsection{Metamodel-Based Simulation Optimization}
\label{section:simulation_opt}

A metamodel can simplify simulation-based optimization, as it is typically a deterministic function rather than a stochastic simulator. Therefore, one way to circumvent the issues of using a microscopic traffic 
simulator in a simulation-optimization loop is to develop a deterministic metamodel to replace the stochastic simulation response. The metamodel is usually less realistic in terms of the modeling capability, but much cheaper to evaluate.

The most common metamodels are \emph{functional} metamodels, which are general-purpose functions and can be used to approximate arbitrary objective functions. Often, they are results of a linear combination of basis functions such as low-order polynomials, spline models, and radial basis functions~\cite{conn2009introduction}. However, they require a large number of decision points to be fitted, since the structure of an underlying problem is not considered. This means that I need to run the simulator many times on various decision points in order to fit a well-performing metamodel. This procedure is expensive and to a certain degree defeats the purpose of using a metamodel. Instead, I use a metamodel that contains not only a \emph{functional} component but also a \emph{physical} component which encodes the underlying problem for achieving high efficiency.

I build this physical component based on the classical \emph{flow conversion equation}, also known as the \emph{traffic equation}~\cite{osorio2011mitigating}:

\begin{equation}
	f(e_1) = \gamma(e_1) + \sum_{e_2 \in \mathcal{C}}p(e_1,e_2)f(e_2), \; \forall e_1 \in \mathcal{C},
	\label{eq:flow}
\end{equation}

\noindent where $ \mathcal{C} $ represents the set of road segments in the simulation region, $ \gamma(e_1) $ is the external flow injected into road segment $ e $, and $ p(e_1,e_2) $ is the transition probability from road segment $ e_1 $ to road segment $ e_2 $.
The exogenous parameters in Equation~\ref{eq:flow} are external flows and transition probabilities. Having them, both the traffic equation and the traffic simulator can be executed to obtain the flows of all road segments in a simulation region.
Denoting the subset of road segments with estimated traffic flows from GPS data as $ \mathcal{A} \in \mathcal{C} $, the estimated flows as $ \mathbf{f}^E = \{f(e)^E\}_{e \in \mathcal{A}} $, the propagated flows using the traffic equation as $ \mathbf{f}^T = \{f(e)^T\}_{e \in \mathcal{A}} $, and the simulated flows using a traffic simulator as $ \mathbf{f}^S = \{f(e)^S\}_{e \in \mathcal{A}} $, the approximation of the objective function derived from the simulator with one simulation run ($r=1$ in Equation~\ref{eq:objective}) is defined as follows:

\begin{equation}
	\hat{F}(\mathbf{x};\rho) = F(\mathbf{x};\rho) = \lVert \mathbf{f}^S- \mathbf{f}^E \rVert_{2}.
\end{equation}

\noindent The approximation of the objective function derived from the traffic equation is defined as follows:

\begin{equation}
	T(\mathbf{x};\rho) = \lVert \mathbf{f}^T- \mathbf{f}^E \rVert_{2}.
\end{equation}

\noindent The metamodel is then constructed as a combination of the physical component $ T $ and a functional component $ \Phi $:

\begin{equation}
	M(\mathbf{x};\alpha,\beta,\rho) = \alpha T(\mathbf{x};\rho) + \Phi(\mathbf{x};\boldsymbol\beta),
	\label{eq:meta}
\end{equation}

\noindent where $ \alpha $ (initially set to 0.5) and $ \boldsymbol{\beta} $ (initially set to $ \boldsymbol{1} $) are parameters of the metamodel. The functional component, $ \Phi $, is chosen to be a quadratic polynomial~\cite{osorio2013simulation}:

\begin{equation}
	\Phi(\mathbf{x};\boldsymbol\beta) = \beta_{1} + \sum_{i=1}^{|\mathbf{x}|}\beta_{i+1}x_{i} + \sum_{i=1}^{|\mathbf{x}|}\beta_{i+|\mathbf{x}|+1}x_{i}^{2},
\end{equation}

\noindent where $|\mathbf{x}|$ is the dimension of $ \mathbf{x} $; $ x_{i} $ and $ \beta_{i} $ are the $ i $th elements of $ \mathbf{x} $ and $ \boldsymbol\beta $, respectively. The quadratic polynomial provides Taylor-type bounds, serves as a general term within a metamodel formulation, and ensures global convergences~\cite{conn2009introduction}. In order to fit the metamodel, I rely on the decision points that have been evaluated using both the traffic equation and the traffic simulator.
Denoting $ \mathcal{X} $ as the pool of the decision points, the metamodel can be fit (i.e., compute $ \alpha $ and $ \boldsymbol\beta $) by solving:

\begin{equation}
	\begin{aligned}
		& \underset{\alpha,\boldsymbol\beta}{\text{minimize}}
		& &
		\sum_{i=1}^{|\mathcal{X}|}\left(w_{i}\left(\hat{F}(\mathbf{x}_{i};\rho)-M(\mathbf{x}_{i};\alpha,\boldsymbol\beta,\rho \right)\right)^2 \\
		& \text{}
		& & + (w_{0} \cdot (\alpha-1))^2 + \sum_{j=1}^{2|\mathbf{x}|+1}(w_{0} \cdot \beta_{j})^2, \\
	\end{aligned}
	\label{eq:fitmeta}
\end{equation}

\noindent where $ w_{0} $ is a fixed constant. $ w_{i} $ is the weight associating each $ \mathbf{x}_{i} $ and a new decision point $ \mathbf{x}_{new} $ at each iteration during the optimization, computed as $ w_{i} = 1/(1+\lVert \mathbf{x}_{new} - \mathbf{x}_{i} \rVert_{2}) $, representing the \emph{inverse distance}~\cite{Atkeson1997}.
The first term in Equation~\ref{eq:fitmeta} represents the weighted distance between simulated results and estimated results.
The remaining terms in Equation~\ref{eq:fitmeta} guarantee the least square matrix to have a full rank.

\subsection{Algorithmic Steps}
\label{section:simulation_alg}

Following the framework proposed by Conn et al.~\cite{conn2009introduction} and its adaptation in Osorio and Bierlaire~\cite{osorio2013simulation}, I combine the metamodel with the derivative-free trust-region algorithm to solve the optimization problem in Equation~\ref{eq:traffic_opt}, which can now be expressed at any given iteration as follows:
\begin{equation}
	\begin{aligned}
		& \mathbf{x}_{new} = \underset{\mathbf{x}}{\text{argmin}}
		& & M(\mathbf{x};\alpha,\boldsymbol\beta,\rho) = \alpha T(\mathbf{x};\rho) + \Phi(\mathbf{x};\boldsymbol\beta), \\
		& \text{subject to}
		& & \lVert \mathbf{x}^*-\mathbf{x} \rVert_{2} \le \Delta, \; 0 \leq x \leq 1, \forall x \in \mathbf{x},
	\end{aligned}
	\label{eq:metaObj}
\end{equation}

\noindent where $\mathbf{x}^*$ is the best decision point so far and $\Delta$ is the current trust-region radius.

The specific algorithmic steps are the following:

\begin{itemize}
\item \textbf{Step 1:} Initialize $\mathcal{X}$ to contain 5 randomly sampled decision points, evaluate each of them using both the traffic equation and \textbf{simulator}, arbitrarily set $\mathbf{x}^*$ as a any element of $\mathcal{X}$, and compute $\alpha$ and $\boldsymbol\beta$ (Equation~\ref{eq:fitmeta}).
\item \textbf{Step 2:} Use Equation~\ref{eq:metaObj} to compute $\mathbf{x}_{new}$.

\item \textbf{Step 3:} Compute $\hat{F}(\mathbf{x}_{new};\rho)$ (\textbf{simulator run}). Compute the relative improvement $\tau = \frac{\hat{F}\left(\mathbf{x}_{new}\right)-\hat{F}\left(\mathbf{x}^*\right)}{M_{i}\left(\mathbf{x}_{new}\right)-M_{i}\left(\mathbf{x}^*\right)}$. If $\tau \geq 1e-3$, \textbf{accept} $\mathbf{x}_{new}$ and set $\mathbf{x}^* := \mathbf{x}_{new}$, otherwise \textbf{reject} $\mathbf{x}_{new}$. In any case, add $\mathbf{x}_{new}$ to $\mathcal{X}$, and compute $\alpha$ and $\boldsymbol\beta$ (Equation~\ref{eq:fitmeta}).

\item \textbf{Step 4:} If $\alpha$ and $\boldsymbol\beta$ have not changed much in \textbf{step 3}, i.e., $\frac{\lVert (\alpha_{new},{\boldsymbol\beta}_{new}) -(\alpha_{old},{\boldsymbol\beta}_{old}) \rVert}{\lVert (\alpha_{old},{\boldsymbol\beta}_{old}) \rVert} \leq 0.1$, add a new randomly sampled decision point to $\mathcal{X}$, evaluate it using both the traffic equation and \textbf{simulator}, and compute $\alpha$ and $\boldsymbol\beta$ (Equation~\ref{eq:fitmeta}).
\item \textbf{Step 5:} Update the trust-region radius:

\begin{equation}
	\Delta :=
	\begin{cases}
		\min\{1.1\times\Delta, 100\} & \text{if } \tau > 1e-3,  \\
		\max\{0.9\times\Delta, 0.1\} & \text{if } \tau \leq 1e-3 \; \text{and} \\
		& \text{$5$ consecutive rejections of $\mathbf{x}$}, \\
		\Delta & \text{otherwise.} \\
	\end{cases}
\end{equation}

\item \textbf{Step 6:} Exit the loop when the maximum number of allowed simulator runs ($20$) is reached; otherwise go to \textbf{step 2}.

\end{itemize}

\noindent When the algorithm stops, $\mathbf{x}^*$ can be taken to generate a simulation, thus ending dynamic data completion.

In practice, instead of running this procedure on the entire road network, I adopt a \emph{decomposition approach} and separate the network into sub-networks, which are going to be modeled independently.
For these sub-networks, I consider nodes with no predecessors as \emph{artificial origins} and nodes without successors as \emph{artificial destinations} (i.e., locations where vehicles respectively enter and exit a sub-network).
Additionally, observing that vehicles rarely navigate in loops, I extract directed acyclic graphs (DAG) from these sub-networks to operate.


%% file: siga/sections/7-results.tex
\section{Results}
\label{sec:siga-results}

In this section, I present evaluations of \emph{Initial Traffic Reconstruction} and \emph{Dynamic Data Completion}, and demonstrate the visualization and animation results.

\subsection{Evaluation of Initial Traffic Reconstruction}

I have generated abundant synthetic data for evaluating my approach. These synthetic datasets are produced via traffic models that have been extensively validated using real-world datasets in transportation engineering. Many newly proposed traffic models are evaluated using these models and real-world datasets. I have conducted my experiments in a similar vein. In this section, I provide details on the generation of the synthetic dataset.

\subsubsection{Road Network and GPS Dataset}
The road network (obtained from \url{http://openstreetmap.org/}) used in testing is from downtown San Francisco (Figure~\ref{fig:apx_sf}), which contains 5407 nodes, 1612 road segments, and 296 TAZs (obtained from \url{https://data.sfgov.org/}). The GPS dataset is obtained from the Cabspotting project~\cite{cabspotting} (Figure~\ref{fig:apx_sf} LEFT), in which the \emph{low-sampling-rate} is reflected as the average timestamp difference between consecutive points is approximately 60 seconds.

\begin{figure*}
	\centering
	\includegraphics[width=\textwidth]{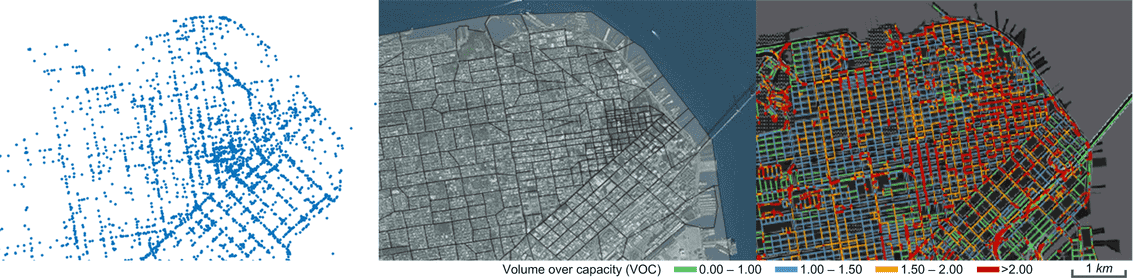}
	\caption{LEFT: Sample GPS points from the Cabspotting dataset. MIDDLE: Road maps of downtown San Francisco overlaid with traffic analysis zones (TAZs). RIGHT: A heuristic traffic condition established via the Timestamp model (the travel times are converted to flows). }
	\label{fig:apx_sf}
\end{figure*}

\subsubsection{Traffic Conditions via System Optimal Model}
I establish the first set of heuristic network travel times by solving the system optimal (SO) model~\cite{sheffi1985urban}. The SO model addresses the \emph{traffic assignment} problem by minimizing the entire travel time of a road network:

\begin{equation}
\begin{aligned}
& {\text{minimize}}
& & z(\mathbf{f}) = \sum_{e \in \mathcal{E}}{f(e)} t(e), \\
& \text{subject to}
& & u_{rs} = \sum_{k \in \mathcal{K}_{rs}} u_{rs}(k), \; \forall (r,s) \in \mathcal{O} \times \mathcal{D},\\
& \text{}
& & f(e) = \sum_{(r,s) \in \mathcal{O}\times\mathcal{D}}\sum_{k \in \mathcal{K}_{rs}} \delta_{e}(k) u_{rs}(k), \; \forall e \in \mathcal{E},\\
& \text{}
& & u_{rs} \geq 0, \; \forall (r,s) \in \mathcal{O} \times \mathcal{D}.
\end{aligned}
\label{eq:so}
\end{equation}

\noindent The solution to Equation~\ref{eq:so} is a set of flows and travel times of all road segments of a network. The key input is the OD pairs estimated by first setting $ \mathcal{O} = \mathcal{V} $ and then recording the number of in-and-out GPS traces for each TAZ. As GPS traces are usually sampled from a small percentage of the entire traffic population, thus representing a partial network flow, I multiply the estimated OD pairs by 10 constants and solve Equation~\ref{eq:so} accordingly. As a result, I have constructed 10 network travel times in which the corresponding congestion levels, measured by volume over capacity (VOC) (computed as $ \sum_{e \in \mathcal{E}} \frac{f(e)}{c(e)} $), range uniformly from 0.19 to 1.85.

\subsubsection{Traffic Conditions via Timestamp Model}
Heuristic network travel times can also be generated based on GPS timestamps. Using the Cabspotting dataset~\cite{cabspotting}, I equally distribute the time difference between two consecutive GPS points to all paths that connect them. For road segments that are covered by multiple GPS traces, the average travel times are adopted. Using this approach, I have produced 24 network travel times representing 24 hours in a typical weekday. An example can be seen at Figure~\ref{fig:apx_sf} RIGHT. I refer to this method of generating network travel times as the \emph{Timestamp} model.

\subsubsection{Synthetic GPS Traces}
Using established network travel times, I can generate synthetic GPS traces in which the true traversed paths and other information are encoded. In order to study the effect of the number of traces used in reconstruction on the estimation accuracy, I have randomly simulated 20 batches of synthetic traces from 50 to 1000 in increments of 50. Each batch contains 30 sets of GPS traces and all set contain 50 traces. As a result, I have generated $315~000$ traces for each traffic condition and over $10$ million traces in total. A synthetic trace is created by selecting a random source and a target in the network and routing with the shortest travel-time strategy. To mimic features of a real-world GPS dataset, the sampling rate is set to 60 seconds, and all coordinates are perturbed by the Gaussian noise $ (0,20)$ in meters~\cite{yuan2010interactive}.

\subsubsection{Evaluation and Comparison}
I compare my technique with two state-of-the-art methods, namely Hunter et al.~\cite{hunter2014large} and Rahmani et al.~\cite{rahmani2015non}. The first method also uses EM algorithm as of the inner loop of my \emph{travel-time estimation} process. In their work, the number of EM iterations is set to 5 and the number of random allocations per aggregate measurement is set to 100. These settings are reported to produce the highest estimation accuracy by Hunter et al.~\cite{hunter2014large}. The second method takes a non-parametric principle, using a kernel-based technique to estimate travel times. The weights used to allocate travel times to individual road segments are set to be the ratio of free-flow travel times among road segments~\cite{hellinga2008decomposing}.

I set parameters of my nested iterative process as follows: retaining the same settings for the inner loop as from Hunter et al.~\cite{hunter2014large}, I empirically set the number of iterations for the outer loop to 10. This setting is based on the results shown in Figure~\ref{fig:convergence}, where the relationship between the normalized convergence rate and the number of iterations for both types of network travel times is shown. Each datum in the plot is the average value computed using all network travel times across all sets of synthetic GPS traces using either the SO model (6000 trials) or the Timestamp model (14400 trials). The measurement of each trial is the mean square error (MSE) between a recovered and a ground-truth traffic condition (i.e., $ \frac{\sum_{e}(t_{e}-\hat{t}_{e})^2}{|\mathcal{E}|} $). As can be seen, the convergence rate decreases quadratically as the number of iterations increases and tends to flatten after 10 iterations.

\begin{figure}
	\centering
	\includegraphics[width=0.9\textwidth]{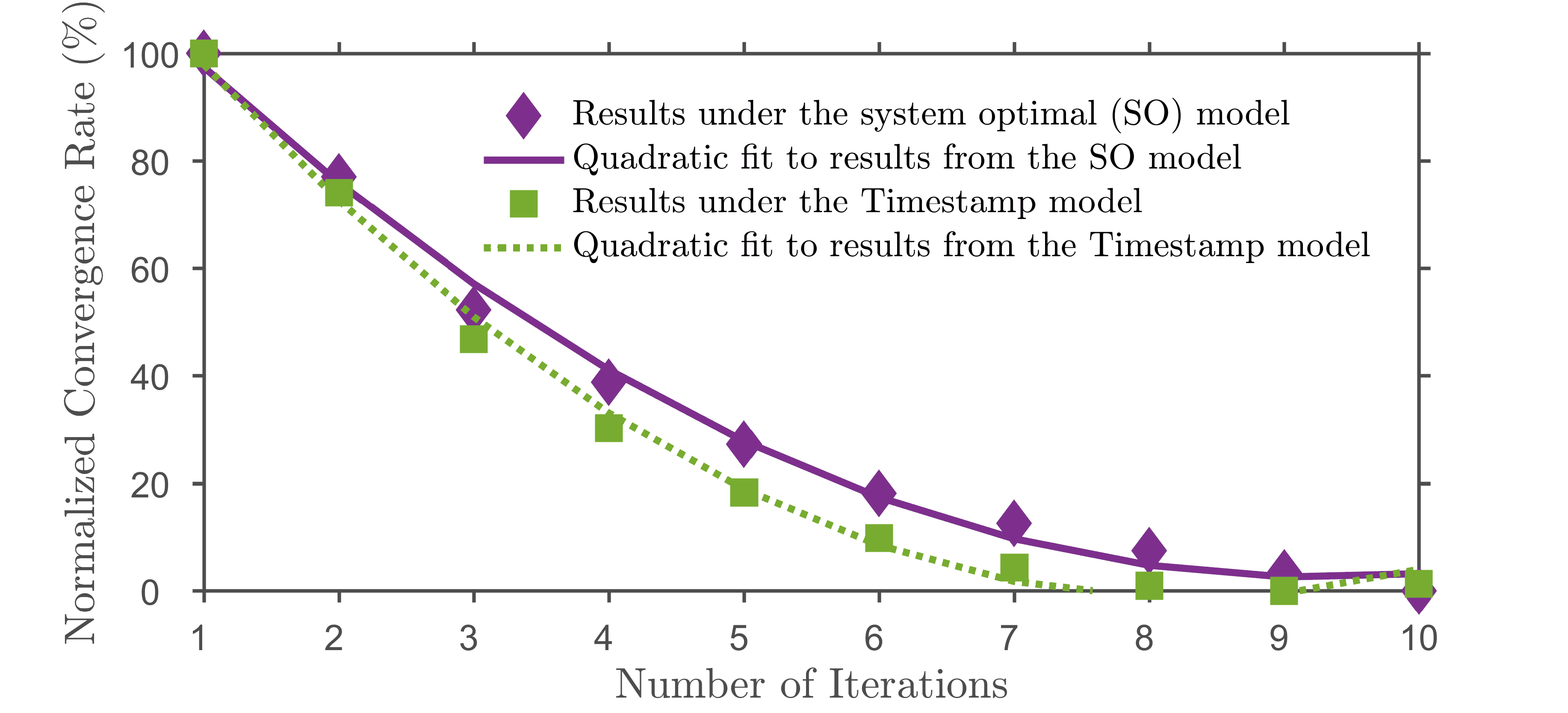}
	\caption{The relationship between the normalized convergence rate and the number of iterations of the outer loop of my iterative process is shown. The convergence rate decreases quadratically as the number of iterations increases and tends to flat after 10 iterations. }
	\label{fig:convergence}
\end{figure}

I evaluate my technique using three metrics. The first metric is the error rate of the aggregate travel time across the entire network, computed as $ \frac{|\sum_{e}\hat{t}_{e} - \sum_{e}t_{e}|}{\sum_{e}t_{e}}, \forall e \in \mathcal{E} $ where $ \hat{t}_{e} $ represents an estimated travel time and $ t_{e} $ represents a ground-truth travel time. The left diagram of Figure~\ref{fig:sf_analysis} TOP shows the results by averaging the experimental outcomes of all network travel times via the SO model.
The minimum error rate of my technique is 18\%, of Rahmani et al.~\cite{rahmani2015non} is 34\%, and of Hunter et al.~\cite{hunter2014large} is 48\%. Experimenting on network travel times generated via the Timestamp model, the corresponding minimum error rates are 8\%, 28\%, and 37\%, respectively, which are shown in the right diagram of Figure~\ref{fig:sf_analysis} TOP. As the number of synthetic GPS traces used in estimation increases, my technique demonstrates consistent advantages in performance over the other two methods.

\begin{figure*}
	\centering
	\includegraphics[width=\textwidth]{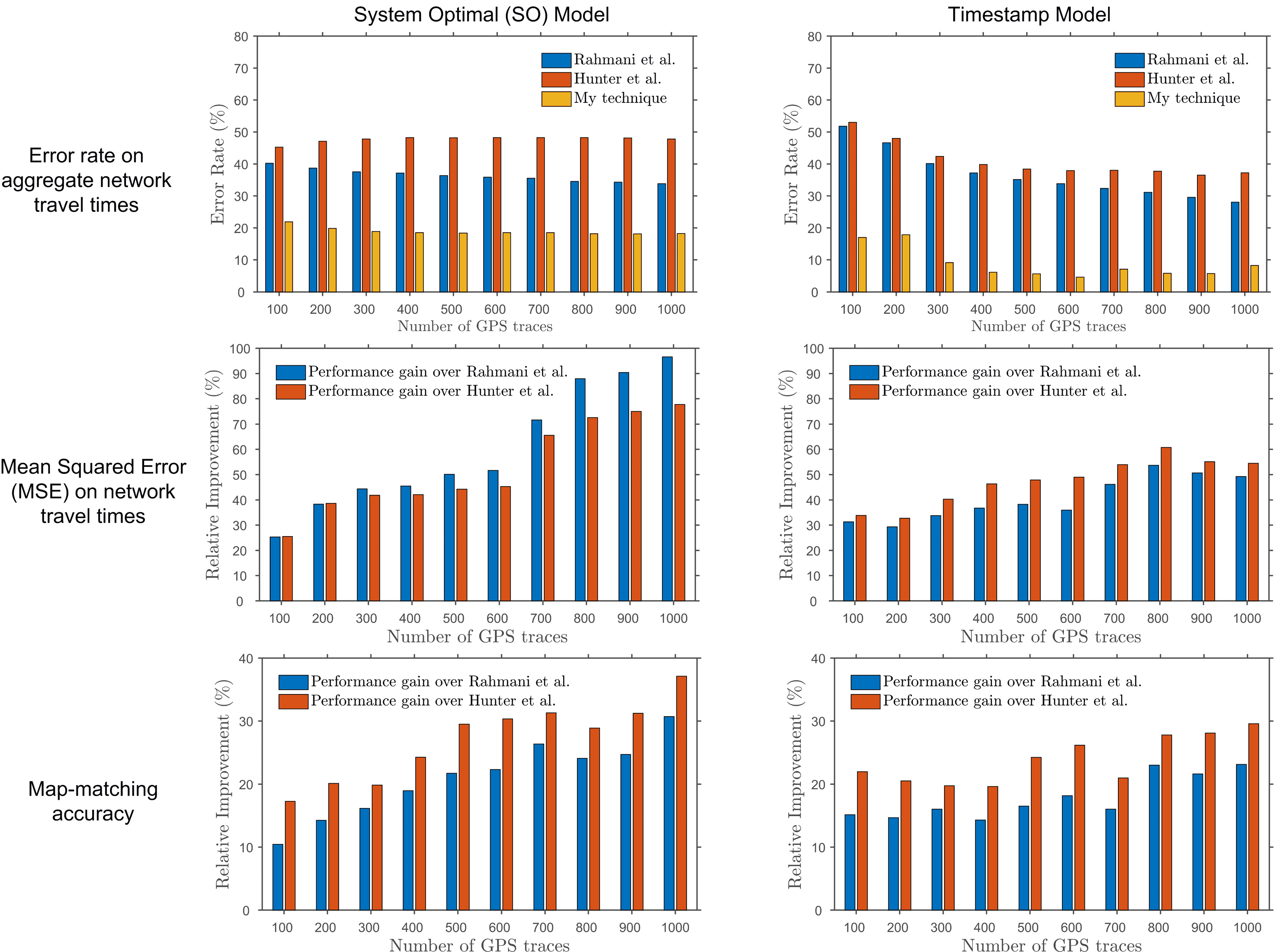}
	\caption{From TOP to BOTTOM, the left diagrams show the results generated using network travel times via the system optimal (SO) model, while the right diagrams show corresponding results of network travel times via the Timestamp model. TOP: The error rates (\%) of various methods of aggregating travel time across the network. MIDDLE: The relative improvements (\%) of travel times of all road segments measured in MSE. BOTTOM: The relative improvements (\%) of map-matching accuracy measured using successfully identification rates of road segments. In summary, my technique achieves consistent improvements over the other two methods as the number of GPS traces used in recovery increases.}
	\label{fig:sf_analysis}
\end{figure*}

The second metric is the relative improvement of my technique over the existing methods on travel times of all road segments. I compute this metric based on $ MSE =   \frac{\sum_{e}(t_{e}-\hat{t}_{e})^2}{|\mathcal{E}|} $ as follows:
\begin{equation}
RelativeImprovement = \frac{MSE_{existing} - MSE_{my}}{MSE_{my}}, 
\end{equation}

\noindent where $ MSE_{my} $ represents the error between a recovered traffic condition using my technique and the ground-truth traffic condition, and $ MSE_{existing} $ represents the error computed using an existing method with the same ground truth. The maximum relative improvements over Hunter et al.~\cite{hunter2014large} and Rahmani et al.~\cite{rahmani2015non} under the SO model (shown in the left diagram of Figure~\ref{fig:sf_analysis} MIDDLE) are 78\% and 97\%, and under the Timestamp model (shown in the right diagram of Figure~\ref{fig:sf_analysis} MIDDLE) are 54\% and 49\%. In general, with more synthetic GPS traces used in estimation, better relative improvements are achieved. Such effects are more apparent on the SO model than the Timestamp model.

The third metric evaluates the map-matching accuracy. For one trace, I calculate the success rate as follows: 

\begin{equation}
SR = \frac{\# \text{successfully identified road segments}}{\# \text{acutual road segments in the trace}}. 
\end{equation}

\noindent I sum all success rates generated using my method and an existing approach, and derive the relative improvement as follows: 

\begin{equation}
\frac{\sum SR_{my} - \sum SR_{existing}}{\sum SR_{my}}. 
\end{equation}

\begin{figure*}
	\centering
	\includegraphics[width=0.6\textwidth]{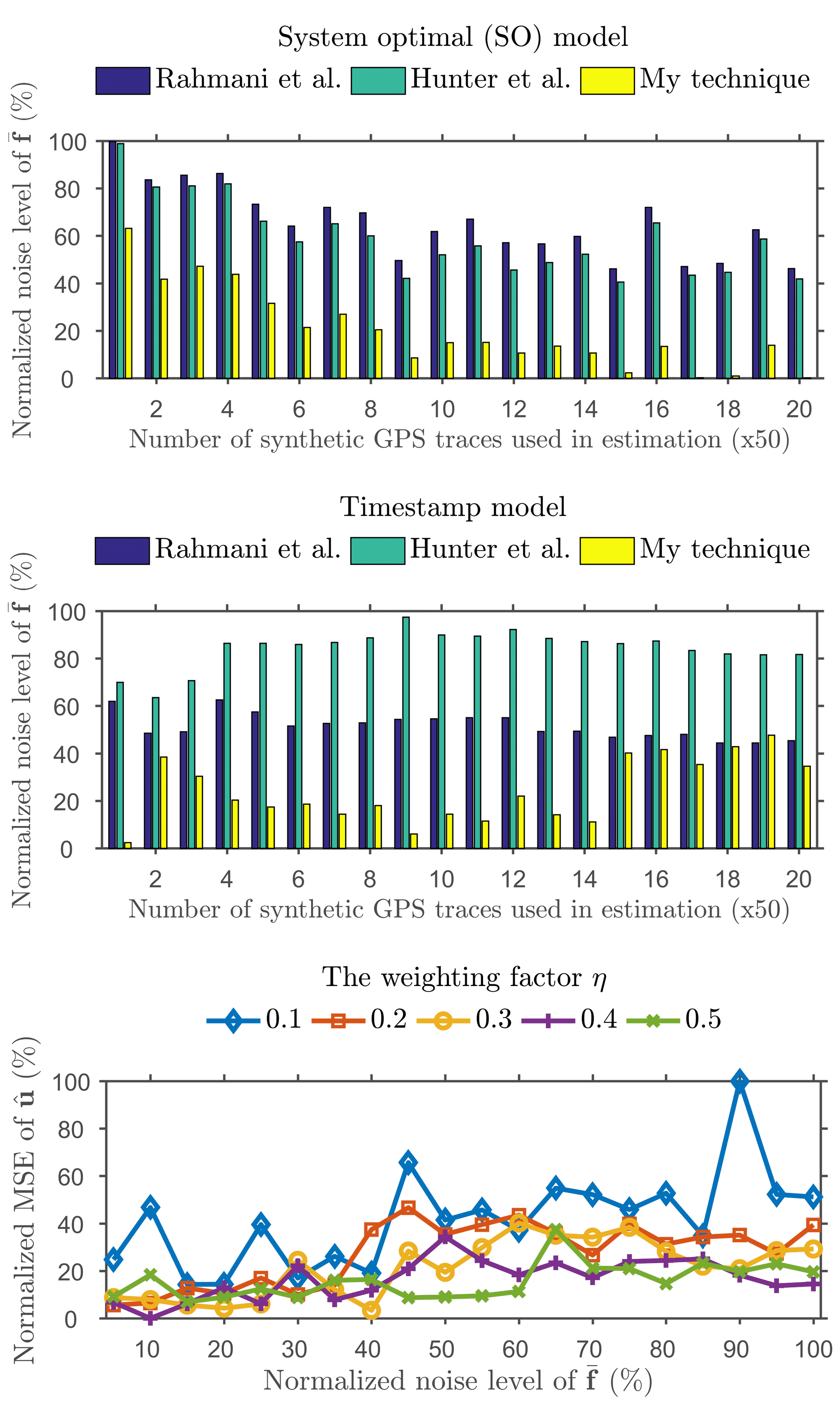}
	\caption{TOP and MIDDLE: The normalized noisy levels (\%) of target road-segment flows $ \bar{\mathbf{f}} $  computed according to the MSE of estimated network travel times to ground truth. In general, for both models, my technique produces lower noisy levels than the other two techniques. BOTTOM: The normalized MSE of target OD pairs $ \hat{\mathbf{u}} $ (\%) under different values of the weighting factor $ \eta $ and various noisy levels of $ \bar{\mathbf{f}} $ (\%). When $ \eta $ is small, the error is more sensitive to perturbations on $ \bar{\mathbf{f}} $. Overall, the error increases as the noisy level of $ \bar{\mathbf{f}} $ increases. For all studies, the normalized noisy level of target OD pairs $ \bar{\mathbf{u}} $ has been set to $ 50\% $. My method has achieved consistently lower error rates compared to the other two methods.}
	\label{fig:verror2}
\end{figure*}

\noindent The maximum relative improvements of my method over Hunter et al.~\cite{hunter2014large} and Rahmani et al.~\cite{rahmani2015non} under the SO model are 28\% and 34\%, and under the Timestamp model are 19\% and 25\%. These results are shown in Figure~\ref{fig:sf_analysis} BOTTOM. Again, as the number of GPS traces used in recovering network travel times increases, gains in the improvements are observed.

\subsubsection{Analysis of Bilevel Optimization}

I have shown that my approach outperforms the existing methods on estimating travel times of a road network. In turn, by inverting Equation~\ref{eq:bpr}, I can obtain better estimations of target flows, which serve as the input to the bilevel optimization program. 

The factors affecting the program are the weighting factor $ \eta $, the noisy level of target OD pairs $ \bar{\mathbf{u}} $, and the noisy level of target road-segment flows $ \bar{\mathbf{f}} $. The noises of $ \bar{\mathbf{u}} $ and $ \bar{\mathbf{f}} $ are assumed to have zero means and diagonal variance-covariance matrices~\cite{cascetta1988unified}. In practice, the noisy level of $ \bar{\mathbf{u}} $ is difficult to assess because not only $ \bar{\mathbf{u}} $ usually comes from existing data but also the true values of $ \bar{\mathbf{u}} $ are usually unknown. Due to these reasons, in the analysis of $ \eta $ and the noisy level of $ \bar{\mathbf{f}} $, I set the normalized noisy level of $ \bar{\mathbf{u}} $ to $ 50\% $. Subsequently, the normalized noisy levels of $ \bar{\mathbf{f}} $ (\%) computed based on MSE of estimated travel times to ground truth are shown in Figure~\ref{fig:verror2} TOP and MIDDLE. In general, my technique produces lower noisy levels of $ \bar{\mathbf{f}} $ than the other two techniques, especially under the SO model which is considered to be a better approximation to real-world traffic than the Timestamp model~\cite{sheffi1985urban}.

In order to evaluate how $ \eta $ and the noisy level of $ \bar{\mathbf{f}} $ affect the estimation accuracy of $ \hat{\mathbf{u}} $ (i.e., the estimated OD pairs), I compute the normalized MSE of $ \hat{\mathbf{u}} $ (\%) under different $ \eta $ and various noisy levels of $ \bar{\mathbf{f}} $ (\%). The results are shown in Figure~\ref{fig:verror2} BOTTOM. When $ \eta $ takes a small value (e.g., 0.1), the impact of $ \bar{\mathbf{u}} $ is restricted, thus the MSE of $ \hat{\mathbf{u}} $ reacts actively to perturbations on $ \bar{\mathbf{f}} $. As I gradually increase the value of $ \eta $, the impact of  $ \bar{\mathbf{f}} $ attenuates. Nevertheless, the MSE of $ \hat{\mathbf{u}} $ increases as the noisy level of  $ \bar{\mathbf{f}} $ progresses.


\subsection{Evaluation of Dynamic Data Completion}
I compare my technique to the approach that only uses traffic simulation on various OD demands and road networks. The experiments are conducted by starting at the center of the road network of San Francisco and gradually increase the radius to retrieve networks with $20$ to $5~000$ road segments. This step results in $48$ road networks. For all DAGs constructed in all networks, the OD demand varies from $1000~vehicle/hour$ to $10000~vehicle/hour$ with an increment of $1000~vehicle/hour$.

In Figure~\ref{fig:meta_error}, I show the accuracy of my technique compared to the approach using only the traffic simulation (i.e., the simulation-only approach). For each road network in the experiment, I randomly select an intersection and assign turning ratios for this intersection. The assigned turning ratios (i.e., a decision point) are treated as the ground truth for the two approaches to recover. Since the most accurate method for recovering a traffic condition is the microscopic simulator, I use the error level of the microscopic simulator to the ground truth as the baseline. For each decision point, I first compute the difference between the recovered value using my technique to the ground truth, as well as the difference between the recovered value by the simulator to the ground truth. Then, I subtract these two error differences and represent the result as a gray cross in Figure~\ref{fig:meta_error}. The mean, minimum, and maximum values of the average deviations to the simulator (indicated by the solid line in Figure~\ref{fig:meta_error}), are respectively 7.8\%, 0\%, and 13\%. In many cases, my technique even outperforms the simulation-only approach with a smaller difference to the ground truth, as indicated by a negative value in Figure~\ref{fig:meta_error}.
3

\begin{figure}
	\centering
	\includegraphics[width=\textwidth]{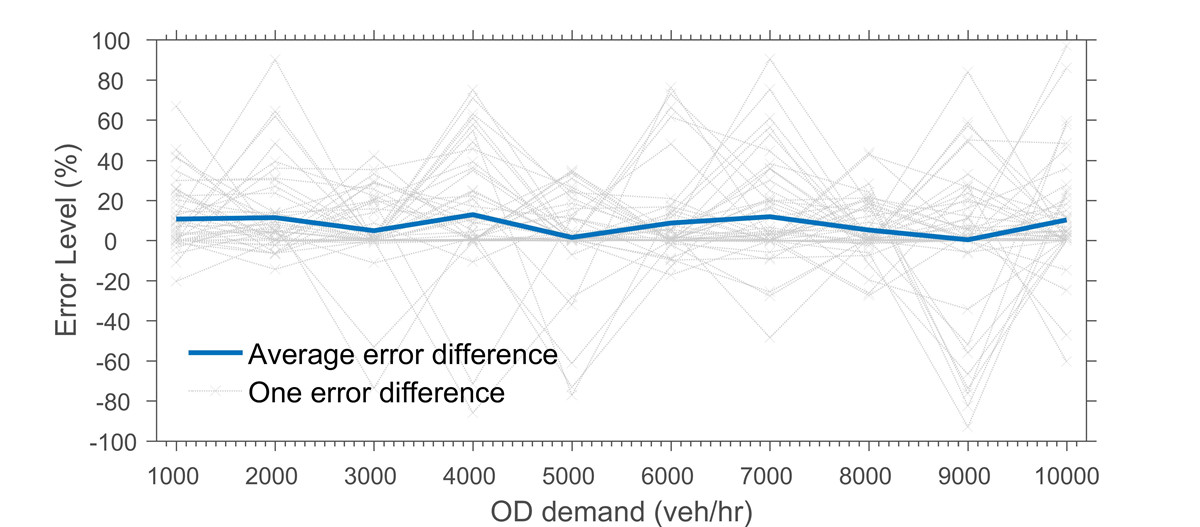}
	\caption{The error level of my technique vs. simulation-only approach: 
For a given road network and a specific OD demand, I first compute the differences between the two methods with respect to the ground truth. Then, I subtract these two differences to obtain one error difference measure (indicated by a gray cross). The mean, minimum, and maximum values of the average error level (indicated by the solid line) are respectively 7.8\%, 0\%, and 13\%. In many cases, my technique even outperforms the simulation-only approach with a smaller difference to the ground truth, as indicated by a negative value in the diagram.}
	\label{fig:meta_error}
\end{figure}

\begin{figure}
	\centering
	\includegraphics[width=\textwidth]{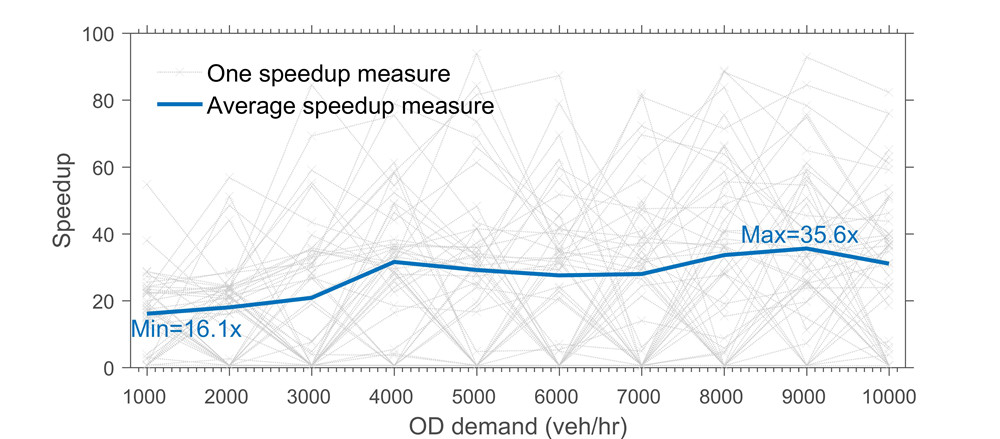}
	\caption{The performance speedup of my technique over the simulation-only approach: my technique is \emph{on average} about $27.2x$ faster, with maximum and minimum performance gains of $35.6x$ and $16.1x$. The maximum observed performance gain of a single speedup measure is \emph{over 90x} (at OD demand = 9000).}
	\label{fig:meta_speedup}
\end{figure}

In the second study, I analyze the performance speedup achieved by my technique over the simulation-only approach. The results are shown in Figure~\ref{fig:meta_speedup}. Each gray cross indicates a speedup measure for a particular road network and a specific OD demand. The highest speedup is over $90x$. On average, the maximum speedup is about $35.6x$ and the minimum speedup is about $16.1x$. There also exist several cases where my technique demonstrates rather negligible speedups (i.e., values close to 0). This is because the initial random guess of a particular decision point is close to the ground truth, in which circumstances either method can achieve a quick convergence. In reality, I expect these situations happen rarely. Finally, in nearly all experiments, my technique converges in the first five iterations.

\subsection{Traffic Visualization and Animation}

To demonstrate the effectiveness of my approach, I use the Cabspotting dataset~\cite{cabspotting} (GPS data) in San Francisco to perform traffic visualization and animation. This dataset recorded driving history from 536 taxicabs, which includes \emph{over 11 million} GPS traces in total.

In Figure~\ref{fig:result_vis} (top row), I show the visualization results over four time intervals: Sunday 9AM, Tuesday 9AM, Thursday Noon, and Friday 7PM, which exemplify \emph{weekend morning traffic}, \emph{weekday morning traffic}, \emph{weekday mid-day traffic}, and \emph{weekday evening traffic}, respectively. In San Francisco, first, the congestion level of Sunday 9AM is low compared to the rest of the time intervals across all areas. Second, the congestion of Tuesday 9AM is more severe in the north-west, central-west, and central-east areas (residential regions) than the same areas in other time intervals. Lastly, the north-east region (downtown commercial and financial districts) of Tuesday 9AM and Friday 7PM appear to have more severe congestion than other time intervals.

\noindent
In addition to the qualitative results, I have quantitatively compared my reconstructed results to the loop-detector data (obtained from \url{http://pems.dot.ca.gov/}) extracted from the same location and time periods as the Cabspotting dataset. The loop-detector data represent relatively accurate measurements and are often regarded as the benchmark for evaluating GPS-based estimations~\cite{work2010traffic}. 

I have additionally used a filtering process based on the observation that traffic exhibits periodic patterns, which I introduced in Chapter~\ref{ch:itsm} and are shown in Figure~\ref{fig:result_vis}. One way to construct such a filter is to transform all traffic signals to the frequency domain and take the average of all frequency components. While this approach can expect to reduce the embedded white noise, it does not consider other types of noise in the data. This is illustrated in the top panel of Figure~\ref{fig:filter}\textbf{(c)}: the most significant frequency component is captured while other harmonics are reduced to various degrees. 

\begin{figure*}[hb]
\centering
\includegraphics[width=\textwidth]{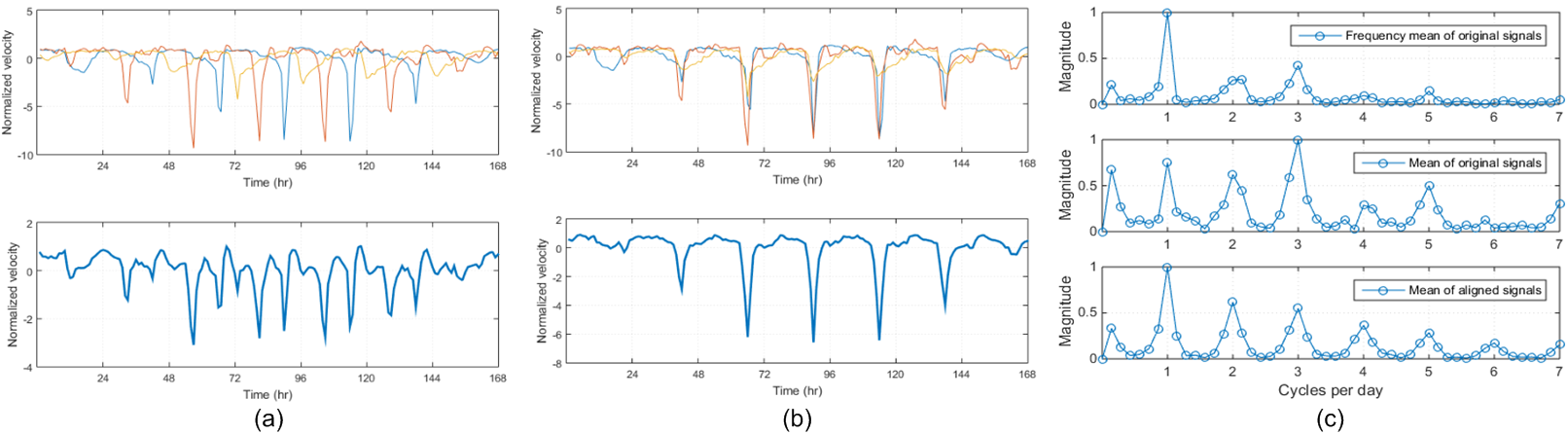}
\caption{(a) TOP: several loop-detector signals are plotted showing phase shifts among them; BOTTOM: the average signal of the signals shown in the top panel. (b) TOP: aligned loop-detector signals according to their phase responses; BOTTOM: the average signal of the aligned signals shown in the TOP panel. (c) TOP: the averaged frequency of the signals in (a) TOP, which shows several frequency component are getting degraded; MIDDLE: the frequency of the signal in (a) BOTTOM, which shows inconsistent magnitude ratios; BOTTOM: the frequency of signal in (b) BOTTOM, which shows prominent frequencies and magnitude ratios.}
\label{fig:filter}
\end{figure*}

Considering time-domain only methods, a na\"{i}ve approach is to get the average signal, transform it to the frequency domain, and analyze its frequency components. Though this is a straightforward way in dealing with univariate and multivariate data, the resulting signal could be a poor summary of original signals both in the time domain (Figure~\ref{fig:filter}\textbf{(a)}) and in the frequency domain (Figure~\ref{fig:filter}\textbf{(c)} MIDDLE). The reason is that this approach does not take the phase variability into account. To address this, I align signals in the time domain, calculate the average signal, transform it to the frequency domain, and extract its major frequencies (see Figure~\ref{fig:filter}\textbf{(b)} and Figure~\ref{fig:filter}\textbf{(c)} bottom panel). Using this approach, I obtain important frequency components and their corresponding magnitudes in the right ratio. The frequency-domain version of this signal then serves as my filter.

\begin{figure*}
	\centering
	\includegraphics[width=\textwidth]{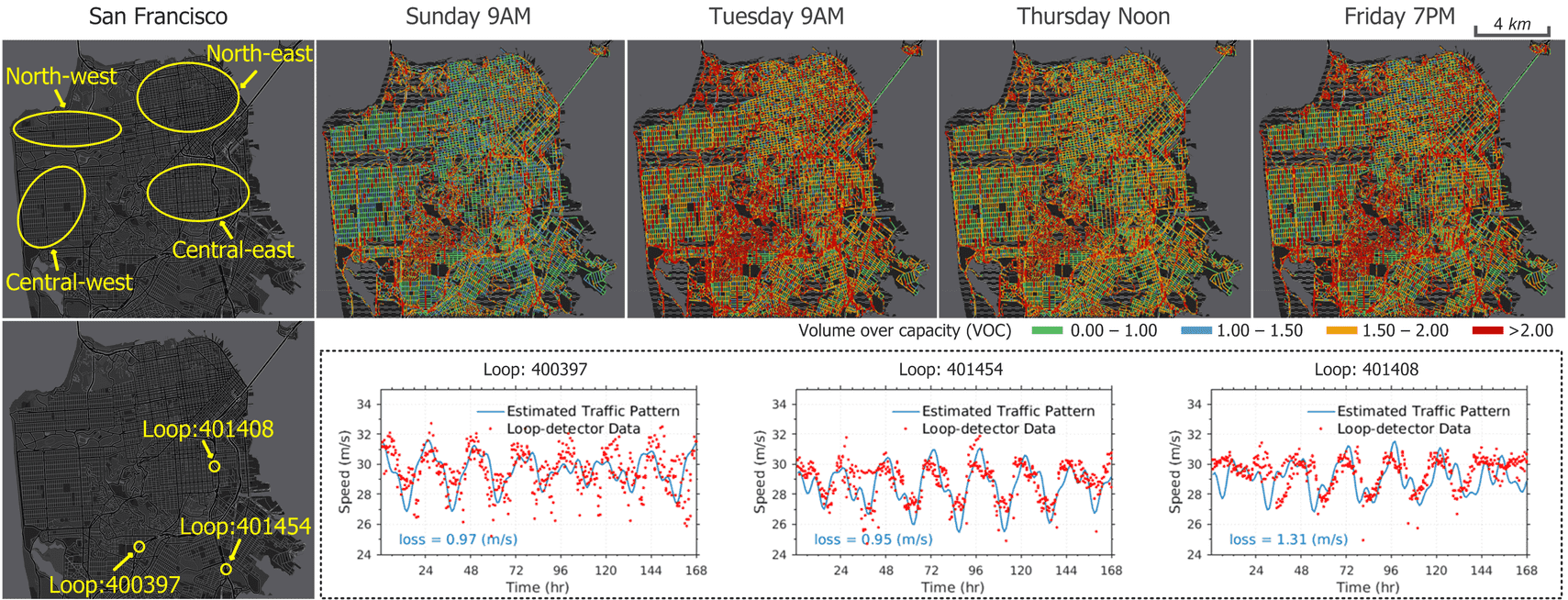}
	\caption{Qualitative (visualization) and quantitative analysis of traffic in San Francisco.  Top: four time periods, namely Sunday 9AM, Tuesday 9AM, Thursday Noon, and Friday 7PM, are adopted to illustrate weekend vs. weekday and morning vs. evening traffic. The traffic is measured by \emph{volume over capacity} (VOC) (computed as $ \sum_{e \in \mathcal{E}}\left(f(e)/c(e)\right) $ where $ c(e) $ is the capacity of the road segment $ e $ defined in Equation~\ref{eq:bpr}). Bottom: I compare my reconstructed results using GPS data (after the filtering process explained in this section) to the data from three loop detectors in San Francisco. The results show small losses (around $ 1~m/s $) in the speed measurement.}
	\label{fig:result_vis}
\end{figure*}

After applying this filtering process, my reconstruction can approximate the accurate loop-detector readings with small losses (around $1~m/s$) in the speed measurement. This validation result can be found in the bottom row of Figure~\ref{fig:result_vis}.

\begin{figure*}
	\centering
	\includegraphics[width=\textwidth]{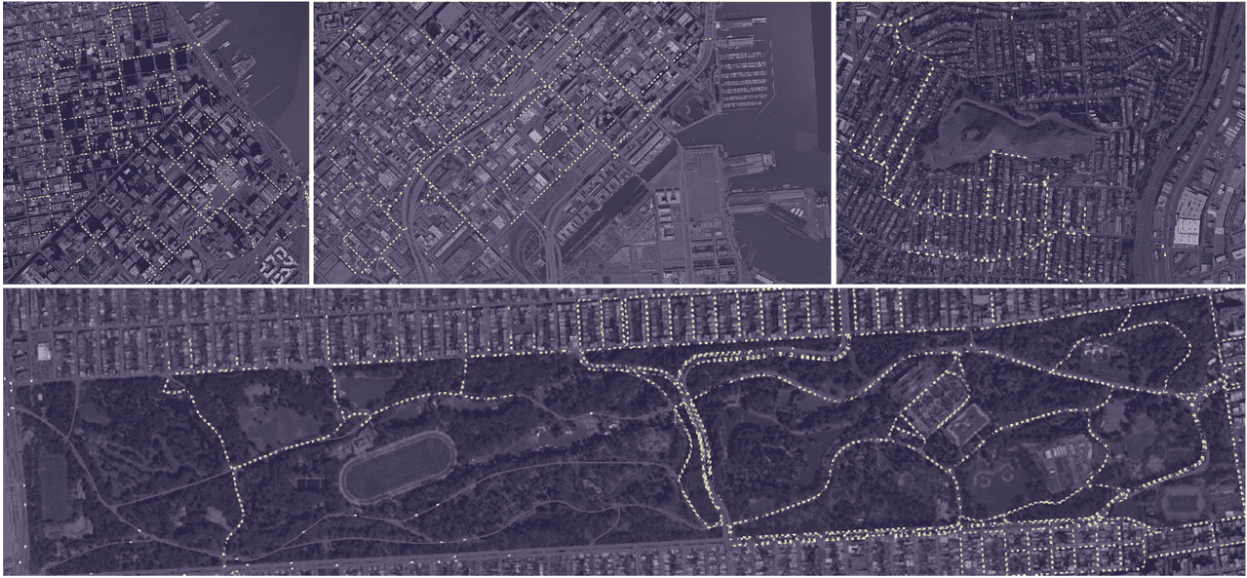}
	\caption{2D traffic animation of regions in San Francisco:
Northeast (top left), Central-East (top center), Central (top Right), 
Northwest (bottom). I have exaggerated the headlights and adopted an evening time period (i.e., Friday 7PM) to make vehicles more visible.}
	\label{fig:result_2d}
\end{figure*}

As a result of the dynamic data completion, I obtain turning ratios that can lead to simulations that respect the estimated traffic conditions in areas with GPS data coverage. The results can be illustrated using both 2D and 3D traffic animation for various virtual-world applications. In Figure~\ref{fig:result_2d}, I provide 2D traffic animation of four regions in San Francisco\footnote{The number of lanes of a road segment is decided by the digital map.} using the reconstructed traffic on Friday at 7PM. This 2D animation can be used to study dynamic traffic patterns at a metropolitan scale. The downtown area is further modeled in a 3D \emph{Virtual San Francisco} to showcase the potential of embedding real-world traffic in a virtual world for immersive VR experiences and virtual tourism applications (see Figure~\ref{fig:result_3d}). It is worth mentioning that although I have used the GPS data from one city throughout my experiments, my approach is independent from data sources, as the effectiveness is demonstrated using both synthetic datasets (Figure~\ref{fig:sf_analysis}) and real-world datasets (Figure~\ref{fig:result_vis}).

\begin{figure*}
	\centering
	\includegraphics[width=\textwidth]{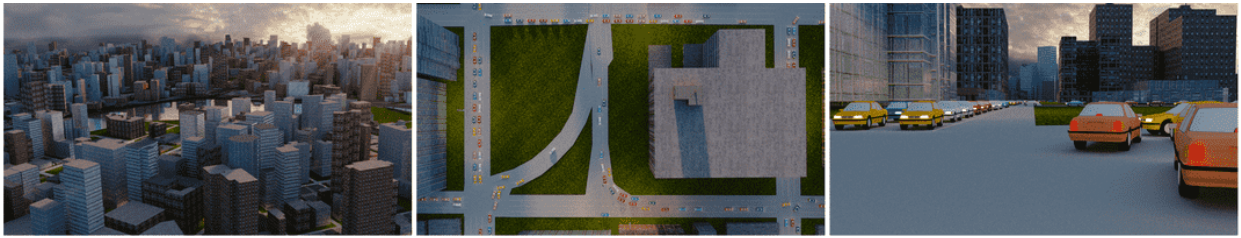}
	\caption{3D traffic animation: a perspective overview (left), a topdown view (center), and a driver's view (right).}
	\label{fig:result_3d}
\end{figure*}

%% file: siga/sections/8-conclude.tex
\section{Summary and Future Work}
\label{sec:siga-conclusion}
I have presented an effective algorithm to reconstruct city-scale traffic from GPS data using statistical learning. To address the issues with incomplete and/or sparse data, a metamodel-based simulation optimization is proposed to dynamically bridge the ``gap'' between the reconstructed traffic learned from GPS data and the simulated traffic in areas with incomplete or missing data. My approach is able to perform visualization of city-scale traffic, as well as data-driven 2D and 3D traffic animation in a virtual environment. With more GPS datasets being made available and released to public, e.g., Mobile Century~\cite{HERRERA2010568}, T-Drive~\cite{t-drive}, GeoLife~\cite{geolife}, and Uber Movement~\cite{ubermovement}, I believe future research would be abundant.

Although the proposed approach is specialized for traffic reconstruction, similar approaches can be developed to reconstruct aggregate dynamics of other multi-agent systems using spatial-temporal data, such as schools of fishes, flocks of birds, and swarms of insects. More importantly, the idea of learning from mobile sensor data as well as the concept of using metamodel-based optimization to refine the simulation parameters and to accelerate local fitting for large-scale motion reconstruction are generalizable to many other topics in computer graphics and beyond.

The main limitation of this work is that the reconstruction accuracy is limited by the available data. While my approach can maintain the accuracy down to the road-segment level, high frequency vehicle positions can not be modeled precisely, as such information is largely missing from the data. Another limitation of this work is that while my technique can satisfy the flow conservation in each sub-network used in computation, such a relationship is difficult to maintain at a city scale. This issue may be alleviated if more in-road sensors are installed on arterial roads to provide broad and accurate traffic measurements. 


There are a number of future directions. Algorithmically, one extension is to combine other data sources (in-road sensors, video streams, or surveying) 
with GPS traces to further improve the reconstruction accuracy. Another possibility is to incorporate a macroscopic traffic simulator, so that we can dynamically switch between simulators of varying fidelities to further reduce the computational cost. Application-wise, a virtual tourism system, a route planning~\cite{wilkie2011self} and navigation system~\cite{wilkie2014participatory}, or an autonomous vehicle training system can immediately benefit from my framework through incorporating and visualizing the reconstructed traffic.

With Chapter~\ref{ch:itsm} and Chapter~\ref{ch:siga}, I have finished introducing my efforts on \emph{city-scale traffic reconstruction}. Traffic is an aggregate form and its dynamics can only be studied collectively. Examining traffic at the scale of a city is necessary if we want to analyze congestion causes, identify network bottlenecks, and experiment with novel transport policies. Furthermore, traffic is form by individual vehicles. If we can improve the safety and efficiency of individual vehicles by converting them from human-driven to autonomous (given over 90\% crashes are due to some kind of human errors), together, these autonomous vehicles have the potential to alleviate the severe traffic condition we are facing. Next chapter introduces my efforts on improving autonomous driving with a focus on enabling the autonomous vehicle to navigate in dangerous situations including accidents.

%% file: chapters/ADAPS.tex
\chapter{ADAPS: AUTONOMOUS DRIVING VIA PRINCIPLED SIMULATIONS}
\label{ch:adaps}

\input{adaps/sections/1-intro}

\input{adaps/sections/2-related}

\input{adaps/sections/3-prelim}
\input{adaps/sections/4-adaps}

\input{adaps/sections/5-policy}
\input{adaps/sections/6-accident}
\input{adaps/sections/7-exp}
\input{adaps/sections/8-conclude}

%% file: adaps/sections/1-intro.tex
\section{Introduction}
\label{sec:intro}

Autonomous driving is a complex task, consider the dynamics of an environment and the lack of accurate definitions of various driving behaviors. These characteristics lead to conventional control methods to suffer subpar performance on the task~\cite{ratliff2009learning,silver2010learning}. Nevertheless, driving can be easily demonstrated by human experts. This observation has inspired \emph{imitation learning}, which leverages expert demonstrations to synthesize a controller. 

While there are many advantages of using imitation learning, it also has drawbacks. For autonomous driving, the most critical one is \emph{covariate shift}, meaning the training and test distributions are different. This could lead autonomous vehicles (AVs) to accidents since a learned policy may fail to respond to unseen scenarios including those dangerous situations that do not occur often. 

In order to mitigate this issue, the training dataset needs to be augmented to cover a wide spectrum of driving scenarios with expert demonstrations---especially ones of significant safety threats to the passengers---so that a policy can learn how to recover from its own mistakes. This is emphasized by Pomerleau~\cite{pomerleau1989alvinn}, who synthesized a neural network based controller for AVs: ``the network must not solely be shown examples of accurate driving, but also how to recover (i.e. return to the road center) once a mistake has been made.'' 

However, obtaining recovery data from accidents in the physical world is impractical, due to the high cost of a vehicle and potential injuries to both passengers and pedestrians. In addition, even one managed to collect accident data, if we need human experts to label them, the process can be inefficient and subject to judgmental errors~\cite{ross2013learning}. 

These difficulties naturally lead us to the virtual world, where accidents can be simulated and analyzed much freely~\cite{Chao2019Survey}. I have developed ADAPS (\textbf{A}utonomous \textbf{D}riving Vi\textbf{a} \textbf{P}rincipled \textbf{S}imulations) to achieve this goal. ADAPS consists of two simulation platforms and a memory-enabled hierarchical control policy. The first simulation platform, referred to as \emph{SimLearner}, runs in a 3D environment and is used to test a learned policy, simulate accidents, and collect training data. The second simulation platform, referred to as \emph{SimExpert}, acts in a 2D environment and serves as the ``expert'' to analyze and resolve an accident via \emph{principled simulations}, which can plan alternative safe trajectories for a vehicle by taking its kinematic and dynamic constraints into account.


Furthermore, ADAPS represents a more efficient online learning mechanism than existing methods such as DAGGER~\cite{ross2011reduction}. This is useful consider learning to drive requires iterative testing and update of a control policy. Ideally, we want to obtain a robust policy using minimal online iterations, as one iteration corresponds to one incident (otherwise there is no need to update a policy by going into the next iteration). This would require the generation of training data at each iteration to be \emph{accurate}, \emph{efficient}, and \emph{sufficient} so that a policy can gain a large improvement after this iteration. ADAPS can be of assistance to achieve this goal.

The \textbf{main contributions} of this chapter are specifically. (1) The accidents generated in \emph{SimLearner} will be analyzed by \emph{SimExpert} to produce alternative safe trajectories. (2) These trajectories will be automatically processed to generate a large number of annotated and segmented training data. Because \emph{SimExpert} is parameterized and has taken the kinematic and dynamic constraints of a vehicle into account (i.e., principled), the resulting training data are not only more heterogeneous than the data collected from running a learned policy multiple times, but also more effective than the data collected through random sampling. (3) Both theoretical and experimental results are provided to show that ADAPS is an efficient online learning mechanism.

%% file: adaps/sections/2-related.tex
\section{Related Work}
\label{sec:related}

I sample previous studies that are related to each aspect of my framework and discuss the differences within.

\textbf{Autonomous Driving}. Among various methods to plan and control an AV~\cite{schwarting2018planning}, I focus on end-to-end imitation learning as it can avoid manually designed features and lead to a more compact policy compared to conventional mediation perception approaches~\cite{chen2015deepdriving}. The early studies done by Pomerleau~\cite{pomerleau1989alvinn} and LeCun et al.~\cite{LeCun2006off} have shown that neural networks can be used for an AV to achieve lane-following and off-road obstacle avoidance. Due to the advancements of deep neural networks (DNNs), a number of studies have emerged~\cite{Bojarski2016,Xu2017end,pan2017agile,codevilla2017end}. While significant improvements have been made, these results mainly restrict a vehicle to either the lane-following behavior~\cite{codevilla2017end} or the \emph{off-road} collision avoidance~\cite{LeCun2006off}. My policy, in contrast, enables an AV to learn from accidents and allows it to achieve \emph{on-road} collision avoidance with both static and dynamic obstacles.

\textbf{Hierarchical Control Policy}. There have been many efforts in constructing a hierarchical policy to control an agent at different stages of a task~\cite{barto2003recent}.
Example studies include the options framework~\cite{sutton1999between} and transferable motor skills~\cite{konidaris2012robot}.
When combined with DNNs, the hierarchical approach has been adopted for virtual characters to learn locomotion tasks~\cite{levine2013guided}.
In these studies, the goal is to discover a hierarchical relationship from complex sensorimotor behaviors. 
I apply a hierarchical and memory-enabled policy to autonomous driving based on multiple DNNs. 
My policy enables an AV to continuously categorize the road condition as safe or dangerous, and execute corresponding control commands to achieve accident-free driving.
 
\textbf{Generative Policy Learning}. Using \emph{principled simulations} to assist learning can be considered as taking a \emph{generative model} approach. Several studies have adopted the same principle for deriving the (near-)optimal policy, examples including Function Approximations~\cite{gordon1995stable}, Sparse Sampling~\cite{kearns2002sparse}, and Fitted Value Iteration~\cite{szepesvari2005finite}. These studies leverage a generative model to \emph{stochastically} generate training samples. The emphasize is to simulate the feedback from an environment assuming the reward function is known. My system, on the other hand, does not assume any reward function of a driving behavior but models the kinematic and dynamic constraints of a vehicle, and uses simulation to plan its trajectories with respect to environment characteristics. In essence, my method acquires a control policy through learning from simulated expert demonstrations rather than from an agent's self-exploration~\cite{lin1992self} as of the previous studies.

%% file: adaps/sections/3-prelim.tex
\section{Preliminaries}
\label{sec:prelim}

Autonomous driving is a \emph{sequential prediction} and \emph{controlled} (SPC) task, for which a system must predict a sequence of control commands based on inputs that depend on past predicted control commands. Because the control and prediction processes are intertwined, SPC tasks often encounter \emph{covariate shift}, meaning the training and test distributions vary. Next, I will introduce notation and definitions to formulate an SPC task, and briefly discuss its existing solutions. Note that, I use ``state'' and "observation'' interchangeably for proofs in this and later sections.

The problem being consider is a $ T $-step control task. Given the observation $ \phi=\phi(s) $ of a state $ s $ at each step $ t \in [\![1, T]\!] $, the goal of the learner is to find a policy $ \pi \in \Pi $ such that its produced action $ a=\pi(\phi)$ will lead to the minimal cost:

\begin{equation}
\hat{\pi} = \argmin_{\pi \in \Pi} \sum_{t=1}^{T} C\left(s_{t}, a_{t}\right),
\label{eq:cost1}
\end{equation}

\noindent where $ C\left(s,a\right) $ is the expected immediate cost of performing $ a $ in  $ s $. For many tasks such as driving, we may not know the true value of $ C$. Instead, the observed surrogate loss $ l(\phi,\pi, a^*) $ is commonly minimized, which is assumed to upper bound $ C $, based on the approximation of the learner's action $ a=\pi(\phi) $ to the expert's action $ a^*=\pi^*(\phi) $. The distribution of observations at $ t $ is denoted as $ d_{\pi}^t $, which is the result of executing $ \pi $ from $ 1 $ to $ t-1 $. $ d_{\pi}=\frac{1}{T}\sum_{t=1}^{T}d_{\pi}^t $ is then the average distribution of observations by executing $ \pi $ for $ T $ steps. My goal is to solve an SPC task by obtaining $ \hat{\pi} $---a policy that can minimize the observed surrogate loss under its own induced observations with respect to expert's actions for those observations:

\begin{equation}
\hat{\pi} = \argmin_{\pi \in \Pi} \mathbb{E}_{\phi \sim d_{\pi}, a^* \sim \pi^*(\phi)} \left[l\left(\phi,\pi, a^*\right)\right].
\end{equation}

\noindent I further denote $ \epsilon =  \mathbb{E}_{\phi \sim d_{\pi^*}, a^* \sim \pi^*(\phi)} \left[l\left(\phi,\pi, a^*\right)\right] $ as the expected loss under the training distribution induced by the expert's policy $ \pi^* $, and the cost-to-go over $ T $ steps of $ \hat{\pi} $ as $ J\left(\hat{\pi}\right) $ and of $ \pi^* $ as $ J\left(\pi^*\right) $. 

By simply treating expert demonstrations as independently and identically distributed (i.i.d.) samples, the discrepancy between $ J\left(\hat{\pi}\right) $ and $ J\left(\pi^*\right) $ is $ \mathcal{O}(T^2\epsilon) $~\cite{syed2010reduction,ross2011reduction}. Given the error of a typical supervised learning is $ \mathcal{O}\left(T\epsilon\right) $, this demonstrates the additional cost due to covariate shift when solving an SPC task via standard supervised learning. I adapt and simplify the proof from Ross et al.~\cite{ross2011reduction} to show that solving an SPC task using standard supervised learning will lead to $ \mathcal{O}(T^2\epsilon) $ error. This proof will better prepare readers for understanding the guarantees of ADAPS presented in Section~\ref{sec:adaps}.

\begin{theorem}
	\cite{ross2011reduction} Consider a $ T $-step control task. Let $ \epsilon = \mathbb{E}_{\phi \sim d_{\pi^*}, a^* \sim \pi^*(\phi)} \left[l\left(\phi,\pi, a^*\right)\right] $ be the observed surrogate loss under the training distribution induced by the expert's policy $ \pi^* $. I assume $ C \in \left[0, C_{max}\right] $ and $ l $ upper bounds the 0-1 loss. $ J\learn $ and $ J\expert $ denote the cost-to-go over $ T $ steps of executing $ \pi $ and $ \pi^* $, respectively. Then, I have the following result:
	\begin{equation*}
	J\learn \le J\expert + C_{max}T^2\epsilon.
	\end{equation*}
	\label{thm:sl}
\end{theorem}

\begin{proof}
	In order to prove this theorem, I introduce the following notation and definitions:
	\begin{itemize}
		\item $ d_{t,c}^{\pi} $: the state distribution at $ t $ as a result of the following event: $ \pi $ is executed and has been choosing the same actions as $ \pi^* $ from time $ 1 $ to $ t-1 $.
		\item $ p_{t-1} \in \left[0,1\right] $: the probability that the above-mentioned event holds true.
		\item $ d_{t,e}^{\pi} $: the state distribution at $ t $ as a result of the following event: $ \pi $ is executed and has chosen at least one different action than $ \pi^* $ from time $ 1 $ to $ t-1 $.
		\item $ (1 - p_{t-1}) \in \left[0,1\right] $: the probability that the above-mentioned event holds true.
		\item $ d_{t}^{\pi} = p_{t-1}d_{t,c}^{\pi} + (1-p_{t-1})d_{t,e}^{\pi} $: the state distribution at $ t $.
		\item $ \epsilon_{t,c} $: the probability that $ \pi $ chooses a different action than $ \pi^* $ in $ d_{t,c}^{\pi} $.
		\item $ \epsilon_{t,e} $: the probability that $ \pi $ chooses a different action than $ \pi^* $ in $ d_{t,e}^{\pi} $.
		\item $ \epsilon_{t} = p_{t-1}\epsilon_{t,c} + (1-p_{t-1})\epsilon_{t,e} $: the probability that $ \pi $ chooses a different action than $ \pi^* $ in $ d_{t}^{\pi} $.
		\item $ C_{t,c} $: the expected immediate cost of executing $ \pi $ in $ d_{t,c}^{\pi} $.
		\item $ C_{t,e} $: the expected immediate cost of executing $ \pi $ in $ d_{t,e}^{\pi} $.
		\item $ C_{t} = p_{t-1}C_{t,c} + (1-p_{t-1})C_{t,e} $: the expected immediate cost of executing $ \pi $ in $ d_{t}^{\pi} $.
		\item $ C_{t,c}^{*} $: the expected immediate cost of executing $ \pi^{*} $ in $ d_{t,c}^{\pi} $.
		\item $ C_{max} $: the upper bound of an expected immediate cost.
		\item $ J\learn  = \sum_{t=1}^{T} C_{t}$: the cost-to-go of executing $ \pi $ for $ T $ steps.
		\item $ J\expert = \sum_{t=1}^{T} C_{t,c}^* $: the cost-to-go of executing $ \pi^* $ for $ T $ steps.
	\end{itemize}
	\noindent The probability that the learner chooses at least one different action than the expert in the first $ t $ steps is: 
	\begin{equation*}
	\left(1-p_{t}\right) = \left(1-p_{t-1}\right) + p_{t-1}\epsilon_{t,c}.
	\label{aux0}
	\end{equation*}
	\noindent This gives $ \left(1-p_{t}\right) \le (1-p_{t-1}) + \epsilon_{t} $ since $ p_{t-1} \in \left[0,1\right]$. Solving this recurrence gives us:
	\begin{equation*}
	1 - p_{t} \le \sum_{i=1}^{t}\epsilon_{i}.
	\label{aux1}
	\end{equation*}
	\noindent Now consider in state distribution $ d_{t,c}^{\pi} $, if $ \pi $ chooses a different action than $ \pi^* $  with probability $ \epsilon_{t,c} $, $ \pi $ will incur a cost at most $ C_{max} $ more than $ \pi^* $. This can be represented as:
	\begin{equation*}
	C_{t,c} \le C_{t,c}^{*} + \epsilon_{t,c}C_{max}.
	\label{aux2}
	\end{equation*}
	\noindent Thus, we have:
	\begin{equation*} 
	\begin{split}
	C_{t} & = p_{t-1}C_{t,c} + (1 - p_{t-1})C_{t,e} \\
	& \le p_{t-1}C_{t,c}^{*} + p_{t-1}\epsilon_{t,c}C_{max} + (1-p_{t-1})C_{max} \\
	& = p_{t-1}C_{t,c}^{*} + (1 - p_{t})C_{max} \\
	& \le C_{t,c}^{*} + (1 - p_{t})C_{max}\\
	& \le C_{t,c}^{*} + C_{max}\sum_{i=1}^{t}\epsilon_{i}.
	\end{split}
	\end{equation*}
	
	\noindent Summing the above result over $ T $ steps and using the fact $ \frac{1}{T}\sum_{t=1}^{T}\epsilon_{t} \le \epsilon $: 
	
	\begin{equation*} 
	\begin{split}
	J\learn & \le J\expert + C_{max}\sum_{t=1}^{T}\sum_{i=1}^{t}\epsilon_{i} \\
	& = J\expert + C_{max}\sum_{t=1}^{T}(T+1-t)\epsilon_{t}\\
	& \le J\expert + C_{max}T\sum_{t=1}^{T}\epsilon_{t}\\
	& \le J\expert + C_{max}T^2\epsilon.
	\end{split}
	\end{equation*}
	
\end{proof}

Several approaches have been proposed to solve SPC tasks using supervised learning while keeping the error growing linearly instead of quadratically with $ T $~\cite{syed2010reduction,ross2011reduction,daume2009search}. Essentially, these methods reduce an SPC task to online learning. By further leveraging interactions with experts and no-regret algorithms that have strong guarantees on convex loss functions~\cite{kakade2009generalization}, at each iteration, these methods train one or multiple policies using standard supervised learning and improve the trained policies as the iteration proceeds. 


To illustrate, I denote the best policy at $ i $th iteration (trained using all observations from the previous $ i-1 $ iterations) as $ \pi_{i} $ and for any policy $ \pi \in \Pi $ we have its expected loss under the observation distribution induced by $ \pi_{i} $ as $ l_{i}\left(\pi\right) =  \mathbb{E}_{\phi \sim d_{\pi_{i}}, a^* \sim \pi^*(\phi)} \left[l_{i}\left(\phi,\pi, a^*\right)\right], l_{i}\in \left[0,l_{max}\right]$\footnote{In online learning, the surrogate loss $ l $ can be seen as chosen by some adversary which varies at each iteration.}. In addition, I denote the minimal loss in hindsight after $ N \ge i $ iterations as $ \epsilon_{min} = \min_{\pi \in \Pi} \frac{1}{N}\sum_{i=1}^{N}l_{i}(\pi) $ (i.e., the training loss after using all observations from $ N $ iterations). Then, the average regret of this online learning program can be represented as $ \epsilon_{regret} = \frac{1}{N}\sum_{i=1}^{N}l_{i}(\pi_{i}) - \epsilon_{min}$. Using DAGGER~\cite{ross2011reduction} as an example method, the accumulated error difference becomes the summation of three terms:

\begin{equation}
J\left(\hat{\pi}\right) \le T\epsilon_{min} + T\epsilon_{regret} +  \mathcal{O}(\frac{f\left(T, l_{max}\right)}{N}),
\label{eq:dagger}
\end{equation}

\noindent where $ f\left(\cdot\right) $ is the function of fixed $ T $ and $ l_{max} $. As $ N \rightarrow \infty $, the third term tends to $ 0 $ so as the second term if a no-regret algorithm such as the Follow-the-Leader~\cite{hazan2007logarithmic} is used.

To further prepare readers for understanding ADPAS, I adapt the proof of DAGGER from Ross et al.~\cite{ross2011reduction} and include it here for completeness. Note that for Theorem~\ref{thm:dagger}, I have arrived at the different third term as of Ross et al.~\cite{ross2011reduction}. 

\begin{lemma}
	\cite{ross2011reduction} Let $ P $ and $ Q $ be any two distributions over elements $ x \in \mathcal{X} $ and $ f: \mathcal{X} \rightarrow \mathbb{R} $, any bounded function such that $ f(x) \in \left[a,b\right] $ for all $ x \in \mathcal{X} $. Let the range $ r = b-a $. Then $ |\mathbb{E}_{x \sim P}\left[f(x)\right] - \mathbb{E}_{x \sim Q}\left[f(x)\right] | \le \frac{r}{2}\lVert P-Q \rVert_{1} $.
	\label{lemma-1}
\end{lemma}

\begin{proof}
	\begin{equation*} 
	\begin{split}
	& |\mathbb{E}_{x \sim P}\left[f(x)\right] - \mathbb{E}_{x \sim Q}\left[f(x)\right] | \\
	& = | \int_{x}P(x)f(x)dx - \int_{x}Q(x)f(x)dx |\\
	& = | \int_{x}f(x)\left(P(x)-Q(x)\right)dx |\\
	& = | \int_{x}\left(f(x)-c\right)\left(P(x)-Q(x)\right)dx |, \forall c \in \mathbb{R}\\
	& \le \int_{x}|f(x)-c||P(x)-Q(x)|dx \\
	& \le \max_{x}|f(x)-c|\int_{x}|P(x)-Q(x)|dx \\
	& = \max_{x}|f(x)-c|\lVert P-Q \rVert_{1} .
	\end{split}
	\end{equation*}
	Taking $ c = a + \frac{r}{2} $ leads to $ \max_{x}|f(x)-c| \le \frac{r}{2} $, thus proves the lemma.
\end{proof}

\begin{lemma}
	\cite{ross2011reduction} Let $ \hat{\pi}_{i} $ be the learned policy, $ \pi^* $ be the expert's policy, and $ \pi_{i} $ be the policy used to collect training data with probability $ \beta_{i} $ executing $ \pi^* $ and probability $  1 - \beta_{i}$ executing $ \hat{\pi}_{i} $ over $ T $ steps. Then, I have $\lVert d_{\pi_{i}} - d_{\hat{\pi}_{i}} \rVert_{1} \le 2\min(1, T\beta_{i})$.
	\label{lemma-2}
\end{lemma}

\begin{proof}
	In contrast to $ d_{\hat{\pi}_{i}} $ which is the state distribution as the result of solely executing $ \hat{\pi}_{i} $, I denote $ d $ as the state distribution as the result of $ \pi_{i} $ executing $ \pi^* $ at least once over $ T $ steps. This gives $ d_{\pi_{i}} = (1-\beta_{i})^Td_{\hat{\pi}_{i}} + \left(1-(1-\beta_{i})^T\right)d $. I also have the facts that for any two distributions $ P $ and $ Q $, $\lVert P-Q \rVert_{1} \le 2$ and $(1-\beta)^T \ge 1-\beta T, \forall \beta \in \left[0,1\right] $. Then, I have $ \lVert d_{\pi_{i}} - d_{\hat{\pi}_{i}} \rVert_{1} \le 2 $ and can further show the following:
	\begin{equation*} 
	\begin{split}
	\lVert d_{\pi_{i}} - d_{\hat{\pi}_{i}} \rVert_{1} & = \left(1-(1-\beta_{i})^T\right) \lVert d - d_{\hat{\pi}_{i}} \rVert_{1}\\
	& \le 2\left(1-(1-\beta_{i})^T\right) \\
	& \le 2T\beta_{i}.
	\end{split}
	\end{equation*}
\end{proof}

\begin{theorem}
	\cite{ross2011reduction} If the surrogate loss $ l \in \left[0,l_{max}\right] $ is the same as the cost function $ C $ or upper bounds it, then after $ N $ iterations of DAGGER:
	\begin{equation}
	J\left(\hat{\pi}\right) \le J\left(\bar{\pi}\right) \le T\epsilon_{min} + T\epsilon_{regret} + \mathcal{O}(\frac{f(T,l_{max})}{N}).
	\end{equation}
	\label{thm:dagger}
\end{theorem}

\begin{proof}
	Let $ l_{i}\left(\pi\right) =  \mathbb{E}_{\phi \sim d_{\pi_{i}}, a^* \sim \pi^*(\phi)} \left[l\left(\phi,\pi, a^*\right)\right]]$ be the expected loss of any policy $ \pi \in \Pi $ under the state distribution induced by the learned policy $ \pi_{i} $ at the $ i $th iteration and $ \epsilon_{min} = \min_{\pi \in \Pi} \frac{1}{N}\sum_{i=1}^{N}l_{i}(\pi) $ be the minimal loss in hindsight after $ N \ge i $ iterations. Then, $ \epsilon_{regret} = \frac{1}{N}\sum_{i=1}^{N}l_{i}(\pi_{i}) - \epsilon_{min}$ is the average regret of this online learning program. In addition, the expected loss of any policy  $ \pi \in \Pi $ under its own induced state distribution is denoted as $ L\left(\pi\right) =  \mathbb{E}_{\phi \sim d_{\pi}, a^* \sim \pi^*(\phi)} \left[l\left(\phi,\pi, a^*\right)\right]]$ and consider $ \bar{\pi} $ as the mixed policy that samples the policies $ \{\hat{\pi}_{i}\}_{i=1}^N $ uniformly at the beginning of each trajectory. Using Lemma~\ref{lemma-1} and Lemma~\ref{lemma-2}, we can show:
	
	\begin{equation*} 
	\begin{split}
	L(\hat{\pi}_{i}) & = \mathbb{E}_{\phi \sim d_{\hat{\pi}_{i}}, a^* \sim \pi^*(\phi)} \left[l\left(\phi,\hat{\pi}_{i}, a^*\right)\right]\\
	& \le \mathbb{E}_{\phi \sim d_{\pi_{i}}, a^* \sim \pi^*(\phi)} \left[l\left(\phi,\hat{\pi}_{i}, a^*\right)\right] + \frac{l_{max}}{2}\lVert d_{\pi_{i}} - d_{\hat{\pi}_{i}} \rVert_{1} \\
	& \le \mathbb{E}_{\phi \sim d_{\pi_{i}}, a^* \sim \pi^*(\phi)} \left[l\left(\phi,\hat{\pi}_{i}, a^*\right)\right] + l_{max}\min\left(1,T\beta_{i}\right) \\
	& = l_{i}\left(\hat{\pi}_{i}\right) + l_{max}\min\left(1,T\beta_{i}\right) 
	\end{split}
	\end{equation*}
	
	\noindent By further assuming $ \beta_{i} $ is monotonically decreasing and $ n_{\beta} = \argmax_{n} (\beta_{n} > \frac{1}{T}), n \le N $, we have the following:
	
	\begin{equation*} 
	\begin{split}
	\min_{i \in 1:N} L(\hat{\pi}_{i}) & \le L(\bar{\pi})\\
	& = \frac{1}{N} \sum_{i=1}^{N}L(\hat{\pi}_{i}) \\
	& \le \frac{1}{N} \sum_{i=1}^{N}l_{i}(\hat{\pi}_{i}) + \frac{l_{max}}{N} \sum_{i=1}^{N}\min\left(1,T\beta_{i}\right)   \\
	& = \epsilon_{min} + \epsilon_{regret} + \frac{l_{max}}{N} \left[n_{\beta} + T\sum_{i=n_{\beta}+1}^{N}\beta_{i}\right].
	\end{split}
	\end{equation*}
	
	\noindent Summing over $ T $ gives us:
	
	\begin{equation*} 
	J(\bar{\pi}) \le T\epsilon_{min} + T\epsilon_{regret} + \frac{Tl_{max}}{N} \left[n_{\beta} + T\sum_{i=n_{\beta}+1}^{N}\beta_{i}\right].
	\end{equation*}
	
	\noindent Define $ \beta_{i} = (1-\alpha)^{i-1} $, in order to have $ \beta_{i} \le \frac{1}{T} $, we need $ (1-\alpha)^{i-1} \le \frac{1}{T} $ which leads $ i \le 1 + \frac{\log{\frac{1}{T}}}{\log{(1-\alpha)}} $. In addition, note now $ i = n_{\beta}$ and $\sum_{i=n_{\beta}+1}^{N}\beta_{i} = \frac{(1-\alpha)^{n\beta} - (1-\alpha)^N}{\alpha} \le \frac{1}{T\alpha}$, continuing the above derivation, we have:
	
	\begin{equation*}
	J\left(\bar{\pi}\right) \le T\epsilon_{min} + T\epsilon_{regret} + \frac{Tl_{max}}{N}\left(1 + \frac{\log{\frac{1}{T}}}{\log{(1-\alpha)}} + \frac{1}{\alpha}\right)
	\end{equation*}
	
	\noindent Given the fact $J\left(\hat{\pi}\right) = \min_{i \in 1:N} J(\hat{\pi}_{i}) \le J(\bar{\pi}) $ and representing the third term as $ \mathcal{O}(\frac{f(T,l_{max})}{N}) $, we have proved the theorem.
\end{proof}

DAGGER offers a practical way to solve SPC tasks. However, it may require many iterations to obtain a robust policy. In addition, usually human experts or pre-defined controllers are needed for labeling the generated training data, which could be inefficient and difficult to generalize. For autonomous driving, we want the iteration number to be minimal since it directly corresponds to the number of accidents. This requires the generation of training data being accurate, efficient, and sufficient.

%% file: adaps/sections/4-adaps.tex
\section{ADAPS}
\label{sec:adaps}
In the following, I present theoretical analysis of ADAPS and introduce its pipeline. 

\subsection{Theoretical Analysis}
I have evaluated my approach against existing learning mechanisms such as DAGGER~\cite{ross2011reduction}, with my method's results proving to be more effective. Specifically, DAGGER~\cite{ross2011reduction} assumes that an underlying learning algorithm has access to a \emph{reset model}. So, the training examples can be obtained only \emph{online} by resetting an agent to its initial state distribution and executing a learned policy, thus achieving  ``small changes'' at each iteration~\cite{ross2011reduction,daume2009search,kakade2002approximately,bagnell2004policy}. In comparison, my method allows a learning algorithm to access a \emph{generative model} so that the training examples can be acquired \emph{offline} by putting an agent to arbitrary states during the analysis of an accident and letting a \emph{generative} model simulate its behavior.  This approach can result in massive training data, thus has the potential to assist a policy achieving ``large improvements'' in one iteration. 

Additionally, existing techniques such as DAGGER~\cite{ross2011reduction} usually incorporate the demonstrations of a few (human) experts into training. Because of the \emph{reset} model assumption and the lack of a diversity requirement on experts, these demonstrations can be homogeneous. In contrast, using my parameterized model to retrace and analyze each accident, the number of recovery actions obtained can be multiple orders of magnitude higher. More importantly, I can treat the generated trajectories and the additional data generated based on them (described in Section~\ref{sec:data_detect}) as running a learned policy to sample independent expert trajectories at different states, since 1) the generated trajectories are results from a principled simulation algorithm which samples its own parameter distributions independently for each simulation run; 2) my model provides near-exhaustive coverage of the configuration space of a vehicle within the boundaries of a road. With these assumptions, I derive the following theorem.

\begin{theorem}
	If the surrogate loss $ l $ upper bounds the true cost $ C $, by collecting $ K $ trajectories using ADAPS at each iteration, with probability at least $ 1 - \mu $, $\mu \in(0,1)$, ADAPS offers the following guarantee:
	\begin{equation*}
	J\left(\hat{\pi}\right) \le J\left(\bar{\pi}\right) \le T\hat{\epsilon}_{min} + T\hat{\epsilon}_{regret} + \mathcal{O}\left(Tl_{max}\sqrt{\frac{\log{\frac{1}{\mu}}}{KN}}\right).
	\end{equation*}
	\label{thm:adaps}
\end{theorem}

\begin{proof}
	 Assuming at the $ i $th iteration, my model generates $ K $ trajectories. These trajectories are independent from each other since they are generated using different parameters and at different states during the analysis of an accident. For the $ k $th trajectory, $ k \in [\![1, K]\!] $, I can construct an estimate $ \hat{l}_{ik}(\hat{\pi}_{i}) = \frac{1}{T}\sum_{t=1}^{T}l_{i}\left(\phi_{ikt},\hat{\pi}_{i}, a^{*}_{ikt}\right) $, where $ \hat{\pi}_{i} $ is the  learned policy from data gathered in previous $ i-1 $ iterations. Then, the approximated expected loss $ \hat{l}_{i} $ is the average of these $ K $ estimates: $ \hat{l}_{i}(\hat{\pi}_{i}) = \frac{1}{K}\sum_{k=1}^{K}\hat{l}_{ik}(\hat{\pi}_{i}) $. I denote $ \hat{\epsilon}_{min} = \min_{\pi \in \Pi} \frac{1}{N}\sum_{i=1}^{N}\hat{l}_{i}(\pi) $ as the approximated minimal loss in hindsight after $ N $ iterations, then $ \hat{\epsilon}_{regret} = \frac{1}{N}\sum_{i=1}^{N}\hat{l}_{i}(\hat{\pi}_{i}) - \hat{\epsilon}_{min}$ is the approximated average regret. 
	 
	 Let $ Y_{i,k} = l_{i}(\hat{\pi}_{i}) - \hat{l}_{ik}(\hat{\pi}_{i}) $ and define random variables $ X_{nK+m} = \sum_{i=1}^{n}\sum_{k=1}^{K}Y_{i,k} + \sum_{k=1}^{m}Y_{n+1,k} $, for $ n \in [\![0, N-1]\!] $ and $ m \in [\![1, K]\!] $. Consequently, $ \{X_{i}\}_{i=1}^{NK} $ form a martingale and $ |X_{i+1} - X_{i}| \le l_{max} $. By Azuma-Hoeffding's inequality, with probability at least $ 1 - \mu $, I have $ \frac{1}{KN}X_{KN} \le l_{max}\sqrt{\frac{2\log{\frac{1}{\mu}}}{KN}} $. 
	 
	 Next, I denote the expected loss of any policy  $ \pi \in \Pi $ under its own induced state distribution as $ L\left(\pi\right) =  \mathbb{E}_{\phi \sim d_{\pi}, a^* \sim \pi^*(\phi)} \left[l\left(\phi,\pi, a^*\right)\right]]$ and consider $ \bar{\pi} $ as the mixed policy that samples the policies $ \{\hat{\pi}_{i}\}_{i=1}^N $ uniformly at the beginning of each trajectory. During the data collection in each iteration, I only execute the learned policy instead of mix it with the expert's policy, which leads to $ L(\hat{\pi}_{i}) = l(\hat{\pi}_{i})$. Finally, I can show:
	 
 	\begin{equation*} 
 	\begin{split}
 	\min_{i \in 1:N} L(\hat{\pi}_{i}) & \le L(\bar{\pi})\\
 	& = \frac{1}{N} \sum_{i=1}^{N}L(\hat{\pi}_{i}) \\
 	& = \frac{1}{N} \sum_{i=1}^{N}l_{i}(\hat{\pi}_{i}) \\
 	& = \frac{1}{KN} \sum_{i=1}^{N}\sum_{k=1}^{K}\left(\hat{l}_{ik}(\hat{\pi}_{i}) +  Y_{i,k}\right)    \\
 	& = \frac{1}{KN} \sum_{i=1}^{N}\sum_{k=1}^{K}\hat{l}_{ik}(\hat{\pi}_{i}) +  \frac{1}{KN}X_{KN}    \\
 	& = \frac{1}{N} \sum_{i=1}^{N}\hat{l}(\hat{\pi}_{i}) + \frac{1}{KN}X_{KN}\\
 	& \le \frac{1}{N} \sum_{i=1}^{N}\hat{l}(\hat{\pi}_{i}) +  l_{max}\sqrt{\frac{2\log{\frac{1}{\mu}}}{KN}}  \\
 	& = \hat{\epsilon}_{min} + \hat{\epsilon}_{regret} + l_{max}\sqrt{\frac{2\log{\frac{1}{\mu}}}{KN}}.
 	\end{split}
 	\end{equation*}
 	Summing over $ T $ proves the theorem.
\end{proof}
  
Theorem~\ref{thm:adaps} provides a bound for the expected cost-to-go of the best learned policy $ \hat{\pi} $ based on the empirical error of the best policy in $ \Pi $ (i.e., $ \hat{\epsilon}_{min} $) and the empirical average regret of the learner (i.e., $ \hat{\epsilon}_{regret} $). The second term can be eliminated if a no-regret algorithm such as Follow-the-Leader~\cite{hazan2007logarithmic} is used and the third term suggests that I need the number of training examples $ KN $ to be $\mathcal{O}\left(T^2l^2_{max}\log{\frac{1}{\mu}}\right) $ in order to have a negligible generalization error, which is easily achievable using ADAPS. Summarizing these changes, I derive the following Corollary.

\begin{corollary}
	If $ l $ is convex in $ \pi $ for any $ s $ and it upper bounds $ C $, and Follow-the-Leader is used to select the learned policy, then for any $ \epsilon > 0 $, after collecting $ \mathcal{O}\left(\frac{T^2l^2_{max}\log{\frac{1}{\mu}}}{\epsilon^2}\right) $ training examples, with probability at least $ 1 - \mu $, $\mu \in(0,1)$,  ADAPS offers the following guarantee:
	\begin{equation*}
	J\left(\hat{\pi}\right) \le J\left(\bar{\pi}\right) \le T\hat{\epsilon}_{min} + \mathcal{O}\left(\epsilon\right).
	\end{equation*}
	\label{col:adaps}
\end{corollary}

\begin{proof}
	Following Theorem~\ref{thm:adaps} and the aforementioned deduction.
\end{proof}

Now I only need the best policy to have a small training error $ \hat{\epsilon}_{min} $, which can be achieved via standard supervised learning.

\subsection{Framework Pipeline}

The pipeline of my framework is the following. First, in \emph{SimLearner}, I test a learned policy by letting it control an AV. During the testing, an accident may occur, in which case the trajectory of the vehicle and the full specifications of the situation (e.g., positions of obstacles, road configuration, etc.) are known. Next, I switch to  \emph{SimExpert} and replicate the specifications of the accident in order to ``solve'' is (i.e., find alternative safe trajectories and dangerous zones). After obtaining the solutions, I then use them to generate additional training data in \emph{SimLearner}, which will be combined with previously generated data to update a policy. Next, I test the updated policy again. This process continues until the policy has reached a pre-specified threshold for the imitation error.

%% file: adaps/sections/5-policy.tex
\section{Policy Learning}
\label{sec:policy}
In this section, I will detail ADAPS's control policy by first explaining my design rationale then formulating the problem and introducing the training data collection.

Driving is a hierarchical decision process. In its simplest form, a driver needs to constantly monitor the road condition, decide it is ``safe'' or ``dangerous'', and make corresponding maneuvers. When designing a control policy for AVs, we need to consider this hierarchical aspect. In addition, driving is a temporal behavior. Drivers need reaction time to respond to various road situations~\cite{johansson1971drivers,mcgehee2000driver}. A Markovian-based control policy will not model this aspect and instead likely to give a vehicle jerky motions. Consider these factors, I propose a \emph{hierarchical} and \emph{memory-enabled} control policy. 

The task I consider is autonomous driving via a single front-facing camera. My control policy consists of three modules: \emph{Detection}, \emph{Following}, and \emph{Avoidance}. The \emph{Detection} module keeps monitoring road conditions and activates either \emph{Following} or \emph{Avoidance} to produce a steering command. All these modules are trained via end-to-end imitation learning and share the same network architecture, which combines Long Short-Term Memory (LSTM)~\cite{hochreiter1997long} and Convolutional Neural Networks (CNNs)~\cite{lecun2015deep}. Images from the front-facing camera will first go through a CNN and then a LSTM. The number of images of a training sample going into the LSTM is empirically set to 5. I use the many-to-many mode of LSTM and set the number of hidden units of the LSTM to 100. The output is the average value of the output sequence.  

My CNN architecture consists of eight layers. The first five are convolutional layers and the last three are dense layers. The kernel size is $ 5\times5 $ in the first three convolutional layers and $ 3\times3 $ in the other two convolutional layers. The first three convolutional layers have a stride of 2 while the last two convolutional layers are non-strided. The filters for the five convolutional layers are 24, 36, 48, 64, 64, respectively. All convolutional layers use VALID padding. The three dense layers have 100, 50, and 10 units, respectively. I use ELU as the activation function and $ \mathcal{L}2 $ as the kernel regularizer, which is set to 0.001 for all layers. 

I train my model using Adam~\cite{kingma2014adam} with initial learning rate set to 0.0001. The batch size is 128 and the number of epochs is 500.
For training \emph{Detection} (a classification task), I use Softmax for generating the output and categorical cross entropy as the loss function.
For training \emph{Following} and \emph{Avoidance} (regression tasks), I use mean squared error (MSE) as the loss function. I have also adopted cross-validation with 90/10 split. The input image data have $ 220 \times 66 $ resolution in RGB channels.

\subsection{End-to-end Imitation Learning}
\label{subsec:end2end}

The objective of imitation learning is to train a model that behaves or makes decisions like an expert through demonstrations. The model could be a classifier or a regresser $ \pi $ parameterized by $ \mathbf{\theta}_{\pi} $:

\begin{equation}
\hat{\theta} = \argmin_{\theta_{\pi}}\sum_{t=1}^{T} \mathcal{F}\left(\pi\left(\phi_{t}; \mathbf{\theta_{\pi}}\right),a^*_{t} \right),
\end{equation}

\noindent where $ \mathcal{F} $ is a distance function.

The end-to-end aspect denotes the mapping from raw observations to decision/control commands. For my policy, I need one decision module $ \pi_{Detection} $ and two control modules $ \pi_{Following} $ and $ \pi_{Avoidance} $. The input to $ \pi_{Detection} $ is a sequence of annotated images while the outputs are binary labels, indicating whether a road condition is dangerous or safe. The inputs to $ \pi_{Following} $ and $ \pi_{Avoidance} $ are both sequences of annotated images while the outputs are steering angles. Together, these learned policies form a hierarchical control mechanism enabling an AV to drive safely on roads and avoid obstacles when needed.

\subsection{Training Data Collection}
\label{subsec:data}

For training \emph{Following}, inspired by the technique from Bojarski et al.~\cite{Bojarski2016}, I collect images from three front-facing cameras behind the main windshield: one at the center, one at the left side, and one at the right side. The image from the center camera is labeled with the exact steering angle while the images from the other two cameras are labeled with adjusted steering angles. However, once \emph{Following} is learned, it only needs images from the center camera to operate.

For training \emph{Avoidance}, I rely on \emph{SimExpert}, which can generate numerous intermediate collision-free trajectories between the first moment and the last moment of a potential accident (see Section~\ref{subsec:solving}). By positioning an AV on these trajectories, I collect images from the center front-facing camera along with corresponding steering angles. The training of \emph{Detection} requires a more sophisticated mechanism and is the subject of the next section.

%% file: adaps/sections/6-accident.tex
\section{Learning from Accidents}
\label{sec:accident}

I explain how I analyze an accident in \emph{SimExpert} and use the generated data to train the \emph{Avoidance} and \emph{Detection} modules of my policy. \emph{SimExpert} is built based on the multi-agent simulator WarpDriver~\cite{wolinski2016warpdriver}.

\subsection{Solving Accidents}
\label{subsec:solving}

When an accident occurs, we know the trajectory of the tested vehicle for the latest $K$ frames, which I note as a collection of states $\mathcal{S} = \bigcup_{k \in [\![1, K]\!]} \mathbf{s}_k$, where each state $\mathbf{s}_k \in \mathbb{R}^4$ contains the 2-dimensional position and velocity vectors of the vehicle.
Then, there are three notable states on this trajectory that worth tracking.
The first is the earliest state where the vehicle involved in an accident (is in a collision) $\mathbf{s}_{k_a}$ (at frame $k_a$).
The second is the last state $\mathbf{s}_{k_l}$ (at frame $k_l$) where the expert algorithm can still avoid a collision.
The final one is the first state $\mathbf{s}_{k_f}$ (at frame $k_f$) where the expert algorithm perceives the interaction leading to the accident with the other involved agent, before that accident.

In order to compute these notable states, I briefly recall the high-level components of WarpDriver~\cite{wolinski2016warpdriver}. This collision-avoidance algorithm consists of two parts.
The first is a function $p$, which given the current state of an agent $\mathbf{s}_k$ and any prediction point $\mathbf{x} \in \mathbb{R}^3$ in 2-dimensional space and time (in this agent's referential), gives the probability of that agent's colliding with any neighbor $p(\mathbf{s}_k, \mathbf{x}) \in [0, \: 1]$.
The second part is the solver, which based on this function, computes the agent's probability of colliding with neighbors along its future trajectory starting from a state $\mathbf{s}_k$ (i.e., computed for $\mathbf{x}$ spanning the future predicted trajectory of the agent, I denote this probability $P\left(\mathbf{s}_k\right)$), and then proposes a new velocity to lower this probability.
Subsequently, I can initialize an agent in this algorithm to any state $\mathbf{s}_k \in \mathcal{S}$ and compute a new trajectory consisting of $\hat{K}$ new states $\mathcal{\hat{S}}_k = \bigcup_{\hat{k} \in [\![1, \hat{K}]\!]} \mathbf{\hat{s}}_{\hat{k}}$, where $\mathbf{\hat{s}}_1 = \mathbf{s}_k$.

Additionally, since $\mathbf{x} = \vecth{0}{0}{0}$ in space and time in an agent's referential represents the agent's position at the current time (I can use this point $\mathbf{x}$ with function $p$ to determine if the agent is currently in a collision with another agent), I find $\mathbf{s}_{k_a}$ where $ k_a = \min(k)$ subject to $k \in [\![1, K]\!]$ and $p(\mathbf{s}_{k}, \vecth{0}{0}{0}) > 0$. I note that a trajectory $\mathcal{\hat{S}}_k$ produced by the expert algorithm could contain collisions (accounting for vehicle dynamics) depending on the state $\mathbf{s}_k$ that it was initialized from.
I denote the set of colliding states along this trajectory as $coll(\mathcal{\hat{S}}_k) = \{ \mathbf{\hat{s}}_{\hat{k}} \in \mathcal{\hat{S}}_k \: | \: p(\mathbf{\hat{s}}_{\hat{k}}, \vecth{0}{0}{0}) > 0 \}$.
Then, I can compute $\mathbf{s}_{k_l}$ where $k_l = max(k)$ subject to $k \in [\![1, k_a]\!]$ and $coll(\mathcal{\hat{S}}_k) = \emptyset$. Finally, I can compute $\mathbf{s}_{k_f}$ with $k_f = 1 + max(k)$ subject to $k \in [\![1, k_l]\!]$ and $P(\mathbf{s}_k) = 0$.

Knowing these notable states, I can solve the accident situation by computing the set of collision-free trajectories $solve(\mathcal{S}) = \{ \mathcal{\hat{S}}_k \: | \: k \in [\![k_f, k_l]\!] \}$. An example is shown in Figure~\ref{fig:trajectories}, in which a set of generated trajectories for a situation where the vehicle had collided with a static obstacle in front of it after driving on a straight road. As expected, the trajectories feature sharper turns (red trajectories) as the starting state tends towards the last moment that the vehicle can still avoid the obstacle. These trajectories can then be used to generate training examples in \emph{SimLearner} in order to train the \emph{Avoidance} module.

\begin{figure}
	\begin{center}
		\includegraphics[width=\textwidth]{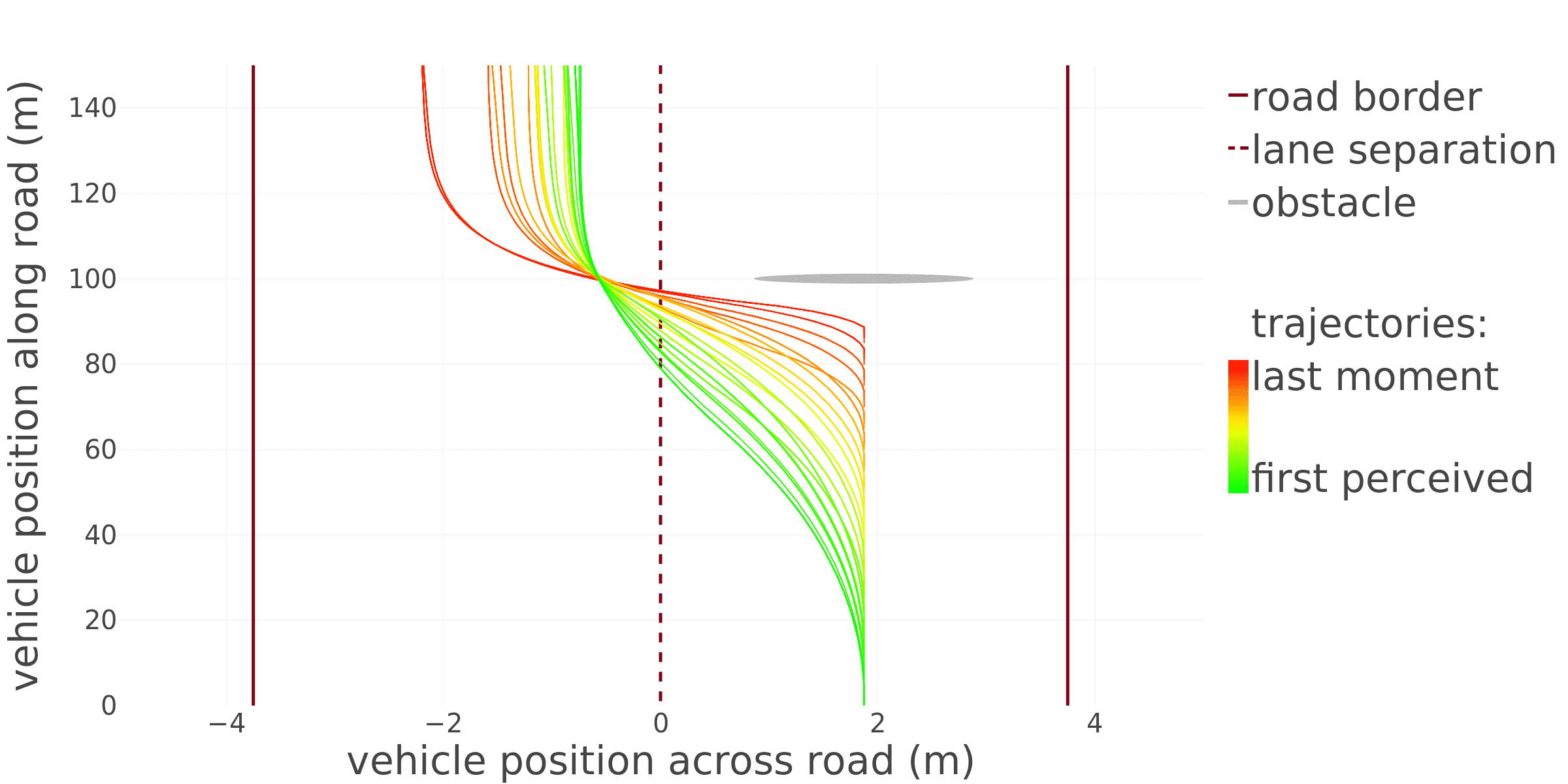}
	\end{center}
	\caption{
		Collision-free trajectories generated by the expert algorithm for a vehicle traveling on the right lane of a straight road, with an obstacle in front: in total, 74 trajectories span from the first moment the vehicle perceives the obstacle (green, progressive avoidance) to the last moment the collision can be avoided (red, sharp avoidance).
	}
	\label{fig:trajectories} 
\end{figure}

\subsection{Additional Data Coverage}
\label{sec:data_detect}

\begin{figure}
\begin{center}
\includegraphics[width=\textwidth]{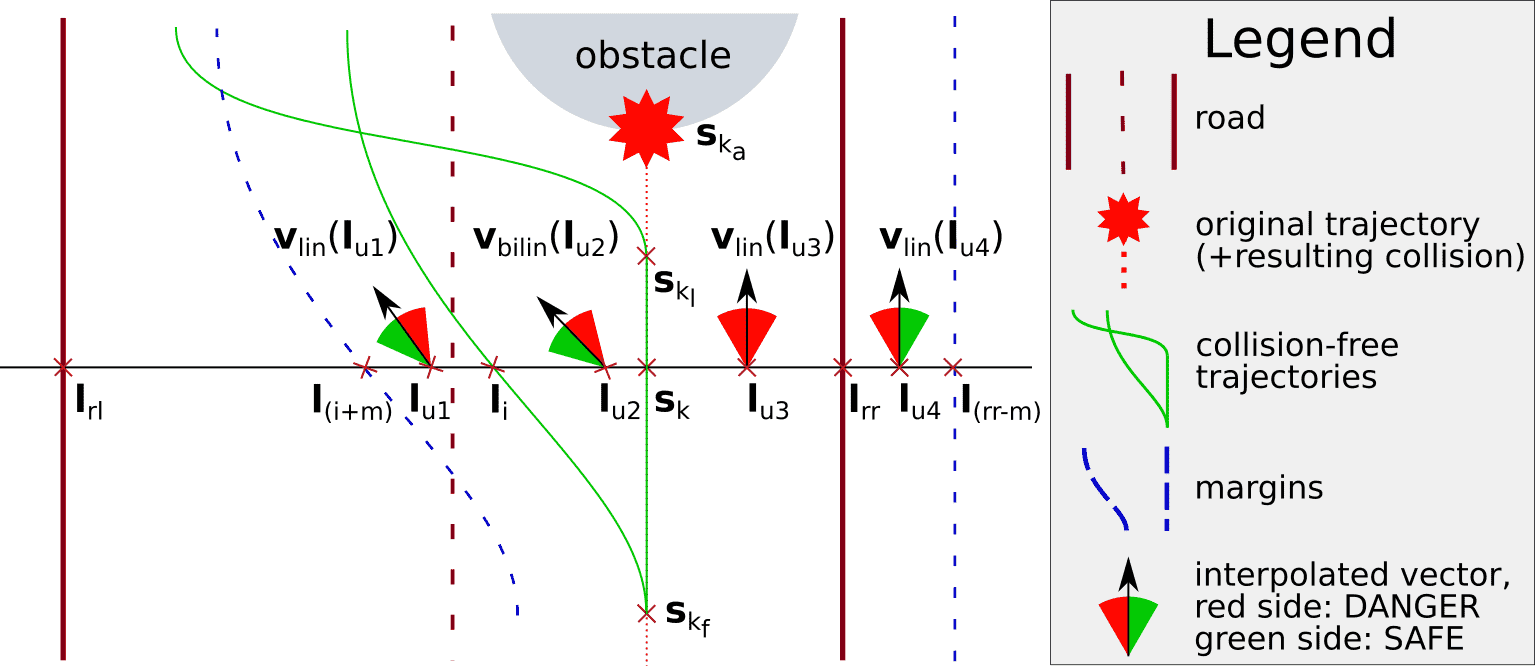}
\end{center}
\caption{
Illustration of important points and DANGER/SAFE labels from Section~\ref{sec:accident} for a vehicle traveling on the right lane of a straight road, with an obstacle in front.
Labels are shown for four points $\{ \mathbf{l}_{u1}, \mathbf{l}_{u2}, \mathbf{l}_{u3}, \mathbf{l}_{u4} \}$ illustrating the four possible cases.
}
\label{fig:points}
\end{figure}

The previous step generates collision-free trajectories $solve(\mathcal{S})$ between $\mathbf{s}_{k_f}$ and $\mathbf{s}_{k_l}$.
It is possible to build on these trajectories if the tested steering algorithm has particular data/training requirements.
Here I detail the data I derived for training the \emph{Detection} module. where the task is to determine if a situation is dangerous and tell \emph{Avoidance} to address it.

To proceed, I essentially generate a number of trajectories parallel to $\{ \mathbf{s}_{k_f}, ..., \mathbf{s}_{k_a} \}$, and for each position on them, generate several images for various orientations of the vehicle.
These images are then labeled based on under-steering/over-steering as compared to the ``ideal'' trajectories in $solve(\mathcal{S})$.
This way, I scan the region of the road before the accident locus, generating several images (different vehicle orientations) for each point in that region. Next, I will explain this process in details. 

For the following paragraph, I abusively note $\mathbf{s}_k.x$, $\mathbf{s}_k.y$ the position coordinates at state $\mathbf{s}_k$, and $\mathbf{s}_k.vx$, $\mathbf{s}_k.vy$ the velocity vector coordinates at state $\mathbf{s}_k$.
Then, for any state $\mathbf{s}_k \in \{ \mathbf{s}_{k_f}, ..., \mathbf{s}_{k_a} \}$ I can define a line $L(\mathbf{s}_k) = \{ \mathbf{l}_u = \vectw{\mathbf{s}_k.x}{\mathbf{s}_k.y} + u \times \vectw{-\mathbf{s}_k.vy}{\mathbf{s}_k.vx} \: | \: u \in \mathbb{R} \}$.
On this line, I note $\mathbf{l}_i$ the furthest point on $L(\mathbf{s}_k)$ from $\vectw{\mathbf{s}_k.x}{\mathbf{s}_k.y}$ which is at an intersection between $L(\mathbf{s}_k)$ and a collision-free trajectory from $solve(\mathcal{S})$.
This point determines how far the vehicle can be expected to stray from the original trajectory $\mathcal{S}$ before the accident, if it followed an arbitrary trajectory from $solve(\mathcal{S})$.
I also note $\mathbf{l}_{rl}$ and $\mathbf{l}_{rr}$ the two intersections between $L(\mathbf{s}_k)$ and the road edges ($\mathbf{l}_{rl}$ is on the ``left'' with $rl > 0$, and $\mathbf{l}_{rr}$ is on the ``right'' with $rr < 0$).
These two points delimit how far from the original trajectory the vehicle could be.
Finally, I define a user-set margin $g$ as outlined below (I set $g=0.5~m$).

Altogether, these points and margin are the limits of the region along the original trajectory wherein I generate images for training: a point $\mathbf{l}_u \in L(\mathbf{s}_k)$ is inside the region if it is between the original trajectory and the furthest collision-free trajectory plus a margin $g$ (if $\mathbf{l}_u$ and $\mathbf{l}_i$ are on the same side, i.e., $sign(u) = sign(i)$), or if it is between the original trajectory and either road boundary plus a margin $g$ (if $\mathbf{l}_u$ and $\mathbf{l}_i$ are not on the same side, i.e., $sign(u) \neq sign(i)$).

In addition, if a point $\mathbf{l}_u \in L(\mathbf{s}_k)$ is positioned between two collision-free trajectories $\mathcal{\hat{S}}_{k_1}, \mathcal{\hat{S}}_{k_2} \in solve(\mathcal{S})$,
I consider the two closest states on $\mathcal{\hat{S}}_{k_1}$, and the two closest states from $\mathcal{\hat{S}}_{k_1}$, and bi-linearly interpolate these four states' velocity vectors, resulting in an approximate velocity vector $\mathbf{v}_{bilin}(\mathbf{l}_u)$ at $\mathbf{l}_u$.
Similarly, if a point $\mathbf{l}_u \in L(\mathbf{s}_k)$ is not positioned between two collision-free trajectories, I consider the two closest states on the single closest collision-free trajectory $\mathcal{\hat{S}}_{k_1} \in solve(\mathcal{S})$, and linearly interpolate their velocity vectors, resulting in an approximate velocity vector $\mathbf{v}_{lin}(\mathbf{l}_u)$ at $\mathbf{l}_u$.

From here, I can construct images at various points $\mathbf{l}_u$ along $L(\mathbf{s}_k)$ (increasing $u$ by steps of $0.1~m$), with various orientation vectors (noted $\mathbf{v}_u$ and within 2.5 degrees of $\vectw{\mathbf{s}_k.vx}{\mathbf{s}_k.vy}$), and label them using the following scheme (also illustrated in Figure~\ref{fig:points}).
If the expert algorithm steers the vehicle to avoid obstacles ($\mathbf{l}_i$ with $i > 0$), there are four cases to consider when building a point $\mathbf{l}_u$:

\begin{itemize}
	\item $u < i + g$ and $u > i$: $\mathbf{l}_u$ is outside of the computed collision-free trajectories $solve(\mathcal{S})$, on the outside of the steering computed by the expert algorithm. The label is SAFE if $det(\mathbf{v}_{lin}(\mathbf{l}_u), \mathbf{v}_u) \geq 0$, and DANGER otherwise.
	\item $u < i$ and $u > 0$: $\mathbf{l}_u$ is inside the computed collision-free trajectories $solve(\mathcal{S})$. The label is SAFE if $det(\mathbf{v}_{bilin}(\mathbf{l}_u), \mathbf{v}_u) \geq 0$ (over-steering), and DANGER otherwise (under-steering).
	\item $u < 0$ and $u > rr$: $\mathbf{l}_u$ is outside the computed collision-free trajectories $solve(\mathcal{S})$ on the inside of the steering computed by the expert algorithm. The label is always DANGER.
	\item $u < rr$ and $u > rr - g$: $\mathbf{l}_u$ is in an unattainable region, but I include it to prevent false reactions to similar (but safe) future situations. The label is DANGER if $det(\mathbf{v}_{lin}(\mathbf{l}_u), \mathbf{v}_u) > 0$, SAFE otherwise.
\end{itemize}

\noindent Here, the function $det(\cdot, \cdot)$ computes the determinant of two vectors from $\mathbb{R}^2$. Conversely, if the expert algorithm made the vehicle avoid obstacles by steering right ($\mathbf{l}_i$ with $i < 0$), there are four cases to consider when building a point $\mathbf{l}_u$:

\begin{itemize}
	\item $u > i - g$ and $u < i$: the label is SAFE if $det(\mathbf{v}_{lin}(\mathbf{l}_u), \mathbf{v}_u) \leq 0$, and DANGER otherwise.
	\item $u > i$ and $u < 0$: the label is SAFE if $det(\mathbf{v}_{bilin}(\mathbf{l}_u), \mathbf{v}_u) \leq 0$, and DANGER otherwise.
	\item $u > 0$ and $u < rl$: the label is always DANGER.
	\item $u > rl$ and $u < rl + g$: the label is DANGER if $det(\mathbf{v}_{lin}(\mathbf{l}_u), \mathbf{v}_u) < 0$, SAFE otherwise.
\end{itemize}

\noindent I then generate images from these (position, orientation, label) triplets which are used to further train the \emph{Detection} module of my policy.

%% file: adaps/sections/7-exp.tex
\section{Experiments}
\label{sec:exp}

In this section, I will first detail my experiment setup then show my evaluation results.

\subsection{Experiment Setup}
\label{sec:app-exp}

\subsubsection{Scenarios}
I have tested my method in three scenarios. The first is a straight road representing a linear geometry, the second is a curved road representing a non-linear geometry, and the third is an open ground. The first two represent on-road situations with a static obstacle while the last represents an off-road situation with a dynamic obstacle. 

Both the straight and curved roads consist of two lanes. The width of each lane is $ 3.75~m $ and there is a $ 3~m $ shoulder on each side of the road. The curved road is half circular with radius at $ 50~m $ and is attached to two straight roads at each end. The open scenario is a $ 1000~m $ $ \times $ $ 1000~m $ ground, which has a green sphere treated as the target for the \emph{Following} module to steer the AV.

\subsubsection{Vehicle Specifications}
The vehicle's speed is set to $ 20~m/s $, which value is used to compute the throttle value in the simulator.
Due to factors such as the rendering complexity and the delay of the communication module, the actual running speed is in the range of $ 20\pm1~m/s $.
The length and width of the vehicle are $ 4.5~m $ and $ 2.5~m $, respectively.
The distance between the rear axis and the rear of the vehicle is $ 0.75~m $.
The front wheels can turn up to 25 degrees in either direction.
I have three front-facing cameras set behind the main windshield, which are at $ 1.2~m $ height and $ 1~m $ front to the center of the vehicle.
The two side cameras (one at left and one at right) are set to be $ 0.8~m $ away from the vehicle's center axis.
These two cameras are only used to capture data for training \emph{Following}.
During runtime, my control policy only requires images from the center camera to operate.

\subsubsection{Obstacles}
For the on-road scenarios, I use a scaled version of a virtual traffic cone as the obstacle on both the straight and curved roads.
This scaling operation is meant to preserve the obstacle's visibility,
since at distances greater than $30~m$ a normal-sized obstacle is quickly reduced to just a few pixels.
This is an intrinsic limitation of the single-camera setup (and its resolution),
but in reality I can emulate this ``scaling'' using the camera's zoom function for instance. 
For the off-road scenario, I use a vehicle with the same specifications as of the AV as the dynamic obstacle. This vehicle is scripted to collide into the AV on its default course when no avoidance behavior is applied by the AV. 

\subsubsection{Training Data}
In order to train \emph{Following}, I have built a waypoint system on the straight road and curved road for AVs to follow, respectively. By running the vehicle for roughly equal distances on both roads, I have gathered in total 65~061 images (33~642 images for the straight road and 31~419 images for the curved road). On the open ground, I have sampled 30~000 positions and computed the angle difference between the direction towards the sphere target and the forward direction. This gives me 30~000 training examples.  

In order to train \emph{Avoidance}, on the straight road, I rewind the accident by 74 frames starting from the frame that the accident takes place, by re-planning the vehicle at each state, I have obtained 74 safe trajectories.
Similarly, on the curved road, I rewind the accident by 40 frames, which results in 40 safe trajectories.
On the open ground, I rewind the accident by 46 frames, which results in 46 safe trajectories.
By positioning the vehicle on these trajectories and capturing the image from the front-facing camera, I have collected 34~516 images for the straight road, 33~624 images for the curved road, and 33~741 images for the open ground. 

For the training of \emph{Detection}, using the mechanism explained in Subsection~\ref{sec:data_detect}, I have collected 32~538 images for the straight road, 71~859 for the curved road, and 67~102 images for the open ground. These statistics are summarized in Table~\ref{tb:train-data}.

\subsection{Evaluation Results}
For evaluation, I compare my policy to the ``flat policy'' that essentially consists of a single DNN ~\cite{chen2015deepdriving,Xu2017end,zhang2017query,codevilla2017end}. Usually, this type of policy contains a few convolutional layers followed by a few dense layers. Although the specifications may vary, without human intervention, they are mainly limited to single-lane following~\cite{codevilla2017end} or \emph{off-road} collision avoidance~\cite{LeCun2006off}. I select Bojarski et al.~\cite{Bojarski2016} as an example network, as it is one of the most tested control policies. In the following, I will first demonstrate the effectiveness of my policy and then qualitatively illustrate the efficiency of my framework.

\begin{figure}[th]
	\centering
	\includegraphics[width=\textwidth]{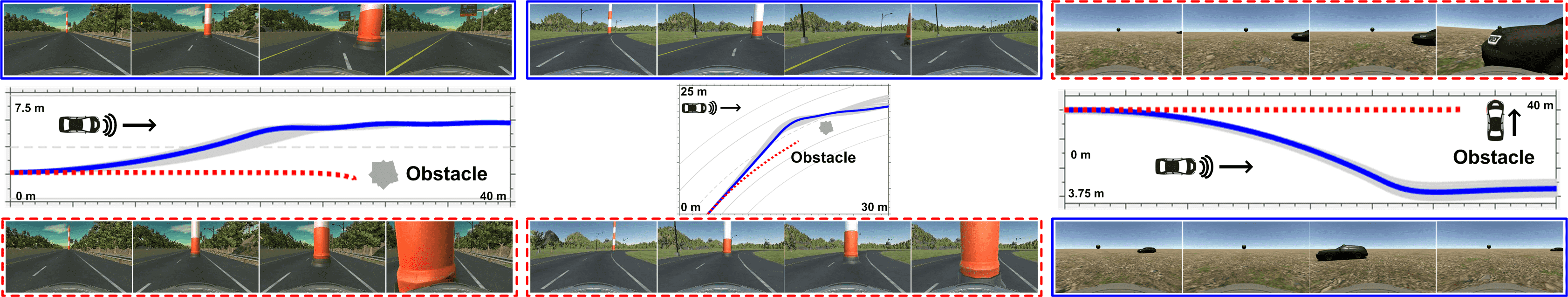}
	\caption{LEFT and CENTER: the comparisons between my policy $ O_{full} $ (TOP) and Bojarski et al.~\protect\cite{Bojarski2016}, $ B_{full} $ (BOTTOM). $ O_{full} $ can steer the AV away from the obstacle while $ B_{full} $ causes collision. RIGHT: the accident analysis results on the open ground. I show the accident caused by an adversary vehicle (TOP); then I show after additional training the AV can avoid the adversary vehicle (BOTTOM).}
	\label{fig:examples}
\end{figure}

\subsubsection{On-road Scenarios}

I derive my training datasets from \emph{straight road with or without an obstacle} and \emph{curved road with or without an obstacle}.
This separation allows me to train multiple policies and test the effect of \emph{learning from accidents} using my policy compared to Bojarski et al.~\cite{Bojarski2016}. By progressively increasing the training datasets, I obtain six policies for evaluation:

\begin{itemize}
	\item My policy: trained with only lane-following data $O_{follow}$; $O_{follow}$ additionally trained after analyzing one accident on the straight road $O_{straight}$; and $O_{straight}$ additionally trained after producing one accident on the curved road $O_{full}$.
	\item Similarly, for the policy from Bojarski et al.~\cite{Bojarski2016}: $B_{follow}$, $B_{straight}$, and $B_{full}$.
\end{itemize}

I first evaluate $B_{follow}$ and $O_{follow}$ using both the straight and curved roads by counting how many laps (out of 50) the AV can finish. As a result, both policies managed to finish all laps while keeping the vehicle in the lane. I then test these two policies on the straight road with a static obstacle added. Both policies result in the vehicle collides into the obstacle, which is expected since no accident data were used during the training.

Having the occurred accident, I can now use \emph{SimExpert} to generate additional training data to obtain $B_{straight}$\footnote{The accident data are used only to perform a regression task as the policy by Bojarski et al.~\cite{Bojarski2016} does not have a classification module.} and $O_{straight}$.
As a result, $B_{straight}$ continues to cause collision while $O_{straight}$ avoids the obstacle. Nevertheless, when testing $O_{straight}$ on the curved road with an obstacle, accident still occurs because of the corresponding accident data are not yet included in training.

By further including the accident data from the curved road in training, I obtain $ B_{full} $ and $ O_{full}$. $ O_{full}$ manages to perform both lane-following and collision avoidance in all runs. $B_{full}$, on the other hand, leads the vehicle to drift away from the road. 

For the studies involved an obstacle, I uniformly sampled 50 obstacle positions on a $3~m$ line segment that is perpendicular to the direction of a road and in the same lane as the vehicle.
I compute the success rate as how many times a policy can avoid the obstacle (while stay in the lane) and resume lane-following afterwards. The results are shown in Table~\ref{tb:acc} and example trajectories are shown in Figure~\ref{fig:examples} LEFT and CENTER.

\begin{table*}[ht!]
	\centering
	\small
	\tabcolsep=0.1cm
	\scalebox{0.75}{
		\begin{tabular}{ccccccccc}
			\toprule
			& \multicolumn{5}{c}{Training Module (Data) } & \multicolumn{3}{c}{Other Specs}  \\        
			\cmidrule(l){2-6} \cmidrule(l){7-9}      
			Scenarios  & \emph{Following (\#Images)} & \emph{Avoidance (\#Images)}  & \emph{Detection (\#Images)} & Total & Data Augmentation & \#Safe Trajectories & Road Type & Obstacle   
			\\
			\midrule
			Straight road & 33~642 & 34~516  & 32~538  & 97~854 & $ 212 $x & 74  &  on-road   & static 
			\\
			\midrule
			Curved road   & 31~419 & 33~624  & 71~859  & 136~855 & $ 98 $x & 40  &  on-road   & static 
			\\
			\midrule
			Open ground   & 30~000 & 33~741  & 67~102 & 130~843 & $ 178 $x & 46  &  off-road  &  dynamic
			\\
			\bottomrule
		\end{tabular}}
	\caption{Training Data Summary: my method can achieve over 200 times more training examples than DAGGER~\protect\cite{ross2011reduction} at one iteration leading to large improvements of a policy.} 
	\label{tb:train-data}
\end{table*}

\subsubsection{Off-road Scenario}

I further test my method on an open ground which involves a dynamic obstacle. The AV is trained heading towards a green sphere while an adversary vehicle is scripted to collide with the AV on its default course. The result showing my policy can steer the AV away from the adversary vehicle and resume its direction to the sphere target. This can be seen in Figure~\ref{fig:examples} RIGHT.

\begin{table}[ht!]
\centering
\small
\tabcolsep=0.1cm
\scalebox{1}{
	\begin{tabular}{ccccccc}
		\toprule
		& \multicolumn{6}{c}{Test Policy and Success Rate (out of 50 runs)}  \\        
		\cmidrule(l){2-7}       
		Scenario & $ B_{follow} $ & $ O_{follow} $  & $ B_{straight} $ & $ O_{straight} $ & $ B_{full} $ & $ O_{full} $   \\
		\midrule
		Straight road / Curved road & 100\% & 100\%  & 100\%  & 100\%  & 100\%  & 100\% 
		\\
		\midrule
		Straight road + Static obstacle & 0\% & 0\%  & \textbf{0\%}  & \textbf{100\%}  & \textbf{0\%}  & \textbf{100\%}
		\\
		\midrule
		Curved road + Static obstacle & 0\% & 0\%  & 0\%  & 0\%  &\textbf{ 0\% } & \textbf{100\%}
		\\
		\bottomrule
	\end{tabular}}
	\caption{Test Results of On-Road Scenarios: my policies $ O_{straight} $ \& $ O_{full} $ can lead to robust collision avoidance and lane-following behaviors. }
	\label{tb:acc}
\end{table}

\subsubsection{Data Heterogeneity}

The key to rapid policy improvement is to generate training data accurately, efficiently, and sufficiently. Using \emph{principled simulations} covers the first two criteria, now I demonstrate the third. Compared to the average number of training data collected by DAGGER~\cite{ross2011reduction} at one iteration, my method can achieve over 200 times more training examples for one iteration\footnote{The result is computed via dividing the total number of training images via my method by the average number of training data collected using the safe trajectories in each scenario.}. This is shown in Table~\ref{tb:train-data}.

In Figure~\ref{fig:straight-tsne}, I show the visualization results of images collected via my method and DAGGER~\cite{ross2011reduction} within one iteration via progressively increasing the number of sampled trajectories. My method generates much more heterogeneous training data, which, when produced in a large quantity can facilitate the update of a control policy.

\begin{figure}[th]
	\centering
	\includegraphics[width=0.7\columnwidth]{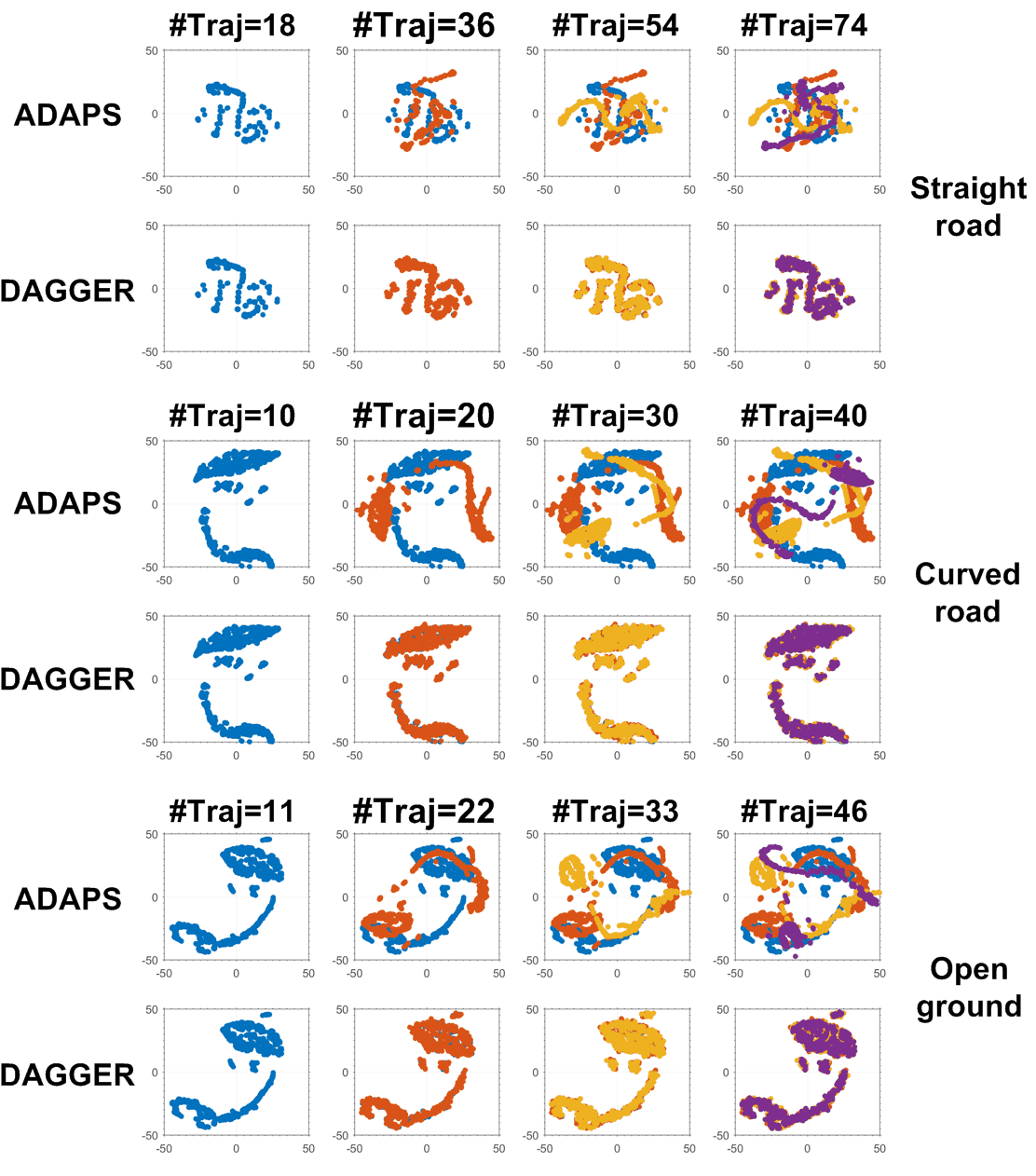}
	\caption{The visualization results of collected images using t-SNE~\protect\cite{maaten2008visualizing}. My method can generate more heterogeneous training data compared to DAGGER~\protect\cite{ross2011reduction} in one learning iteration as the sampled trajectories progress. }
	\label{fig:straight-tsne}
\end{figure}

%% file: adaps/sections/8-conclude.tex
\section{Summary and Future Work}
\label{sec:conclude}

In this chapter, I have proposed ADAPS, a framework that consists of two simulation platforms and a control policy. Using ADAPS, one can easily simulate accidents. Then, ADAPS can retrace each accident, analyze it, and plan alternative safe trajectories. With an additional training data generation technique, my method can produce a large number of heterogeneous training examples compared to existing methods such as DAGGER~\cite{ross2011reduction}, thus representing a more efficient learning mechanism. My hierarchical and memory-enabled policy offers robust collision avoidance that previous policies fail to achieve. I have evaluated my method using multiple simulated scenarios, showing a variety of benefits.

There are many future directions. First of all, I would like to combine long-range vision into ADAPS so that an AV can plan ahead in time. Secondly, the generation of accidents can be parameterized using knowledge from traffic engineering studies. Thirdly, I would like to combine more sensors and fuse their inputs so that an AV can navigate in more complicated traffic scenarios. Lastly, in order to improve the safety of the AV, it is critical to model the interactions between the AV and pedestrians. There has been a number of studies making virtual humans more intelligent and life-like~\cite{Li2011Purpose,Li2012Distribution,Li2012Apprentice,Li2012Commonsense,Li2013Memory}. These results can be incorporated into the virtual environment for improving the training of the AV either directly or indirectly.  

To this end, I have finished introducing my efforts on advancing autonomous driving in terms of enabling an AV to navigate safely in dangerous situations including accidents. In the next chapter, I will conclude my dissertation with a discussion of future research directions. 

%% file: chapters/Conclusions.tex
\chapter{CONCLUSION}
\label{ch:Conclusion}

In this dissertation, I have studied traffic at both the macroscopic level and the microscopic level. Macroscopically, I have developed methods to accurately and efficiently estimate and reconstruct city-scale traffic using mobile sensor data while generating visual analytics in various forms. Microscopically, I have developed a framework to simulate, analyze, and produce traffic accident data, and have proposed an efficient online learning mechanism for learning control policies for autonomous driving. The proposed techniques can enable many Intelligent Transportation System (ITS) applications including travel time estimation, route planning, visualization, traffic optimization, traffic management, and safety and control of autonomous vehicles (AVs).

\section{Summary of Results}
I described a deterministic approach to estimate traffic conditions based on Wardrop's Principles~\cite{wardrop1952road} and the shortest travel-time criterion. In order to compute the shortest travel-time path (for realizing the shortest travel-time criterion), the travel time of a road network is required, which information is usually lacking. This issue is addressed by adapting a travel-time allocation method from Hellinga et al.~\cite{hellinga2008decomposing}, which was derived from the observations of real-world traffic conditions. Compared to other state-of-the-art approaches that use shortest distance criterion, my approach results in less estimation bias in congested environments for map-matching and overall better estimation accuracy on various synthetic GIS data.   

I presented an approach to interpolate temporal missing measurements by exploring sparsity embedded in traffic patterns. To be specific, observing the sparse representation in the frequency domain of traffic patterns, I have proposed a method based on the Compressed Sensing algorithm~\cite{1614066,1580791} for recovering traffic patterns using GPS traces. My approach can provide accurate recovery when tested using the ground-truth traffic data from loop detectors, and is among few that have exploited the effectiveness of Compressed Sensing on traffic pattern processing~\cite{Lin2019ComSenseTRB}. 


In order to further improve the estimation accuracy of traffic conditions for spatial data interpolation, I proposed an iterative approach, which embeds map-matching and travel-time estimation as its sub-routines. This process is further improved by the statistical modeling and learning of traffic conditions of a road segment. The results of my approach are accurate estimations of traffic conditions in areas with GPS data coverage---achieving up to 97\% relative improvement in estimation accuracy over previous techniques---and coarse estimations of traffic conditions in areas without GPS data coverage (of a city). 

For achieving accurate estimations of traffic conditions in data-deficient areas, using the abovementioned results, I presented a method to dynamically interpolate spatial missing traffic data. In particular, I have leveraged traffic simulation to ensure the consistency of traffic flows on the boundaries of areas with and without GPS data coverage. A metamodel-based simulation optimization is further developed to save the computational cost of using traffic simulation in optimization. Compared to the simulation-only approach, my technique has achieved on average a 7\% error rate and up to 90 times speedup. My approach is the first dynamical and efficient method for interpolating large-scale traffic data while ensuring the flow consistency on city-scale boundaries. 

After fully reconstructing spatial-temporal traffic at a city scale, I visualized the reconstructed traffic in various forms such as 2D flow map, 2D animation, and 3D animations. These visual representations can be adopted to improve many ITS applications including the analysis of traffic patterns at street level, region level, and the city level, and enrich virtual environment applications such as virtual tourism and the training of general driving behaviors or autonomous driving.  

Regarding autonomous driving, I presented ADAPS, a framework that consists of two simulation platforms and a hierarchical control policy. ADAPS can be used to simulate, analyze various traffic scenarios, especially accidents, and automatically produce labeled training data. In addition, ADAPS represents a more efficient online learning mechanism compared to previous techniques, attributing to the switch from the reset modeling approach to the generative modeling approach. Using the hierarchical control policy and the efficient online learning mechanism of ADAPS, robust control policies for autonomous driving can be learned and applied to obtain normal driving and safe navigation in dangerous situations including accidents.

\section{Future Work}
There are many future development and research directions can stem from this dissertation, at the macroscopic level of traffic (i.e., \emph{city-scale traffic}), the microscopic level of traffic (i.e., \emph{autonomous driving}), the connection between the two levels, and beyond. I will discuss a few of them in each category in the following.

\subsection{Macroscopic Level}
On the ``city-scale traffic'' side, first of all, it would be useful to develop an interactive simulation platform. Using the platform, policy makers, city planners, and other users can easily edit a road network and alter a transport policy in order to test the effectiveness of these changes via observing the response of simulated traffic flows propagating in a city. Building such a platform would require several elements: a 3D virtual environment with a user interface, a road network construction mechanism~\cite{wilkie2012transforming,musialski2013survey}, a road network editing mechanism~\cite{chen2008interactive}, and a real-time traffic simulation technique~\cite{Sewall:2011:IHS:2070781.2024169,Wilkie:2013:FRD:2461912.2462021,garcia2014designing}. The 3D virtual environment can be built using a game engine such as Unity (\url{https://unity.com/}) or Unreal (\url{https://www.unrealengine.com}). The rest of the elements have been explored to various degrees in existing studies. Unifying these elements would be an interesting topic. 

Secondly, simulating city-scale traffic, depending on the levels of detail, can be computationally prohibitive. A scalable approach that can combine modern machine learning techniques and traffic flow models is highly desirable. Such a approach can especially benefit applications with highly interactive and real-time demands such as the simulation platform mentioned above. The metamodel-based simulation optimization presented in Chapter~\ref{ch:siga} is an example work in this direction. However, the functional component of the metamodel is currently chosen to be quadratic polynomial, which offers limited expressiveness. With the emergence of deep learning~\cite{lecun2015deep}, it would be interesting to replace the quadratic polynomial with a deep neural network or an LSTM network for potential improvements on the traffic reconstruction accuracy. 

Thirdly, traffic participants are not limited to cars, mixed traffic involving cars, pedestrians, cyclists, and other motorists are commonly seen in many regions across the globe. A simulation model that can encompass all these different traffic modalities can enrich various real-world and virtual-world applications mentioned in the previous chapters. In order to achieve this goal, knowledge from traffic engineering literature can be exploited. Mixture models, although not necessarily built for simulation, have been explored and developed~\cite{faghri1999development,laxman2010pedestrian}. Examining the possibility of extending these models for mixture traffic simulation is a promising research direction.

\subsection{Microscopic Level}
On the ``autonomous driving'' side, the system presented in Chapter~\ref{ch:adaps} is an end-to-end system, which means a single model is trained to map the sensor input directly to the control command output. Such an approach is straightforward and usually results in a more compact model as it does not contain intermediate steps. However, an end-to-end system based on deep learning can be hard to interpret. Also, a large number of training examples are often needed to train an end-to-end model. 

In contrast, the traditional engineering pipeline, which consists of several modules, can be adopted for autonomous driving. In this approach, the sensor input will get processed and passed to its subsequent modules such as detection, tracking, and planning, before the final control command is produced. This conventional approach has both advantages and disadvantages. Advantages include, since there are multiple modules in the pipeline, less training examples are needed to learn a model; the prior knowledge can be incorporated into the problem; the explainability is improved as the final control command is produced by a planning algorithm rather than directly from the raw sensor input. Disadvantages include, the uncertainty and errors of each module are difficult to propagate backwards to its preceding modules, thus causing the system to suffer potential compounding errors; the computation is not shared between modules: each module is trained independently for a different objective; human experts are usually needed to tailor each module so that the system can achieve maximum performance.

While both end-to-end and traditional engineering approaches have their own characteristics, given that the safety is of the leading concerns these days regarding autonomous driving, the traditional engineering approach is likely to prevail in the near future due to its superior explainability and controllability. Hence, it would be interesting to develop the ``traditional engineering'' version of ADAPS. The only element needs to be modified is the hierarchical control policy, which currently is represented by a single model with three neural networks. The other elements such as the simulation platforms and online learning mechanism remain applicable.

Another aspect can be improved in ADAPS is the generation of accident data. Currently, the accident simulation is simple and arbitrary. However, in traffic engineering, there exists rich literature on accident analysis and prevention. By exploring which, a systematically way of simulating accidents can be developed, which can bring further justifications to ADAPS on autonomous driving training and testing. One imminent research direction is to incorporate the pre-crash scenarios published by the National Highway Traffic Safety Administration of the United States~\cite{najm2013description} into our simulation platform, and then develop a sampling mechanism to produce accident data for learning a control policy. 

Beyond the abovemetioned immediate research topics that can be built on top of ADAPS, there are many other interesting research directions. In general, the safety, control, and coordination aspects of autonomous driving all need further exploration and development. One future research direction can be exploring the possibility of using simulations to assist sample-efficient learning for a control policy. Another direction, inspired by the observation that the training of autonomous driving is largely context-dependent, is to develop theory and practice in transferring the learned behaviors of an autonomous vehicle from one environment to other environments. This generalization ability, observed in humans, is largely missing in autonomous driving at the moment.

\subsection{Connection Between The Two Levels}
Although this dissertation has been addressing the macroscopic level and the microscopic level of traffic as two separate topics, the two aspects have tight connection, where many applications and developments can be drawn. From the macro-to-micro perspective, the estimated city-scale traffic conditions can be immediately adopted for better routing and planning of AVs. The reconstructed traffic can be incorporated into virtual environments to provide rich traffic semantics for training the navigation and decision-making of autonomous driving. From the micro-to-macro perspective, AVs can be treated as probe vehicles to gather traffic information in a city so that traffic reconstruction can be achieved with higher accuracy. The AV can also be dispatched to multiple road users as a sharing transportation tool. This way not only the number of vehicles on the road is reduced, which can assist in alleviating traffic jams, but also less space is needed to physically accommodate the large number of vehicles, which implies additional socio-economic benefits. Back to the macro-to-micro perspective, an efficient traffic reconstruction technique can contribute to the design of the dispatching algorithm for AVs to maximize their sharing functionality.  


In conclusion, now and into the near future, AVs will be operating not only in traffic but also along with human-driven vehicles. This assembly brings many challenges as well as research opportunities. It would be interesting to develop simulation models for the mixture of human-driven and autonomous vehicles, with the flexibility to choose the percentage of each type of the vehicle. As AVs can be considered part of the overall cyber-physical system, they can serve as additional ``degrees of freedom'' to the traffic system, which can be potentially ``tuned'' to regulate traffic flows~\cite{Wu2018Stablizing}. The applications range from alleviating traffic congestion to assisting flow distribution in social gatherings or evacuation situations. Lastly, it would be imperative to consider human factors in addition to technology development, given essentially autonomous and intelligent systems are designed and built to improve people's life. As technology advances, developing collaborative rather than competitive relationships between autonomous systems and humans is the challenge that scientists and engineers will be facing. My future efforts will be centered around this challenge, with the focus on building an cooperative mixed traffic system. 


%% file: main.bbl
\begin{thebibliography}{}

\bibitem[Abadi et~al., 2015]{AbadiRajabiounIoannou2015}
Abadi, A., Rajabioun, T., and Ioannou, P.~A. (2015).
\newblock Traffic flow prediction for road transportation networks with limited
  traffic data.
\newblock {\em IEEE Transactions on Intelligent Transportation Systems},
  16(2):653--662.

\bibitem[Agarwal et~al., 2016]{agarwal2016dynamic}
Agarwal, S., Kachroo, P., and Contreras, S. (2016).
\newblock A dynamic network modeling-based approach for traffic observability
  problem.
\newblock {\em IEEE Transactions on Intelligent Transportation Systems},
  17(4):1168--1178.

\bibitem[Anava et~al., 2015]{anava2015online}
Anava, O., Hazan, E., and Zeevi, A. (2015).
\newblock Online time series prediction with missing data.
\newblock In {\em Proceedings of the 32nd International Conference on Machine
  Learning (ICML)}, pages 2191--2199.

\bibitem[Andrienko and Andrienko, 2006]{andrienko2006exploratory}
Andrienko, N. and Andrienko, G. (2006).
\newblock {\em Exploratory analysis of spatial and temporal data: a systematic
  approach}.
\newblock Springer Science \& Business Media.

\bibitem[Asif et~al., 2016]{asifmatrix}
Asif, M.~T., Mitrovic, N., Dauwels, J., and Jaillet, P. (2016).
\newblock Matrix and tensor based methods for missing data estimation in large
  traffic networks.
\newblock {\em IEEE Transactions on Intelligent Transportation Systems},
  17(7):1816--1825.

\bibitem[Atkeson, 1994]{atkeson1994using}
Atkeson, C.~G. (1994).
\newblock Using local trajectory optimizers to speed up global optimization in
  dynamic programming.
\newblock In {\em Advances in neural information processing systems}, pages
  663--670.

\bibitem[Atkeson et~al., 1997]{Atkeson1997}
Atkeson, C.~G., Moore, A.~W., and Schaal, S. (1997).
\newblock Locally weighted learning.
\newblock {\em Artificial Intelligence Review}, 11(1-5):11--73.

\bibitem[Bagnell et~al., 2004]{bagnell2004policy}
Bagnell, J.~A., Kakade, S.~M., Schneider, J.~G., and Ng, A.~Y. (2004).
\newblock Policy search by dynamic programming.
\newblock In {\em Advances in neural information processing systems}, pages
  831--838.

\bibitem[Barto and Mahadevan, 2003]{barto2003recent}
Barto, A.~G. and Mahadevan, S. (2003).
\newblock Recent advances in hierarchical reinforcement learning.
\newblock {\em Discrete Event Dynamic Systems}, 13(4):341--379.

\bibitem[Bayarri et~al., 1996]{bayarri1996}
Bayarri, S., Fernandez, M., and Perez, M. (1996).
\newblock Virtual reality for driving simulation.
\newblock {\em Commun. ACM}, 39(5):72--76.

\bibitem[Bera and Rao, 2011]{bera2011estimation}
Bera, S. and Rao, K. V.~K. (2011).
\newblock Estimation of origin-destination matrix from traffic counts: the
  state of the art.
\newblock {\em European Transport\textbackslash Trasporti Europei n.}, pages
  3--23.

\bibitem[Bi et~al., 2016]{bi2016data}
Bi, H., Mao, T., Wang, Z., and Deng, Z. (2016).
\newblock A data-driven model for lane-changing in traffic simulation.
\newblock In {\em Proceedings of the ACM SIGGRAPH/Eurographics Symposium on
  Computer Animation}, pages 149--158.

\bibitem[Bojarski et~al., 2016]{Bojarski2016}
Bojarski, M., Del~Testa, D., Dworakowski, D., Firner, B., Flepp, B., Goyal, P.,
  Jackel, L.~D., Monfort, M., Muller, U., Zhang, J., et~al. (2016).
\newblock End to end learning for self-driving cars.
\newblock {\em arXiv preprint arXiv:1604.07316}.

\bibitem[Candes et~al., 2006]{1580791}
Candes, E., Romberg, J., and Tao, T. (2006).
\newblock Robust uncertainty principles: exact signal reconstruction from
  highly incomplete frequency information.
\newblock {\em Information Theory, IEEE Transactions on}, 52(2):489--509.

\bibitem[Cand{\`e}s and Wakin, 2008]{candes2008introduction}
Cand{\`e}s, E. and Wakin, M. (2008).
\newblock An introduction to compressive sampling.
\newblock {\em IEEE signal processing magazine}, 25(2):21--30.

\bibitem[Cascetta and Nguyen, 1988]{cascetta1988unified}
Cascetta, E. and Nguyen, S. (1988).
\newblock A unified framework for estimating or updating origin/destination
  matrices from traffic counts.
\newblock {\em Transportation Research Part B: Methodological}, 22(6):437--455.

\bibitem[Castro et~al., 2012]{castro2012urban}
Castro, P.~S., Zhang, D., and Li, S. (2012).
\newblock Urban traffic modelling and prediction using large scale taxi gps
  traces.
\newblock In {\em International Conference on Pervasive Computing}, pages
  57--72.

\bibitem[Celikoglu, 2007]{celikoglu2007dynamic}
Celikoglu, H.~B. (2007).
\newblock A dynamic network loading model for traffic dynamics modeling.
\newblock {\em IEEE Transactions on Intelligent Transportation Systems},
  8(4):575--583.

\bibitem[Celikoglu et~al., 2009]{celikoglu2009node}
Celikoglu, H.~B., Gedizlioglu, E., and Dell'Orco, M. (2009).
\newblock A node-based modeling approach for the continuous dynamic network
  loading problem.
\newblock {\em IEEE Transactions on Intelligent Transportation Systems},
  10(1):165--174.

\bibitem[Celikoglu and Silgu, 2016]{celikoglu2016extension}
Celikoglu, H.~B. and Silgu, M.~A. (2016).
\newblock Extension of traffic flow pattern dynamic classification by a
  macroscopic model using multivariate clustering.
\newblock {\em Transportation Science}, 50(3):966--981.

\bibitem[Chao et~al., 2019]{Chao2019Survey}
Chao, Q., Bi, H., Li, W., Mao, T., Wang, Z., Lin, M.~C., and Deng, Z. (2019).
\newblock A survey on visual traffic simulation: Models, evaluations, and
  applications in autonomous driving.
\newblock {\em Computer Graphics Forum}.

\bibitem[Chao et~al., 2018]{chao2017realistic}
Chao, Q., Deng, Z., Ren, J., Ye, Q., and Jin, X. (2018).
\newblock Realistic data-driven traffic flow animation using texture synthesis.
\newblock {\em IEEE Transactions on Visualization and Computer Graphics},
  24(2):1167--1178.

\bibitem[Chao et~al., 2013]{chao2013video}
Chao, Q., Shen, J., and Jin, X. (2013).
\newblock Video-based personalized traffic learning.
\newblock {\em Graphical Models}, 75(6):305--317.

\bibitem[Charalambous and Chrysanthou, 2014]{Charalambous2014}
Charalambous, P. and Chrysanthou, Y. (2014).
\newblock The pag crowd: A graph based approach for efficient data-driven crowd
  simulation.
\newblock {\em Computer Graphics Forum}, 33(8):95--108.

\bibitem[Chen et~al., 2014]{chen2014map}
Chen, B.~Y., Yuan, H., Li, Q., Lam, W.~H., Shaw, S.-L., and Yan, K. (2014).
\newblock Map-matching algorithm for large-scale low-frequency floating car
  data.
\newblock {\em International Journal of Geographical Information Science},
  28(1):22--38.

\bibitem[Chen et~al., 2015]{chen2015deepdriving}
Chen, C., Seff, A., Kornhauser, A., and Xiao, J. (2015).
\newblock Deepdriving: Learning affordance for direct perception in autonomous
  driving.
\newblock In {\em Computer Vision, 2015 IEEE International Conference on},
  pages 2722--2730.

\bibitem[Chen et~al., 2008]{chen2008interactive}
Chen, G., Esch, G., Wonka, P., M{\"u}ller, P., and Zhang, E. (2008).
\newblock Interactive procedural street modeling.
\newblock {\em ACM Trans. Graph.}, 27(3):103:1--103:10.

\bibitem[Codevilla et~al., 2017]{codevilla2017end}
Codevilla, F., M{\"u}ller, M., Dosovitskiy, A., L{\'o}pez, A., and Koltun, V.
  (2017).
\newblock End-to-end driving via conditional imitation learning.
\newblock In {\em Robotics and Automation (ICRA), 2017 IEEE International
  Conference on}, pages 746--753. IEEE.

\bibitem[Conn et~al., 2009]{conn2009introduction}
Conn, A., Scheinberg, K., and Vicente, L.~N. (2009).
\newblock {\em Introduction to derivative-free optimization}, volume~8.
\newblock {SIAM}.

\bibitem[Daum{\'e} et~al., 2009]{daume2009search}
Daum{\'e}, H., Langford, J., and Marcu, D. (2009).
\newblock Search-based structured prediction.
\newblock {\em Machine learning}, 75(3):297--325.

\bibitem[Donoho, 2006]{1614066}
Donoho, D. (2006).
\newblock Compressed sensing.
\newblock {\em Information Theory, IEEE Transactions on}, 52(4):1289--1306.

\bibitem[Faghri and Egyh{\'a}ziov{\'a}, 1999]{faghri1999development}
Faghri, A. and Egyh{\'a}ziov{\'a}, E. (1999).
\newblock Development of a computer simulation model of mixed motor vehicle and
  bicycle traffic on an urban road network.
\newblock {\em Transportation research record}, 1674(1):86--93.

\bibitem[Ferreira et~al., 2013]{ferreira2013visual}
Ferreira, N., Poco, J., Vo, H.~T., Freire, J., and Silva, C.~T. (2013).
\newblock Visual exploration of big spatio-temporal urban data: A study of new
  york city taxi trips.
\newblock {\em IEEE Transactions on Visualization and Computer Graphics},
  19(12):2149--2158.

\bibitem[Friedman et~al., 2001]{friedman2001elements}
Friedman, J., Hastie, T., and Tibshirani, R. (2001).
\newblock {\em The elements of statistical learning}.
\newblock Springer series in statistics Springer, Berlin.

\bibitem[Gao, 2012]{gao2012modeling}
Gao, S. (2012).
\newblock Modeling strategic route choice and real-time information impacts in
  stochastic and time-dependent networks.
\newblock {\em IEEE Transactions on Intelligent Transportation Systems},
  13(3):1298--1311.

\bibitem[Garcia-Dorado et~al., 2017]{garcia2017fast}
Garcia-Dorado, I., Aliaga, D., Bhalachandran, S., Schmid, P., and Niyogi, D.
  (2017).
\newblock Fast weather simulation for inverse procedural design of 3d urban
  models.
\newblock {\em ACM Trans. Graph.}, 36(2):21:1--21:19.

\bibitem[Garcia-Dorado et~al., 2014]{garcia2014designing}
Garcia-Dorado, I., Aliaga, D., and V.~Ukkusuri, S. (2014).
\newblock Designing large-scale interactive traffic animations for urban
  modeling.
\newblock {\em Computer Graphics Forum}, 33(2):411--420.

\bibitem[Gning et~al., 2011]{gning2011interval}
Gning, A., Mihaylova, L., and Boel, R.~K. (2011).
\newblock Interval macroscopic models for traffic networks.
\newblock {\em IEEE Transactions on Intelligent Transportation Systems},
  12(2):523--536.

\bibitem[Gordon, 1995]{gordon1995stable}
Gordon, G.~J. (1995).
\newblock Stable function approximation in dynamic programming.
\newblock In {\em Proceedings of the 12th International Conference on Machine
  Learning (ICML)}, pages 261--268.

\bibitem[Greenshields et~al., 1935]{greenshields1935study}
Greenshields, B., Bibbins, J., Channing, W., and Miller, H. (1935).
\newblock A study of traffic capacity.
\newblock {\em Highway Research Board Proceedings}, 14(1):448--477.

\bibitem[Hajiahmadi et~al., 2016]{hajiahmadi2016integrated}
Hajiahmadi, M., van~de Weg, G.~S., Tamp{\`e}re, C.~M., Corthout, R., Hegyi, A.,
  De~Schutter, B., and Hellendoorn, H. (2016).
\newblock Integrated predictive control of freeway networks using the extended
  link transmission model.
\newblock {\em IEEE Transactions on Intelligent Transportation Systems},
  17(1):65--78.

\bibitem[Hato et~al., 1999]{hato1999incorporating}
Hato, E., Taniguchi, M., Sugie, Y., Kuwahara, M., and Morita, H. (1999).
\newblock Incorporating an information acquisition process into a route choice
  model with multiple information sources.
\newblock {\em Transportation Research Part C: Emerging Technologies},
  7(2):109--129.

\bibitem[Hazan et~al., 2007]{hazan2007logarithmic}
Hazan, E., Agarwal, A., and Kale, S. (2007).
\newblock Logarithmic regret algorithms for online convex optimization.
\newblock {\em Machine Learning}, 69(2-3):169--192.

\bibitem[Hellinga et~al., 2008]{hellinga2008decomposing}
Hellinga, B., Izadpanah, P., Takada, H., and Fu, L. (2008).
\newblock Decomposing travel times measured by probe-based traffic monitoring
  systems to individual road segments.
\newblock {\em Transportation Research Part C: Emerging Technologies},
  16(6):768--782.

\bibitem[Herrera et~al., 2010]{HERRERA2010568}
Herrera, J.~C., Work, D.~B., Herring, R., Ban, X.~J., Jacobson, Q., and Bayen,
  A.~M. (2010).
\newblock Evaluation of traffic data obtained via gps-enabled mobile phones:
  The mobile century field experiment.
\newblock {\em Transportation Research Part C: Emerging Technologies},
  18(4):568--583.

\bibitem[Herring, 2010]{herring2010real}
Herring, R. (2010).
\newblock {\em Real-time traffic modeling and estimation with streaming probe
  data using machine learning}.
\newblock PhD thesis, University of California, Berkeley.

\bibitem[Herring et~al., 2010]{herring2010estimating}
Herring, R., Hofleitner, A., Abbeel, P., and Bayen, A. (2010).
\newblock Estimating arterial traffic conditions using sparse probe data.
\newblock In {\em Intelligent Transportation Systems (ITSC), 13th International
  IEEE Conference on}, pages 929--936.

\bibitem[Hochreiter and Schmidhuber, 1997]{hochreiter1997long}
Hochreiter, S. and Schmidhuber, J. (1997).
\newblock Long short-term memory.
\newblock {\em Neural computation}, 9(8):1735--1780.

\bibitem[Hofleitner et~al., 2012a]{hofleitner2012learning}
Hofleitner, A., Herring, R., Abbeel, P., and Bayen, A. (2012a).
\newblock Learning the dynamics of arterial traffic from probe data using a
  dynamic bayesian network.
\newblock {\em IEEE Transactions on Intelligent Transportation Systems},
  13(4):1679--1693.

\bibitem[Hofleitner et~al., 2012b]{hofleitner2012probability}
Hofleitner, A., Herring, R., and Bayen, A. (2012b).
\newblock Probability distributions of travel times on arterial networks: a
  traffic flow and horizontal queuing theory approach.
\newblock In {\em Transportation Research Board 91st Annual Meeting}.

\bibitem[Hofleitner et~al., 2012c]{hofleitner2012large}
Hofleitner, A., Herring, R., Bayen, A., Han, Y., Moutarde, F., and
  De~La~Fortelle, A. (2012c).
\newblock Large scale estimation of arterial traffic and structural analysis of
  traffic patterns using probe vehicles.
\newblock In {\em Transportation Research Board 91st Annual Meeting}.

\bibitem[Hunter et~al., 2014]{hunter2014path}
Hunter, T., Abbeel, P., and Bayen, A. (2014).
\newblock The path inference filter: model-based low-latency map matching of
  probe vehicle data.
\newblock {\em Intelligent Transportation Systems, IEEE Transactions on},
  15(2):507--529.

\bibitem[Hunter, 2014]{hunter2014large}
Hunter, T.~J. (2014).
\newblock {\em Large-Scale, Low-Latency State Estimation Of Cyberphysical
  Systems With An Application To Traffic Estimation}.
\newblock PhD thesis, University of California, Berkeley.

\bibitem[Johansson and Rumar, 1971]{johansson1971drivers}
Johansson, G. and Rumar, K. (1971).
\newblock Drivers' brake reaction times.
\newblock {\em Human factors}, 13(1):23--27.

\bibitem[Ju et~al., 2010]{Ju:2010:MC:1882261.1866162}
Ju, E., Choi, M.~G., Park, M., Lee, J., Lee, K.~H., and Takahashi, S. (2010).
\newblock Morphable crowds.
\newblock {\em ACM Trans. Graph.}, 29(6):140:1--140:10.

\bibitem[Kachroo and Sastry, 2016]{kachroo2016travel}
Kachroo, P. and Sastry, S. (2016).
\newblock Travel time dynamics for intelligent transportation systems: Theory
  and applications.
\newblock {\em IEEE Transactions on Intelligent Transportation Systems},
  17(2):385--394.

\bibitem[Kakade and Langford, 2002]{kakade2002approximately}
Kakade, S. and Langford, J. (2002).
\newblock Approximately optimal approximate reinforcement learning.
\newblock In {\em Proceedings of the 30th International Conference on Machine
  Learning (ICML)}, pages 267--274.

\bibitem[Kakade and Tewari, 2009]{kakade2009generalization}
Kakade, S.~M. and Tewari, A. (2009).
\newblock On the generalization ability of online strongly convex programming
  algorithms.
\newblock In {\em Advances in Neural Information Processing Systems}, pages
  801--808.

\bibitem[Kearns et~al., 2002]{kearns2002sparse}
Kearns, M., Mansour, Y., and Ng, A.~Y. (2002).
\newblock A sparse sampling algorithm for near-optimal planning in large markov
  decision processes.
\newblock {\em Machine learning}, 49(2-3):193--208.

\bibitem[Khosravi et~al., 2011]{khosravi2011prediction}
Khosravi, A., Mazloumi, E., Nahavandi, S., Creighton, D., and Van~Lint, J.
  (2011).
\newblock Prediction intervals to account for uncertainties in travel time
  prediction.
\newblock {\em IEEE Transactions on Intelligent Transportation Systems},
  12(2):537--547.

\bibitem[Kingma and Ba, 2015]{kingma2014adam}
Kingma, D. and Ba, J. (2015).
\newblock Adam: A method for stochastic optimization.
\newblock In {\em International Conference on Learning Representations (ICLR)}.

\bibitem[Kong et~al., 2013]{kong2013efficient}
Kong, Q.-J., Zhao, Q., Wei, C., and Liu, Y. (2013).
\newblock Efficient traffic state estimation for large-scale urban road
  networks.
\newblock {\em IEEE Transactions on Intelligent Transportation Systems},
  14(1):398--407.

\bibitem[Konidaris et~al., 2012]{konidaris2012robot}
Konidaris, G., Kuindersma, S., Grupen, R., and Barto, A. (2012).
\newblock Robot learning from demonstration by constructing skill trees.
\newblock {\em The International Journal of Robotics Research}, 31(3):360--375.

\bibitem[Krajzewicz et~al., 2012a]{SUMO2012}
Krajzewicz, D., Erdmann, J., Behrisch, M., and Bieker, L. (2012a).
\newblock Recent development and applications of {SUMO - Simulation of Urban
  MObility}.
\newblock {\em International Journal On Advances in Systems and Measurements},
  5(3-4):128--138.

\bibitem[Krajzewicz et~al., 2012b]{krajzewicz2012recent}
Krajzewicz, D., Erdmann, J., Behrisch, M., and Bieker, L. (2012b).
\newblock Recent development and applications of {SUMO}--simulation of urban
  mobility.
\newblock {\em International Journal On Advances in Systems and Measurements},
  5(3\&4):128--138.

\bibitem[Kuhi et~al., 2015]{kuhi2015using}
Kuhi, K., Kaare, K.~K., and Koppel, O. (2015).
\newblock Using probabilistic models for missing data prediction in network
  industries performance measurement systems.
\newblock {\em Procedia Engineering}, 100:1348--1353.

\bibitem[Kuhl et~al., 1995]{kuhl1995}
Kuhl, J., Evans, D., Papelis, Y., Romano, R., and Watson, G. (1995).
\newblock The iowa driving simulator: An immersive research environment.
\newblock {\em Computer}, 28(7):35--41.

\bibitem[Laxman et~al., 2010]{laxman2010pedestrian}
Laxman, K.~K., Rastogi, R., and Chandra, S. (2010).
\newblock Pedestrian flow characteristics in mixed traffic conditions.
\newblock {\em Journal of Urban Planning and Development}, 136(1):23--33.

\bibitem[LeCun et~al., 2015]{lecun2015deep}
LeCun, Y., Bengio, Y., and Hinton, G. (2015).
\newblock Deep learning.
\newblock {\em nature}, 521(7553):436.

\bibitem[LeCun et~al., 2005]{LeCun2006off}
LeCun, Y., Muller, U., Ben, J., Cosatto, E., and Flepp, B. (2005).
\newblock Off-road obstacle avoidance through end-to-end learning.
\newblock In {\em Advances in neural information processing systems}, pages
  739--746.

\bibitem[Leduc, 2008]{leduc2008road}
Leduc, G. (2008).
\newblock Road traffic data: Collection methods and applications.
\newblock {\em Working Papers on Energy, Transport and Climate Change}, 1(55).

\bibitem[Lee et~al., 2007]{Lee2007}
Lee, K.~H., Choi, M.~G., Hong, Q., and Lee, J. (2007).
\newblock Group behavior from video: a data-driven approach to crowd
  simulation.
\newblock In {\em Proceedings of the ACM SIGGRAPH/Eurographics symposium on
  Computer animation}, pages 109--118.

\bibitem[Lerner et~al., 2007]{Lerner2007}
Lerner, A., Chrysanthou, Y., and Lischinski, D. (2007).
\newblock Crowds by example.
\newblock {\em Computer Graphics Forum}, 26(3):655--664.

\bibitem[Levine and Koltun, 2013]{levine2013guided}
Levine, S. and Koltun, V. (2013).
\newblock Guided policy search.
\newblock In {\em Proceedings of the 30th International Conference on Machine
  Learning (ICML)}, pages 1--9.

\bibitem[Li et~al., 2014]{li2014multimodel}
Li, L., Chen, X., and Zhang, L. (2014).
\newblock Multimodel ensemble for freeway traffic state estimations.
\newblock {\em IEEE Transactions on Intelligent Transportation Systems},
  15(3):1323--1336.

\bibitem[Li et~al., 2016]{li2016manhattan}
Li, M., Wonka, P., and Nan, L. (2016).
\newblock Manhattan-world urban reconstruction from point clouds.
\newblock In {\em European Conference on Computer Vision}, pages 54--69.

\bibitem[Li and Allbeck, 2011]{Li2011Purpose}
Li, W. and Allbeck, J.~M. (2011).
\newblock Populations with purpose.
\newblock In {\em Proceedings of the 4th International Conference on Motion in
  Games}, pages 132--143.

\bibitem[Li and Allbeck, 2012a]{Li2012Apprentice}
Li, W. and Allbeck, J.~M. (2012a).
\newblock The virtual apprentice.
\newblock In {\em Proceedings of the 12th International Conference on
  Intelligent Virtual Agents}, pages 15--27.

\bibitem[Li and Allbeck, 2012b]{Li2012Commonsense}
Li, W. and Allbeck, J.~M. (2012b).
\newblock Virtual humans: Evolving with common sense.
\newblock In {\em Proceedings of the 5th International Conference on Motion in
  Games}, pages 182--193.

\bibitem[Li et~al., 2013]{Li2013Memory}
Li, W., Balint, T., and Allbeck, J.~M. (2013).
\newblock Using a parameterized memory model to modulate npc {AI}.
\newblock In {\em Proceedings of the 13th International Conference on
  Intelligent Virtual Agents}, pages 1--14.

\bibitem[Li et~al., 2012]{Li2012Distribution}
Li, W., Di, Z., and Allbeck, J.~M. (2012).
\newblock Crowd distribution and location preference.
\newblock {\em Computer Animation and Virtual Worlds}, 23(3-4):343--351.

\bibitem[Li et~al., 2018]{Li2018CityEstIter}
Li, W., Jiang, M., Chen, Y., and Lin, M.~C. (2018).
\newblock Estimating urban traffic states using iterative refinement and
  wardrop equilibria.
\newblock {\em IET Intelligent Transport Systems}, 12(8):875--883.

\bibitem[Li et~al., 2017a]{Li2017CityEstSparse}
Li, W., Nie, D., Wilkie, D., and Lin, M.~C. (2017a).
\newblock Citywide estimation of traffic dynamics via sparse {GPS} traces.
\newblock {\em IEEE Intelligent Transportation Systems Magazine},
  9(3):100--113.

\bibitem[Li et~al., 2017b]{Li2017CityFlowRecon}
Li, W., Wolinski, D., and Lin, M.~C. (2017b).
\newblock City-scale traffic animation using statistical learning and
  metamodel-based optimization.
\newblock {\em ACM Trans. Graph.}, 36(6):200:1--200:12.

\bibitem[Li et~al., 2019]{Li2019ADAPS}
Li, W., Wolinski, D., and Lin, M.~C. (2019).
\newblock {ADAPS}: Autonomous driving via principled simulations.
\newblock In {\em IEEE International Conference on Robotics and Automation
  (ICRA)}.

\bibitem[Li et~al., 2015]{Li2015Insects}
Li, W., Wolinski, D., Pettr\'{e}, J., and Lin, M.~C. (2015).
\newblock Biologically-inspired visual simulation of insect swarms.
\newblock {\em Computer Graphics Forum}, 34(2):425--434.

\bibitem[Lin et~al., 2019]{Lin2019ComSenseTRB}
Lin, L., Li, W., and Peeta, S. (2019).
\newblock A compressive sensing approach for connected vehicle data capture and
  its impact on travel time estimation.
\newblock In {\em Transportation Research Board 98th Annual Meeting}.

\bibitem[Lin, 1992]{lin1992self}
Lin, L.-J. (1992).
\newblock Self-improving reactive agents based on reinforcement learning,
  planning and teaching.
\newblock {\em Machine learning}, 8(3-4):293--321.

\bibitem[Lin et~al., 2016]{lin2016generating}
Lin, W.-C., Wong, S.-K., Li, C.-H., and Tseng, R. (2016).
\newblock Generating believable mixed-traffic animation.
\newblock {\em IEEE Transactions on Intelligent Transportation Systems},
  17(11):3171--3183.

\bibitem[Lou et~al., 2009]{lou2009map}
Lou, Y., Zhang, C., Zheng, Y., Xie, X., Wang, W., and Huang, Y. (2009).
\newblock Map-matching for low-sampling-rate {GPS} trajectories.
\newblock In {\em Proceedings of the 17th ACM SIGSPATIAL International
  Conference on Advances in Geographic Information Systems}, pages 352--361.

\bibitem[Lu et~al., 2014]{lu2014personality}
Lu, X., Wang, Z., Xu, M., Chen, W., and Deng, Z. (2014).
\newblock A personality model for animating heterogeneous traffic behaviors.
\newblock {\em Computer animation and virtual worlds}, 25(3-4):361--371.

\bibitem[Maaten and Hinton, 2008]{maaten2008visualizing}
Maaten, L. v.~d. and Hinton, G. (2008).
\newblock Visualizing data using t-{SNE}.
\newblock {\em Journal of machine learning research}, 9:2579--2605.

\bibitem[Mao et~al., 2015]{mao2015efficient}
Mao, T., Wang, H., Deng, Z., and Wang, Z. (2015).
\newblock An efficient lane model for complex traffic simulation.
\newblock {\em Computer Animation and Virtual Worlds}, 26(3-4):397--403.

\bibitem[McGehee et~al., 2000]{mcgehee2000driver}
McGehee, D.~V., Mazzae, E.~N., and Baldwin, G.~S. (2000).
\newblock Driver reaction time in crash avoidance research: validation of a
  driving simulator study on a test track.
\newblock {\em Proceedings of the human factors and ergonomics society annual
  meeting}, 44(20).

\bibitem[Microsoft, 2009]{geolife}
Microsoft (2009).
\newblock Geolife project.
\newblock
  \href{https://www.microsoft.com/en-us/research/project/geolife-building-social-networks-using-human-location-history/}{www.microsoft.com/en-us/research/project/geolife-building-social-networks-using-human-location-history/}.

\bibitem[Microsoft, 2010]{t-drive}
Microsoft (2010).
\newblock T-drive project.
\newblock
  \href{https://www.microsoft.com/en-us/research/project/t-drive-driving-directions-based-on-taxi-traces/}{www.microsoft.com/en-us/research/project/t-drive-driving-directions-based-on-taxi-traces/}.

\bibitem[MIT, 2011]{mit}
MIT (2011).
\newblock Mit intelligent transportation systems.
\newblock \href{https://its.mit.edu/}{its.mit.edu/}.

\bibitem[Mitrovic et~al., 2015]{mitrovic2015low}
Mitrovic, N., Asif, M.~T., Dauwels, J., and Jaillet, P. (2015).
\newblock Low-dimensional models for compressed sensing and prediction of
  large-scale traffic data.
\newblock {\em IEEE Transactions on Intelligent Transportation Systems},
  16(5):2949--2954.

\bibitem[Miwa et~al., 2012]{miwa2012development}
Miwa, T., Kiuchi, D., Yamamoto, T., and Morikawa, T. (2012).
\newblock Development of map matching algorithm for low frequency probe data.
\newblock {\em Transportation Research Part C: Emerging Technologies},
  22:132--145.

\bibitem[Musialski et~al., 2013]{musialski2013survey}
Musialski, P., Wonka, P., Aliaga, D.~G., Wimmer, M., van Gool, L., and
  Purgathofer, W. (2013).
\newblock A survey of urban reconstruction.
\newblock {\em Computer Graphics Forum}, 32(6):146--177.

\bibitem[Najm et~al., 2013]{najm2013description}
Najm, W.~G., Ranganathan, R., Srinivasan, G., Smith, J.~D., Toma, S., Swanson,
  E., Burgett, A., et~al. (2013).
\newblock Description of light-vehicle pre-crash scenarios for safety
  applications based on vehicle-to-vehicle communications.
\newblock Technical report, United States. National Highway Traffic Safety
  Administration.

\bibitem[Osorio, 2010]{osorio2011mitigating}
Osorio, C. (2010).
\newblock {\em Mitigating network congestion: analytical models, optimization
  methods and their applications}.
\newblock PhD thesis, \'{E}cole polytechnique f\'{e}d\'{e}rale de Lausanne
  (EPFL).

\bibitem[Osorio and Bierlaire, 2013]{osorio2013simulation}
Osorio, C. and Bierlaire, M. (2013).
\newblock A simulation-based optimization framework for urban transportation
  problems.
\newblock {\em Operations Research}, 61(6):1333--1345.

\bibitem[Osorio et~al., 2015]{osorio2015metamodel}
Osorio, C., Fl{\"o}tter{\"o}d, G., and Zhang, C. (2015).
\newblock A metamodel simulation-based optimization approach for the efficient
  calibration of stochastic traffic simulators.
\newblock {\em Transportation Research Procedia}, 6:213--223.

\bibitem[Pan et~al., 2018]{pan2017agile}
Pan, Y., Cheng, C.-A., Saigol, K., Lee, K., Yan, X., Theodorou, E., and Boots,
  B. (2018).
\newblock Agile off-road autonomous driving using end-to-end deep imitation
  learning.
\newblock In {\em Robotics: Science and Systems}.

\bibitem[Perttunen et~al., 2015]{perttunen2015urban}
Perttunen, M., Kostakos, V., Riekki, J., and Ojala, T. (2015).
\newblock Urban traffic analysis through multi-modal sensing.
\newblock {\em Personal and Ubiquitous Computing}, 19(3-4):709--721.

\bibitem[Phan and Ferrie, 2011]{phan2011interpolating}
Phan, A. and Ferrie, F.~P. (2011).
\newblock Interpolating sparse {GPS} measurements via relaxation labeling and
  belief propagation for the redeployment of ambulances.
\newblock {\em IEEE Transactions on Intelligent Transportation Systems},
  12(4):1587--1598.

\bibitem[Piorkowski et~al., 2009]{cabspotting}
Piorkowski, M., Sarafijanovic-Djukic, N., and Grossglauser, M. (2009).
\newblock {CRAWDAD} dataset epfl/mobility (v. 2009-02-24).
\newblock Downloaded from \url{http://crawdad.org/epfl/mobility/20090224}.

\bibitem[Pomerleau, 1989]{pomerleau1989alvinn}
Pomerleau, D. (1989).
\newblock {ALVINN}: An autonomous land vehicle in a neural network.
\newblock In {\em Advances in neural information processing systems}, pages
  305--313.

\bibitem[Quddus and Washington, 2015]{quddus2015shortest}
Quddus, M. and Washington, S. (2015).
\newblock Shortest path and vehicle trajectory aided map-matching for low
  frequency {GPS} data.
\newblock {\em Transportation Research Part C: Emerging Technologies},
  55:328--339.

\bibitem[Rahmani et~al., 2015]{rahmani2015non}
Rahmani, M., Jenelius, E., and Koutsopoulos, H. (2015).
\newblock Non-parametric estimation of route travel time distributions from
  low-frequency floating car data.
\newblock {\em Transportation Research Part C: Emerging Technologies},
  58:343--362.

\bibitem[Rahmani and Koutsopoulos, 2013]{rahmani2013path}
Rahmani, M. and Koutsopoulos, H. (2013).
\newblock Path inference from sparse floating car data for urban networks.
\newblock {\em Transportation Research Part C: Emerging Technologies},
  30:41--54.

\bibitem[Ratliff, 2009]{ratliff2009learning}
Ratliff, N. (2009).
\newblock {\em Learning to search: structured prediction techniques for
  imitation learning}.
\newblock PhD thesis, Carnegie Mellon University.

\bibitem[Ratliff et~al., 2007]{ratliff2007imitation}
Ratliff, N., Bagnell, J.~A., and Srinivasa, S.~S. (2007).
\newblock Imitation learning for locomotion and manipulation.
\newblock In {\em 7th IEEE-RAS International Conference on Humanoid Robots},
  pages 392--397.

\bibitem[Ross et~al., 2011]{ross2011reduction}
Ross, S., Gordon, G., and Bagnell, D. (2011).
\newblock A reduction of imitation learning and structured prediction to
  no-regret online learning.
\newblock In {\em Proceedings of the 14th International Conference on
  Artificial Intelligence and Statistics}, pages 627--635.

\bibitem[Ross et~al., 2013]{ross2013learning}
Ross, S., Melik-Barkhudarov, N., Shankar, K.~S., Wendel, A., Dey, D., Bagnell,
  J.~A., and Hebert, M. (2013).
\newblock Learning monocular reactive {UAV} control in cluttered natural
  environments.
\newblock In {\em Robotics and Automation, IEEE International Conference on},
  pages 1765--1772.

\bibitem[Schrank et~al., 2015]{schrank2015}
Schrank, D., Eisele, B., Lomax, T., and Bak, J. (2015).
\newblock 2015 urban mobility scorecard.
\newblock {\em Texas A\&M Transportation Institute and INRIX}.

\bibitem[Schwarting et~al., 2018]{schwarting2018planning}
Schwarting, W., Alonso-Mora, J., and Rus, D. (2018).
\newblock Planning and decision-making for autonomous vehicles.
\newblock {\em Annual Review of Control, Robotics, and Autonomous Systems}.

\bibitem[Sewall et~al., 2011a]{sewall2011virtualized}
Sewall, J., Van Den~Berg, J., Lin, M., and Manocha, D. (2011a).
\newblock Virtualized traffic: Reconstructing traffic flows from discrete
  spatiotemporal data.
\newblock {\em IEEE Transactions on Visualization and Computer Graphics},
  17(1):26--37.

\bibitem[Sewall et~al., 2011b]{Sewall:2011:IHS:2070781.2024169}
Sewall, J., Wilkie, D., and Lin, M.~C. (2011b).
\newblock Interactive hybrid simulation of large-scale traffic.
\newblock {\em ACM Trans. Graph.}, 30(6):135:1--135:12.

\bibitem[Sewall et~al., 2010]{sewall10}
Sewall, J., Wilkie, D., Merrell, P., and Lin, M.~C. (2010).
\newblock Continuum traffic simulation.
\newblock {\em Computer Graphics Forum}, 29(2):439--448.

\bibitem[Sheffi, 1985]{sheffi1985urban}
Sheffi, Y. (1985).
\newblock Urban transportation network.
\newblock {\em Pretince Hall}.

\bibitem[Shen and Jin, 2012]{shen2012detailed}
Shen, J. and Jin, X. (2012).
\newblock Detailed traffic animation for urban road networks.
\newblock {\em Graphical Models}, 74(5):265--282.

\bibitem[Silver, 2010]{silver2010learning}
Silver, D. (2010).
\newblock {\em Learning preference models for autonomous mobile robots in
  complex domains}.
\newblock PhD thesis, Carnegie Mellon University.

\bibitem[Sugiyama and Kawanabe, 2012]{sugiyama2012machine}
Sugiyama, M. and Kawanabe, M. (2012).
\newblock {\em Machine learning in non-stationary environments: Introduction to
  covariate shift adaptation}.
\newblock MIT press.

\bibitem[Sun and Work, 2014]{sun2014distributed}
Sun, Y. and Work, D. (2014).
\newblock A distributed local kalman consensus filter for traffic estimation.
\newblock In {\em Decision and Control (CDC), 53rd IEEE Annual Conference on},
  pages 6484--6491.

\bibitem[Sutton et~al., 1999]{sutton1999between}
Sutton, R.~S., Precup, D., and Singh, S. (1999).
\newblock Between mdps and semi-mdps: A framework for temporal abstraction in
  reinforcement learning.
\newblock {\em Artificial intelligence}, 112(1-2):181--211.

\bibitem[Syed and Schapire, 2010]{syed2010reduction}
Syed, U. and Schapire, R.~E. (2010).
\newblock A reduction from apprenticeship learning to classification.
\newblock In {\em Advances in Neural Information Processing Systems}, pages
  2253--2261.

\bibitem[Szepesv{\'a}ri and Munos, 2005]{szepesvari2005finite}
Szepesv{\'a}ri, C. and Munos, R. (2005).
\newblock Finite time bounds for sampling based fitted value iteration.
\newblock In {\em Proceedings of the 22nd international conference on Machine
  learning}, pages 880--887.

\bibitem[TAC, 2017]{tacvr}
TAC (2017).
\newblock {TAC}'s virtual reality driving school.
\newblock
  \href{http://www.smh.com.au/victoria/virtual-reality-traffic-school-to-give-youth-a-glimpse-of-roads-of-the-future-20170214-gud7ic.html}{www.smh.com.au/victoria/}.

\bibitem[Tang et~al., 2016a]{tang2015estimating}
Tang, J., Song, Y., Miller, H., and Zhou, X. (2016a).
\newblock Estimating the most likely space--time paths, dwell times and path
  uncertainties from vehicle trajectory data: A time geographic method.
\newblock {\em Transportation Research Part C: Emerging Technologies},
  66:176--194.

\bibitem[Tang et~al., 2016b]{tang2016estimating}
Tang, J., Song, Y., and Miller, Harveyand~Zhou, X. (2016b).
\newblock Estimating the most likely space--time paths, dwell times and path
  uncertainties from vehicle trajectory data: A time geographic method.
\newblock {\em Transportation Research Part C: Emerging Technologies},
  66:176--194.

\bibitem[Thomas and Donikian, 2000]{thomas2000}
Thomas, G. and Donikian, S. (2000).
\newblock Modelling virtual cities dedicated to behavioural animation.
\newblock {\em Computer Graphics Forum}, 19(3):71--80.

\bibitem[Tibshirani et~al., 2005]{tibshirani2005sparsity}
Tibshirani, R., Saunders, M., Rosset, S., Zhu, J., and Knight, K. (2005).
\newblock Sparsity and smoothness via the fused lasso.
\newblock {\em Journal of the Royal Statistical Society: Series B (Statistical
  Methodology)}, 67(1):91--108.

\bibitem[Uber, 2017]{ubermovement}
Uber (2017).
\newblock Uber movement.
\newblock \href{https://movement.uber.com/}{movement.uber.com/}.

\bibitem[van~den Berg et~al., 2009]{Berg2009}
van~den Berg, J., Sewall, J., Lin, M., and Manocha, D. (2009).
\newblock Virtualized traffic: Reconstructing traffic flows from discrete
  spatio-temporal data.
\newblock In {\em Virtual Reality Conference, IEEE}, pages 183--190.

\bibitem[Wang et~al., 2005]{wang2005steering}
Wang, H., Kearney, J.~K., Cremer, J., and Willemsen, P. (2005).
\newblock Steering behaviors for autonomous vehicles in virtual environments.
\newblock In {\em Virtual Reality Conference, IEEE}, pages 155--162.

\bibitem[Wang et~al., 2016]{7501704}
Wang, Y., Cao, J., Li, W., and Gu, T. (2016).
\newblock Mining traffic congestion correlation between road segments on {GPS}
  trajectories.
\newblock In {\em 2016 IEEE International Conference on Smart Computing
  (SMARTCOMP)}, pages 1--8.

\bibitem[Wang et~al., 2014a]{wang2014travel}
Wang, Y., Zheng, Y., and Xue, Y. (2014a).
\newblock Travel time estimation of a path using sparse trajectories.
\newblock In {\em Proceedings of the 20th ACM SIGKDD international conference
  on Knowledge discovery and data mining}, pages 25--34.

\bibitem[Wang et~al., 2014b]{wang2014visual}
Wang, Z., Ye, T., Lu, M., Yuan, X., Qu, H., Yuan, J., and Wu, Q. (2014b).
\newblock Visual exploration of sparse traffic trajectory data.
\newblock {\em IEEE Transactions on Visualization and Computer Graphics},
  20(12):1813--1822.

\bibitem[Wardrop, 1952]{wardrop1952road}
Wardrop, J. (1952).
\newblock Some theoretical aspects of road traffic research.
\newblock {\em Proceedings of the Institution of Civil Engineers},
  1(3):325--362.

\bibitem[Westgate et~al., 2013]{westgate2013travel}
Westgate, B., Woodard, D., Matteson, D., and Henderson, S. (2013).
\newblock Travel time estimation for ambulances using bayesian data
  augmentation.
\newblock {\em The Annals of Applied Statistics}, 7(2):1139--1161.

\bibitem[Wilkie et~al., 2014]{wilkie2014participatory}
Wilkie, D., Baykal, C., and Lin, M.~C. (2014).
\newblock Participatory route planning.
\newblock In {\em Proceedings of the 22nd ACM SIGSPATIAL International
  Conference on Advances in Geographic Information Systems}, pages 213--222.

\bibitem[Wilkie et~al., 2015]{Wilkie2015Virtual}
Wilkie, D., Sewall, J., Li, W., and Lin, M.~C. (2015).
\newblock Virtualized traffic at metropolitan scales.
\newblock {\em Frontiers in Robotics and AI}, 2:11.

\bibitem[Wilkie et~al., 2012]{wilkie2012transforming}
Wilkie, D., Sewall, J., and Lin, M.~C. (2012).
\newblock Transforming {GIS} data into functional road models for large-scale
  traffic simulation.
\newblock {\em IEEE transactions on visualization and computer graphics},
  18(6):890--901.

\bibitem[Wilkie et~al., 2013]{Wilkie:2013:FRD:2461912.2462021}
Wilkie, D., Sewall, J., and Lin, M.~C. (2013).
\newblock Flow reconstruction for data-driven traffic animation.
\newblock {\em ACM Trans. Graph.}, 32(4):89:1--89:10.

\bibitem[Wilkie et~al., 2011]{wilkie2011self}
Wilkie, D., van~den Berg, J.~P., Lin, M.~C., and Manocha, D. (2011).
\newblock Self-aware traffic route planning.
\newblock In {\em AAAI}, volume~11, pages 1521--1527.

\bibitem[Willemsen et~al., 2006]{willemsen2006}
Willemsen, P., Kearney, J.~K., and Wang, H. (2006).
\newblock Ribbon networks for modeling navigable paths of autonomous agents in
  virtual environments.
\newblock {\em IEEE Transactions on Visualization and Computer Graphics},
  12(3):331--342.

\bibitem[Wolinski et~al., 2016]{wolinski2016warpdriver}
Wolinski, D., Lin, M., and Pettr{\'e}, J. (2016).
\newblock Warpdriver: context-aware probabilistic motion prediction for crowd
  simulation.
\newblock {\em ACM Trans. Graph.}, 35(6):164:1--164:11.

\bibitem[Work et~al., 2010]{work2010traffic}
Work, D., Blandin, S., Tossavainen, O.-P., Piccoli, B., and Bayen, A. (2010).
\newblock A traffic model for velocity data assimilation.
\newblock {\em Applied Mathematics Research Express}, 2010(1):1--35.

\bibitem[Wu et~al., 2018]{Wu2018Stablizing}
Wu, C., Bayen, A.~M., and Mehta, A. (2018).
\newblock Stabilizing traffic with autonomous vehicles.
\newblock In {\em 2018 IEEE International Conference on Robotics and Automation
  (ICRA)}, pages 1--7.

\bibitem[Xu et~al., 2017]{Xu2017end}
Xu, H., Gao, Y., Yu, F., and Darrell, T. (2017).
\newblock End-to-end learning of driving models from large-scale video
  datasets.
\newblock In {\em IEEE Conference on Computer Vision and Pattern Recognition
  (CVPR)}, pages 3530--3538.

\bibitem[Yang et~al., 1992]{yang1992estimation}
Yang, H., Sasaki, T., Iida, Y., and Asakura, Y. (1992).
\newblock Estimation of origin-destination matrices from link traffic counts on
  congested networks.
\newblock {\em Transportation Research Part B: Methodological}, 26(6):417--434.

\bibitem[Yang et~al., 2017]{yang2016}
Yang, X., Lu, Y., and Hao, W. (2017).
\newblock Origin-destination estimation using probe vehicle trajectory and link
  counts.
\newblock {\em Journal of Advanced Transportation}.

\bibitem[Yuan et~al., 2010]{yuan2010interactive}
Yuan, J., Zheng, Y., Zhang, C., Xie, X., and Sun, G. (2010).
\newblock An interactive-voting based map matching algorithm.
\newblock In {\em Proceedings of the 11th International Conference on Mobile
  Data Management}, pages 43--52.

\bibitem[Zhang and Cho, 2017]{zhang2017query}
Zhang, J. and Cho, K. (2017).
\newblock Query-efficient imitation learning for end-to-end simulated driving.
\newblock In {\em AAAI}, pages 2891--2897.

\bibitem[Zhang et~al., 2013]{zhang2013aggregating}
Zhang, J.-D., Xu, J., and Liao, S.~S. (2013).
\newblock Aggregating and sampling methods for processing {GPS} data streams
  for traffic state estimation.
\newblock {\em IEEE Transactions on Intelligent Transportation Systems},
  14(4):1629--1641.

\bibitem[Zheng et~al., 2011]{zheng2011urban}
Zheng, Y., Liu, Y., Yuan, J., and Xie, X. (2011).
\newblock Urban computing with taxicabs.
\newblock In {\em Proceedings of the 13th International Conference on
  Ubiquitous Computing}, pages 89--98.

\bibitem[Zhu et~al., 2013]{zhu2013compressive}
Zhu, Y., Li, Z., Zhu, H., Li, M., and Zhang, Q. (2013).
\newblock A compressive sensing approach to urban traffic estimation with probe
  vehicles.
\newblock {\em Mobile Computing, IEEE Transactions on}, 12(11):2289--2302.

\end{thebibliography}
